\documentclass[11pt,a4paper]{book}
\pdfoutput=1
\def\p{{_{\bf P}}}
\usepackage{phdthesis}
\usepackage{amssymb}
\usepackage{amsmath}
\usepackage{sidecap}
\usepackage{bibunits}
\usepackage[]{titlepage_melbourne_uni}
\usepackage[a4paper=true,pagebackref=true]{hyperref}
\usepackage[T1]{fontenc}
\usepackage[sc]{mathpazo}
\hypersetup{
    pdftitle =Phd Thesis,
    pdfauthor =Tridib Sadhu,
    pdfsubject=Thesis,
    pdfkeywords =tridibthesis
}
\hypersetup{
    colorlinks = true,
    linkcolor = red,
    anchorcolor = red,
    citecolor = blue,
    filecolor = red,
    pagecolor = red,
    urlcolor = red
}
\begin{document}

\author{Tridib Sadhu}
\title{Emergence and Complexity in Theoretical Models of
Self-Organized Criticality }
\date{July 2011}
\maketitle
%


\chapter*{Statutory Declarations}
\addcontentsline{toc}{chapter}{Statutory Declarations}
\begin{tabular}{p{2.2in}p{3in}}
Name of the Candidate &: Tridib Sadhu\\
\\ & \\
Title of the Thesis &: Emergence and Complexity in\\
& $~$ Theoretical Models of Self-\\
& $~$ Organized Criticality.\\
\\ & \\
Degree &: Doctor of Philosophy (Ph.D.)\\
\\ & \\
Subject &: Physics \\
\\ & \\
Name of the advisor &: Prof. Deepak Dhar\\
\\ & \\
Registration number &: PHYS--154\\
\\ & \\
Place of Research &: Tata Institute of Fundamental\\
 & $~$ Research, Mumbai
400005
\end{tabular}

\chapter*{Declaration Of Authorship}
\addcontentsline{toc}{chapter}{Declaration of Authorship}
This thesis is a presentation of my original research work. Wherever
contributions of others are involved, every effort is made to indicate
this clearly, with due reference to the literature, and
acknowledgement of collaborative research and discussions.

The work was done under the guidance of Professor Deepak Dhar, at the
Tata Institute of Fundamental Research, Mumbai.\\
\vspace{0.5in}\\
Signed: \dotfill \hspace{2.0cm}  Date:\dotfill\\
\vspace{0.02in}
Name: {\Large \textbf{Tridib Sadhu}}
\newline
\vspace{1.0in}\\
In my capacity as supervisor of the candidate's thesis, I certify that
the above statements are true to the best of my knowledge.\\
\vspace{0.5in}\\
Signed: \dotfill \hspace{2.0cm}  Date:\dotfill\\
\vspace{0.02in}
Name: {\Large \textbf{Prof. Deepak Dhar}}
\noindent

\chapter*{Acknowledgements}
\addcontentsline{toc}{chapter}{Acknowledgements}
I would like to express my gratitude to all those who helped me in the
completion of this work. Most importantly, I gratefully acknowledge
the guidance and support of my thesis advisor Prof. Deepak Dhar. His
emphasis on perfection and strong sense of work ethics have had a deep
influence on me.

I also acknowledge the financial support I received from TIFR, Mumbai.

My special thanks to Shaista for proofreading, and for helping
me in many technical aspects of preparing the thesis. Finally, I
would like to thank my family for their unconditional love and support.

\noindent

\tableofcontents

\listoffigures

\listoftables

\chapter{Synopsis}\label{synopsis}
\section{Introduction}\label{sec:intro}
The concept of self-organized criticality (SOC) was introduced by Bak
Tang and Wiesenfeld in 1987 \cite{btw}, to explain the abundant fractal
structures in nature, e.g. mountain ranges, river networks,
power law tails in the distribution of earthquake intensities etc.
The SOC refers to the non-equilibrium steady state of
slowly driven systems, which show irregular burst like relaxations
with a wide distribution of event sizes. The power law correlations of
different physical quantities, extending over a wide range of length and time scales is a signature of
criticality. Usually, reaching a critical state requires fine tuning
of some control parameters e.g. temperature and magnetic field for the Ising model with a
given interaction strength.
However in SOC the systems reach a critical state under their
own dynamics, irrespective of the initial states and without any
obvious fine tuning of parameters.

In the last two decades a large amount of study is focused on
understanding the mechanism of SOC. The questions regarding the
universality classes of the critical
states has still not been completely settled.
Many theoretical models have been studied to address these issues.
Most of these are cellular automata models with discrete or
continuous variables, evolving under deterministic or
stochastic evolution rules (see \cite{dharphysica06} for a review).
Among them the Abelian Sandpile Model (ASM) is studied the most, mainly
because of its analytical tractability using the Abelian property\cite{dharprl}.

The standard ASM first proposed in \cite{btw} is defined on a
lattice with height variables $z_{i}$ at each site $i$, which is equal to the number of
sand grains at that site. There is a threshold value $z_{c}$ for
each site, and any site with height $z_{i}\ge z_{c}$ is said to be unstable.
The system in a stable configuration is driven by adding a single grain at
a randomly chosen site. If this addition makes the system unstable, it
relaxes by the following toppling rule: All the unstable sites at one
time step transfers one grain each to all its nearest neighbors. Grains can move out of the
lattice by toppling at the boundary. When the system reaches a stable
configuration it is again driven by adding a grain and the process is
repeated. The model reaches a steady state, in which the
probability distribution of the size of events has power law tail.

The ASM defined on an \textbf{infinite} lattice when driven by adding grains
only at a \textbf{single} site and relaxed, produces beautiful and complex patterns in
height variables (see figure $1.1$). They are one of the
examples where complexity arises from simple rules. In section \ref{sec:first} we
analyze some of these patterns and develop a detailed and exact mathematical
characterization of them.

There are different variants of the ASM that have been proposed. We study
two well known models among them.

First is the Zhang model. This is similar to the ASM, but with
continuous non-negative variables, usually referred as energy.
At any time step all the unstable sites relax by equally distributing all
its energy amongst its nearest neighbors, with their energy reducing to
zero. Energy can also move out of the system by toppling at the
boundary. The driving is done by adding energy to a randomly chosen
site and the amount of the energy is chosen at random from a
distribution.

The second is a stochastic variant of the ASM.
The first stochastic sandpile model was proposed
by Manna and it is known as Manna model\cite{manna}. The model is non-Abelian, but one
can construct stochastic relaxation rules with Abelian character. We
consider one such stochastic Abelian sandpile model introduced in
\cite{dm}. The model is similar to the ASM with non-negative integer height variables
$z_{i}$ and a threshold value $z_{c}$ defined at each site. The
driving is also done by adding one sand grain at a
randomly chosen site in a stable configuration. The difference is in
the relaxation rules: On toppling $z_{c}$ number of grains are
transfered, each grain moving independent of others to the nearest neighbors with equal
probability and the height at the toppling site reduces by
$z_{c}$. For the one dimensional model defined on a linear chain,
with $z_{c}=2$, there are three possible events in a toppling at site
$i$: Both the neighbors $\left( i-1 \right)$ and $\left( i+1 \right)$
gets one grain each. Probability of this event is $1/2$. Other two
possibilities are that both the grains move either to the left or to
the right neighbor, each event with probability $1/4$.

\section[Spatial patterns]{The spatial patterns in theoretical sandpile models}\label{sec:first}
While real sand, poured at one point on a flat substrate, produces a
rather simple conical pyramid shape, nontrivial patterns
are generated this way in the ASM on an infinite lattice.
One such pattern on a square lattice with threshold height
$z_{c}=4$, produced by adding grains at the origin in an initial
uniform distribution of heights $z=2$, is shown in the figure
$1.1$.
\begin{figure}
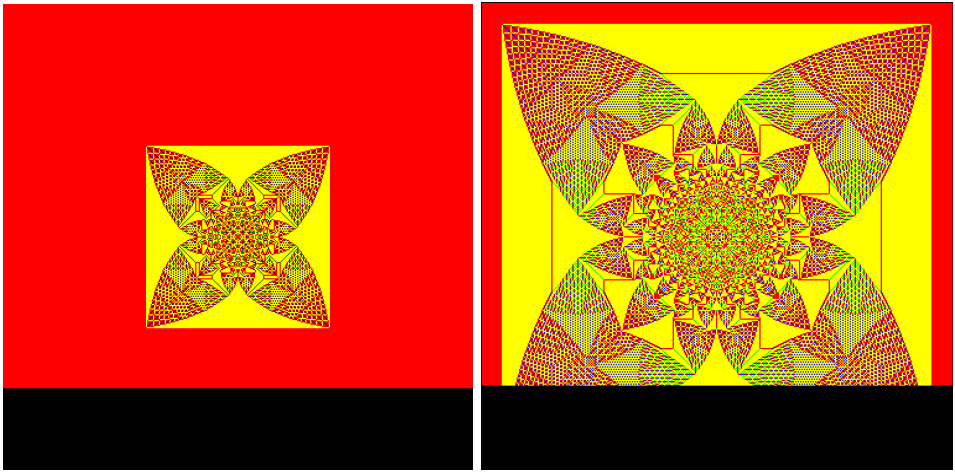

\includegraphics[scale=0.35,angle=0,clip]{Images/Synopsis/Output50k}
\includegraphics[scale=0.35,angle=0,clip]{Images/Synopsis/btw}
\label{fig:btw}
\caption{The stable configurations for the Abelian sandpile model, obtained
by adding $10^4$ and $5 \times 10^{4}$ grains, respectively at one
site on a square lattice. Initial configuration is with all heights 2. Color code:
blue =0, green = 1, red = 2, yellow = 3. Both patterns are on the same
scale. (Details can be seen in the
electronic version using zoom in).}
\end{figure}
\begin{SCfigure}
\includegraphics[scale=0.68,clip]{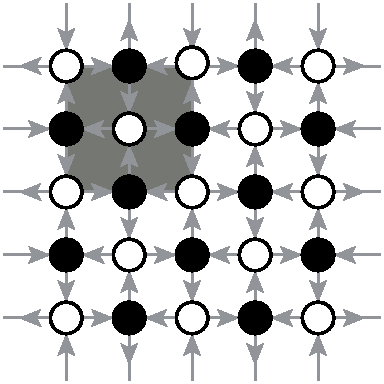}
\caption{\label{fig:flatt} F-lattice with checker board distribution
of grains. Unfilled circles denote height $z=1$ and filled ones $z=0$.
The gray area denotes a unit cell of the periodic distribution.}
\end{SCfigure}
\begin{SCfigure}
\includegraphics[scale=0.47,angle=0]{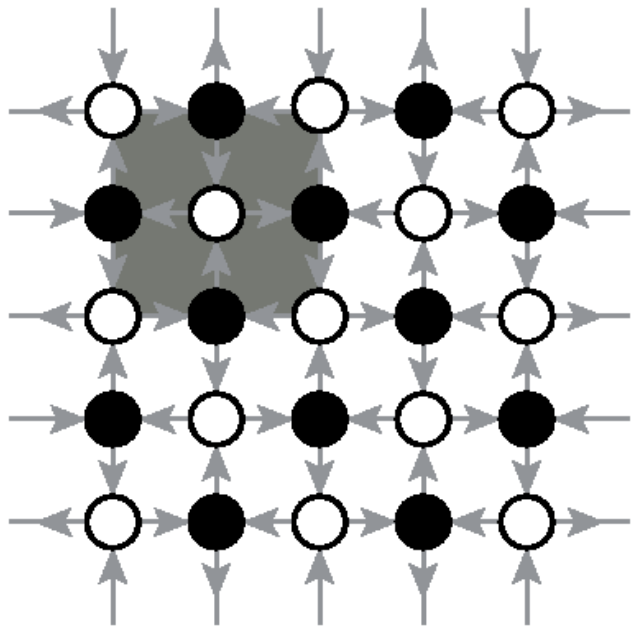}
\label{fig:flatone}
\caption{The stable configuration for the Abelian sandpile model, obtained
by adding $5\times10^4$ grains at one site, on the F-lattice of
figure \ref{fig:flatt} with initial checkerboard configuration. Color code: red =
0, yellow = 1. The apparent orange regions in the picture represent the
patches with checkerboard configuration. (Details can be seen in the
electronic version using zoom in.)}
\end{SCfigure}

The reason for interest in these patterns is two fold.

Firstly, these are analytically tractable examples of complex patterns that are obtained from simple
deterministic evolution rules. Here complexity means that we have
structures with variations, and a complete description of which is
long. Thus, a living organism is complex because it has many different
working parts, each formed by variations in the working out of the
same, but relatively much simpler genetic coding.

Secondly, these patterns have the very interesting property of
\textit{proportionate growth}. This is a well-known feature of
biological growth in animals, where different parts of the growing
animal grow at roughly the same rate, keeping their shape almost the
same. Our
interest in studying the sandpile patterns comes from these being the
simplest model of proportionate growth with non-trivial patterns. Compare the two patterns in figure
$1.1$ produced on
the same background but with different values of $N$. The pattern
grows in size and finer features become discernible at the center, but the overall shape of
the pattern remains same. Most of the other growth models studied in physics
literature, such as the Eden model, the diffusion limited aggregation,
or the surface
deposition, do not show this property \cite{eden,dla,barabasi}. In these
models, the growth is
confined to some active outer region. The inner structures, once
formed are frozen in and do not evolve further in time.

The standard square lattice produces complicated patterns and it has
not been possible to characterize them so far. We consider a pattern which is simpler but
still complex. The pattern is produced on the F-lattice which is a
variant of the square lattice with directed edges. The lattice with
checker board distribution of heights is shown in figure
\ref{fig:flatt}. Each site has
equal number of incoming and outgoing arrows. The threshold height is
$2$ and any unstable site relaxes by giving away one particle each in
the direction of the outgoing arrows. The pattern produced by
centrally seeding grains on the checkerboard background is shown in
the figure \ref{fig:flatone}.

\subsection{The characterization of the pattern}
We take some qualitative features of the observed pattern as input
and show how one can get a complete and quantitative characterization
of the pattern in the asymptotic limit of $N\rightarrow \infty$.

In models with proportionate growth, it is natural to describe the
asymptotic pattern in terms of the rescaled coordinate
$\mathbf{r}=\mathbf{R}/\Lambda\left( N \right)$ where $\Lambda\left( N
\right)$ is the diameter of the pattern, suitably
defined, and $\mathbf{R}$ is the position vector of a site on the
lattice. The function $\Lambda\left( N \right)$ increases in steps
with $N$ and goes to infinity as $N\rightarrow\infty$. In the
asymptotic limit the pattern can be characterized by a function
$\rho\left( \mathbf{r} \right)$ which gives the local density of sand
grains in a small rectangle of size $\delta\xi\delta\eta$ about the
point $\mathbf{r}$, with $1/\Lambda \ll \delta\xi, \delta\eta\ll 1$
where $\xi$ and $\eta$ are the x and y components of the rescaled
position vector $\mathbf{r}$. We define $\Delta\rho\left(
\mathbf{r}\right)$ as the change in density $\rho\left(
\mathbf{r}\right)$ from its background value. 

The pattern is composed of large regions where the heights are
periodic and we call these regions as patches. Inside each patch
$\Delta\rho\left( \mathbf{r} \right)$ is constant and takes only two
values, either $0$ or $1/2$.

Let $T_{N}\left( \mathbf{R} \right)$ be the number of
topplings at site $\mathbf{R}$ when $N$ number of grains have been
added and then relaxed. Define a rescaled toppling function
\begin{equation}
\phi\left( \mathbf{r}
\right)=\lim_{N\rightarrow\infty}\frac{T_{N}\left(\lfloor\Lambda\mathbf{r}\rfloor
\right)}{\Lambda\left( N \right)^{2}},
\end{equation}
where the floor function $\lfloor x \rfloor$ is the largest integer less than or equal to $x$.
From the conservation of sand grains in the toppling process, it
follows that $\phi$ satisfies the Poisson equation
\begin{equation}
\nabla^{2}\phi\left( \mathbf{r} \right)=\Delta\rho\left(
\mathbf{r}\right)-\frac{N}{\Lambda^2}\delta\left( \mathbf{r} \right).
\label{eq:poisson}
\end{equation}

The complete specification of $\phi\left( \mathbf{r} \right)$
determines the density function $\Delta\rho\left( \mathbf{r} \right)$
and hence the asymptotic pattern. The condition
that determines $\phi\left( \mathbf{r} \right)$ is the requirement
that inside each patch of constant density, it is a quadratic function
of $\xi$ and $\eta$. Considering that there are only two types of
patches and that $\phi\left(
\mathbf{r}\right)$ satisfies (\ref{eq:poisson}) we write
\begin{equation}
\phi\left( \mathbf{r}
\right)=\frac{1}{8}\left( 1+m\right)\xi^{2}+\frac{1}{4}n\xi\eta+\frac{1}{8}\left( 1-m \right)\eta^{2}+d\xi+e\eta+f,
\end{equation}
for the patches with $\Delta\rho=1/2$ and
\begin{equation}
\phi\left( \mathbf{r}
\right)=\frac{1}{8}m\xi^{2}+\frac{1}{4}n\xi\eta-\frac{1}{8}m \eta^{2}+d\xi+e\eta+f,
\end{equation}
for the patches with $\Delta\rho=0$. Each patch is
characterized by the values of the parameters $m$,$n$,$d$,$e$ and $f$. The continuity of
$\phi\left( r \right)$ and its derivatives along the boundary between
two adjacent patches imposes linear relations among the corresponding
parameters. Using these relations we show that $m$ and $n$ take
only integer values. Each patch with its
parameters $d$, $e$ and $f$ can be labeled by
the pair $\left( m,n \right)$. The pair can be taken as the Cartesian coordinates of the
adjacency graph of the patches, which for this pattern is a square lattice on a two
sheeted Riemann surface. We show that the function $D\left(
m,n \right)=d\left( m,n \right)+ie\left( m,n \right)$, with
$i=\sqrt{-1}$,
satisfies the discrete Laplace's equation on the adjacency
graph. Using the asymptotic dependence of $\phi\left(
\mathbf{r} \right)$ close to the site of addition, we show that 
for large $|m|+|n|$,
\begin{equation}
D\left( m,n \right)\simeq\pm\frac{1}{2\pi}\sqrt{m+in}.
\end{equation}
Solution of the discrete Laplace's equation on this adjacency graph
with the above boundary condition is
difficult to determine. We numerically calculate the solution on
finite adjacency graphs and extrapolate our results to the asymptotic
limit. We also show that the pattern has an exact eight-fold rotational
symmetry.
\subsection{The effect of multiple sources or sinks}
\begin{figure}[t]
\begin{center}
\includegraphics[scale=0.22,angle=0,clip]{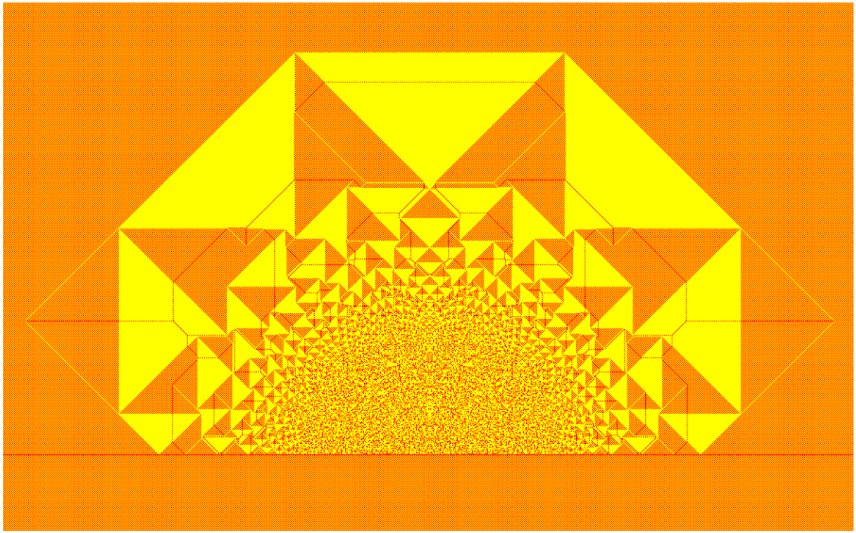}
\label{fig:fsink}
\caption{The pattern produced by adding $N=14\times 10^6$ grains at
the site
$\left( 0,1 \right)$ in the checkerboard background on the F-lattice with
the sink sites placed along the X-axis. Color code: yellow=1, and
red=0.
(Details can be seen in the electronic version using zoom in.)}
\end{center}
\end{figure}
We also studied the patterns where the grains are added at more than
one site or those formed in presence of sink sites. One such
pattern on the F-lattice with the checker board background in presence
of a line of sink sites is shown in the figure ${\color{red}1.4}$.
There are still only two types of patches and like
the single source case, the spatial
distances can be expressed in terms of the solution of the discrete
Laplace's equation on the adjacency graph. However, the structure of
the adjacency graph changes. For the pattern in figure
${ \color{red}1.4 }$, the adjacency graph is still a square lattice but on a
Riemann surface of three sheets.
We have explicitly worked out the spatial distances by numerically
solving the Laplace's equation on this graph. We have also studied the case with two sites of
addition and quantitatively characterized the pattern.

The most interesting effect of the sink sites is that it changes how
different spatial lengths in the pattern scale with the number of added
grains $N$. For example, in the absence of sink sites, the diameter
$\Lambda\left( N \right)$ of the pattern in figure \ref{fig:flatone}
grows as $\sqrt{N}$, for large $N$, whereas in presence of a line of
sink sites next to the site of addition, it changes to $N^{1/3}$. More
precisely, we show that in this case
\begin{equation}
C_{1}\Lambda^{3}+C_{2}\Lambda^{2}\simeq N,
\label{eq:lrel}
\end{equation}
where $C_{1}$ and $C_{2}$ are numerical constants.
For $C_{1}=0.1853$ and $C_{2}=0.528$ this relation describes the $N$ dependence of $\Lambda\left( N \right)$ for
$N$ in the range of $100$ to $10^{5}$, with unexpectedly high accuracy where both sides of the above equation
differs by at most $1$.

We have also studied the case in which the source site is at the corner of a
wedge angle $\omega$, where the wedge boundaries are absorbing. We
show that the relation similar to (\ref{eq:lrel}) is
\begin{equation}
C_{1}\Lambda^{2+\alpha}+C_{2}\Lambda^{2}\simeq N,
\end{equation}
where $\alpha=\omega/\left( \pi +2\omega \right)$. This analysis is extended
to other lattices with different initial height distribution, and also to
higher dimensions.
\subsection{The compact and non-compact growth}
\begin{SCfigure}[][t]
\includegraphics[scale=0.38,clip]{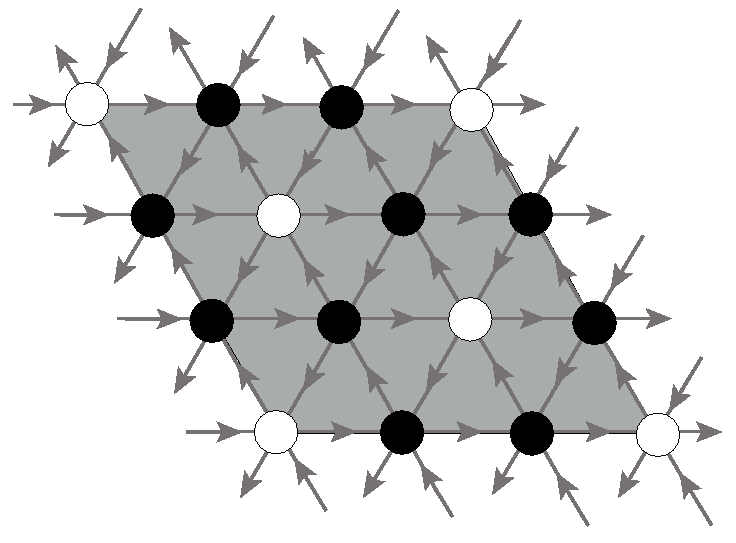}
\caption{\label{fig:trilatt} Directed triangular lattice. Unfilled
circles represent $z=1$ and filled ones $z=2$. The gray area denotes a unit
cell of the periodic distribution.}
\end{SCfigure}
\begin{SCfigure}
\includegraphics[scale=0.5,angle=0,clip]{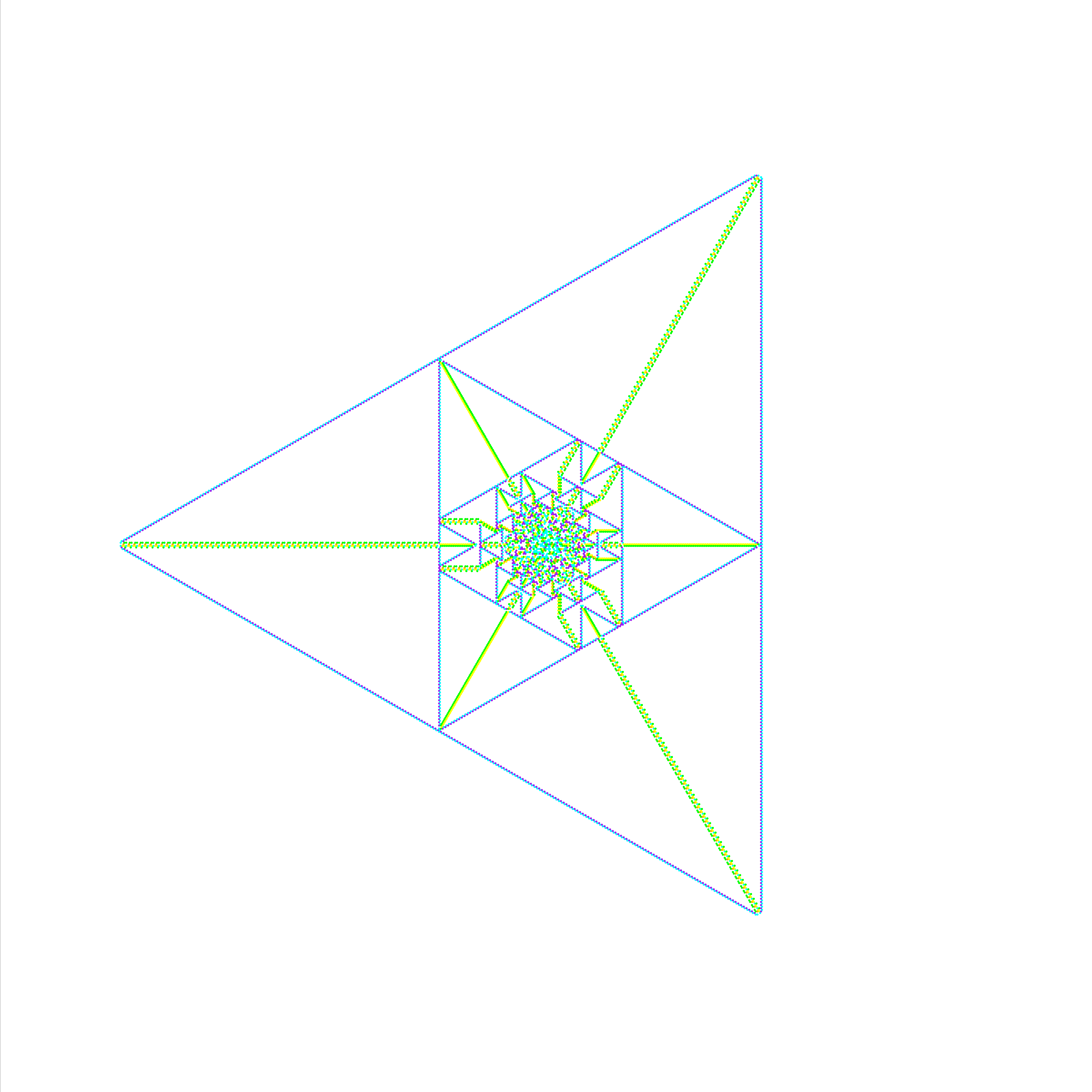}
\label{fig:trione}
\caption{The pattern produced by adding $N=520$ grains at a single site on
the directed triangular lattice. Only
the patch boundaries are shown by colored lines.
(Details can be seen in the electronic version using zoom in.)}
\end{SCfigure}
The growth rate of the patterns closely depends on the background
height distribution.
When the heights at all sites on the background are low
enough, one gets patterns with $\Lambda\left( N \right)$ growing as
$N^{1/d}$ in d-dimensions. We refer to this growth as the compact growth.
However if sites with maximum stable height
in the background form an infinite cluster we get avalanches that do
not stop, and the pattern is not-well defined.
We describe our unexpected finding of an
interesting class of backgrounds, that show an intermediate behavior.
For any $N$, the avalanches are finite, but the diameter of the
pattern increases as $N^{\alpha}$, for large $N$, with $1/2 <\alpha
\le 1$. We call this as non-compact growth. The exact value of $\alpha$ depends on the background. These
patterns still show proportionate growth.

We characterize one such pattern in the asymptotic limit for which $\alpha=1$.
This pattern is produced on a triangular lattice with directed
edges with the background shown in the figure \ref{fig:trilatt}. The
threshold height is $3$ and the
toppling rules are similar to the model on the F-lattice. The corresponding
pattern is shown in figure \ref{fig:trione}. We
show that for this pattern the rescaled toppling function $\phi$ is
piece-wise linear in $\xi$ and $\eta$. The adjacency graph is also simpler,
it is a hexagonal lattice and like the previous examples, the spatial
distances in the pattern are expressed in terms of the solution of the Laplace's equation on
this graph. We determine the solution in a closed integral form.

\section[Zhang model]{Emergence of quasi-units in the Zhang model}\label{sec:zhang}
The Zhang model in one dimension has the remarkable property that in spite of
the randomness in the amount of energy added during driving, the
steady state energy per site has a very sharply peaked
distribution in which the width of the peak is much less than the
spread in the input amount. One such distribution is shown
in figure \ref{fig:emergence}. The threshold energy $E_{c}=1.5$ and
the driving energy is chosen from a uniform distribution in the range
$\left[ 0.76:1.24 \right]$. The distribution has a spike at $E=0$ and
a peak at $E=1.0$ with a standard deviation $\sigma=0.0135$. In general
the width of the peak decreases with the increase of
the system size.
\begin{SCfigure}
\includegraphics[scale=0.7,angle=0,clip]{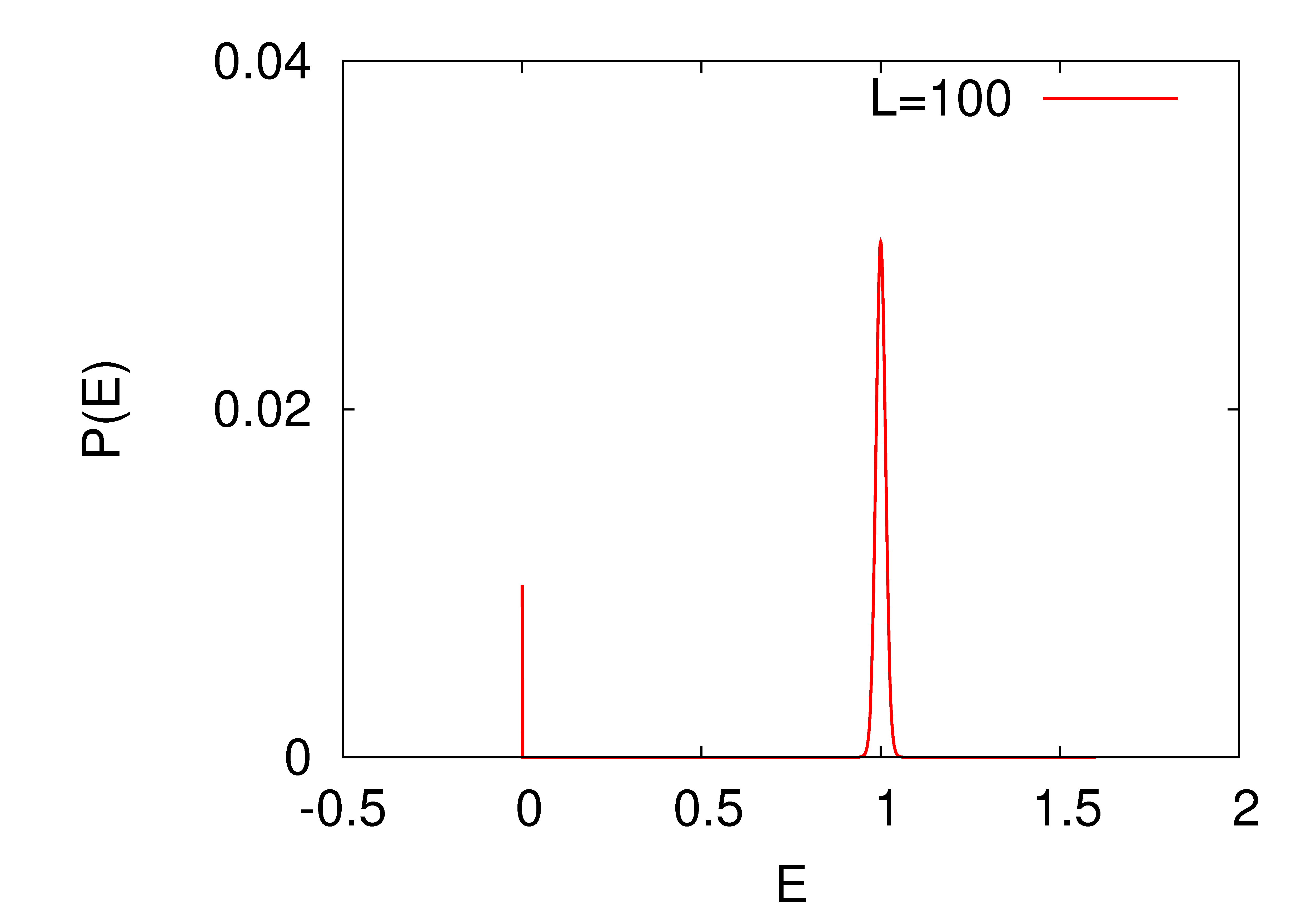}
\label{fig:emergence}
\caption{The probability distribution of energy per site in the Zhang model defined on
a linear chain of length $L=100$.}
\end{SCfigure}

This behavior was noticed in numerical simulations in
both one and two dimensions, and it is called as ``emergence of
quasi-units''. It is argued that for large systems, the behavior would
be the same as in the discrete model \cite{zhang}. Fey et. al.
\cite{fey} have
shown that for some choices of the distribution of input energy, in
one dimension, the variance of energy does go to zero as length of the
chain $L$ goes to infinity. However they did not show how fast the
variance decreases with $L$.

We study the emergent behavior in the one-dimensional
model by looking at how the added energy is redistributed among
different sites in the relaxation process. Let the amount of energy
for driving at time $t$ be chosen randomly from a uniform interval $\left
[1-\epsilon, 1+\epsilon \right ]$. The time is counted in terms of
the relaxation steps and at one time step all the unstable sites relax
together. We write the amount of driving energy at time $t$ as
\begin{equation}
\Delta_{t}=1+\epsilon u_{t},
\end{equation}
where $u_{t}$ is uniformly distributed in the interval $\left[ -1,1
\right]$. We decompose the energy variable at a site $x$ in a
relaxation time step $t$ as
\begin{equation}
E\left( x,t \right)=Nint[E\left( x,t \right)]+\epsilon \eta\left( x,t
\right)
\end{equation}
where $Nint$ refers to the nearest integer value. The integer part
evolves as the integer heights in the ASM. The function $\eta\left(
x,t \right)$ is independent of $\epsilon$ and is a linear function of
$u_{t}$. The precise function depends on the evolution history
$H_{t}$ which
is determined by the initial configuration and the sequence of
addition sites $d_{t}$ up to time $t$. We assume that at starting time
$t=0$, the variable $\eta\left( x,t=0 \right)$ are zero for all
$x$, and the initial configuration is a recurrent configuration of the
ASM. We define $G\left( x, t|d_{t'},t',H_{t} \right)$ by
\begin{equation}
\eta\left( x,t| \left\{ u_{t} \right\},H_{t}
\right)=\sum_{t'=1}^{t}G\left( x,t|d_{t'},t',H_{t} \right)u_{t'}.
\end{equation}
We show that this is equal to the probability $Prob\left( x, t
|d_{t'},t',H_{t}
\right)$ of a marked grain in
the corresponding ASM, added at the site $d_{t'}$ at time $t'$ following
the history $H_{t}$, to be found at site $x$ at time $t$
\begin{equation}
G\left( x, t|a_{t'},t',H_{t} \right)=Prob\left( x, t|a_{t'},t',
H_{t} \right).
\end{equation}
Then we show that the variance of energy in the steady state at a site $x$
can be written as
\begin{equation}
Var\left[ E\left( x \right)
\right]=L/(L+1)^{2}+\epsilon^{2}\Sigma^{2}\left( x \right),
\end{equation}
where
\begin{equation}
\Sigma^{2}\left( x \right)=\lim_{t\rightarrow \infty}Var\left[
\eta\left( x,t \right)
\right]=\frac{1}{3}\sum_{\tau=0}^{\infty}\frac{1}{L}\lim_{t'\rightarrow\infty}\sum_{x'}\overline{G^{2}
(x,t'+\tau|x',t',H_{t} )}.
\end{equation}
The overbar denotes averaging over the evolution histories $H_{t}$.

The function $\overline{G}$ has been studied in \cite{punya}, but for
$\overline{G^{2}}$ the calculation is much more difficult. However
analyzing the behavior in different limits we introduce a
phenomenological expression for it and show that in the large
$L$ limit the variance of energy at a site $x$ has a scaling form
$L^{-1}g\left( x/L \right)$. We determine an approximate form of the
scaling function 
\begin{equation}
g\left( \xi \right)=A\ln\left( 1+\frac{1}{B\sqrt{\xi\left( 1-\xi
\right)}} \right),
\end{equation}
where $A$ and $B$ are numerical constants. This expression agrees very
well with the results of our numerical simulation.

\section[Stochastic models]{Stochastic models of SOC}\label{sec:manna}
The sandpile models with stochastic toppling rules are important subclass
of SOC models. These models are able to describe the avalanche behavior seen experimentally in the
piles of granular media much better than the deterministic models \cite{osloexpt}.
Also in the numerical studies, one gets better scaling collapse, and
consequently, more reliable estimates for the values of the critical
exponents, than for the models with deterministic toppling rules
\cite{stochnum}. There is good numerical evidence that these models
constitute a universality class different from their deterministic
counter parts \cite{univ1,univ2,numbtw,numsand1,numsand2}.

Unfortunately, at present, the theoretical understanding of the models
with stochastic toppling rules is much less than the
deterministic models, even the characterization of the
steady state is not known for the one-dimensional Manna model.

We consider the Manna model with Abelian toppling rule on a linear
chain of length $L$, described in section \ref{sec:intro}. Any stable
state of the model is expressed as a linear combination of the stable
height configurations with the coefficients being the probability of finding
the system in that configuration. We define an addition operator
$\mathbf{a}_{i}$
corresponding to the addition of grains at the site $i$, which acting on a
stable state takes it to another stable state achieved following the
stochastic relaxation rule.

The Abelian property of the operators is shown in \cite{dm}.
However, unlike the deterministic ASM, the inverse operator $\left\{
\mathbf{a}_{i}^{-1} \right\}$ need not exist. This makes
the determination of the matrix form of the operators difficult for
this model. In a straightforward exact numerical calculation of the steady state
one needs to invert a matrix of size $4^{L}\times 4^{L}$. We use the
operator algebra of the addition operators to obtain an efficient
method which requires inverting a matrix only of size $2^{L}\times
2^{L}$.

Using the conservation of sand grains during toppling we show that
\begin{equation}
\mathbf{a}_{i}^{2}=\frac{1}{4}\left(
\mathbf{a}_{i-1}+\mathbf{a}_{i+1} \right)^{2}.
\label{eg}
\end{equation}
In general the operators $\mathbf{a}_{i}$ need not be diagonalizable.
However, using the Abelian property we construct a common set of
generalized eigenvectors for all the operators such that in this basis the
matrices simultaneously reduce to the Jordan block form. Corresponding to
each block there is at least one common eigenvector. Then from
eq. (\ref{eg}), the eigenvalues satisfy
\begin{equation}
a_{i}^{2}=\frac{1}{4}\left( a_{i-1} +
a_{i+1}\right)^{2},
\end{equation}
with the boundary condition $a_{0}=a_{L+1}=1$. We reduce this set of
coupled quadratic equations to a
set of linear equations by taking square root
\begin{equation}
\eta_{i}a_{i}=\frac{1}{2}\left( a_{i-1}+a_{i+1} \right),
\end{equation}
where $\eta_{i}=\pm 1$. There are $2^{L}$ different choices for the
set of $L$ different $\eta$'s and for each such choice, we get a set
of eigenvalues $\left\{ a_{i} \right\}$. In general there will be
degenerate sets of eigenvalues. We show that degeneracies are possible
only for $L= 3 \left(\rm{mod ~} 4 \right)$ and the set can at most be doubly
degenerate. This implies that the largest dimension of a Jordan
block is $2$. We determine the matrix elements inside each block in
terms of the solution of a set of coupled linear equations.

We define a transformation matrix between this generalized
eigenvector basis and the height configuration
basis.
\begin{equation}
|\left\{ z_{i} \right\}\rangle = \sum_{j} \mathbf{M}_{\left\{
z_{i}\right\},j} |\psi_{j}\rangle,
\end{equation}
where $|\left\{ z_{i} \right\}\rangle$ is the basis vector
corresponding to the height configuration $\left\{ z_{i} \right\}$ and
$|\psi_{j}\rangle$ is the $j$ th generalized eigenvector.
Any height configuration can be generated by an appropriate sequence
addition operators acting on the all empty configuration. As the action of the
addition operators on the generalized eigenvectors are known, this
determines the elements of the transformation matrix $\mathbf{M}$.

Given $\mathbf{M}$ we can get the eigenvectors of the addition
operators in the configuration basis, in particular, the steady state
vector, by the inverse transformation
\begin{equation}
|\psi_{j}\rangle=\mathbf{M}^{-1}|\left\{ z_{i} \right\}\rangle.
\end{equation}

We numerically calculate the inverse matrix $\mathbf{M}^{-1}$ and
determine the exact steady state for systems of small sizes. This
result is then extrapolated to determine the asymptotic density
profile in the steady state. Our results suggest that
the steady state density averaged over all sites approaches the
asymptotic density as
\begin{equation}
\frac{1}{\rho_{L}}=\frac{1}{\rho_{\infty}}+\frac{B}{\left( L+\delta,
\right)^{\nu}}
\end{equation}
with $\rho_{\infty}=0.953$ which is close to the Monte Carlo estimate
$0.949$. We also find that the ratio of probabilities of the most
probable to the least probable configuration varies as $\exp{\left(
0.94 L \log L \right)}$. We show that the steady state is not a
product measure state.

The method described is easily generalized to other stochastic Abelian sandpile
models and we discuss some examples of them.

\chapter{List of publications}
\begin{itemize}
\item[(1)]Tridib Sadhu and Deepak Dhar, \textit{Emergence of quasiunits
in the one-dimensional zhang model}, Phys. Rev. E. $\mathbf{77}$
(2008). no. 3, 031122.
\item[(2)]Tridib Sadhu and Deepak Dhar, \textit{Steady state of
stochastic sandpile models}, Journal of Statistical Physics
$\mathbf{134}$ (2009), 427.
\item[(3)]Deepak Dhar, Tridib Sadhu, and Samarth Chandra,
\textit{Pattern formation in growing sandpiles}, Europhysics
Letters $\mathbf{85}$ (2009), no. 4, 48002.
\item[(4)] Tridib Sadhu and Deepak Dhar, \textit{Pattern formation in
growing sandpiles with multiple sources and sinks}, Journal of
Statistical Physics $\mathbf{138}$ (2010), 815.
\item[(5)] Tridib Sadhu and Deepak Dhar, \textit{Pattern formation in
fast growing sandpiles}, Phys. Rev. E. $\mathbf{85}$
(2012), 021107.
\end{itemize}

\chapter{Introduction}\label{ch:intro}
\section{Self-organized criticality (SOC)}
One of the most striking aspects of physics is the simplicity of its
laws. Maxwell's equations, Schrodinger's equation, and Hamiltonian
mechanics are simple and expressible in few lines. However every place
we look, outside the textbook examples, we see a world of amazing complexity:
huge mountain ranges, scale free coastlines, the delicate ridges on the surface of sand dunes,
the interdependencies of financial markets, the
diverse ecologies formed by living organisms are few examples. Each situation is
highly organized and distinctive, but extremely complex.
So why, if the basic laws are simple, is the world so complicated? The idea of Self
Organized Criticality was born aiming to give an explanation for this
ubiquitous complexity \cite{jensen}. In this chapter the basic
concepts related to SOC, that will be
important for this thesis, are introduced.

The examples, cited above, share
a common feature: a power-law tail of the correlations. Consider the
two point correlation of a quantity $\Delta h\left(\mathbf{x} \right)=h\left( \mathbf{x} \right)-\bar{h}$, where $h\left( \mathbf{x}
\right)$ is the height at a place $\mathbf{x}$ in a mountain range, and
$\bar{h}$ is its average value. The function $\langle \Delta
h(\mathbf{x+r})\Delta h(\mathbf{x}) \rangle$
increases as $r^{\delta}$, with the exponent $\delta$ varying very
little for different mountain ranges. Similar distribution with
extended tails is observed in many other natural phenomena:
Gutenberg-Richter law in earth quake \cite{gutenberg}, Levy
distribution in stock market price variations \cite{levy}, Hacks law in River networks \cite{hackslaw1,hackslaw2} etc. Such power-law
distributions entail scale invariance --- there are no macroscopic
spatial scales other than the system size,
in terms of which one can describe the system, making it complex.

Such features are familiar to physicist in equilibrium systems undergoing phase
transition. In standard critical phenomena there are control
parameters such as temperatures, magnetic field, which requires to be
fine tuned by an external agent, to reach the critical point. This is
unlikely to happen in naturally occurring processes such as formation
of mountain ranges, earth quakes or even stock markets. These are
non-equilibrium systems brought to their present states, by their
intrinsic dynamics --- and not by a delicate selection of temperature,
pressure or similar control parameters. \footnote{Per Bak, in his
book \cite{blind}, puts this in an interesting comment---``The nature is
operated by a 'blind watchmaker' who is unable to make continuous fine
adjustments''}.

In the summer of 1987, Bak, Tang and Wiesenfeld(BTW) published a paper
\cite{btw}
proposing an explanation to such ubiquitous scale invariance. They argued that the dynamic which gives
rise to the robust power-law correlations seen in the non-equilibrium
steady states in nature must not involve any fine tuning of
parameters. It must be such that the systems under their \textit{natural
evolution} are driven to a state at the boundary between the stable and
unstable states. Such a state then shows long range spatio-temporal
fluctuations similar to those in equilibrium critical phenomena. The
complex features appear spontaneously due to a cooperative behavior
between the components of the system. They called this
self-organized criticality as the system self-organizes to its
critical steady state.

SOC nicely compliments the idea of chaos. In the latter,
dynamical systems with a few degrees of freedom, say as little as
three, can display extremely complicated behavior. However,
a statistical description of this randomness is predictable in the
sense that, the signals have a white noise spectrum, and not a
power law tail. A Chaotic system has little memory of the past, and it
is easy to
give a statistical description of such behavior. In short, chaos does not explain complexity.
On the other hand, in SOC, generally, we start with systems of many
degrees of freedom, and find a few general features which are also
statistically predictable, but has a power-law spectra leading to
complex behavior. In certain dynamical systems, \textit{e.g.},
logistic maps, there are points (the Feigenbaum point \cite{chaos}) in the parameter space, which
separates states with a predictable periodic behavior and chaos. At
this transition point there is complex behavior, with power-law
correlations. SOC gives description of how systems, under their
own dynamics, without external monitoring, reaches this very special
point (``edge of chaos''), explaining the robust complex behavior in natural systems.

In the book ``How nature works?'', Per Bak gives various kinds of
natural examples of SOC, of which the canonical one is the sandpile.
On slowly adding grains of sand to an empty table, a pile will grow
until its slope becomes critical and avalanches start spilling over
the sides. If the slope is too small, each grain just
stays at the place where it lands or creates a small avalanche. One can understand the motion of
each grain in terms the local properties, like place, the neighborhood
around it etc. As the process continues, the slope of the pile become
steeper and steeper. If the slope becomes too large, a large
catastrophic avalanche is likely, and the slope will reduce. Eventually, the slope reaches a critical value
where there are avalanches of all sizes. At this point,
the system is far out of balance, and its behavior can no longer be
understood in terms of the behavior of localized events. The
system is invariably driven towards its critical state.

\section{Theoretical models}
In order to have a mathematical formulation of SOC,
BTW studied a so-called cellular automata known
as the sandpile model \cite{btw}, which is discrete in space, time and in its
dynamical variables. The model is defined on a two dimensional square
lattice where each site $i$ has a state variable $z_{i}$ referred as
height, which takes only positive integer values. This integer can be thought
of as representing the amount of sand at that location
or in another
sense it represents the slope of the sandpile at that point. Neither
of these analogies is fully accurate, the model has aspects of both.
One should consider it as a mathematical model of SOC, rather than an
accurate model of physical sand.

A set of local dynamical rules defines the evolution of the model: At each time
step a site is picked randomly, and its height $z_{i}$ is increased by
unity. In this thesis, this step will be referred as the
\textit{driving}. If the height now is greater than or equal to a threshold value
$z_{c}=4$, the site is said to be unstable. It relaxes by toppling
whereby four sand grains leave the site, and each of the four
neighboring sites gets one grain. If there are any unstable sites
remaining, they too are toppled, all in parallel. In case of toppling at a site at the
boundary of the lattice, grains falling outside the lattice are
removed from the system. This process continues until all sites are
stable. This completes one time step. Then, another site is picked randomly, its height is increased by
$1$, and so on.

The following example illustrates the dynamics. Let the lattice size
be $3\times 3$ and suppose at some time step the following
configuration is reached where all sites are stable.
$$
\begin{tabular}{|c|c|c|}
\hline
2 & 3 & 2 \\
\hline
3 & 3 & 0 \\
\hline
1 & 2 & 3 \\
\hline
\end{tabular}
$$
We now add a grain of sand at randomly selected site: let us say the
central site is chosen. Then the configuration becomes the
following
$$
\begin{tabular}{|c|c|c|}
\hline
2 & 3 & 2 \\
\hline
3 & {\color{red}4} & 0 \\
\hline
1 & 2 & 3 \\
\hline
\end{tabular}
$$
The central site is not stable, and therefore it will topple and
the configuration becomes
$$
\begin{tabular}{|c|c|c|}
\hline
2 & {\color{red}4} & 2 \\
\hline
 {\color{red}4} & 0 & 1 \\
\hline
1 & 3 & 3 \\
\hline
\end{tabular}.
$$
This configuration has two unstable sites, so both
will topple in parallel. Since these are at the boundary, two grains
will be lost, on toppling. The new result is
$$
\begin{tabular}{|c|c|c|}
\hline
{\color{red}4}  & 0 & 3 \\
\hline
0 & 2 & 1 \\
\hline
2 & 3 & 3 \\
\hline
\end{tabular},
$$
and further toppling leads to
$$
\begin{tabular}{|c|c|c|}
\hline
0 & 1 & 3 \\
\hline
1 & 2 & 1 \\
\hline
2 & 3 & 3 \\
\hline
\end{tabular}
$$
This is a configuration with all sites stable. One speaks in this case
of an avalanche of size $s=4$, since there are four topplings. Another
measure is the number of steps required for relaxation, which in this case
is $t=3$.
For large lattices, in the steady state, the distribution of avalanche sizes and durations
display a long power-law tail, with an eventual cutoff determined by
the finite size of the system.

Since the original sandpile model by BTW a large number of variations
of the model have been studied (see \cite{dharphysica06,jensen} for
reviews). These are mostly extended
systems with many components, which under steady drive reaches a
steady state where there are irregular burst like relaxations and long
ranged spatio temporal correlations.
It is to be noted that
in these models the complexity is not contained in the evolution rules itself, but
rather emerges as a result of the repeated local interactions among
different variables in the extended system.

In the rest of this chapter, I will introduce three of the most
studied models of sandpile and the techniques used to analyze them.

\subsection{Deterministic abelian Sandpile Model
(DASM)}\label{sec:dasm}
This is the most studied model due to it analytical tractability. In a
series of papers, Deepak Dhar and his collaborators have shown that
this model has some remarkable mathematical properties. In
particular, the critical state of the system has been
well characterized in terms of an abelian group. In the following I
will generally follow the discussion in \cite{dharphysica06}.

The model is a generalized BTW model on any general graph with $N$ sites labeled
by integers $1,2,3\cdots N$. To make things precise, I will start with
some definitions. A configuration $C$ for the sandpile model is
specified by a set
of integer heights $\left\{ z_{i} \right\}$ defined on the $N$ sites
of the graph. We denote a threshold value of the height at a site $i$ as $z_{i}^{c}$. The system is driven like the BTW
model by adding one sand grain at a randomly chosen site which
increases the height at that site by $1$. The toppling rules are
specifies by a $N\times N$ toppling matrix $\Delta$ such that on toppling at site
$i$, heights at all sites are updated according to the rule: 
\begin{equation}
z_{j}\rightarrow  z_{j}-\Delta_{i,j}  \rm{~ ~ for~ ~ every
~}j.
\end{equation}
For example in the BTW model on a square lattice
\begin{equation}
\Delta_{i,j}=\left\{ \begin{array}{rl}4 &\mbox{for $i=j$}\\-1
&\mbox{for $i,j$ nearest neighbors}\\0 &\mbox{otherwise}\end{array}
\right.
\end{equation}
Evidently the matrix $\Delta$ has to satisfy some conditions to ensure
that the model is well behaved. These are
\begin{itemize}
\item[1.] $\Delta_{i,i}>0$, for all $i$. (Height decreases at the
toppled site)
\item[2.] $\Delta_{i,j}\le0$, for all $j\ne i$. (Heights at other
sites are increased or unchanged)
\item[3.] $\sum_{j}\Delta_{i,j}\ge 0$ for all $i$. (Sand is not
generated in toppling)
\item[4.] Each site is connected through toppling events to at least
one site where sand can be lost, such as the boundary.
\end{itemize}
Without loss of generality we choose $z_{i}^{c}=\Delta_{i,i}$ (This
only amounts to defining the reference level for the height variables).

With this convention, if all $z_{i}$ are initially non-negative they
will remain so, and we restrict ourself to configurations $C$ belonging to
that space, denoted by $\Omega$. Let $\mathbb{S}$ be the space of stable
configurations denoted by $C_{s}$ where the height variables at each
site are below threshold. The property $4$ above ensures that
stability will always be achieved in a finite time.

We formalize the addition of sand to a stable configuration by
defining an ``addition operator'' $a_{i}$ so that
$a_{i}C_{s}$ is the new stable configuration obtained by taking
$z_{i}\rightarrow z_{i}+1$ and then relaxing.

The mathematical treatment of ASM relies on one simple property it
possesses: The order in which the operations of particle addition and
site toppling are performed does not matter. Thus the operators
$a_{i}$ commute \textit{i.e.}
\begin{equation}
a_{i}a_{j}C_{s}=a_{j}a_{i}C_{s} \mbox{~ ~  for every $i,j$ and
$C_{s}$}.
\end{equation}
The proof uses the linearity of the toppling processes \cite{dharphysica06}.
In the relaxation processes represented by the two sides of the above
equation, the order of topplings can be changed, but the final
configurations are equal. An example of this abelian nature is the
process of long addition of multi-digit numbers. In this example the
toppling process is like carrying.

%
Note that, there are some ``garden of Eden'' configurations that once
exited can not be reached again. For example, in the
BTW model on square lattice, system can never reach a state with two
adjacent $z_{i}=0$. This is because in trying to topple a site to
zero, the neighbor gains a grain, and vice versa. This leads to the
definition of the recurrent state space $\mathcal{R}$ which consists of any stable
configuration that can be achieved by adding sand to some other
recurrent configuration. This
set is not empty since one can always reach a minimally stable
configuration defined by having all $z_{i}=z_{c}-1$.

Dhar proved \cite{dharprl} another remarkable property that the addition operators
$a_{i}$ have unique inverses when restricted to the recurrent space;
that is, there exists a unique operator $a_{i}^{-1}$ such that
$a_{i}\left (a_{i}^{-1}C_{s}\right )=C_{s}$ for all $C_{s}$ in $\mathcal{R}$.
This can be easily seen from the fact that there are finite number of
configurations in $\mathcal{R}$, so for some positive period
$p$, $a_{i}^{p}C_{s}=C_{s}$ with $C_{s}$ a recurrent configuration.
Using the abelian property it can be shown that the period $p$ is same
for all $C_{s}\in \mathcal{R}$. Then $a_{i}^{p-1}$ is the inverse for
$a_{i}$.

These properties of $a_{i}$ have some interesting
consequences \cite{dharprl}. One is that in the steady state all the recurrent
configurations are equally probable. Also, the number of recurrent
states is simply the determinant of the toppling matrix
$\Delta$. For large square lattices of $N$ sites this determinant can be
found easily by Fourier transform. In particular, whereas there are
$4^N$ stable states, there are only $\left( 3.2102\cdots
\right)^{N}$ recurrent states. Thus starting from an arbitrary state
and slowly adding sand, the system self-organizes to an exponentially small
subset of states, which are the attractor of this dynamics.

There are many more interesting properties of the DASM,
\textit{e.g.}, using a burning algorithm \cite{burning}, it is
possible to test whether an arbitrary
configuration is recurrent. Using this
algorithm it can also be shown that the model is related to statistics
of spanning-trees
on the lattice, as well as with the $q\rightarrow 0$ limit of the
Potts model \cite{burning,dharphysica06}. As several results are
known about spanning tree these equivalence help in relating properties
of DASM to known properties of spanning trees.

In spite of these interesting mathematical properties, the exponents
characterizing the power-law tail in the distribution of
avalanches are still difficult to
determine analytically on most lattices, and computer
simulations are still needed. In fact, on a square lattice, the
numerical values estimated by different people have shown a wide range
of values. It has been argued that the simple finite size scaling does
not work for the avalanche distribution and instead it has a
multi-fractal character \cite{kadanoff}. In some simpler quasi-one dimensional
lattices it has been shown that simple linear combination of two
scaling forms provides an adequate description \cite{ali}.

For higher dimensional lattices it has been shown by Priezzhev that
the upper critical dimension for the models is $4$ \cite{ucd}. For
lattices of dimension above
$4$, the avalanche exponents take mean field values and can be deduced
from the exact solution of the model on a Bethe lattice
\cite{bethe}.

\subsection{Zhang model}
The Zhang model, introduced by Zhang in 1989, differs from the DASM in
two aspects: first, the height variables $z_{i}$ are continuous and
takes non-negative real values. A site is unstable if its height is above threshold, and
it topples by equally dividing its entire content among
its nearest neighbors, and itself becoming empty. Second, the external
perturbation is not by adding height $1$, but by a quantum chosen randomly
from an interval $\left[ a,b \right)$, where $0\le a < b$ are
positive real numbers.

Here, is an example of the Zhang model in one dimension. Let the
threshold height is $z_{i}^{c}=1.5$, same for all $i$, and an initial configuration is
$$
\begin{tabular}{|c|c|c|}
\hline
0.8 & 1.4 & 0 \\
\hline
\end{tabular}
$$
Now a time step begins by an addition to a random site, of a random
amount chosen from the interval $\left[ 0,1.5 \right)$. Let the amount
is $1.0$ and the site is the central site. After addition the result
is
$$
\begin{tabular}{|c|c|c|}
\hline
0.8 & 2.4 & 0 \\
\hline
\end{tabular}
$$
Because the middle site is unstable, an avalanche starts:
$$
\begin{tabular}{|c|c|c|}
\hline
2.0 & 0 & 1.2 \\
\hline
\end{tabular}
\rightarrow
\begin{tabular}{|c|c|c|}
\hline
0 & 1.0 & 1.2 \\
\hline
\end{tabular}
$$
In case of two or more unstable sites, all are toppled in parallel.

Since the addition amount is random, a stable site could in principle
have any height between zero and the threshold and the stationary
distribution could be very different from that of the DASM, where
only
discrete values are encountered. Nevertheless, when one simulates the
model on large lattices in one and two dimensions, the stationary
heights tend to concentrate around nonrandom discrete values. This is
known as the ``emergence of quasi-units'' \cite{zhang}. It appears that altering
the \textit{local} toppling rules of the DASM, does not have that much effect on the
\textit{global} behavior after all, if the lattice size is large.

This behavior led to the conjecture that in the thermodynamic limit
the critical behavior is identical to that of DASM. However, due to the
changed toppling rules, the dynamics is no longer
abelian,
and determining the steady state is quite difficult, even in
one-dimension. In fact, Blanchard \textit{et. al.} have shown that the
probability distribution of heights in the steady state, even for the
two site problem, has a multi-fractal character
\cite{blanchard}.

This status was unchanged for over a decade until Fey
\textit{et. al.} showed that on a one-dimensional lattice, for some
specific choice of the amount of addition, the toppling becomes
abelian. Using this they showed that, indeed, the model is on the same
universality class of the DASM. However, the analysis in two dimension
is still an open problem.

\subsection{Manna model}
Another important class of the sandpile models are with stochastic
toppling rules. The first model of its kind was studied by Manna in
1991, and is known as the Manna model \cite{manna}.

The evolution rules of this sandpile in $d$-dimensions are very
similar to the ones defined for the DASM. In fact, the driving rule
and the dissipation rules at the ``boundary'' remain the same. But in a toppling, an
unstable site redistributes \textit{all the sand grains} between sites randomly
chosen amongst its $2d$ nearest neighbors.
\begin{eqnarray}
z_{i}&\rightarrow&0\nonumber\\
z_{j}&\rightarrow&z_{j}+1 \mbox{~ ~ for $z_{i}$ sites chosen
randomly amongst n.n. of $i$.}\nonumber
\end{eqnarray}

The randomness in the evolution rule is a relevant change in the
dynamics, which makes it non-abelian. It is possible to get back the
abelianness by a simple modification in the toppling rule, which I
will discuss in detail in the later chapters. However, the addition
operators defined appropriately do not form a group anymore and this
makes the analysis less tractable even for a linear chain.

The steady state is very different from its deterministic counter part
\textit{e.g.} on
a simple linear chain the different recurrent states are not equally
probable, unlike the deterministic model. Also the avalanche
distribution can be satisfactorily described by simple finite
size scaling.
Another evidence of the differences between these models is in the way
the avalanches spread over the lattice \cite{milshtein}. The
distribution of number of toppling per site in a typical avalanche for
both DASM and Manna model on a square lattice are shown in Fig. \ref{fig:ava}. For the
DASM, one can see a shell structure in which all sites that toppled
$T$ times form a connected cluster with no holes, and these sites are
contained in the cluster of sites that toppled $T-1$ times, and so on.
On the other hand, the toppling distribution exhibits a random
avalanche structure with many peaks and holes.
\begin{figure}
\begin{center}
\includegraphics[width=6.0cm,clip]{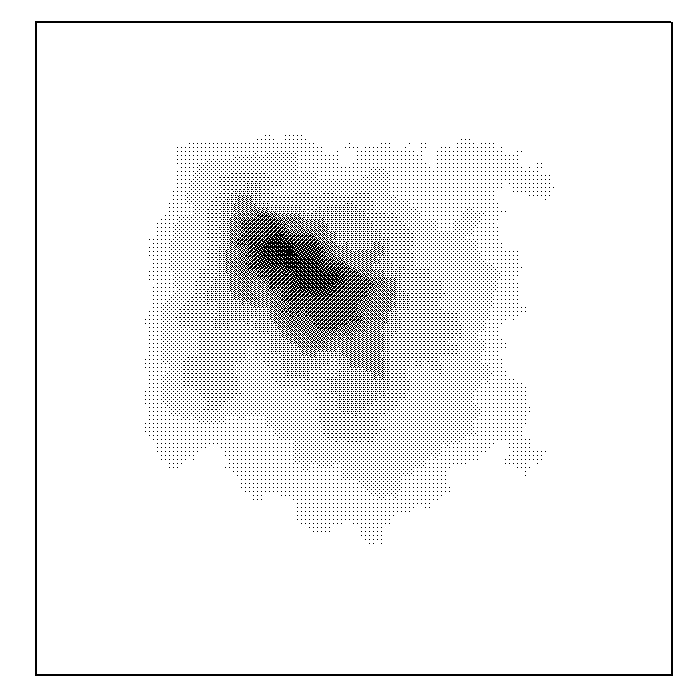}
\includegraphics[width=6.0cm,clip]{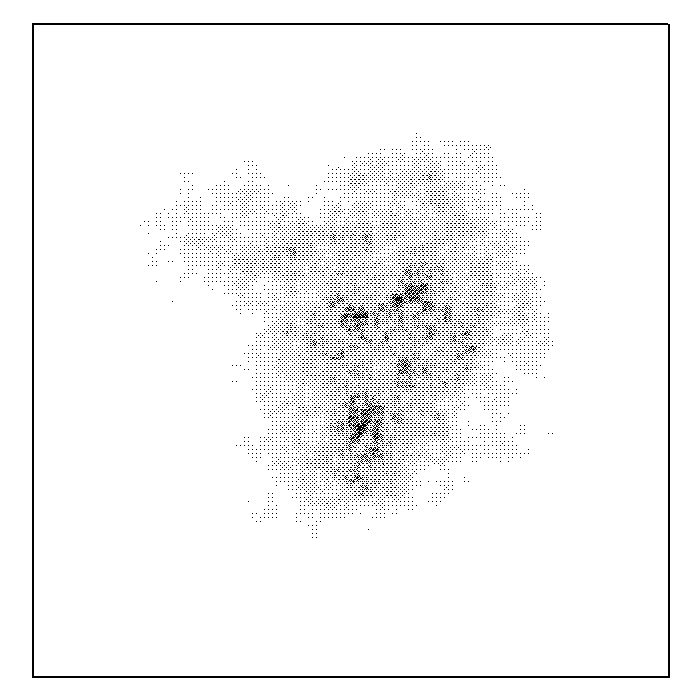}
\label{fig:ava}
\caption{Number of toppling per site for a typical avalanche in
(a) DASM
and (b) Manna model. The darker shades denote more topplings. (Courtesy
\cite{milshtein})}
\end{center}
\end{figure}

For many years, the universality of the manna model was a controversial
question. At present there are convincing numerical evidences that in
dimension up to $3$, they have a different critical behavior, from its
deterministic counterpart, with a different set of critical exponents.
However in dimensions $d\ge 4$, both DASM and Manna model take the same
mean-field values of critical exponents.

\section{Universality in the sandpile models.}
Since the work by BTW, a large number of different models have been
studied \textit{e.g.} sandpile models with many variations of the BTW toppling
matrix \cite{kadanoff} or sand grain distribution rules
\cite{maslov}, stochastic topplings \cite{manna}, with
activity inhibition \cite{manna2}, continuous
height models \cite{zhang}, loop erased random walk \cite{loop},
Takayasu aggregation model \cite{takayasu}, train model
\cite{train1,train2}, non-abelian sandpile directed sandpile model
\cite{lee,pan,threshold,pruessner}, forest-fire model \cite{forest},
Olami-Feder-Christensen model \cite{ofc} \textit{etc} (and many more could
have been defined). Most of these models could only be studied numerically,
and for a while it seemed that each new variation studied belong to a
new universality class each with its own set of critical exponents. It is a fair
question to ask, what is the generic behavior?

Although this question is not yet resolved
completely, by now, there has been a fair amount of understanding of
this problem. The universality classes with renormalization group flow
in these models can be summarized in the Fig. \ref{fig:universality}.
\begin{SCfigure}
\includegraphics[width=9.0cm,clip]{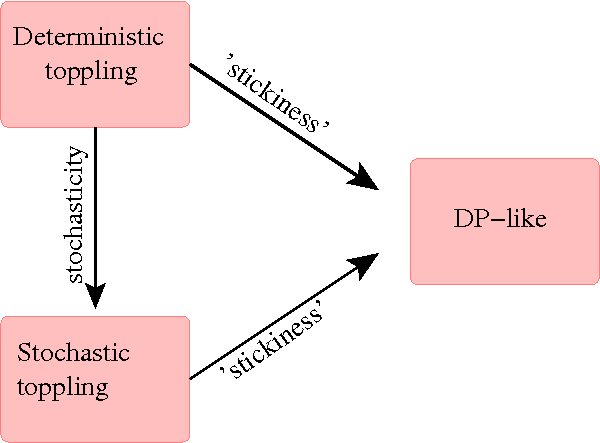}
\caption{ A schematic flow diagram of renormalization
group flows between different fixed points of sandpile models.}
\label{fig:universality}
\end{SCfigure}

There are sufficient numerical evidence that sandpile models with
deterministic toppling rules (DASM) and stochastic toppling rules
(Manna model)
constitute different universality classes. There are also several
other model which show critical exponents different from these two
\cite{sneppen1,sneppen2,grass}.
They are related to the directed
percolation (DP) \cite{dp}, which describes the
active-absorbing state phase transition in a wide class of
reaction-diffusion systems. The activity in avalanches in sandpile
can grow, diffuse or die, and any stable configuration is an absorbing
state. Thus one would expect that in general the sandpile should belong to the
universality class of active-absorbing state transition with many
absorbing states \cite{path}. However, these models do not involve any conserved
fields. In Manna and DASM-type models, it is this presence of conservation laws of sand that makes the
critical behavior different from DP \cite{vespignani}.

Recently, the effect of non-conservation has been explicitly studied
\cite{mohanty1,mohanty2}
by introducing stickiness in the toppling rules (\textit{i.e.} there
is small probability that the incoming particles to a site get stuck
there, and do not cause any toppling until the next avalanche hits the
site, thus in effect there is no conservation of grains within an
avalanche). It has been argued that as long as the sand grains have
non-zero stickiness, the distribution of avalanche sizes follows directed percolation
exponents. The DASM, and the stochastic Manna-type models are
unstable to this perturbation, and the renormalization group flows are
directed away from them to the directed percolation fixed point, as
schematically shown in the figure. This picture is exactly verified in
\textit{directed} sandpiles. However, the argument is less convincing
for undirected models, and the issues is not settled \cite{bonachela2}.

\section{Experimental models of SOC}
Soon after the sandpile model was introduced, several experimental
groups measured the size distribution of avalanches in granular
materials. Unfortunately, real sandpile do not seem to behave as the
the theoretical models. Experiments show large periodic avalanches
separated by quiescent states with only limited activity. While for
small piles one could try to fit the avalanche distribution with
power-law over a limited range, the behavior would eventually cross
over, on increasing the system size, to a state which is not
scale-invariant \cite{sandexpt1,sandexpt2}. It is later realized that inertia of rolling grains
is the reason for non-criticality.
A large avalanche propagating over a surface with slope
$\theta_{c}$ scours the surface, and takes away materials from it. The
final angle, after the avalanche has stopped, is bellow
$\theta_{c}$.
So if we want to see power-law
avalanches, we have to minimize the effect of inertia of the grains. This is
achieved in an experiment in Oslo by using rice grains. Because of the
elongated shape of the rice grains (Fig. \ref{fig:ricepile})
frictional forces are stronger and these
poured at very small rate gave rise to a convincing power-law
avalanche distribution \cite{osloexpt}.
\begin{SCfigure}
\includegraphics[width=8.0cm,clip]{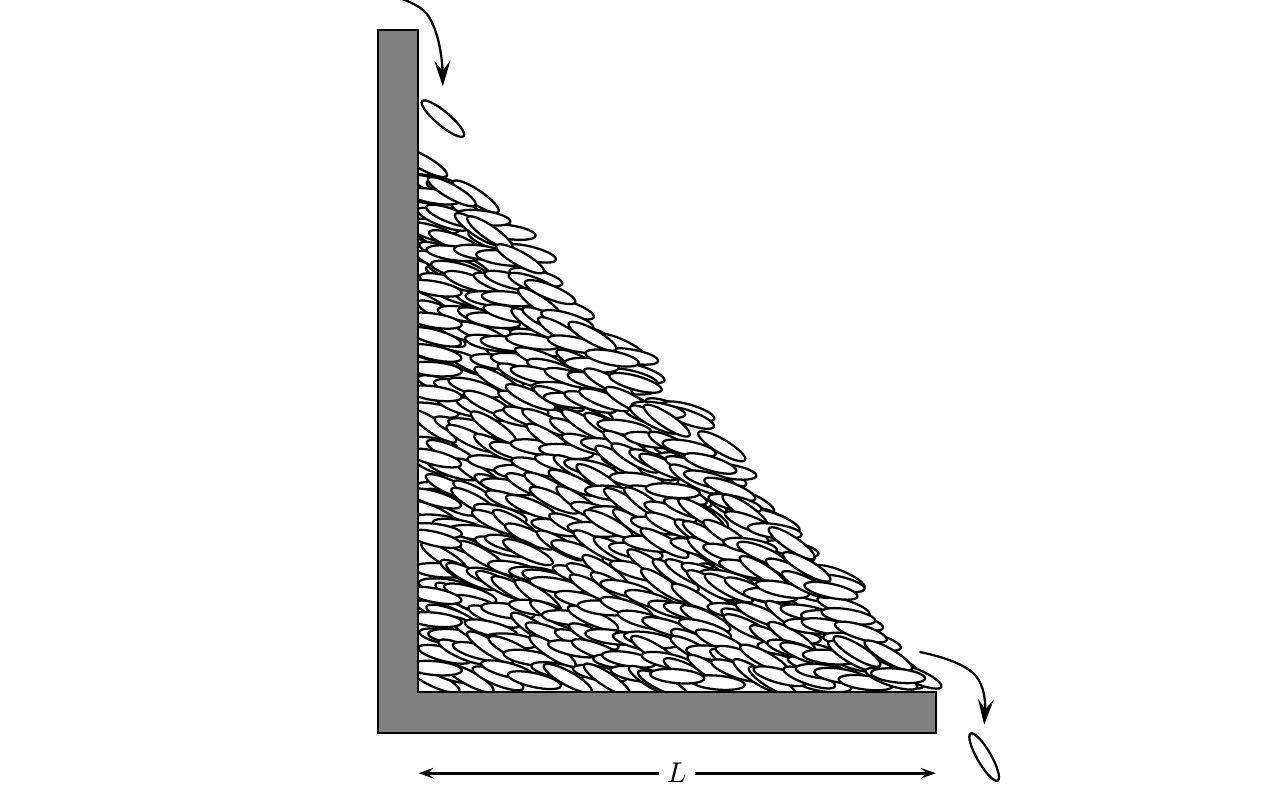}
\caption{A schematic of a rice-pile. The elongated shapes of the rice
grains reduces the inertial effect in an avalanche. (Courtesy K.
Christensen.)}
\label{fig:ricepile}
\end{SCfigure}

A similar power law distribution of avalanche sizes are obtained in
motion of domain walls in ferro-magnets \cite{feromag1,feromag2} and
flux lines in type II superconductors \cite{sup1,sup2}.
A more recent
experimental realization of SOC is obtained in suspensions of
sedimenting non-brownian particles by slow
periodic shear \cite{corte}.

It is worth mentioning that SOC has been invoked in
several other situations in geophysics (atmospheric precipitation
\cite{presp}, river
pattern due to erosion \cite{erotion}, landslides \cite{land}),
biology (neural-network \cite{nunet}), economics (stock-market crashes
\cite{smcr}) and many
more. I have deliberately chosen the above experimental examples for which
experimental observations are accurate and reproducible.
\section{Remarks}
Originally, SOC was proposed with the aim of providing an explanation
of the ubiquitous complexity in nature \cite{btw}. The abundance of fractal
structures in nature, temporal as well as spatial, was considered to be an
effect of a generic tendency --- pertinent to most many-body systems
--- to develop by themselves in to a critical scale-invariant state.

However, certainly not all systems that organize themselves into one
specific state will, when gently driven, exhibit scale invariance in
that self-organized state. The experiments of real sandpiles referred
earlier are a prime example. Neither is all observed power law
behavior are an effect of dynamical self-organization into a critical
state. The work by Sethna and co-workers on Barkhausen noise
\cite{sethna} is an interesting example of this, what Didier Sornette has called ``Power laws by
sweeping of an instability'' \cite{sornet}.

Since the introduction of the idea, a large amount of discussion went into
understanding the minimum conditions for a model to be
self-organized critical. Though a broad picture has emerged in last
two decades, it is still not complete and controversial. In the rest of this section, I will
discuss two of the most discussed properties, using both examples and counter examples.

\begin{itemize}
\item
\textbf{Slow driving limit.} There is a strong belief in the community that an essential ingredient of SOC is
\textit{slow driving} (driving and dynamics operating at two
infinitely separated timescales, \textit{i.e.} avalanches are
instantaneous relative to the time scale of driving). This idea got
widely accepted after an argument given by Dickman \textit{et.
al.}.\cite{path}
They argued that the dynamics in the sandpile model implicitly involve
tuning of the \textit{density} of grains to a value where a phase
transition takes place between an active state, where topplings take
place, and a
quiescent ``absorbing'' state. When the system is
quiescent, addition of new particles increases the density. When the
system is active, particles are lost to the sinks via toppling,
decreasing the density. The slow driving ensures that these two
density changing mechanisms balance one another out, driving the
system to the threshold of instability.

However in the Takayasu model of aggregation \cite{takayasu} the driving is fast. A
simple example, it can be defined on a linear chain on which particles
are continuously injected, diffuse and coalesce. One can write the
explicite rules as follows:
\begin{itemize}
\item
At each time step, each particle in the system moves by a single step,
to the left or to the right, taken with equal probability, independent
of the choice at other sites. 
\item
A single particle is added at every site at each time step.
\item
If there are more than one particle at one site, they coalesce and
become a single particle whose mass is the sum of the masses of the
coalescing parts. In all subsequent events, the composite particles
acts as a single particle.
\end{itemize}
The probability distribution of total mass at a
randomly chosen site, has power law tail, with an upper cutoff that
increases with time. This is a signature of criticality.
The analogue of avalanches in this model is the event of coming
together of large masses. In fact, it can be shown \cite{dharphysica06} that the model is
equivalent to a directed version of sandpile. In this example, it is
clear that the driving is fast, and the rate is comparable to the
local movements of the particles.

\item
\textbf{Conservation.} Conservation of grains is also considered as a key property for the
criticality to emerge in sandpile models. A simple intuitive argument
goes as follows: the sand grains introduced in the pile can dissipate only by
reaching the diffusive ``boundary'' of the lattice. Owing to this and
because of the vanishing rate of sand addition, arbitrarily large
avalanches (of all possible sizes) must exist for an arbitrarily large
system size, yielding a power-law size distribution. In contrast, in
the presence of non-vanishing bulk dissipation, grains disappear at
some finite rate, and avalanches stop after some characteristic size
determined by the dissipation rate. This clearly says that bulk
dissipation is a relevant perturbation in the \textit{sandpile
dynamics} and breaks the criticality \cite{bonachela1}. 

There are also some other models of SOC like Forest fire
\cite{forest}, OFC model \cite{ofc}
where it was shown, mostly numerically, that non-conservation in the
dynamics leads to non-critical steady state.

However, extrapolating these results and considering conservation as a
neccesarry criteria for SOC, in general, is not correct. A model which
is clearly non-conservative and still, when slowly driven displays
power-law in the avalanche size distribution is discussed in
\cite{ice}. Another two
non-conservative models of SOC are a sandpile model with threshold
dissipation\cite{thresdis}, and Bak-Sneppen model of evolution
\cite{sneppen2}.
\end{itemize}

Finally, one could ask: Has SOC, taught us anything about the world that
we did not know prior to it? Jensen addresses this question very
nicely in his book. The most important lesson is that, in a
great variety of systems, particularly for slowly-driven-interaction-
dominated-threshold systems, it is misleading to neglect
fluctuations. In these systems, sometimes, the fluctuations are so
large that the fate of a major part of the system can be determined by
a single burst of activity. Dinosaurs may have become extinct simply
as a result of an intrinsic fluctuation in a system consisting of a
highly interconnected and interacting web of species; there may be
no need for an explanation in terms of external bombardment by
meteorites. \textit{Fluctuations are so large that the "atypical" events
decides the future of the system}.
This new insight is sufficiently important to justify and inspire more
theoretical, and experimental research in SOC.
\section{Overview of the later chapters}
The work in this thesis ranges from characterizing the spatial
patterns in sandpile model, to quasi-units in the stationary
distribution of Zhang's model, and determining exact steady state of
Manna model. The first three chapters in the following are about
sandpile as a growth model, where we show how well-structured
non-trivial patterns emerge at large length scales, due to local interactions in cellular
lengths where the patterns are not obvious. In chapter \ref{ch:zhang}
we discuss another emergent behavior in the Zhang model. Chapter
\ref{ch:manna} contains an operator algebraic analysis of the
stochastic sandpile models. Here is a brief summary of these chapters.
\begin{itemize}
\item[]{\textbf{Chapter 4:} While a considerable amount of research
went into characterizing the universality classes of sandpile models
and understanding the mechanism of SOC, very limited work is done
about spatial patterns in sandpile models. Such patterns were noted around the
time when sandpile was first introduced \cite{liu}. Yet, very little
is known about them.

This chapter is devoted to the study of a class of such spatial patterns produced by
adding sand at a \textit{single} site on an \textit{infinite} lattice
with initial periodic distribution of grains and then
relaxing using the DASM toppling rules.
We present a complete quantitative characterization
of \textit{one} such patterns. We show that the spatial distances in the
asymptotic (in the limit when large number of grains are added)
patterns produced by adding a large number of grains, can be expressed
in terms solution of discrete Laplace's equation (discrete holomorphic
functions \cite{duffin,mercat,laszlo}) on a grid on
two-sheeted Riemann surface.

We also discuss the importance of these patterns as a paradigmatic
model of growth where different parts of the structure grow in
proportion to each other, keeping their shape the same. We call these
kind of growth as \textit{proportionate growth}. We discuss the
importance of such growth in real world examples.}

\item[]\textbf{Chapter 5:} In this chapter we describe how the
pattern changes in presence of absorbing sites, reaching which the
grains get lost and no longer participate in the avalanches. We show
that, again, the \textit{asymptotic} pattern can be characterized in
terms of discrete holomorphic functions, but on a different lattice.
Similar effects of multiple sites of addition on the pattern are
also calculated.

The most interesting effect of the sink sites is the
change in the rate of growth of the pattern. In absence of sink sites
the diameter of the pattern, suitably defined, increases as
$\sqrt{N}$ where $N$ is the number of sand grains added in the
lattice. When the pattern grows with the sink sites inside, the growth
rate of the diameter changes, in general, to $N^{\alpha}$, where
the exponent $\alpha$ depends on the sink geometries. For example,
$\alpha=1/3$, when the sink sites are along an infinite line adjacent
to the site where grains are dropped. When the site of addition is
inside a wedge of angle $\pi/2$ with the sink sites along the wedge
boundary, this value of the exponent is $1/4$. We use an scaling
argument and determine $\alpha$, for some simple sink-geometries.

\item[]\textbf{Chapter 6:} The growth rate also depends on the
arrangement of heights in the background, and this dependence is quite intriguing. When the
initial heights are low \textit{enough} at all sites, one gets
patterns with $\alpha={1/d}$, in $d$-dimension. If sites with maximum
stable height $\left( z^{c}-1 \right)$ in the starting configuration
form an infinite cluster, we get avalanches that do not stop, and the
model is not well-defined. In this chapter, we study backgrounds in
two dimensions. We describe our unexpected finding of an interesting
class of backgrounds, that show an intermediate behavior: For any $N$,
the avalanches are finite, but  the diameter of the pattern increases
as $N^{\alpha}$, for large $N$, with $1/2 < \alpha \leq 1$, the exact
value of $\alpha$ depending on the background. It still shows
proportionate growth. We characterize the asymptotic pattern exactly
for one illustrative example.

\item[]{\textbf{Chapter 7:} As mentioned, the Zhang model on one
and two dimensional lattices displays a remarkable property: emergence
of quasi-units, where the continuous heights, in spite of the
randomness in the driving, are peaked around a few discrete non-random
values. Fey \textit{et. al} have shown that on a linear chain the
width of the distribution vanishes in the infinite volume limit.
However they did not show how it approaches zero.

In this chapter, we show that, the sequence of toppling of the continuous height
variables, when suitably discretized, have an one-to-one relation with
that of integer heights in the corresponding DASM. We use this
relation to show that the width of the distribution of heights
decreases in inverse power of the length of the chain. We also
determine how the variance of height at a site, changes with position
of the site along the length of the system.}

\item[]{\textbf{Chapter 8:} This chapter contains an algebraic approach of
determining the steady state of a class of sandpile models with
stochastic toppling rules. The original Manna model, as
discussed in section 1.2.3, does not have the abelian property of its
deterministic
counterpart. However, a simple modification
of the toppling rules makes the model abelian \cite{dm}. A similar construction is
possible for other stochastic toppling rules. However, analysis of
these
models are still difficult as the corresponding addition operators
(see section \ref{sec:dasm}), in general, does not have an inverse, and are not diagonalizable.

We show that, in principle, the operators can be reduced to a Jordan
block form, using the algebra satisfied by these.  These are then used
to determine the steady state of the models. We illustrate this
procedure by explicitly determining the numerically exact steady for
a stochastic model on a linear chain. Using the desktop computers at
our disposal, we have been able to perform the calculation for systems
of size $\le 12$ and also studied the density profile in the steady
state. }


\end{itemize}
\chapter{Pattern formation on growing sandpiles}\label{ch2}
\textit{Based on the paper \cite{myepl}} by Deepak Dhar, Tridib Sadhu and Samarth
Chandra.

\begin{itemize}\item[\textbf{Abstract}]{\small{
Adding grains at a single site on a flat substrate in the abelian 
sandpile models produces beautiful complex patterns. We study in detail 
the pattern produced by adding grains on a two-dimensional square 
lattice with directed edges (each site has two arrows directed inward 
and two outward), starting with a periodic background with half the 
sites occupied.  The model shows proportionate growth and the size of the pattern formed by adding $N$ 
grains scales as $\sqrt{N}$.  We give exact 
characterization of the asymptotic pattern, in terms of the position
and shape of different features in the pattern.}
}
\end{itemize}

\section{Introduction}\label{ch2intro}
As discussed in details in chapter \ref{ch:intro}, the sandpile models were introduced in physics in the context of 
self-organized criticality, where the main interest has been the 
power-law tail in the distribution of avalanche sizes \cite{btw}.
In this chapter our interest is different. We study the 
pattern produced by adding large number of sand grains at a single 
site in a two dimensional DASM starting from a 
periodic background, and allowing the system to relax using the
sandpile toppling rule (See chapter 1). For example, consider the ASM
on an \textit{infinite} square lattice with initial heights
$z_{i}=2$, for all sites, and add large $N$ number of grains at the origin.
The distribution of heights $z_{i}$ in the relaxed configuration, produces a 
very beautiful, but complex pattern (Fig. \ref{asm})
\footnote{A natural sandpile, formed by pouring sand grains at a
constant rate on a flat table with boundaries, gives rise to singular
structures like ridges, in the stationary state. This has been
attracted much attention recently \cite{hadeler,falcone}}. In this chapter, we give a
detailed quantitative characterization of similar patterns produced on two directed lattices, starting from a
simple periodic background.
\begin{figure}
\begin{center}
  \includegraphics[scale=0.70]{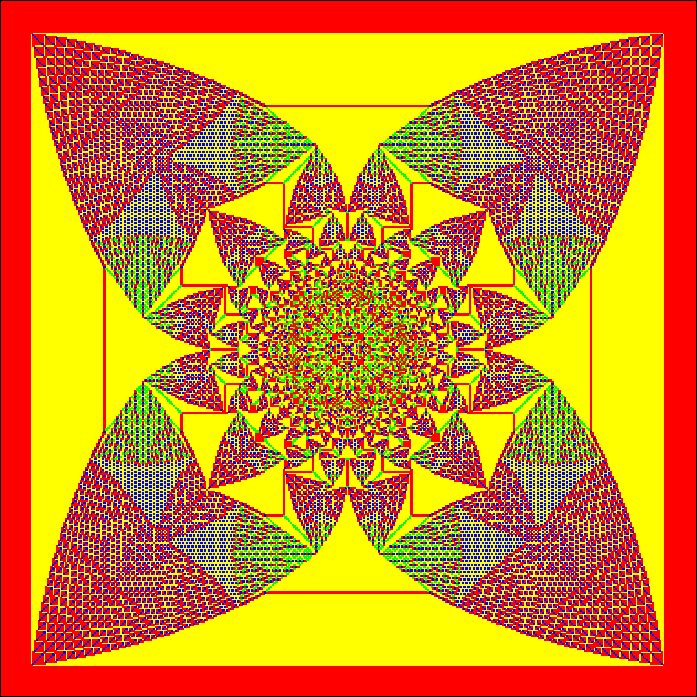}
  \caption{A stable configuration for the abelian sandpile model, obtained 
  by adding $5\times 10^4$ sand grains at one site, on a square
lattice, and relaxing.
  Initial configuration with all heights $2$. Color code: blue=0, green=1, 
  red=2, yellow=3. (Details can be seen in the electronic version using zoom in.)}
  \label{asm}
\end{center}
\end{figure}
\begin{figure}
\includegraphics[width=7.0cm,angle=0,clip]{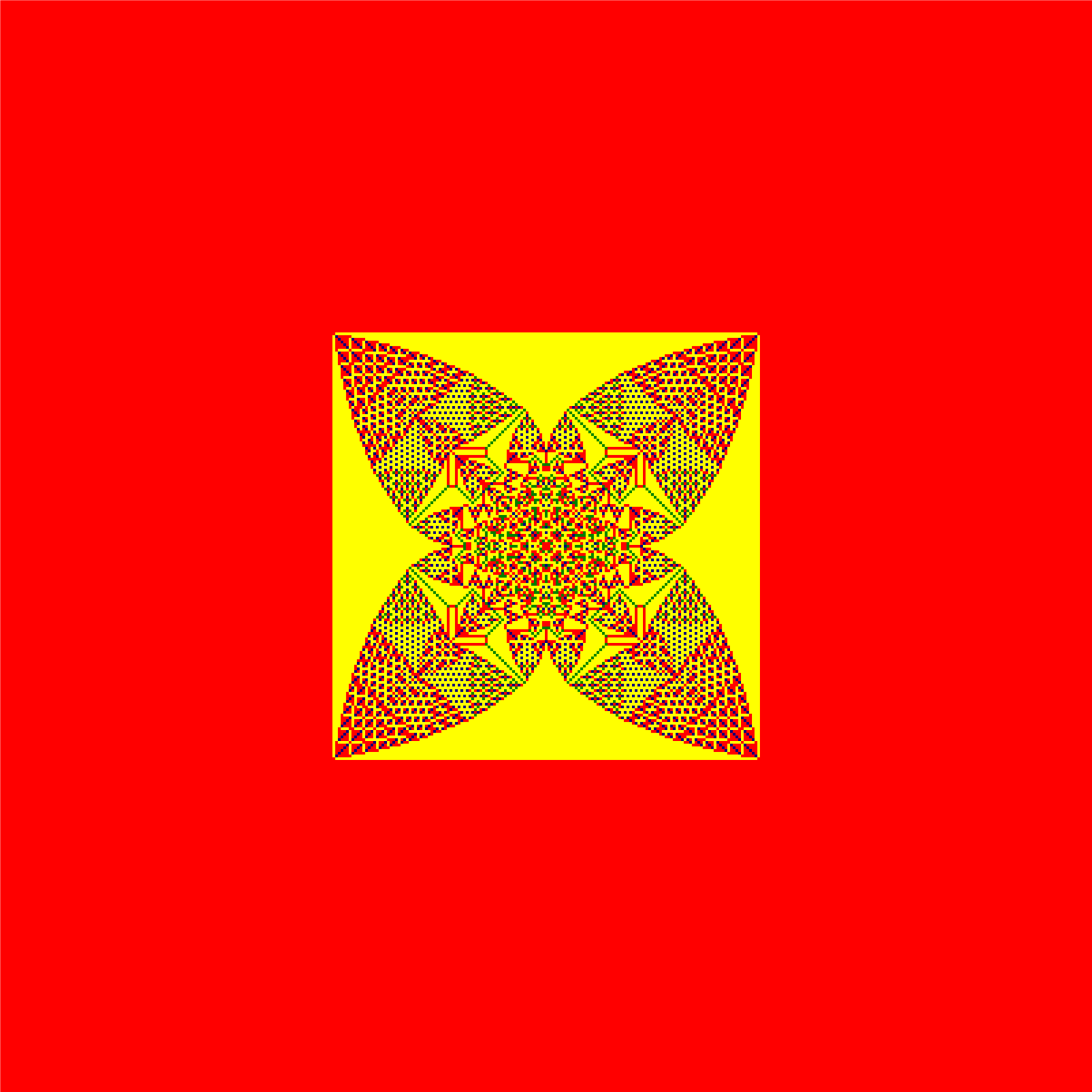}
\includegraphics[width=7.0cm,angle=0,clip]{Images/ch2/fig1}
\caption{The proportionate growth of the pattern for DASM, obtained by adding
$N$ grains at \textit{one} site on a square lattice. Initial configuration
is with all heights $2$. The two patterns correspond to $N=10^{4}$ and
$5\times 10^{4}$, respectively.  Color
code same as in the Fig. \ref{asm}. Both patterns are on the same
scale. (Details can be seen in the electric version using zoom in).}
\label{fig:prop}
\end{figure}

These patterns are examples of complex patterns obtained from simple deterministic 
evolution rules, and are also analytically tractable.
Here complexity means that we have structures with variations, and a
complete description of which is long. Thus,
a living organism is complex because it has many different working
parts, each formed by variations in the working out of the same, but
relatively much simpler genetic coding. There are some earlier known
theoretical examples of
complex patterns obtained from local deterministic evolution rules.
Most studied among them are the
Conway's game of life \cite{earlierone}, which is a cellular automata
model with local update rules, and Turing patterns \cite{earliertwo}, which are
reaction-diffusion models. In general, a detailed and exact 
mathematical characterization of such patterns has not been possible so 
far. In this aspect the sandpile patterns are important as they are
analytically tractable. Understanding these should also help in studying the more general 
problem.
The most important aspect of these patterns is that, these are the simplest examples that show nontrivial spatial 
features with {\it proportionate growth} (see Fig. \ref{fig:prop}),
where these features grows in proportion, keeping the overall shape
same.
Examples of proportionate
growth are abundant in the animal kingdom, where a 
young animal, typically, grows in size with time,
with different parts of the body growing roughly at the same rate. 
Obviously, this kind of growth requires some coordination and communication 
between different parts.
While there are many models of growing objects studied in 
physics literature so far, e.g. the Eden model, Diffusion-Limited 
Aggregation (DLA), invasion-percolation etc. \cite{herrmann}, none has this
property. All of these are 
mainly models of aggregation, where growth occurs by accretion on the 
surface of the object, and inner parts do not evolve significantly
(Fig. \ref{fig:dla}). It is worth mentioning, that modelling proportionate
growth with simple structures is easy, \textit{e. g.}, growth of a
water droplet as one injects more water into it. However, generating a
complex pattern with large number of structures inside, all growing at
same rate, maintaining their relative shape, is highly nontrivial.
This is what happens in the patterns produced in the sandpile models.
\begin{SCfigure}
\includegraphics[width=7.0cm,clip]{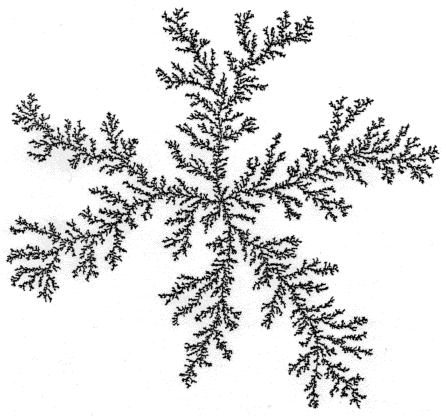}
\caption{Growth in a DLA occurs by accretion of particles, doing random
walk, and attaching to the cluster connected to the origin, when it
comes close to it. }
\label{fig:dla}
\end{SCfigure}

Also, there is an astonishing qualitative similarity between the formation of these
patterns and the way a fertilized egg develops into a well
structured multicellular organism. The development starts with a single cell which
divides into more cells, and then they divide in turn. At some stage
of the development, few cells generate newer types of cells and form
organs. This process continues and after a long time it forms a large
complex, but highly coordinated cellular assembly. The same genetic code in each cell is
responsible for the cell differentiation and structure formation. In
the abelian sandpile when the sand grains are added at the same
site, the grains redistribute themselves and a pattern emerges around
the addition site, which grows as more and more sands are added. The
same redistribution rule is used for all sites, and yet the pattern
has large visibly distinguishable structures inside, which, one can
think of as different organs in an animal.
This simple mathematical model of growing sandpile captures these
qualitative features, even though the actual mechanism of growth is
much more complicated in the biological world.
\begin{figure}
 \begin{center}
 \includegraphics[scale=0.27,angle=90]{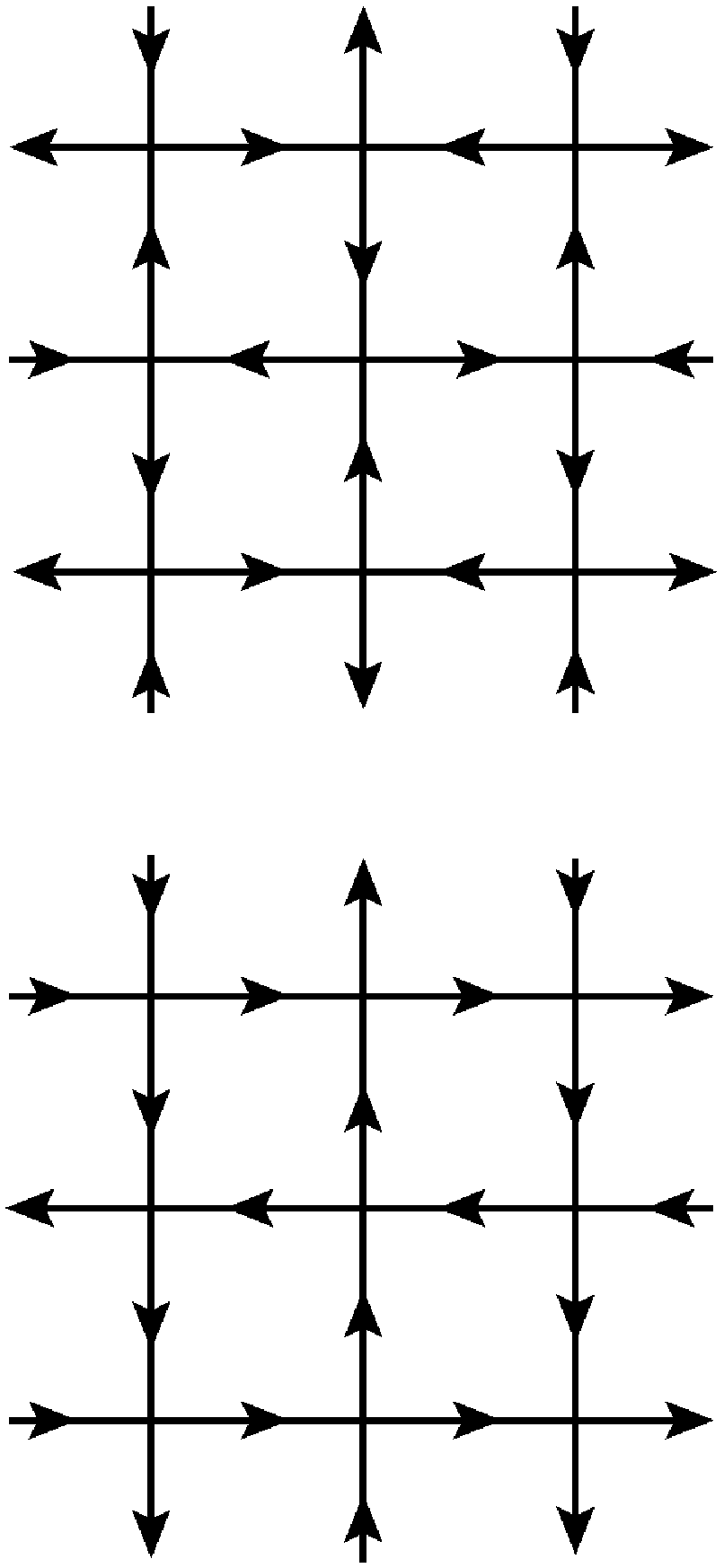}
 \caption{The directed square lattices studied in this chapter: $(a)$ the  
 F-lattice, $(b)$ the Manhattan lattice.}
 \label{lattice}
 \end{center}
\end{figure}

The spatial patterns in
sandpile models were first discussed by Liu \textit{et al.} \cite{liu}. The asymptotic
shape  of the boundaries of the sandpile patterns produced by adding grains at 
a single site in different periodic backgrounds was discussed in \cite{dhar99}.
Borgne \textit{et al.} \cite{borgne} obtained bounds on the rate of growth of these boundaries
and later these bounds were improved by Fey \textit{et al.} \cite{redig} and
Levine \textit{et al.} \cite{lionel}. The first detailed analysis of different
periodic structures found in the patterns were carried out by Ostojic  in
\cite{ostojic}. Other spatial configurations in the abelian sandpile models, like the
identity \cite{borgne,identity,caracciolo} or the stable state produced
from special unstable states, also show complex internal self-similar structures
\cite{liu}, which share common features with the patterns studied here. In particular, the 
identity configuration on the F-lattice has recently been shown to have 
spatial structure similar to what we study here \cite{caracciolo}.
There are other models, which are related to the abelian sandpile model, e.g.
the Internal Diffusion-Limited Aggregation (IDLA), Eulerian walkers (also called
the rotor-router model), and the infinitely-divisible sandpile, which  also show
similar structure. For the IDLA,  Gravner and Quastel showed that the
asymptotic shape of the growth pattern is related to the classical Stefan
problem in hydrodynamics, and determined the exact radius of the pattern with
a single point source \cite{gravner}. Levine and Peres have studied patterns
with multiple sources in these models recently, and proved the existence
of a limit shape\cite{levine_peres}.
\begin{figure}
  \begin{center}
  \includegraphics[scale=0.70]{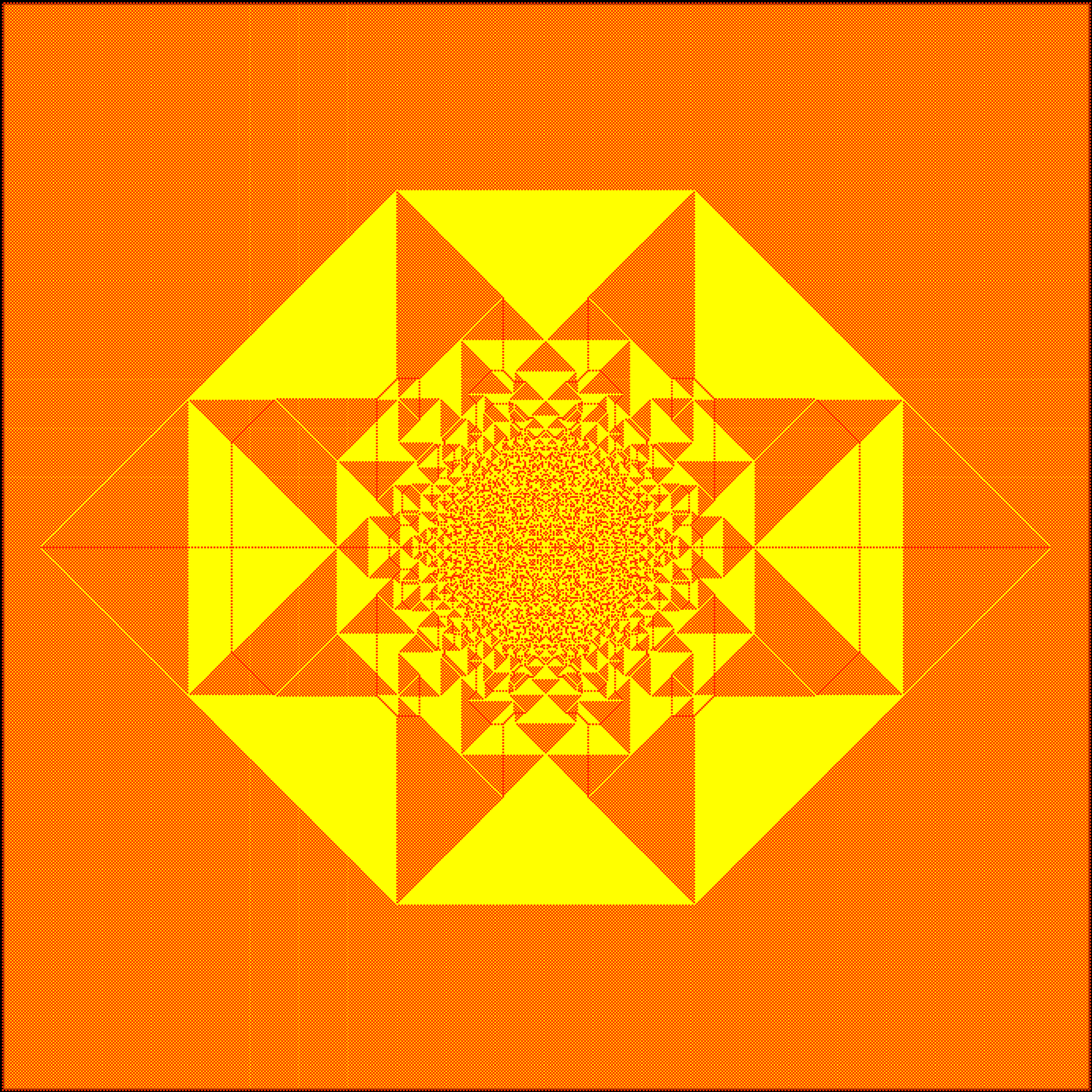}
  \caption{The stable configuration for the DASM, obtained 
  by adding $5\times10^4$ particles at one site, on the F-lattice of Fig.\ref{lattice}$a$ with initial 
  checkerboard configuration. Color code: 
  red=0, yellow=1. The apparent orange regions in the picture represent the
  patches with checkerboard configuration. (Details can be seen in the electronic version using zoom in.)}
  \label{flattice}
  \end{center}
\end{figure}

\section{Definition of the model}
The pattern on a standard square lattice is rather complicated
(Fig.\ref{asm}), and it has not been possible to characterize 
it so far. In this chapter, we consider two variations of the square
lattice, assigning orientations to 
its edges, such that each site has two inward
and two outward arrows, as shown in Fig.\ref{lattice}a and
Fig.\ref{lattice}b. They are known as F-lattice and Manhattan lattice, respectively.

We define a position vector on the lattice, $\mathbf{R}\equiv\left( x, y \right)$.
The DASM is defined on the lattice by a height variable
$z\left( \mathbf{R} \right)$, at each site $\mathbf{R}$.
In a stable configuration all sites have height $z\left( \mathbf{R} \right)<2$.
The system is driven by adding grains at a single site and if this addition
makes the system unstable it relaxes by the toppling rule: each unstable 
site transfers one grain each in the direction of its outward arrows. 
We start with an initial checkerboard configuration in which  $z\left( \mathbf{R} \right)=1$
for sites with $(x+y)=$ even, and $0$ otherwise. Clearly, the average density of
sand grains for the initial configuration is $1/2$ per site.
For numerical purpose we use a lattice large enough so that
none of the avalanches reaches the boundary.
The result of adding $N=5\times 10^4$ grains at the origin is shown in
Fig. \ref{flattice} and Fig. \ref{manhattan}, for the two lattices.

The asymptotic pattern in large $N$ limit, for the two lattices
are indistinguishable from each other, at large scales, except that the thin
lines of $1$'s forming two triangles outside the octagon are rotated by 
$45^{\circ}$ (Fig.\ref{manhattan}).  Since the lattices are different, 
this is quite intriguing. Specially there is no obvious lattice symmetry, and
it is easily checked that patterns produced for small $N$ are quite different 
(Fig.\ref{small}).

This pattern is somewhat simpler than in Fig.\ref{asm}, which makes its study easier. We shall discuss here only the F-lattice,
but the discussion is equally applicable to the Manhattan lattice. Taking some qualitative features 
of the observed pattern
(e.g. only
two types of patches are present, and they are all $3$- or $4$- sided 
polygons) as input, we show how one can get a complete and {\it 
quantitative} characterization of the pattern.  We also  show that the 
\textit{asymptotic} pattern 
has an unexpected  exact $8$-fold rotational symmetry, and determine the 
exact 
coordinates of all the boundaries in it.  We will also discuss 
some other cases, where exactly the same asymptotic pattern is obtained.

\begin{SCfigure}
 \includegraphics[scale=0.45,angle=0]{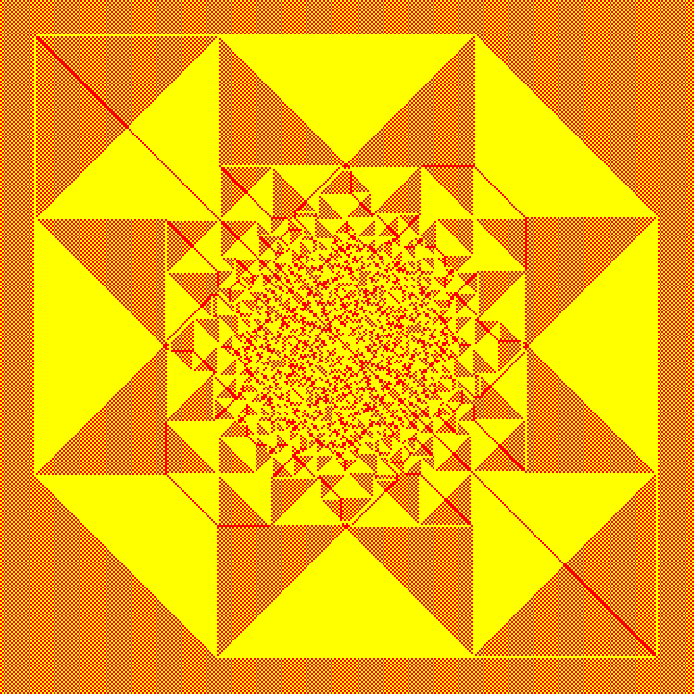}
 \caption{A stable configuration for the DASM on the 
Manhattan lattice of Fig.\ref{lattice}b, obtained by adding $25\times10^3$ 
particles at one site, with initial checkerboard configuration. Color code: red=0, yellow=1.
  (Details can be seen in the electronic version using zoom in.)}
 \label{manhattan} 
\end{SCfigure}
\begin{figure}
 \includegraphics[scale=0.35]{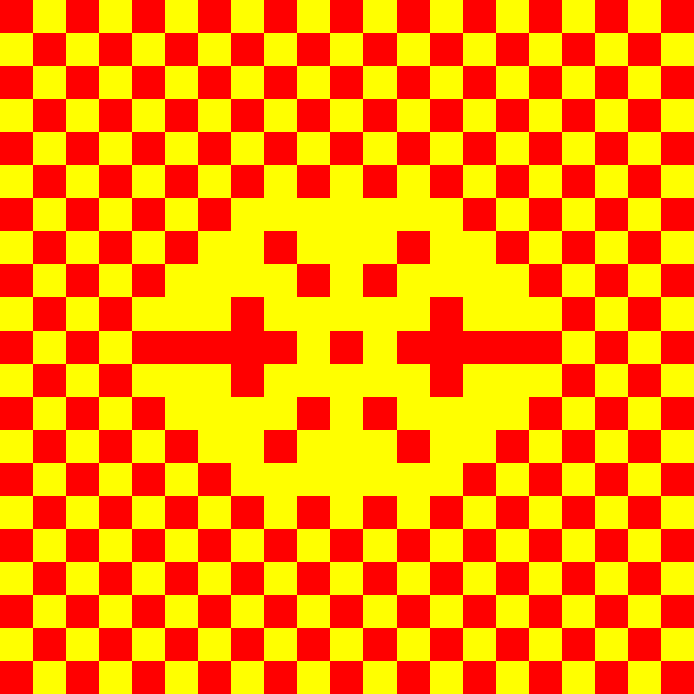}
 \includegraphics[scale=0.35]{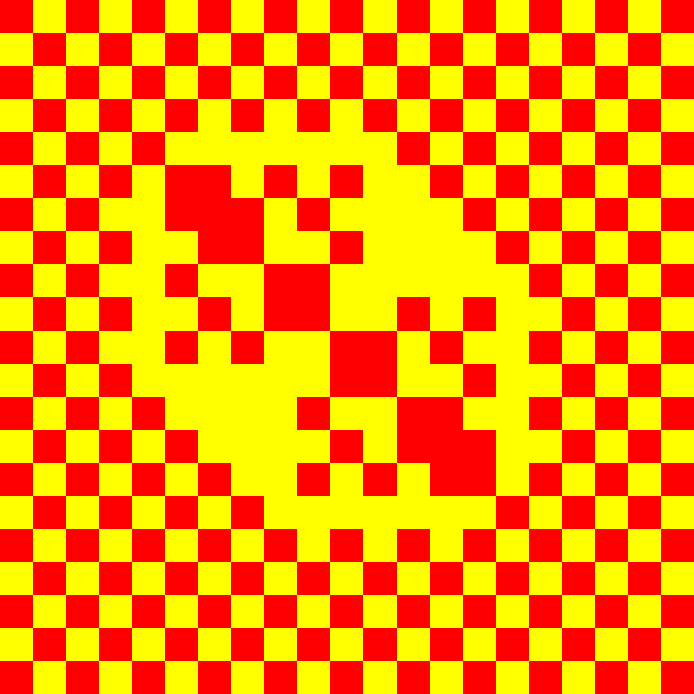}
 \caption{Very different stable configurations for the abelian sandpile model on
$(a)$ the F-lattice, and $(b)$ the Manhattan lattice, 
obtained by adding only $20$ particles at one site, with initial 
checkerboard configuration. Color code: red$=0$, yellow$=1$.}
 \label{small}
\end{figure}

\section{Characterizing asymptotic pattern: A general theory}
We first describe a general method of formally characterizing a
large number of sandpile patterns, not just the three patterns shown
till now.
We start by defining $2\Lambda\left( N \right)$ as the diameter of the
pattern when $N$ grains have been added. The exact definition of $\Lambda$ is flexible, and the
characterization does not depend on the choice. For the patterns in
Fig. \ref{asm}, \ref{flattice}, and \ref{manhattan} we choose $2\Lambda$ as the width of the
smallest rectangle that encloses all sites that have toppled at least
once.

As mentioned before, the patterns exhibit proportionate growth,
\textit{i.e.}, all structures in the pattern grows at the same rate to
the diameter. While there is as yet no rigorous proof of this
important property, we assume this in the following discussion. Then,
it is natural to describe the patterns in the reduced coordinates
defined by $\xi=x/\Lambda$ and $\eta=y/\Lambda$. A position
vector in these reduced coordinates is defined by $\mathbf{r}=\mathbf{R}/\Lambda\equiv\left( \xi, \eta \right)$.
Then in the limit $\Lambda\rightarrow\infty$, the patterns can be characterized by a
function $\rho(\mathbf{r})$ which gives the local density of sand
grains in a small rectangle of size $\delta\xi\delta\eta$
about the point $\mathbf{r}$, with $1/\Lambda\ll\delta\xi$, $\delta\eta\ll1$.
We define $\Delta\rho\left( \mathbf{r} \right)$ as the change in density $\rho\left( \mathbf{r} \right)$
from its initial background value.

A large number of sandpile patterns, including the one in Fig.
\ref{asm}, \ref{flattice}, and \ref{manhattan}, are made of a union of distinct regions,
called ``Patches'', inside which the heights are periodic in
space. Then, inside each patch $\Delta\rho\left( \mathbf{r} \right)$
is constant.  For example, in the pattern in Fig. \ref{flattice},
there are only two types of patches, and the
change in density takes only two possible values, $1/2$ in a high-density patch
(color yellow in Fig. \ref{flattice}) and $0$ in a low-density patch (color orange). There are 
few defect-lines, which move with $N$, and can also be seen in Fig.\ref{asm} and Fig.\ref{flattice}. 
But these can be ignored in discussing the asymptotic pattern.

Let $T_{\Lambda}\left( \mathbf{R} \right)$ be the number of topplings at site
the $\mathbf{R}$ when the diameter reaches the value $2\Lambda$ for the first time.
Define
\begin{equation}
\phi\left( \mathbf{r} \right)=\lim_{\Lambda\rightarrow\infty}\frac{1}{2\Lambda^{2}}T_{\Lambda}\left( \mathbf{R}' \right),
\label{phi}
\end{equation}
where $\mathbf{R}'\equiv\left(\lfloor\Lambda \xi \rfloor, \lfloor\Lambda \eta \rfloor\right)$,
with $\lfloor x \rfloor$ being the floor function which gives the largest
integer $\le x$. 
From the conservation of sand-grains in the toppling process, it is easy to see that
\begin{equation}
\sum_{\mathbf{R'}\in n.n.}T_{\Lambda}\left( \mathbf{R'}
\right)-\alpha T_{\Lambda}\left( \mathbf{R} \right)= \Delta z\left(
\mathbf{R} \right)-N \delta_{\mathbf{R},0},
\end{equation}
where the sum is over the sites nearest neighbors of
$\mathbf{R}$, and $\alpha$ is the number of them. Then in the rescales coordinate,
$\phi$ satisfies the Poisson equation
\begin{equation}
\nabla^{2}\phi\left( \mathbf{r} \right)=\Delta\rho\left( \mathbf{r} \right)-\frac{N}{\Lambda^{2}}\delta\left( \mathbf{r} \right).
\label{poisson0}
\end{equation}
In an electrostatic analogy, we can think of $\Delta\rho(\mathbf{r})$
as an areal 
charge density, and $\phi(\mathbf{r})$ as the corresponding 
electrostatic potential. \textit{A complete specification of $\phi\left( \mathbf{r} \right)$ determines
the density function $\Delta\rho\left( \mathbf{r} \right)$ which in turn
characterizes the asymptotic pattern.}

The key observation that allows us to determine the asymptotic pattern 
is the following proposition.
\begin{proposition}
Inside each patch of periodic heights, 
$\phi(\mathbf{r})$ is a quadratic function of $\xi$ and $\eta$.
\label{lemma1}
\end{proposition}

A proof can be done in the following way.
Within a patch, the function $\phi(\mathbf{r})$ is
Taylor expandable around any point $\mathbf{r}_{O}\equiv\left(
\xi_{o},\eta_{O}\right)$ inside the patch.
\begin{equation}
\phi\left( \mathbf{r}
\right)=f+d\Delta\xi+e\Delta\eta+a\left(\Delta\xi\right)^{2}+2h\Delta\xi\Delta\eta+b\left(\Delta\eta\right)^{2}+\mathcal{O}\left(
\Delta\xi^{3},\Delta\eta^{3} \right)\dots,
\end{equation}
where $\Delta\xi=\xi-\xi_{o}$ and $\Delta\eta=\eta-\eta_{o}$.
Consider any term
of order $\ge 3$ in the expansion, for example the term
$\sim(\Delta\xi)^3$. This can only arise due to a term $\sim (\Delta
x)^3/\Lambda$ in $T(x,y)$. Then, considering the fact that $T\left(
\mathbf{R} \right)$ is an integer function of the coordinates, it is easy
to see that it will 
change discontinuously at intervals of $\Delta x \sim
\mathcal{O}(\Lambda^{1/3})$.
This leads to change in the periodicity of heights at such intervals
inside each patch which themselves are of size $\sim\Lambda$. This
would then result in an infinitely many defect-lines in the asymptotic pattern. However
there are no such features in Fig.\ref{asm} or Fig.\ref{flattice}.
Therefore inside each periodic patch of constant $\Delta\rho(\mathbf{r})$,
$\phi(\mathbf{r})$ can at most be quadratic in $\xi$ and $\eta$.

The argument finally boils down to proving the two features of the
pattern, \textit{i.e.}, there is proportionate growth, and that the
pattern can be decomposed in terms of periodic patches.
\qed

In
each periodic patch the toppling function $T(x,y)$ is a sum of two 
terms: a 
part, that is a simple quadratic function of $x$ and $y$, and another
is a periodic 
part. The periodic part averages to zero, and does not contribute to the 
coarse-grained function $\phi(\mathbf{r})$ \footnote{In some patterns, 
with other backgrounds (not discussed here) there are 
regions that occupy finite fraction area of the full pattern, which show aperiodic height 
patterns. These cases are harder to analyse.}. The quadratic part,
when rescaled, can be written as
\begin{equation}
\phi\left( \mathbf{r}\right)=f+d\xi+e\eta+a\xi^{2}+2h\xi\eta+b\eta^{2},
\end{equation}
where $a,h,b,d,e,f$ are constants inside a patch, and
$a+b=\Delta\rho/2$, corresponding to the patch. Then each patch is
characterized by the values of these parameters.

Now we will show that the continuity of $\phi$ and its first
derivatives along the boundary between adjacent patches imposes linear
relations among the corresponding parameters. Consider two neighboring periodic patches ${\bf P}$ and ${\bf P'}$ 
with mean densities $\rho$ and $\rho'$ respectively. Let the rescaled quadratic 
toppling function be $Q(\mathbf{r})$ and $Q'(\mathbf{r})$ in these 
patches. Then the boundary between the patches is given by the equation 
$Q(\mathbf{r}) = Q'(\mathbf{r})$. As the derivatives of $\phi$ are also 
continuous across the boundary, the boundary between two periodic 
patches must be a straight line, and
\begin{equation} 
Q'(\mathbf{r}) = Q(\mathbf{r})+\frac{1}{2}(\rho'-\rho) l_{\perp}^2,  
\label{continuity}
\end{equation} 
where $l_\perp$ is the perpendicular distance of $(\mathbf{r})$ from the 
boundary. We can start with a periodic patch ${\bf P}$, and go to another 
patch ${\bf P'}$ by more than one path. Since the final quadratic function at 
${\bf P'}$ should be the same whichever path we take, this imposes consistency 
conditions which restrict the allowed values of slopes of the boundaries. 
Consider a point $z_0$ where $n$ periodic patches meet, with $n > 2$
(Fig.\ref{matching}). If the $j$th boundary at this point makes an angle 
$\theta_j$ with the $x$-axis, and the density of the patch in the wedge 
$\theta_j \le \theta \le \theta_{j+1}$ is $\rho_{j+1}$ (Fig.\ref{matching}), 
then the condition that the net change in the quadratic form is zero if we 
go around $z_0$ once, reduces to the following condition:
\begin{equation}
\sum_{j=1}^n  (\rho_{j+1} -\rho_{j}) e^{ 2 i \theta_j}  =0,
\end{equation}
with $\rho_{n+1}=\rho_1$.
For $n=3$, with $\rho_1 \neq \rho_2 \neq \rho_3$, this equation has
only trivial solutions with $\theta_{j}$ equal to $0$ or $\pi$ for all
$j$. Hence, only  $n \geq 4$ are allowed.
\begin{SCfigure}
 \includegraphics[scale=0.50]{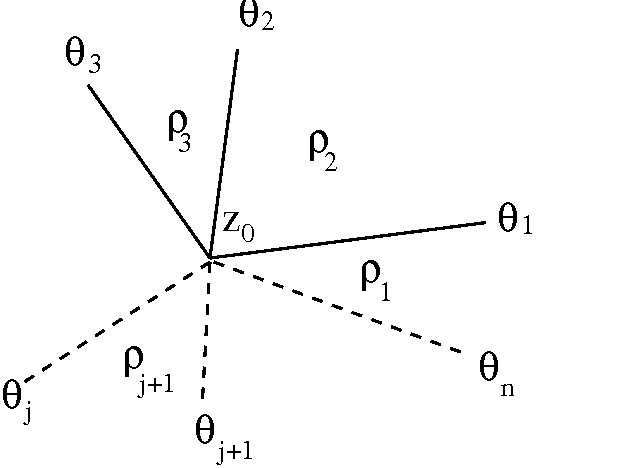}
 \caption{$n$ different periodic patches of density
 $\rho_1$,$...$,$\rho_n$ meeting at point $z_0$.}
 \label{matching}
\end{SCfigure}

These linear equations amongst the parameters corresponding to
neighboring patches, can be solved and this will determine the
complete potential function $\phi$, giving a quantitative characterization of
the pattern.

\section{Determination of the potential function}
We now apply this method, in the last section, to the F-lattice pattern
in Fig. \ref{flattice}, and determine the exact potential function $\phi\left(
\mathbf{r}\right)$. We note that in this
pattern, there are no aperiodic patches, 
only two types of periodic patches, where $\rho(\mathbf{r})$ only
takes values $1$ or $1/2$. Also, the slopes of the boundaries between patches 
only take values $0$, $\pm1$ or $\infty$. The patches are typically dart 
shaped quadrilaterals, and some triangles. These 
simplifications, not present in Fig.\ref{asm}, make possible a full 
characterization of the pattern in Fig.\ref{flattice}.

We start by determining the exact asymptotic size of the pattern. We
note from Fig.\ref{flattice} that the boundary of the pattern is an octagon
(we shall prove later that this is a regular octagon). In fact, there
are four lines of $1$'s outside the octagon. But these have zero areal
density in the limit $N\rightarrow \infty$, and do not contribute to 
$\rho(\mathbf{r})$. We will ignore these  in the following discussion.
\begin{SCfigure}
 \includegraphics[scale=0.14]{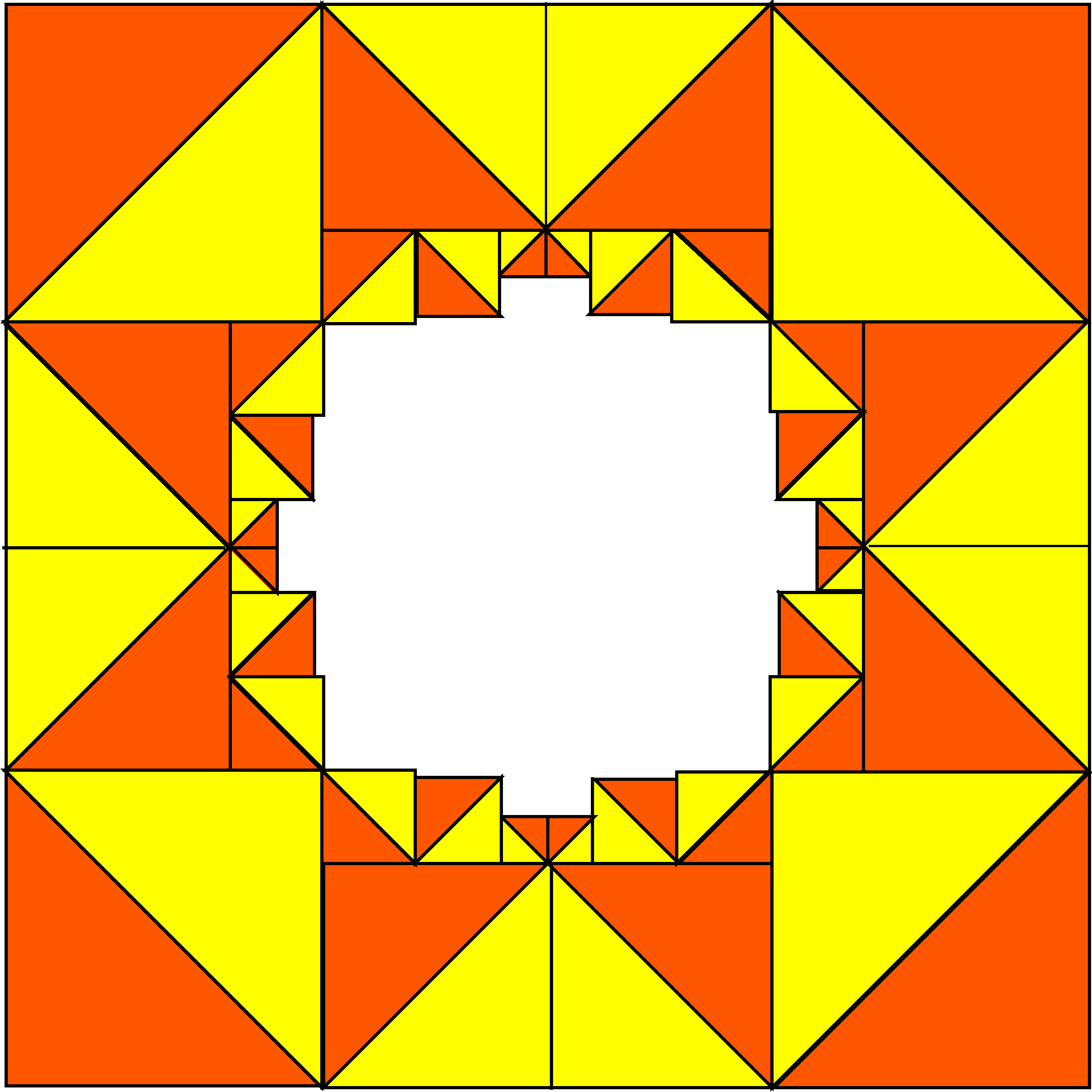}
 \caption{The pattern in 
 Fig.\ref{flattice} is obtainable by putting together square tiles of 
 different sizes. Each of the tiles is divided into two halves of 
 different density.}
 \label{tile}
\end{SCfigure}

Let $B$ be the minimum boundary square containing all
($\mathbf{r}$)
that have a non-zero charge density $\Delta\rho(\mathbf{r})$. We
observe that $B$ can be considered as a union of
disjoint smaller squares, each of which is divided by diagonal into two 
parts where $\Delta \rho(\mathbf{r})$ takes values $1/2$ and $0$
(Fig.\ref{tile}). This is seen to be true for the
outer layer patches. Towards the center, the squares are not so well
resolved. Assuming that this construction remains true all the way to the
center, in the limit of large $N$, the mean density of negative 
charge in the bounding square $=1/4$. Given that the total amount of 
negative charge is $-1$, the area of the bounding square should be $4$. 
Hence, the boundaries of the minimum bounding square are
\begin{equation}
|\xi|=1, \textrm{          } |\eta|=1.
\end{equation}
This means, with our choice of the diameter as the width of the box
$B$, we have
\begin{equation}
\Lambda\left( N \right)=\sqrt{N}+\textrm{lower order terms}.
\label{eq:boundary}
\end{equation}
In Fig. \ref{fig:bdiff}, we have shown, the correction term appears to
grow as $N^{1/4}$.
\begin{figure}
\begin{center}
\includegraphics[width=6.0cm,clip]{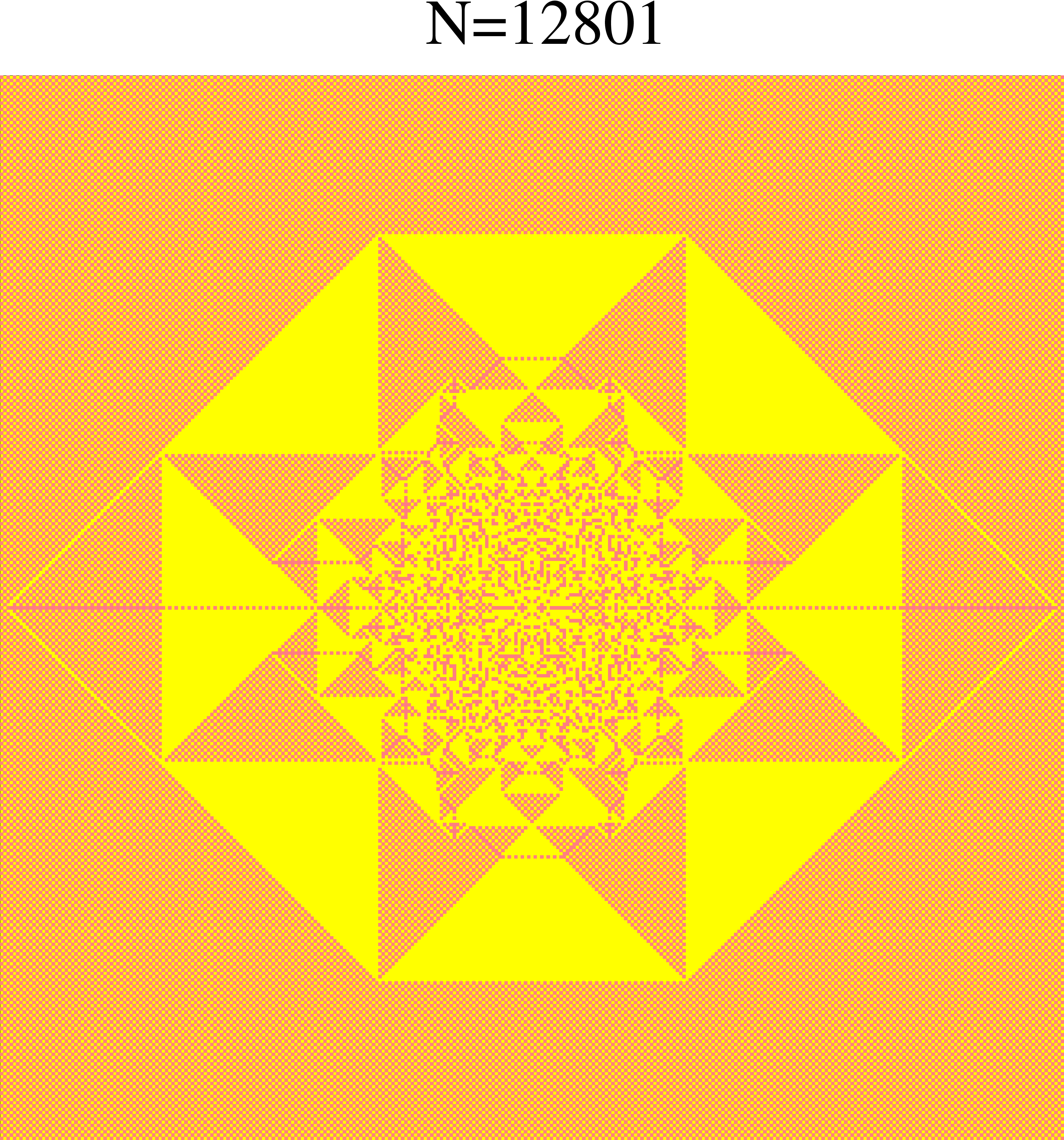}
\includegraphics[width=6.0cm,clip]{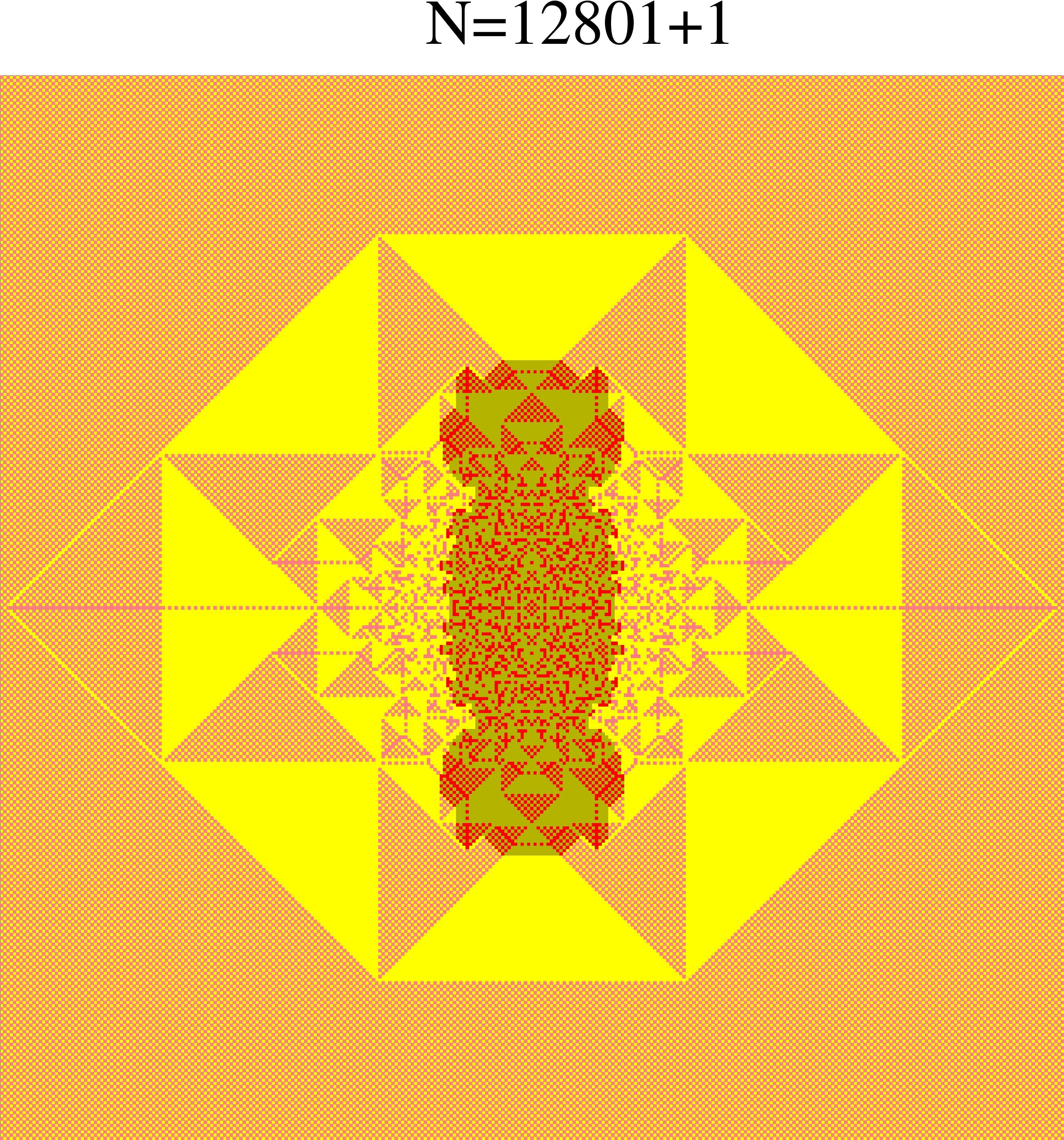}
\caption{The patterns produced on the F-lattice, by adding $N=12801$ and
$N+1$ particles. The shaded region on the second pattern represents the
toppled sites. Notice that the avalanches are stopped by some of the defect
lines in the first picture.}
\label{succ}
\end{center}
\end{figure}

Most of the time the avalanches does not
reach the boundary. They are often stopped by the defect-lines inside the patches,
which breaks the periodicity of the heights. For example, there are
lines of alternating $1$'s and $0$'s inside the dense (all $1$)
patches. When an avalanche enters the patch, the defect line shifts
its position, partially increasing the size of the patch. An example
of such event is shown in the Fig.\ref{succ}. Because of this, the diameter
increases in steps with the increase of $N$ (see
Fig.\ref{fig:boundary}).
\begin{SCfigure}
\includegraphics[width=9.0cm,clip]{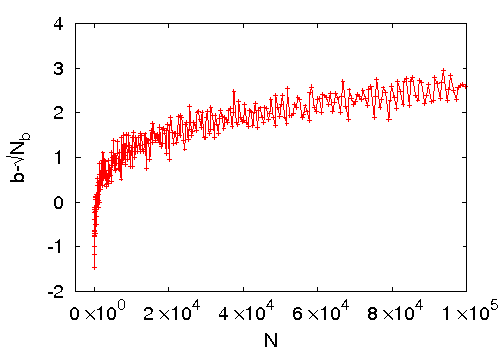}
\caption{Difference of the diameter to $\sqrt{N}$ is less than
$3$ for $N$, at least, up to $10^{5}$.}
\label{fig:bdiff}
\end{SCfigure}
\begin{SCfigure}
\includegraphics[width=9.0cm,clip]{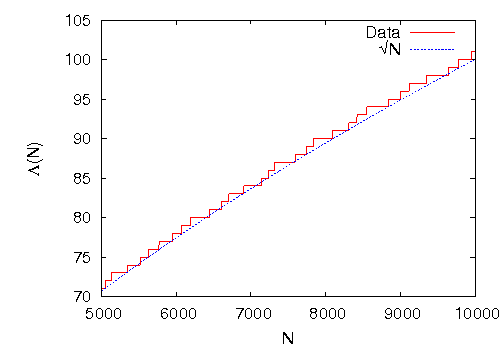}
\caption{$\Lambda$ for the F-lattice pattern (Fig. \ref{flattice}) as
a function of $N$, in the range $5000\le N \le 10000$.}
\label{fig:boundary}
\end{SCfigure}

Let $N_b$ be the minimum number of particles that have to be added so
that at least one site at $y=b$ topples. We find that for $b=10$,
$50$, $100$, and $300$, $\sqrt{N_b}=10.770$, $49.436$, $98.894$ and
$297.798$.
This is consistent with Eq. (\ref{eq:boundary}).

Then, the Poisson equation, in Eq. (\ref{poisson0}), for this pattern becomes,
\begin{equation}
\nabla^{2}\phi\left( \mathbf{r} \right)=\Delta\rho\left( \mathbf{r}
\right)-\delta\left( \mathbf{r} \right),
\label{poisson1}
\end{equation}
\textit{i.e.}, there is unit amount of point charge at the origin.

\begin{figure}
 \begin{center}
 \includegraphics[scale=0.70]{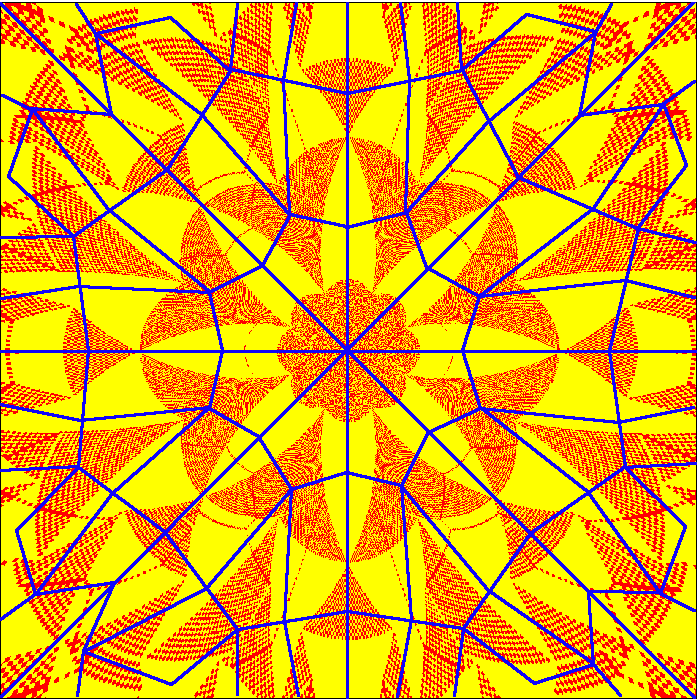}
 \caption{$1/r^2$ transformation of the pattern in Fig. \ref{flattice}. Blue lines
are drawn between the patches if they are neighbor of each other, with triangular patches 
considered as degenerate quadrilaterals.}
 \label{adj}
 \end{center}
\end{figure}
\begin{figure}
 \begin{center}
 \includegraphics[scale=0.35]{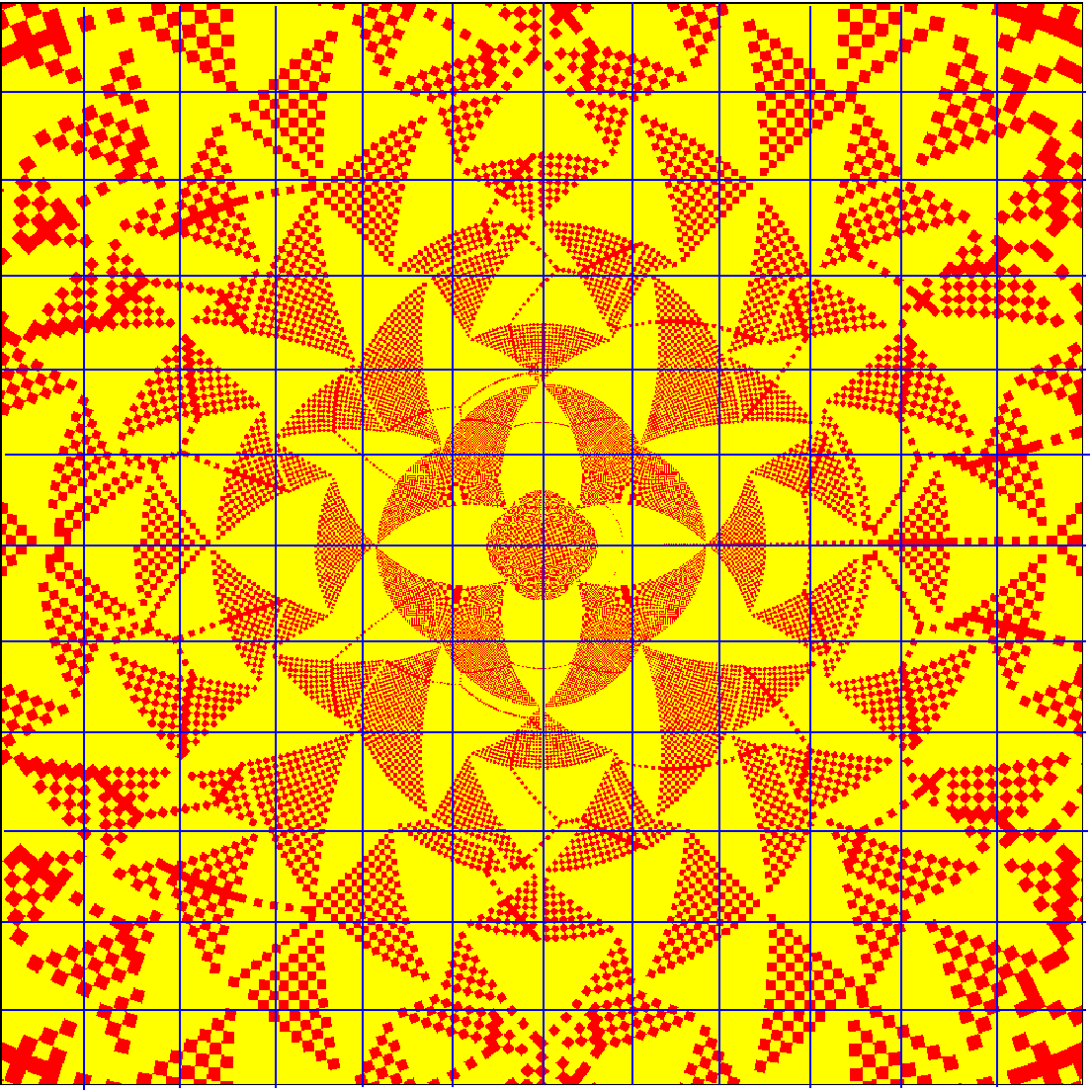}
 \caption{$1/z^2$ transformation of the section of the pattern in Fig.
\ref{flattice}, on the half-plane corresponding to positive values of
$\eta$. The regions along the positive and negative $\xi$-axis are glued together in the
transformed picture. The blue lines are drawn in a way similar those in the
previous picture, and they form a square grid.}
 \label{zsqr}
 \end{center}
\end{figure}
We now determine the parameters in the quadratic form of the potential
function. In order to do that in a consistent way, we first
look at the topological structure of the pattern. We note 
that the patches become smaller, and there are more of them in number, as 
we move towards the center. One can use a coordinate
transformation $r'= 1/r^2$,  $\theta' =\theta$ to avoid this 
overcrowding (Fig.\ref{adj}). We can now draw the  adjacency graph (Fig.\ref{fig5}$a$) 
of the pattern, where each vertex 
denotes a patch, and a bond between the vertices is drawn if the 
vertices share a common boundary. It is convenient to think of the 
triangular patches in the pattern as degenerate quadrilaterals, with one 
side of length zero. Then we see that the adjacency graph is planar with 
each vertex of degree four, except a single vertex of coordination 
number eight corresponding to the exterior of the pattern. The graph has 
the structure of a square lattice wedge of wedge angle $4\pi$. The 
square lattice structure of the adjacency graph is seen more clearly, if 
rather than $1/r^2$ transformation, the transformation used is  $z' = 1/z^2$  ( 
this has been used earlier in 
\cite{ostojic}), where $z=\xi+ i \eta$, and view it in the complex 
$z'$-plane (see Fig.\ref{zsqr}). Thus, one can equivalently represent the graph as a square 
grid on a Riemann surface of two sheets (fig.\ref{fig5}$b$).

\begin{figure}
\begin{center}
    \includegraphics[scale=0.21]{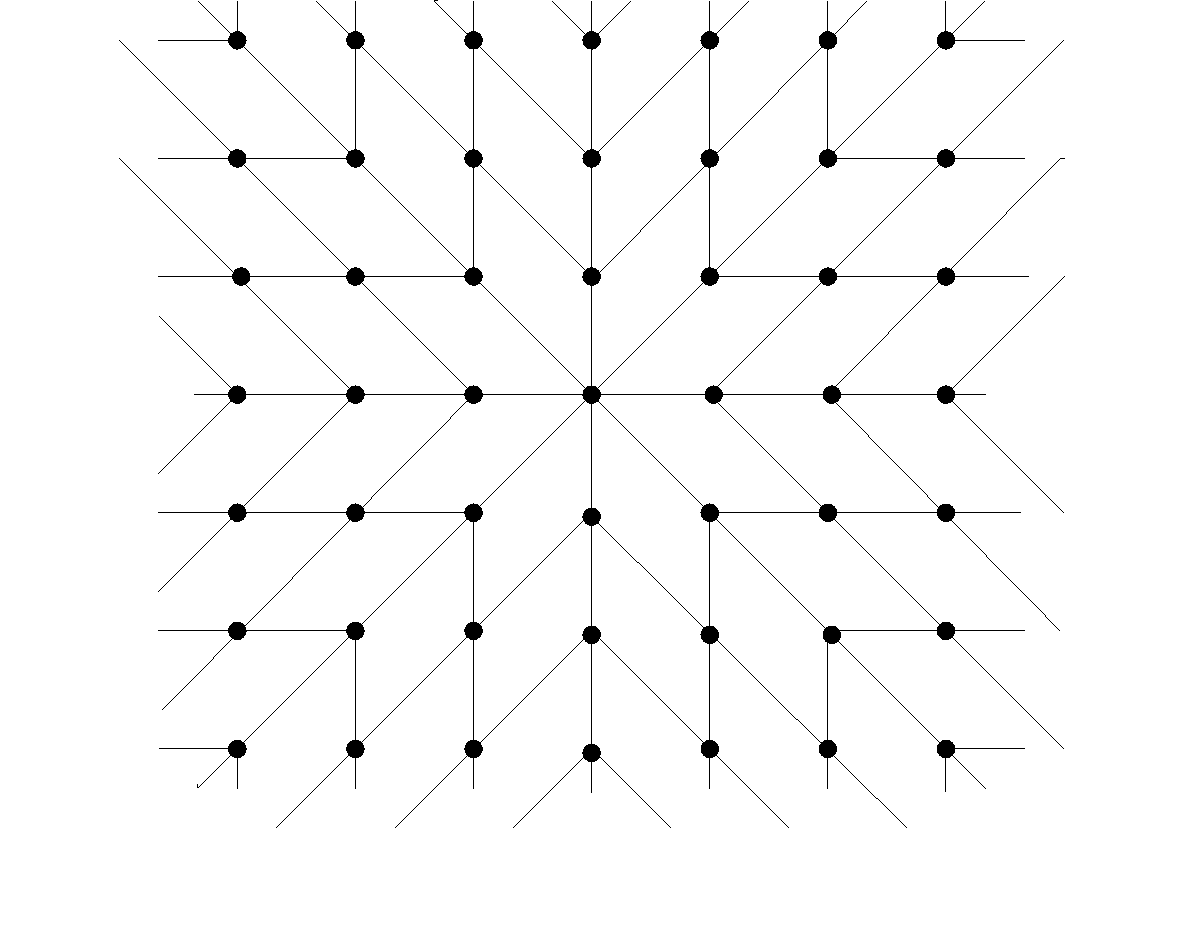}    
    \includegraphics[scale=0.35]{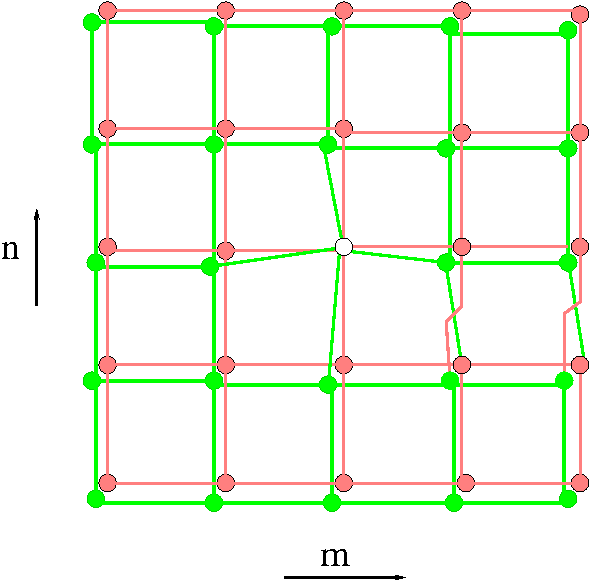} 
    \caption{Two representations of the adjacency graph of the pattern. 
    Here the vertices are the patches, and the edges connect the adjacent
    patches.
    $(a)$ Representation as a planar graph, $(b)$ as a graph of wedge of angle 
    $4 \pi$ formed by glueing together the eight quadrant graphs at the
    origin.}
    \label{fig5}
\end{center}
\end{figure}

We now use the qualitative information obtained from the adjacency 
matrix of the observed pattern, to obtain quantitative prediction of the 
exact coordinates of all the patches. Consider an arbitrary patch ${\bf P}$, 
having an excess density $1/2$. The potential function in this patch is a 
quadratic function of $(\xi,\eta)$ and we parametrize it as
\begin{eqnarray}
\phi_\p(\mathbf{r}) &=& \frac{1}{8}(m_\p+1)\xi^2 + \frac{1}{4}n_\p\xi\eta
+\frac{1}{8}(1-m_\p)\eta^2 \nonumber \\
&&+ d_\p\xi + e_\p \eta + f_\p.
\label{eq:odd}
\end{eqnarray}
The potential function in another patch ${\bf P}$ having zero excess density 
is parametrized as
\begin{equation}
\phi_\p(\mathbf{r})=\frac{1}{8}m_\p(\xi^2-\eta^2) + \frac{1}{4}n_\p\xi\eta
 + d_\p\xi + e_\p\eta + f_\p.
\label{eq:even}
\end{equation}
Now consider two neighboring patches ${\bf P}$ and ${\bf P'}$ with   
excess densities  $1/2$ and $0$ respectively. Then using the matching 
conditions (see Eq. (\ref{continuity})), it  is easy to show that if the boundary between  them 
is a horizontal line $\eta =\eta_{\p}$, we must have 
\begin{eqnarray}
m_{\p'} &=& m_\p+1 \textrm{,   }n_{\p'} = n_\p \textrm{,   } d_{\p'} = d_\p,
\nonumber \\
e_{\p'} &=& e_\p + \eta_\p/2 \textrm{,   } f_{\p'} = f_\p - {\eta_\p}^2/4.
\label{a2}
\end{eqnarray}

There are similar conditions for other boundaries. These result a 
coupled set of linear equations for the coefficients $\{m_\p, n_\p, 
d_\p,e_\p,f_\p\}$. The equations for $m_\p$ and $n_\p$ do not involve 
other variables. In the outermost patch, clearly $\phi(\mathbf{r})=0$, and for this patch
both $m$ and $n$ are zero. It follows that $m_\p$ and 
$n_\p$ are integers, equal to the Cartesian coordinates of the 
vertex corresponding to the patch 
${\bf P}$ in  the discretized Riemann surface in Fig.\ref{fig5}b.  In the
following, we denote a patch by integers
$(m,n)$, and write the corresponding coefficients $d_\p$, $e_\p$, and
$f_\p$ as $d_{m,n}$, $e_{m,n}$ and $f_{m,n}$.   
With this convention,
the matching conditions in Eq.(\ref{a2}) can be rewritten as 
\begin{equation}
d_{m+1,n}=d_{m,n} \textrm{, } e_{m+1,n}-e_{m,n}=\eta_{m,n}/2 \textrm{, $(m+n)$ odd}.
\label{10}
\end{equation}
Using similar matching conditions for the boundary
of patch $(m$, $n)$ with slope $\pm1$, we get the conditions
\begin{eqnarray}
d_{m,n+1}-d_{m,n} &=& e_{m,n} -e_{m,n+1},\textrm{    }(m+n)\rm{~odd},
\nonumber \\
d_{m,n-1}-d_{m,n} &=& e_{m,n-1} - e_{m,n},\textrm{    }(m+n)\rm{~odd}.
\label{11}
\end{eqnarray}
We can eliminate the variables $d_{m,n}$ and $e_{m,n}$ with $(m+n)$
even using Eq. (\ref{10}) and Eq.(\ref{11}). Then the equations become 
\begin{eqnarray}
e_{m+2,n}-e_{m,n} &=& \eta_{m,n}/2, \label{a6}\\
d_{m-2,n}-d_{m,n} &=& \xi_{m,n}/2, \label{a7}\\
d_{m-1,n-1}-d_{m,n} &=& e_{m+1,n-1}-e_{m,n}, \label{a8}\\
d_{m-1,n+1}-d_{m,n} &=& -[e_{m+1,n+1}-e_{m,n}]. \label{a9} 
\end{eqnarray}

It is convenient to introduce the complex variables $z = \xi + i \eta$,
$ M = m + i n$ and $D = d + i e$. In these variables we can write
the potential function, in Eq. (\ref{eq:odd}) and (\ref{eq:even}), as
\begin{equation}
\phi(z) =  \frac{1}{8} z \bar{z} + \frac{1}{8} Re[ z^2 \bar{M} + \bar{D} 
z] +f,
\end{equation}
where overbar denotes complex conjugation.

On the $(m,n)$ lattice, with $(m + n)$ odd, the natural basis vectors are 
$(1,1)$ and $(1,-1)$. Let us call these $\alpha$ and $\beta$. We define 
the finite difference operators  $\Delta_{\pm \alpha}$ and $\Delta_{\pm 
\beta}$ by 
\begin{eqnarray}
\Delta_{\pm \alpha} f(z) = f( z \pm \alpha) - f(z),\nonumber\\
\Delta_{\pm \beta} f(z) = f(z \pm \beta) -f(z).
\end{eqnarray}
Then the equations (\ref{a6}-\ref{a9}) can be written as
\begin{eqnarray}
\Delta_{-\alpha} d = \Delta_{\beta} e, \nonumber \\
\Delta_{-\beta} d = -\Delta_{\alpha} e.
\end{eqnarray}

These equations are the discrete analog of the familiar Cauchy-Riemann
conditions connecting the partial derivatives of real and imaginary 
parts of an analytic function where the role of the analytic function is 
played by   $ D = d + i e$.   

From Eq.$(\ref{a6})$ and Eq.(\ref{a9}), it is easy to deduce that $D$   
satisfies  the  discrete Laplace's equation
\begin{equation}
[\Delta_{\alpha} \Delta_{-\alpha} +\Delta_{\beta} \Delta_{-\beta}] D =0.
\label{a10}
\end{equation}

If $m$ and $n$ are large, the corresponding patch is near the origin
($|\xi|+|\eta|$ is small), and where the leading behavior of 
$\phi(\mathbf{r})$ is given by 
$\tilde{\phi}(\mathbf{r}) \sim -\frac{1}{4\pi}\log(\xi^2+\eta^2)$
(see Eq. \ref{poisson1}).
Consider a point $z_0$, such that at $z_0$
\begin{equation}
\partial^2 \tilde{\phi}/\partial \xi^2 \approx m/4;~~~~~
\partial^2 \tilde{\phi}/\partial \xi \partial \eta \approx n/4.
\end{equation}
Then, $z_0$ would be expected to lie in the patch labeled by $(m,n)$.
This gives $z_0 \approx \pm ( \pi \bar{M}/2)^{-1/2}$. Then, setting
$\partial \tilde{\phi}/\partial z$ equal to $\bar{D}/2$ gives us
\begin{equation}
D_{m,n} \simeq \pm \frac{1}{\sqrt{2\pi}} \sqrt{m + i n}.
\label{a12}
\end{equation}
The equation (\ref{a10}), subjected to the behavior at large
$|m|+|n|$ given by Eq.(\ref{a12}) on the $4 \pi$-wedge graph 
(for each value of $(m,n)$, $D_{m,n}$ has two values) has a unique
solution. Clearly the solution has eight fold rotational symmetry
about the origin in the $(m, n)$ space. This implies that 
\begin{equation}
D_{-n,m}=i^{1/2} D_{m,n}; \rm{~for~all~} (m,n).
\label{a22} 
\end{equation}
Given $D_{m,n}$, its real and imaginary parts determine $d_{m,n}$ and
$e_{m,n}$, and using Eq.(\ref{a6}, \ref{a7}) we determine the exact
positions of all the patch corners. 
The exact eight-fold rotational symmetry of the adjacency graph of the 
pattern, and the fact that $D$ satisfies Eq.(\ref{a22}) on the adjacency graph 
together imply the eight-fold rotational symmetry of all the distances 
in the pattern. 

Note that for the usual square lattice, the solution of Eq.(\ref{a10})
is the well known 2-dimensional lattice Greens function, that is explicitly 
calculable for any finite $(m,n)$, and is a simple polynomial of $1/\pi$ with rational coefficiants 
\cite{spitzer}. However, for our case of the two sheeted Riemann surface, 
we have not been able to find a closed-form formula for $D_{m,n}$.
But the solution can be determined 
numerically to very good precision by
solving it on a finite grid $-L \le m$, $n\le L$ with the condition in
Eq.(\ref{a12}) imposed exactly at the boundary. We determined $d_{m,n}$ 
and $e_{m,n}$ numerically for $L=100,200,400$, and extrapolated our 
results for $L\rightarrow\infty$. We find $d_{1,0}= 0.5000 $ and 
$d_{2,1} = 0.6464$, in perfect agreement with the exact theoretical 
values $1/2$ and $1- 1/2\sqrt{2}$, respectively, determined using the
Fig. \ref{tile}.
\begin{figure}
\begin{center}
    \includegraphics[scale=0.7]{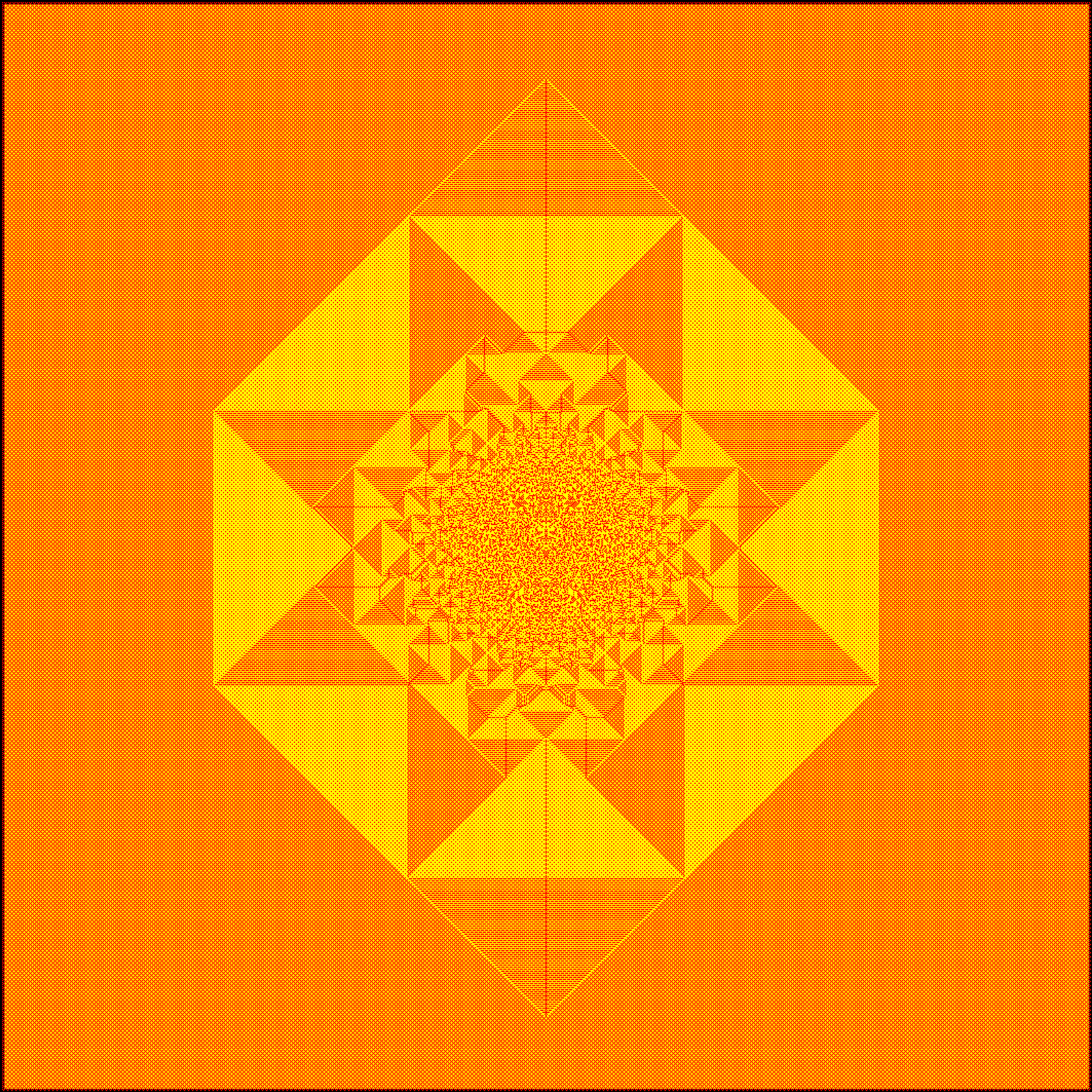} 
    \caption{The stable configuration for the abelian sandpile model on F-lattice, 
obtained by adding $5\times10^4$ particles at one site, initial configuration
with average height $5/8$. Color code: red=0, yellow=1. 
(Details can be seen in the electronic version using zoom in.)}
    \label{f5by8}
\end{center}
\end{figure}
\begin{figure}
\begin{center}
    \includegraphics[scale=0.7]{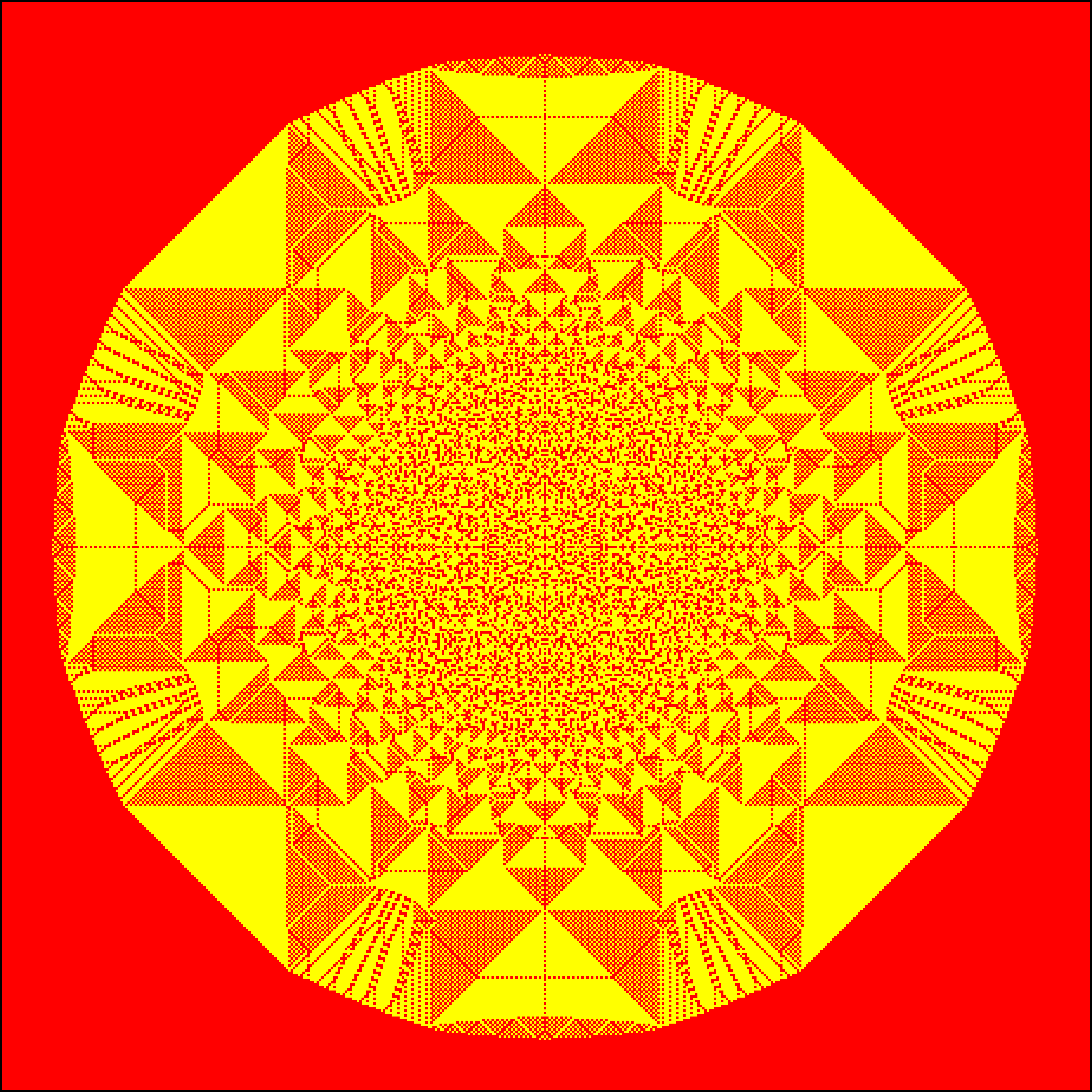} 
    \caption{The stable configuration for the abelian sandpile model on F-lattice, 
obtained by adding $10^5$ particles at one site, 
initial configuration with all heights $0$. Color code: red=0, yellow=1. 
(Details can be seen in the electronic version using zoom in.)}
    \label{fzero}
\end{center}
\end{figure}

An interesting question to ask is, what is the size distribution of
the patches? This can be easily determined from the $1/z^{2}$
transformation of the pattern, as shown in Fig.\ref{zsqr}. In this
representation, the patches are positioned around the sites of the square grid, and all the patches
have similar sizes. The Jacobian $J(\xi,\eta)$ of the coordinate transformation is an estimate
of the area on the $\xi$-$\eta$ plane, corresponding to a unit area around
$(m,n)$ site of the grid on the $\xi'$-$\eta'$ plane. For large $m$ and $n$, it can be shown that
\begin{equation}
J(\xi,\eta)\sim |\xi+i\eta|^{3}\sim \frac{1}{|m+in|^{3/2}}.
\end{equation}
On the other hand, from the same picture, the number of patches of
area larger than or equal to the area of the patch $(m,n)$ is equal to the
number of sites inside the circle of radius $|m+in|$, on the square
grid. This number increases linearly with $ |m+in|$. Then, the number of patches of
area greater than or equal to $A$ is $\sim A^{-2/3}$, which in turn
implies that the number of patches of area lying between $A$ and
$A+dA$ would vary as $A^{-5/3}dA$.
\section{Other patterns}
Our calculations above can be easily extended to the patterns on any other
background, on any other two dimensional lattices, 
so long as there are only patches with two values of $\Delta \rho$. The
matching conditions along the patch boundaries, ensure that the
boundaries are straight lines. Although, we do not have a complete understanding
of what determines the slope of the patch boundaries and the number of them
for a patch, it seems, that as long as there are only two types of
patches, the patches are always quadri-laterals, and the slopes of
boundaries are integer multiples of $\pi/4$. Then our analysis shows
that the asymptotic pattern is same as the one for the F-latiice. For example,
as mentioned already, this is true for the Manhattan lattice (Fig.\ref{manhattan}), for initial 
density $1/2$.  Same happens for the the F-lattice itself, with a different periodic 
background of initial density 5/8 ($z_{i,j}=1$ if $i+j$ 
even, or $(i,j)$ congruent to $(0$, $1)$ or $(2$, $3)$ mod $4$). The
pattern for this case is shown in (Fig.\ref{f5by8}). In this case,
only the density of patches are different from the one on the
checkerboard background, but the patch boundaries for the asymptotic
pattern, in the rescaled coordinate, are at the identical positions.

In some other cases, like the F-lattice, with initially all sites empty, the pattern 
is very similar, but there are some aperiodic patches in the outermost 
ring (Fig.\ref{fzero}). Since the behavior of $\phi(\mathbf{r})$ in
such patches are not known, the equations for $D_{m,n}$ do not close in this case. 

Finally, how much of this analysis applies to the pattern in Fig.
\ref{asm}? As noted in \cite{ostojic}, there are large number,
possibly
infinitely many periodic patches in the asymptotic pattern.
Characterization of such patterns remains an interesting open problem.

\chapter{Effect of multiple sources
and sinks on the growing sandpile pattern\label{ch3}}
\textit{Based on the paper \cite{myjsp2}} by Tridib Sadhu and Deepak Dhar.

\begin{itemize}
\item[\textbf{Abstract}]
In this chapter, we study the effect of sink sites on DASM patterns, discussed in chapter
\ref{ch2}. Sinks
change the scaling of the diameter of the pattern with the number $N$ of sand
grains added. For example, in two dimensions, in the presence of a sink site, the
diameter of the pattern grows as $\sqrt{(N/\log N)}$ for large $N$, whereas it
grows as $\sqrt{N}$ if there are no sink sites. In the presence of a line of sink
sites, this rate reduces to $N^{1/3}$. We determine the growth rates for various sink geometries along with the case when there are
two lines of sink sites forming a wedge, and generalizations to higher dimensions.
We characterize the asymptotic pattern in the large $N$ limit for one such case, the two-dimensional F-lattice
with a single source adjacent to a line of sink sites. The characterization is done in terms of the positions of different
spatial features in the pattern. For this lattice, we also provide an exact
characterization of the pattern with two sources, when the line joining them is along one of the axes of the lattice. 
\end{itemize}

\section{Introduction}
In the previous chapter, we studied growing sandpiles in the abelian model on
the F-lattice and the Manhattan lattice.
We were able to characterize the pattern corresponding to the initial configuration in which each alternate
site of the lattice is occupied, forming a checkerboard pattern. The full characterization of this
pattern reveals an interesting underlying mathematical structure,
which seems to deserve further exploration. This is what we do in
this chapter by adding sink sites or multiple sources.  

The presence of sink sites changes the pattern in interesting ways. In
particular, it changes how different spatial lengths in the pattern
scale with the number of added grains $N$. For example, in the absence of
sink sites, the diameter of the pattern grows as $\sqrt{N}$ for large
$N$, whereas in the presence of a single sink site,
this changes to a $\sqrt{N/\log{N}}$ growth. If there is a line of
sink sites next to the site of addition, the growth rate is $N^{1/3}$.
We also study the case in which the source site is at the corner of a
wedge-shaped region of wedge angle $\omega$, where the wedge boundaries
are absorbing. We show that for any $\omega$ the pattern grows as
$N^{\alpha}$, with $\alpha=\omega/(\pi+2\omega)$. This analysis
is extended to other lattices with different initial
height distributions, and to higher dimensions.

We also study the exact characterization of the asymptotic pattern in the infinite $N$ limit for
the pattern with a line of sink sites. For a single point source, as
discussed in chapter \ref{ch2},
the determination of the different distances in the pattern requires a solution
of the Laplace equation on a discrete Riemann surface of two-sheets.
Interestingly, for the pattern with a line sink, we still have to solve the
discrete Laplace equation, but the structure of the Riemann surface
changes from two-sheets to three-sheets.
 
We then study the effect on the pattern of having multiple sites of addition.
For multiple sources, the pattern of small patches near each source is not
substantially different from a single-source pattern, but some rearrangements
occur in the larger outer patches. Two patches may sometimes join into one,
or, conversely, a patch may break up into two. While the number of patches
undergoing such changes is finite, the sizes and positions of all the patches are
affected by the presence of the other source, and we show how these changes
can be calculated exactly for the asymptotic pattern.

This chapter is organized as follows. In Section \ref{ch3.2}, we discuss scaling of the diameter of the patterns with $N$ for different
sink geometries. First, we consider the pattern in the presence of a line of sink sites. Then,
this analysis is extended to other sink geometries: two intersecting line
sinks in two dimensions and two or three intersecting planes of sink sites in three
dimensions. The problem of a single sink site is a bit different from the others, and
is discussed separately in Section \ref{ch3.3}. In Section
\ref{ch3.4}, we numerically verify the
growth rates. The remaining sections are devoted to
a detailed characterization of some of these patterns. In Section
\ref{ch3.5}, we
characterize the pattern in the presence of a line sink. In Section
\ref{ch3.6},
we discuss the case when there are two sources present. These analytical
calculations for the metric properties of the asymptotic pattern are
compared in Section \ref{ch3.7}, with the measured values for the patterns with finite but large $N$.
Section \ref{ch3.8}, contains a summary and some concluding remarks.

\section{Rate of growth of the patterns}\label{ch3.2}
For the single source pattern, discussed in chapter \ref{ch2}, the diameter $2\Lambda\left( N \right)\simeq2\sqrt{N}$,
for large $N$. We want to study how this dependence gets modified in the presence
of sink sites.

First, consider the pattern formed by adding sand grains at a single site in
the presence of a line of sink sites. In the rest of this chapter we
will use the same notations defined in chapter \ref{ch2}. Any grain reaching
a sink site gets absorbed, and is removed from the system. For simplicity let us consider the
source site at $\mathbf{R}_{o}\equiv\left(x_{o}, 0  \right)$ and the sink sites along
the $y$-axis. A picture of the  pattern produced by adding $14336000$ grains at
$\left( 1, 0 \right)$ is shown in Fig.\ref{fig:lsone}.
\begin{figure}
\begin{center}
\includegraphics[width=18cm,angle=-90]{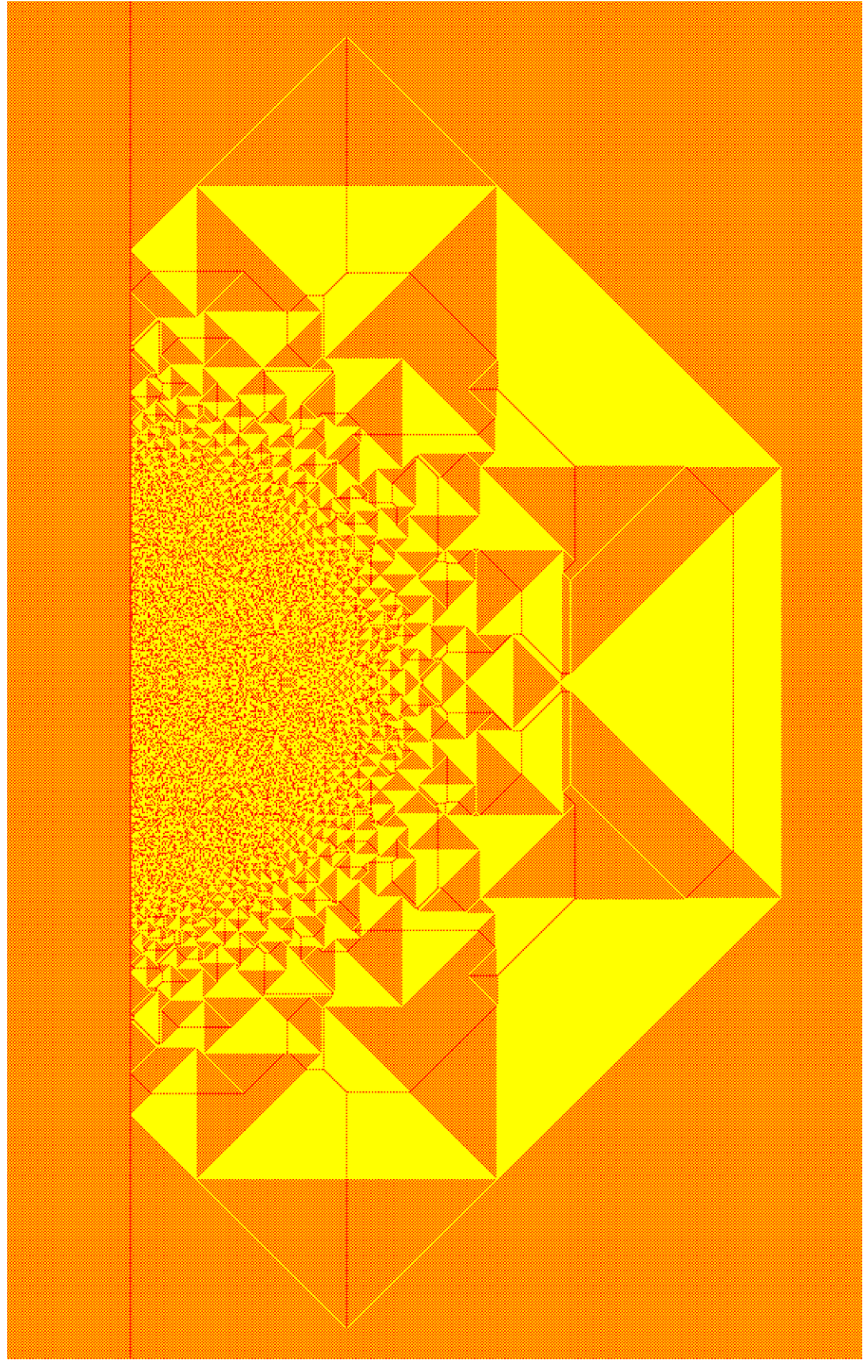}
\caption{ Pattern produced by adding grains at a single site adjacent to a line
of sink sites. Color code: red=0 and yellow=1. Apparent orange regions in the picture represent
patches with checkerboard configuration (Zoom in for details in the
electronic version).}
\label{fig:lsone}
\end{center}
\end{figure}

The equation analogous to Eq. (\ref{poisson0}) for this problem is
\begin{equation}
\nabla^{2}\phi\left( \mathbf{r} \right)=\Delta\rho\left( \mathbf{r} \right)-\frac{N}{\Lambda^{2}}\delta\left( \mathbf{r} - \mathbf{r}_{o} \right),
\label{poisson}
\end{equation}
for all $\mathbf{r}$ in the right-half plane with $\xi > 0$,
where $\mathbf{r}_{o}$ is the position of the source in reduced
coordinates. Also, as there is no toppling at the sink sites, $\phi$
must satisfy the boundary condition
\begin{equation}
\phi\left( \mathbf{r} \right)=0 \textrm{~ ~ for~all~}\mathbf{r}\equiv
\left(0, \eta \right).
\label{bc}
\end{equation}

We can think of $\phi$ as the potential due to a point charge $N/\Lambda^{2}$ at $\mathbf{r}_{o}$
and an areal charge density $-\Delta\rho\left( \mathbf{r} \right)$, in the presence
of a grounded conducting line along the $\eta$-axis. This problem can be solved
using the well-known method of images in electrostatics. Let $\mathbf{r}'$ be
the image point of $\mathbf{r}$ with respect to the $\eta$-axis.
Define $\Delta\rho\left( \mathbf{r} \right)$ in the left half plane as
\begin{equation}
\Delta\rho\left( \mathbf{r}' \right)=-\Delta\rho\left( \mathbf{r} \right).
\label{imden}
\end{equation}
Then the Poisson equation for this new charge configuration is
\begin{equation}
\nabla^2\phi(\mathbf{r})=\Delta\rho(\mathbf{r})-\frac{N}{\Lambda^2}\delta\left( \mathbf{r} - \mathbf{r}_{o} \right)+\frac{N}{\Lambda^2}\delta\left( \mathbf{r} - \mathbf{r'}_{o} \right).
\label{poisson4}
\end{equation}
As the function $\Delta\rho\left( \mathbf{r} \right)$ is odd under reflection,
$\phi$ automatically vanishes along the $\eta$-axis.

We define $N_{r}$ as the number of sand grains that remain unabsorbed. Then
\begin{equation}
N_{r}=\sum_{x>0}\sum_{y}\Delta z\left( x, y \right),
\label{nr1}
\end{equation}
where $\Delta z\left(x, y  \right)$ is the change in the height variables
between its values before and after the system relaxes. Clearly, for large $\Lambda$, we can
write
\begin{equation}
N_{r}\simeq\Lambda^{2}\int_{\mathbb{H}}d\tau\Delta\rho\left( \mathbf{r} \right),
\label{nr2}
\end{equation}
where $d\tau=d\xi d\eta$ is the infinitesimal area around $\mathbf{r}\equiv \left( \xi, \eta \right)$.
The integration is performed over the right half-plane $\mathbb{H}$ with $\xi>0$.
We shall use the sign $\simeq$ to denote equality up to leading order in $\Lambda$. 
Since $\Delta\rho\left( \mathbf{r} \right)$ is a non-negative bounded function,
exactly zero outside a finite region, this integral exists. Let its value be $C_{2}$, then we have
\begin{equation}
N_{r}\simeq C_{2}\Lambda^{2}.
\label{nr3}
\end{equation}

Let $N_{a}$ denote the number of grains that are absorbed by the sink
sites. Then considering that the grains can reach the sink sites only by toppling at its neighbors we have
\begin{equation}
N_{a}\simeq\frac{1}{2}\sum_{y}T_{\Lambda}\left( 1, y \right).
\label{na1}
\end{equation}
The factor $1/2$ comes from the fact that in the F-lattice, only half of the sites
on the column $x=1$ would have arrows going out to the sink sites. Then using
our scaling ansatz in equation (\ref{phi}), for $\Lambda$ large,
\begin{equation}
T_\Lambda\left( 1, y \right)\simeq 2 \Lambda  \left.\frac{\partial \phi}{\partial \xi}\right \arrowvert_{\xi=0}.
\label{T1}
\end{equation}
Hence
\begin{equation}
N_{a}\simeq\Lambda^{2}\int_{-\infty}^{\infty}d\eta\left.\frac{\partial \phi}{\partial \xi}\right \arrowvert_{\xi=0}.
\label{na2}
\end{equation}
Now from equation (\ref{poisson4}) the potential $\phi$ can be written as the sum of two
terms: $\phi_{dipole}$ due to two point charges $N/\Lambda^{2}$ and $-N/\Lambda^{2}$
at $\mathbf{r}_{o}\equiv\left( \xi_{o}, 0 \right)$ and its image point
$\mathbf{r}'_{o}\equiv\left( -\xi_{o}, 0 \right)$ respectively, and the term 
$\phi_{rest}$ due to the areal charge density.
\begin{equation}
\phi\left( \mathbf{r} \right)=\phi_{dipole}\left( \mathbf{r} \right)+\phi_{rest}\left( \mathbf{r} \right),
\label{phi2}
\end{equation}
where
\begin{eqnarray}
\nabla^{2}\phi_{dipole}\left( \mathbf{r} \right)&=&-\frac{N}{\Lambda^2}\delta\left( \mathbf{r} - \mathbf{r}_{o} \right) + \frac{N}{\Lambda^2}\delta\left( \mathbf{r} - \mathbf{r}'_{o} \right),\nonumber \\
\nabla^{2}\phi_{rest}\left( \mathbf{r} \right)&=&\Delta\rho\left( \mathbf{r} \right).
\label{poisson5}
\end{eqnarray}
We first consider the case where $R_{o}$ is finite and $r_{o}=R_{o}/\Lambda$ vanishes
in the large $\Lambda$ limit. Then $\phi_{dipole}$ reduces to a dipole potential,
and it diverges near the origin. However,  $\phi_{rest}\left( \mathbf{r} \right)$ is a
continuous and differentiable function for all $\mathbf{r}$. From the solution of
the dipole potential, it is easy to show that
\begin{equation}
\phi_{dipole}\left( r, \theta \right) =  A \frac{\cos\theta}{r},
\label{dipole1}
\end{equation}
for $1\gg r\gg 1/\Lambda$, where we have used polar coordinates $\left( r, \theta \right)$
with $\theta$ being measured with respect to the $\xi$-axis. Here 
$A$ is a numerical  constant, which is a property of the asymptotic pattern. Then
\begin{equation}
\left.\frac{\partial \phi}{\partial \xi}\right \arrowvert_{\xi=0}=\frac{A}{\eta^{2}},
\label{phi3}
\end{equation}
and the integral in equation (\ref{na2}) diverges as $A/\eta_{min}$, where
$\eta_{min}$ is the cutoff introduced by the lattice. Using $\eta_{min}=\mathcal{O}\left( 1/\Lambda \right)$
it is easy to show that
\begin{equation}
N_{a}\simeq C_{1}\Lambda^{3},
\label{na3}
\end{equation}
where $C_{1}$ is a constant. Then using equations (\ref{nr3}) and (\ref{na3}) and that $N_{a}$ and $N_{r}$ add up to $N$,
we get
\begin{equation}
C_{1}\Lambda^{3}+C_{2}\Lambda^{2}\simeq N.
\label{scalethree}
\end{equation}
Considering the dominant term in the expression for large $\Lambda$, it follows
that $\Lambda$ increases as $N^{1/3}$.

For the patterns in the other limit where the source is placed at a distance
$\mathcal{O}\left( \Lambda \right)$ such that $r_{o}$ is non-zero for $\Lambda\rightarrow\infty$,
$\phi_{dipole}$ is non-singular along  the  sink line. Then, clearly  $N_{a}\sim \Lambda^{2}$ and as a result $\Lambda\left( N \right)\sim N^{1/2}$.

The above analysis can be easily generalized to a case
with the sink sites along two straight lines intersecting at an angle $\omega$
and a point source inside the wedge. For a square lattice,
$\omega=0, \pi/2, \pi, 3\pi/2$ and $2\pi$ are most easily constructed,
and avoid the problems of lines with irrational slopes, or rational numbers 
slopes with large denominators. The wedge with wedge-angle 
$\omega=\pi/2$ is obtained by placing the sink sites along the $x$ and $y$-axis
and the source site at $\mathbf{R}_{o}\equiv\left( 1, 1 \right)$ in
the first quadrant. The pattern with a line sink, discussed in
previous section, corresponds to $\omega=\pi$. 

For the general $\omega$, the corresponding electrostatic problem reduces to
determining the potential function $\phi$ inside a wedge formed by two
intersecting grounded conducting lines. Again the potential has two
contributions: the potential $\phi_{point}\left(\mathbf{r}  \right)$
due to a point charge at the source site and the potential $\phi_{rest}\left( \mathbf{r} \right)$
due to the areal charge density. We first consider the case where the source site is placed at a finite
distance from the wedge corner such that the distance in reduced
coordinates vanishes in the large $\Lambda$ limit. In this limit $\phi_{rest}$ is a non-singular
function of $\mathbf{r}$ while $\phi_{point}$ diverges close
to the origin. A simple  calculation of the electrostatic problem gives
\begin{equation}
\phi_{point}\left( r, \theta \right) \approx A \frac{\sin \alpha \theta}{r^{\alpha}},
\label{mpole}
\end{equation}
where $\alpha=\pi/\omega$ and we have used polar coordinates $\left( r, \theta \right)$ with
the polar angle $\theta$ measured from one of the absorbing lines.
Again $A$ is a constant independent of $N$ or $\Lambda$ and is a property of the asymptotic pattern.
Then arguing as before, we get
\begin{equation}
N_{a}\simeq C_{1}\Lambda^{2+\alpha}\rm{~ ~ and ~ ~ }N_{r}\simeq C_{2}\Lambda^{2}.
\label{na4}
\end{equation}
So the equation analogous to equation (\ref{scalethree}) is
\begin{equation}
C_{1}\Lambda^{2+\alpha}+C_{2}\Lambda^{2}\simeq N.
\label{scalefour}
\end{equation}
For a wedge angle $\omega=\pi$, $\alpha = 1$, and the above equation reduces to 
Eq.(\ref{scalethree}).

Similar arguments involving conformal transformation
have been used earlier in the context of equilibrium statistical physics to determine the wedge-angle dependence of surface critical
exponents near a wedge \cite{dup}.

For the problem where the source site is at a distance
$\mathcal{O}\left( \Lambda \right)$ from the wedge corner 
both the functions $\phi_{rest}$ and $\phi_{point}$ are nonsingular
close to the origin. It is easy to show that $\Lambda\left( N \right)$ grows as $N^{1/2}$.

These arguments can be easily extended to other lattices with
different initial height distributions, or to higher dimensions.
Consider, for example, an abelian sandpile model defined on the  cubic 
lattice. The allowed heights are from $0$ to $5$, and a site topples if the height
exceeds $5$, and sends one particle to each neighbor. The sites are labelled
by the Cartesian coordinates $(x,y,z)$, where $x,y$ and $z$ are integers. We
consider the infinite octant defined by $x\geq 0,y \geq 0,z \geq 0$. We start
with all heights equal to $4$, and add sand grains at the site $(1,1,1)$.  We assume
that the sites on planes $x=0$, $y=0$ and $z=0$ are all sink sites, and any grain reaching there is lost.
We add $N$ grains and determine the diameter of the resulting stable pattern. 

We again write the potential function in two parts: $\phi_{point}$ due to a point charge at 
$\left( 1/\Lambda,1/\Lambda,1/\Lambda \right)$ and $\phi_{rest}$ due to the bulk charge density in the presence
of three conducting grounded planes. Then, a simple electrostatic
calculation shows that the potential $\phi_{point}$ is the octapolar
potential with it's form in spherical polar coordinate as
\begin{equation}
\phi\left( r, \theta, \Phi \right)\approx \frac{f\left( \theta,\phi \right)}{r^{4}}.
\label{octupole}
\end{equation}
This then implies that the equation determining the dependence
of $\Lambda$ on $N$ is
\begin{equation}
C_{1}\Lambda^{6}+C_{2}\Lambda^{3}\simeq N
\label{scalefive}
\end{equation}

\section{A single sink site}\label{ch3.3}
Let the site of addition be the origin, with the sink site placed at $\mathbf{R}_{o}$.
We shall show that when  $\mathbf{R}_{o}$ lies in a high-density patch (color yellow in
Fig.\ref{fig:psone}), the asymptotic patterns are identical to the one
produced in the absence of the sink site.

The patterns, produced for $r_{o}$ close to $1$, with the sink sites placed deep inside 
a high-density patch are simple to analyze, even for finite but large $\Lambda$. One such pattern
is presented in Fig.\ref{fig:psone}. 
\begin{figure}
\begin{center}
\includegraphics[scale=0.30]{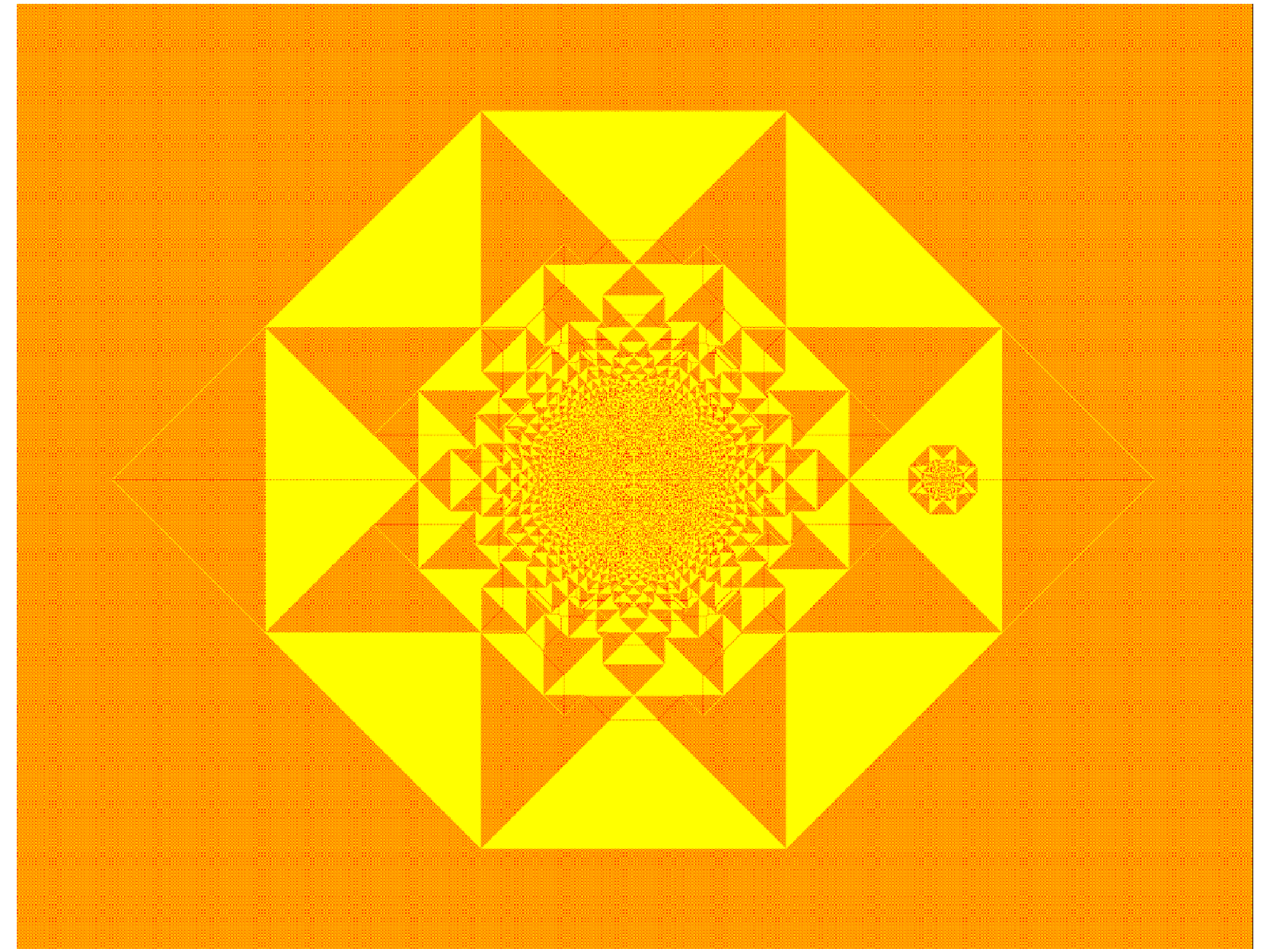}
\caption{ The pattern produced by adding $224000$ grains at the origin
with a sink site at $\left( 400,0 \right)$, inside a patch of density $1$ (color yellow). Color code: red$=0$ and
yellow$=1$. The apparent orange regions correspond to the checkerboard height distribution. (Zoom in for details in the electronic version.)}
\label{fig:psone}
\end{center}
\end{figure}
\begin{figure}[t]
\begin{center}
\includegraphics[scale=0.30]{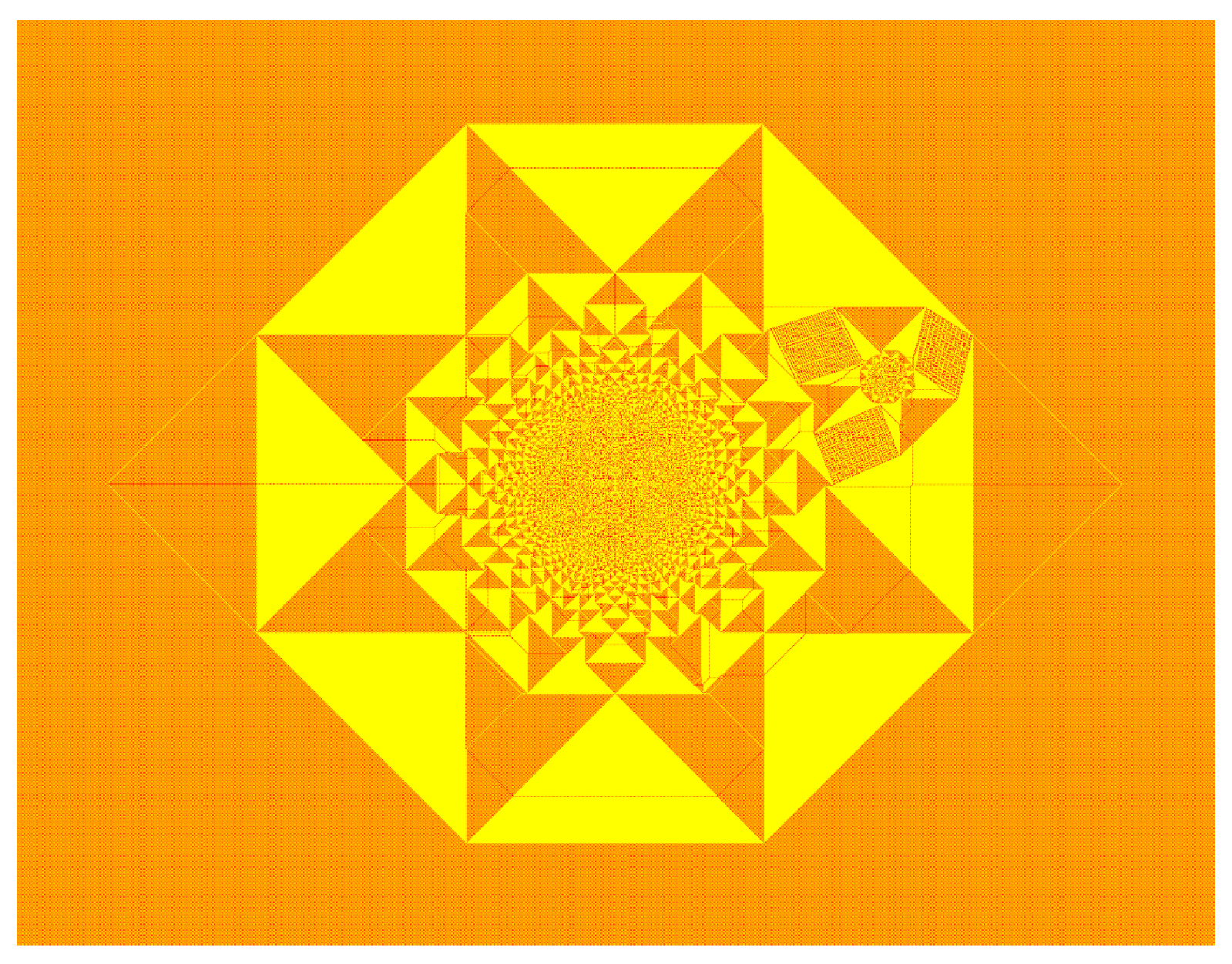}
\caption{ The pattern produced by adding $224000$ grains at the origin
with a sink site placed at $\left( 360,140 \right)$, inside a low-density patch. 
Color code red$=0$ and
yellow$=1$. The apparent orange regions correspond to the checkerboard height distribution. (Details can be seen in the electronic version using zoom in.)}
\label{fig:pstwo}
\end{center}
\end{figure}

We see that the effect of the sink site on the pattern is to produce a
depletion pattern centered at this site. The depletion pattern
is a smaller negative copy of the single source pattern, where
negative means $\Delta\rho$ is negative of the original pattern. We define the function
$\Delta z_{sink}\left( \mathbf{R}; N \right)$ as the difference between
the heights at ${\mathbf R}$ in the final stable configuration produced
by adding $N$ grains at the origin, with and without the sink site. 
\begin{equation}
\Delta z_{sink} \left( \mathbf{R}; N \right)= \Delta z_{source+sink}\left( {\mathbf R}; N \right)- \Delta z_{source}\left( \mathbf{R}; N \right). 
\end{equation}
From the figure it is seen that, in this case, $\Delta z_{sink}\left( \mathbf{R}; N \right)$
is the negative of the pattern produced by a smaller source, centered at $\mathbf{R}_0$.
The number of grains required to produce this smaller pattern is exactly the number of grains
$N_a$ absorbed at the sink site.
\begin{equation}
\Delta z_{sink}\left( \mathbf{R}; N \right)= - \Delta z_{source}\left( \mathbf{R}-\mathbf{R}_{o}; N_{a} \right).
\label{zsink}
\end{equation}
This is immediately seen from the fact that the 
toppling function $T_\Lambda\left( \mathbf{R} \right)$ satisfies 
\begin{equation}
\Delta T_{\Lambda}\left( \mathbf{R} \right)=\Delta z_{source+sink}\left( \mathbf{R}; N \right)-N \delta_{\mathbf{R}, \mathbf{0}} + N_{a}\delta_{\mathbf{R}, \mathbf{R}_{o}},
\label{discrpois}
\end{equation}
where  $\Delta$ is the toppling matrix for the sandple model on the
F-lattice (see chapter \ref{ch:intro}).
Let $T_{source} ({\mathbf R}; N)$ be the number of topplings at ${\mathbf R}$, when we
add $N$ particles at the origin in the absense of any sink site. Since Eq. (\ref{discrpois}) is  a linear equation,
it follows that a solution of this equation is
\begin{equation}
T_{\Lambda}\left( \mathbf{R} \right)=T_{source}\left( \mathbf{R}; N \right)-T_{source}\left( \mathbf{R}-\mathbf{R}_{o}; N_{a} \right).
\label{lcomb}
\end{equation}
This is a valid solution for our problem, if the corresponding heights in the final configuration with the sink are all non-negative.
This happens when the region with nonzero $\Delta z_{sink}$ is confined within a high-density patch of the single
source pattern.

The number $N_{a}$  can be determined from the requirement that the number of topplings at the sink site is zero.  
The potential function for the single source problem diverges as $\left( 4\pi \right)^{-1}\log r$
near the source. Considering the ultraviolate cutoff due to the lattice,
$T_{source}\left( \mathbf{R}, N \right)$ at $\mathbf{R}=0$ can be approximated by
$\left( 4\pi \right)^{-1}N \log N$ to leading order in $N$.
Then at $\mathbf{R} = \mathbf{R}_0$, $T_{source}\left( \mathbf{R}-\mathbf{R}_{o}; N_{a} \right)$ is approximately equal to
$(4\pi)^{-1}N_{a}\log N_{a}$ whereas $T_{source}\left( \mathbf{R}_0; N \right) \approx N \phi_{source}( \mathbf{r}_0) $,
where $\phi_{source}\left( \mathbf{r} \right)$ is the potential function for the problem without a sink. Then 
from the equation (\ref{lcomb}) we have
\begin{equation}
\frac{1}{4\pi}N_a \log N_a \simeq N \phi_{source}\left( \mathbf{r}_{o} \right). 
\end{equation}
For large $N$, this implies that  
\begin{equation}
N_a \simeq 4\pi\phi_{source}\left( \mathbf{r}_{o} \right)N/\log N.
\label{nss1}
\end{equation}
Then, in the large $N$ limit, for a sink at a fixed reduced coordinate $\mathbf{r}_o$,
the relative size of the defect produced by the sink site decreases as
$1/\sqrt{\log N}$. Hence asymptotically, the fractional area of the defect
region will decrease to zero, if the sink position $\mathbf{r}_o$ is inside a high-density patch.

When the sink site is inside a low-density patch, the subtraction procedure in
equation (\ref{zsink}) gives negative heights, and no longer gives the correct solution.  
However it is observed for the patches in the outer layer,
where the patches are large, that the effect of the sink site is confined within the neighboring 
high-density patches (Fig.\ref{fig:pstwo}) and rest of the pattern in the asymptotic 
limit remains unaffected.

The pattern in which the source and the sink sites are adjacent to each other, appears to be
very similar to the one produced without the sink site. This is easy to see. 
The Poisson equation analogous to equation (\ref{poisson0}) for this problem is
\begin{equation}
\nabla^2\phi(\mathbf{r})=\Delta\rho(\mathbf{r})-\frac{N}{\Lambda^2}\delta(\mathbf{r})+\frac{N_{a}}{\Lambda^2}\delta(\mathbf{r}-\mathbf{r}_o),
\label{poisson3}
\end{equation}
where $N_a$ is the number of grains absorbed in the sink site at $\mathbf{r}_{o}$. 
In an electrostatic analogy, as discussed earlier, $\phi$ can be considered
as the potential due to a distributed charge of density $-\Delta\rho\left( \mathbf{r}_{o} \right)$
and two point charges of strength $N/\Lambda^{2}$ and $-N_{a}/\Lambda^{2}$, placed
at the origin and at $\mathbf{r}_{o}$ respectively.
It is easy to see  that the dominant contribution to the potential is the  monopole term with net
charge $(N- N_a)/\Lambda^2$. The contribution due to other terms decreases as
$1/\Lambda$ for large $\Lambda$, and the asymptotic pattern is the same as without a sink, with $N-N_a$ particles added.

The number $N_a$ of particles absorbed is determined by the condition that the number
of topplings at $(1,0)$ (the sink position) is zero. The potential produced at $(1,0)$ and $(0,0)$, by the
areal charge density is nearly the same. The number of topplings
at $(1,0)$, if we add $N_a$ particles at the sink site, is approximately $\left (4 \pi\right )^{-1} N_a \log N_a$.
Now, from the solution of the discrete Laplacian, the number of topplings produced at
$(1,0)$ due to $N$ particles added at $(0,0)$ is approximately $\left( 4\pi \right)^{-1}\left (N \log N -C N\right )$ with
$C$ being an undetermined constant. Equating these two, we get
\begin{equation}
 N_a \log N_a  \simeq   N \log N - C N.
\label{nans}
\end{equation}
As the asymptotic pattern is the same as that produced by adding $\left( N-N_{a} \right)$
grains at the origin without a sink, we have $N - N_a \simeq \Lambda^2$, and
\begin{equation}
(N -\Lambda^{2})\log ( N -\Lambda^2) \simeq  N \log N - C N.
\label{scaleone}
\end{equation} 
Simplification of this equation for large $N$ shows that $\Lambda$ grows as $\sqrt{N/\log N}$ with $N$. 

For finite $N$, the leading  correction to $\phi\left( \mathbf{r} \right)$ comes from
the dipole term in the potential. This term breaks the reflection 
symmetry of the pattern about the origin. A measure of the bilateral asymmetry is the
difference of the boundary distances on two opposite sides of the source.
As the relative contribution of the dipole potential compared to the monopole
term decays as $\log \Lambda/\Lambda$, for large $\Lambda$, this difference
vanishes in the asymptotic pattern in the reduced coordinates.

\section{Numerical results}\label{ch3.4}
All the above scaling behaviors are verified by the measurement of lengths in the patterns
for finite, but large $N$. Let $\Lambda^{\ast}_{line}\left( N \right)$ be the real positive root of
equation (\ref{scalethree}) for a given integer value of $N$. As
$\Lambda_{line}$ takes only the integer values on the lattice, an estimate of
it would be $Nint\left[ \Lambda^{\ast}_{line}(N) \right]$, the integer nearest
to $\Lambda^{\ast}_{line}\left( N \right)$.
Interestingly, we found that for
a choice of $C_{1}=0.1853$ and $C_{2}=0.528$, this estimate gives
values which differ from the measured values
at most by $1$ for all $N$ in the range of $100$ to $3\times10^{6}$.
We rewrite the equation (\ref{scalethree}) as
\begin{equation}
0.1853\Lambda_{line}^{3}+0.528\Lambda_{line}^{2}\doteqdot N,
\label{scalethreehalf}
\end{equation}
where we used the symbol $\doteqdot$ to denote that both sides differ at most by $1$.
Clearly more precise estimates of $C_{1}$ and $C_{2}$ would be
required if we want this to work for larger $N$.

Similarly for the other two equations (\ref{scalefour}) and (\ref{scalefive})
we find that they are in very good agreement with our numerical data.
We consider the case of wedge angle $\omega=2\pi$. 
This corresponds to the case with the source site next to an infinite half-line of sink sites.
Here $\alpha=1/2$ and equation (\ref{scalefour}) reduces to
\begin{equation}
C_{1}\Lambda_{\omega=2\pi}^{5/2}+C_{2}\Lambda_{\omega=2\pi}^{2}\simeq N.
\end{equation}
Choosing $C_{1}=0.863408$ and $C_{2}=0.043311$, we find that the function
$Nint\left[ \Lambda^{\ast}_{\omega=2\pi}\left( N \right) \right]$
differs from the measured values by at most $1$ for all $N$ in the range of $100$ to $2\times10^{5}$.
Then, as in equation (\ref{scalethreehalf}), we write
\begin{equation}
0.863408\Lambda_{\omega=2\pi}^{5/2}+0.043311\Lambda_{\omega=2\pi}^{2}\doteqdot N.
\label{scalefourhalf}
\end{equation}
Similarly, for the three dimensional abelian sandpile model with the source site inside the first octant
and $x=0$, $y=0$, and $z=0$ as the absorbing planes, the equation determining
the dependence of the diameter on $N$ is
\begin{equation}
0.0159\Lambda_{3d}^{6}+88\Lambda_{3d}^{3}\doteqdot N
\label{scalefivehalf}
\end{equation}
We have verified this equation for $N$ between $5\times10^{5}$ to $5\times10^{8}$.

We obtained these equations by determining the number of absorbed grains $N_{a}$ and the
remaining grains $N_{r}$ from dimensional counting grounds, and
the final equations are then only a statement
of the conservation of the sand grains. It is quite remarkable that this scaling
analysis gives almost the exact values of the diameter.
In addition, these equations have an important feature that they include a ``correction to
scaling'' term whereas the usual scaling analysis ignores the sub-leading powers.

We also verify equation (\ref{nss1}) using patterns with fixed $\mathbf{r}_{o}$
and the sink site inside a high-density patch in the outer layer of the pattern.
It is found that for a change of $N$ from $224000$ to $896000$, $N_{a}\log{N}/N$
changes by less than $7\%$, which is consistent with the scaling relation. 

In the other limit, where the sink site is next to the source, the dependence of $\Lambda$
on $N$ is given in equation (\ref{scaleone}). We measure $\Lambda\left( N \right)$ for the patterns with
the sink site at $\left( 1, 0 \right)$ and the source at the origin. For $N$ in the
range of $100$ to $5\times 10^{5}$ we find
that the function $Nint[\Lambda^{\ast}_{point}\left( N \right)]$ with $C=2.190$ in equation (\ref{scaleone}),
gives almost exact values of $\Lambda\left( N \right)$, with their difference being at most $1$. Then we write
\begin{equation}
(N -\Lambda_{point}^{2})\log ( N -\Lambda_{point}^2) \doteqdot  N \log N - 2.190 N.
\end{equation} 

In the last case, let $R_{1}$ and $R_{2}$ be the boundary distances measured along the
positive and the negative $x$ axis. The difference $R_{2}-R_{1}$ is plotted in 
Fig.\ref{fig:asymmetry} where the data is found to fit to the
function $1.22\log{(R_{2}+0.5)}$. This confirms the result that the relative
bilateral asymmetry $(R_{2}-R_{1})/R_{2}$ vanishes in the asymptotic pattern as $\log\Lambda/\Lambda$.
\begin{SCfigure}
\includegraphics[scale=0.75]{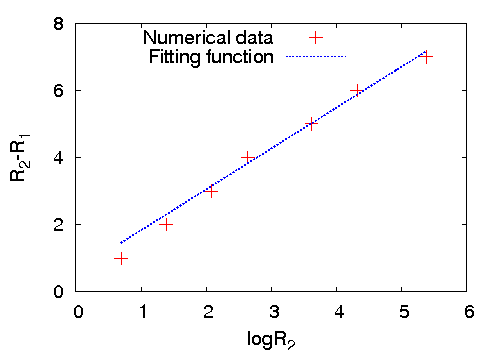}
\caption{The bilateral asymmetry due to the presence of a sink site in Fig.\ref{fig:psone}.}
\label{fig:asymmetry}
\end{SCfigure}

\section{Characterization of the pattern with a line
sink}\label{ch3.5}
The pattern with a line sink (Fig.\ref{fig:lsone}), discussed in
Section \ref{ch3.2}, retained two important properties present in 
the single source pattern (Fig. \ref{flattice}). These are: The asymptotic pattern is made of
the union of two types of patches of excess density $1/2$ and $0$
and the separating boundaries of the patches are straight lines of slope $0$, $\pm1$ or
$\infty$. However the adjacency graph is changed significantly and this changes the sizes
of the patches as well. In this section we show how to explicitly determine the potential
function on this adjacency graph.

The adjacency graph of the patches is shown in Fig.\ref{fig:adjLS}. 
This representation of the graph is easier to see by taking the $1/r^{3}$
transformation of the pattern and then joining neighboring patches by straight lines (Fig.\ref{fig:1byrsqr}).
Each vertex in the graph is connected to four neighbors except for the vertices corresponding to
the patches next to the absorbing line. These have coordination number $3$. Also the vertex
at the center corresponding to the exterior of the pattern is connected to seven neighbors.
\begin{figure}[t]
\begin{center}
\includegraphics[scale=0.43]{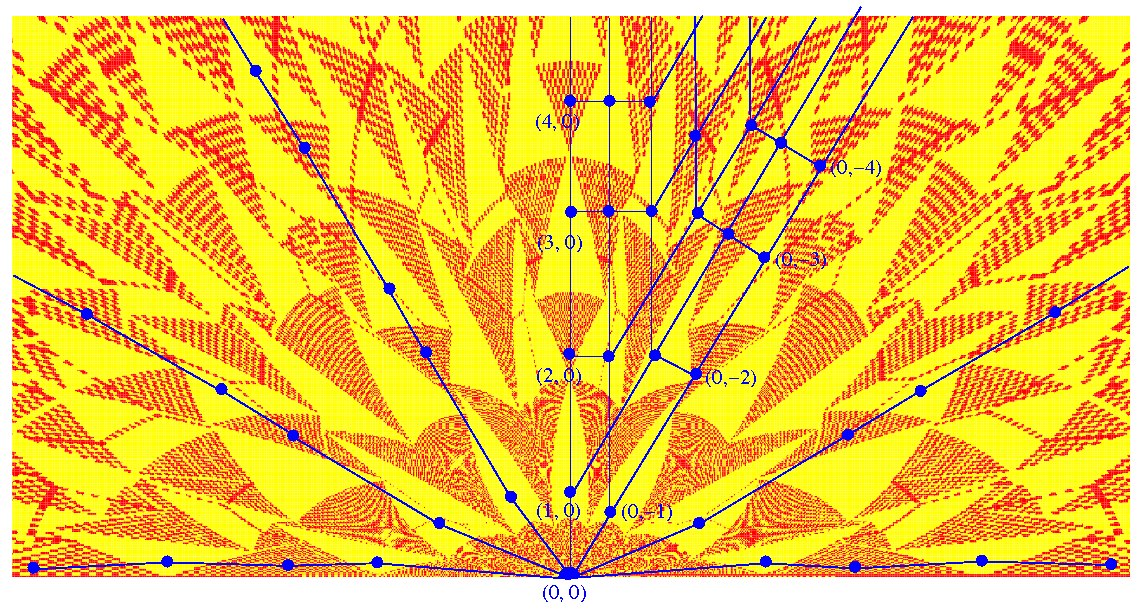}
\caption{ $1/r^{3}$ transformation of the pattern in Fig.\ref{fig:lsone}. Two adjoining patches are connected by drawing a straight line.}
\label{fig:1byrsqr}
\end{center}
\end{figure}
\begin{figure}
\begin{center}
\includegraphics[scale=0.30]{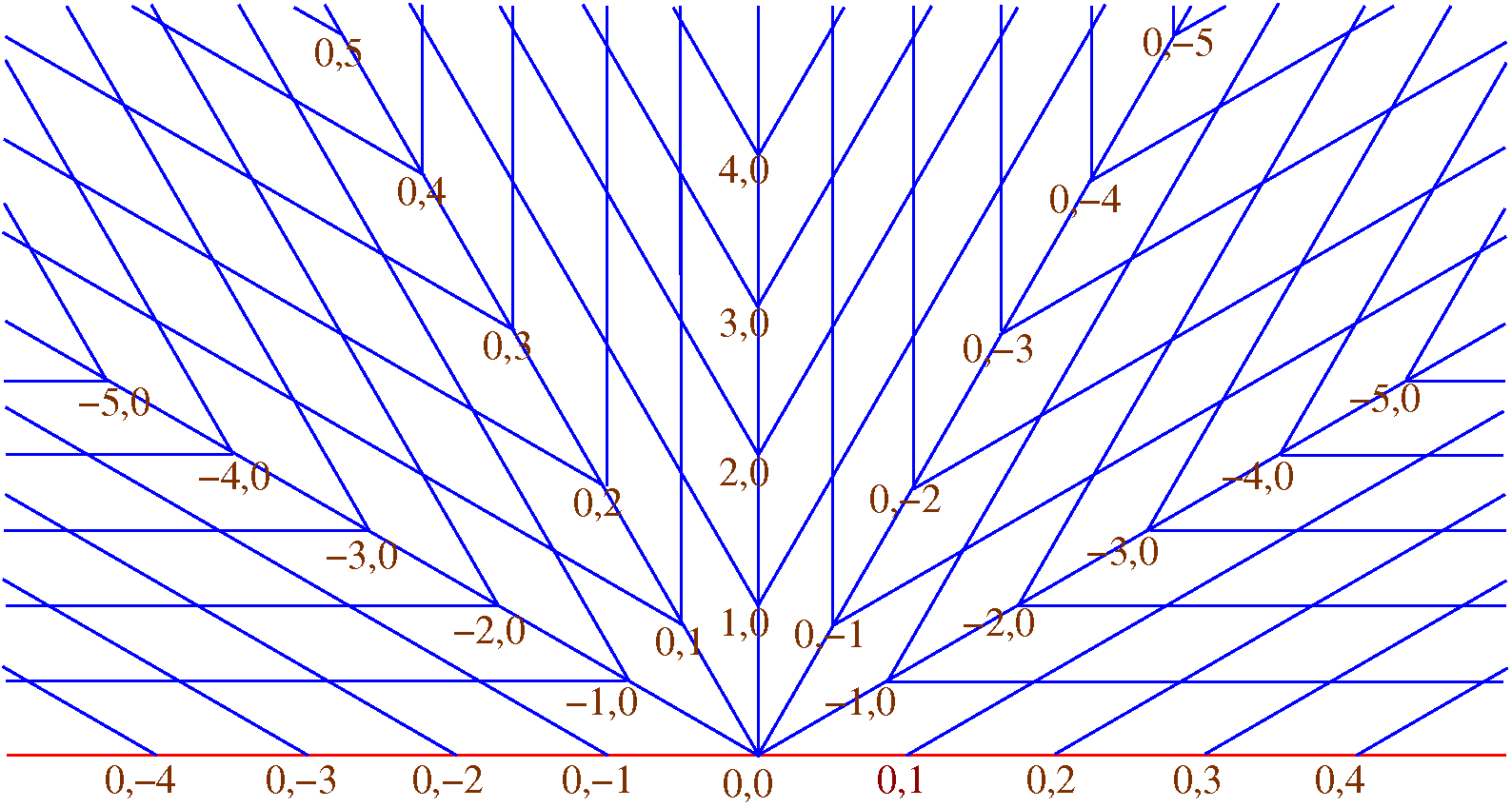}
\caption{ The adjacency graph of the patches corresponding to the pattern in Fig.\ref{fig:lsone}.}
\label{fig:adjLS}
\end{center}
\end{figure}

Let us write the quadratic potential function in a patch $P$ having excess density $1/2$ as 
\begin{equation}
\phi_{_P}(\mathbf{r})=\frac{1}{8}(m_{_P}+1)\xi^2+\frac{1}{4}n_{_P}\xi\eta+\frac{1}{8}(1-m_{_P})\eta^2 + d_{_P}\xi+e_{_P}\eta+f_{_P},
\label{fform1}
\end{equation}
where the parameters $m$, $n$, $d$, $e$ and $f$ take 
constant values within a patch. Similarly for a low-density patch $P'$
\begin{equation}
\phi_{_{P'}}(\mathbf{r})=\frac{1}{8}m_{_{P'}}(\xi^{2}-\eta^2)+\frac{1}{4}n_{_{P'}}\xi\eta+d_{_{P'}}\xi+e_{_{P'}}\eta+f_{_{P'}}.
\label{fform2}
\end{equation}

Using the continuity of $\phi\left( \mathbf{r} \right)$ and its first derivatives along the common boundaries
between neighboring patches it has been shown in chapter
\ref{ch2}, that for the single source pattern
without sink sites $m$, and $n$ take integer values. The same argument also applies to this
problem and it can be shown that $\left( m, n \right)$ are the coordinates of the patches in the
adjacency graph in Fig.\ref{fig:adjLS}. These coordinates are shown next to some of the vertices.
There are two different patches corresponding to the same set of $\left( m, n \right)$
values. In fact, as in the single source pattern the adjacency graph
forms a square grid
on a two sheeted Riemann surface, the same is formed for this pattern, but on a
three sheeted Riemann surface. This can be constructed by modifying the graph
in Fig \ref{fig:adjLS} keeping its topology the same. In this representation the pattern covers half of the surface with
$\left( m, n \right)$ being the Cartesian coordinates on the surface.

Define function $D\left( m, n \right)=d\left( m, n \right)+ie\left( m, n \right)$ on this
lattice. As discussed in \cite{myepl}, the continuity of $\phi(\mathbf{r})$ and its first derivatives along the
common boundary between neighboring patches imposes linear relations between $d$ and $e$ of the corresponding
patches. Using these matching conditions it can be
shown that $d$ and $e$ satisfy the discrete Cauchy-Riemann conditions \cite{myepl}
\begin{eqnarray}
d\left( m+1, n+1 \right) - d\left( m, n \right)&=&e\left( m, n+1 \right) - e\left( m+1, n \right), \nonumber \\
e\left( m+1, n+1 \right) - e\left( m, n \right)&=&d\left( m+1, n \right) - d\left( m, n+1 \right),
\label{cr}
\end{eqnarray}
and then the function $D$ satisfies the discrete Laplace equation
\begin{equation}
\sum_{i=\pm1}\sum_{j=\pm1}D(m+i,n+j)-4D(m,n)=0,
\label{laplace}
\end{equation} 
on this adjacency graph.

Let us define $M=m+in$ and $z=\xi+i\eta$. As argued before, close to the origin the potential $\phi$ diverges as
$1/r$ (equation (\ref{dipole1})). Then, the corresponding complex potential function
$\Phi\left( z \right) \sim 1/z$. As $M \sim {d^2 \Phi}/{dz^2} $, and
$D \sim d \Phi /dz$,  it follows that for large $|M|$,
\begin{equation}
D \sim M^{2/3}.
\label{asdls}
\end{equation}
Also, the condition that on the absorbing line $\phi\left( \mathbf{r} \right)$ must vanish implies that for the vertices
with even $n$ along the red line in Fig.\ref{fig:adjLS} $e(0,n)$ vanishes. These
vertices correspond to the patches with the absorbing line as the horizontal boundary in Fig.\ref{fig:lsone}.

Equation (\ref{laplace}) with the above constraint and the boundary condition (Eq.(\ref{asdls}))
has a unique solution. The normalization of $\phi$ is fixed by the requirement that
$d(1,0)=-1$, which fixes the diameter of the pattern to be $2$ in reduced units. All the
spatial distances in the pattern can be expressed in terms of this solution $D\left( m, n \right)$
using the matching conditions between two neighboring patches. As an example, consider
the boundary between the patches corresponding to $\left( m, n \right)$ and $\left( m+1, n \right)$
with $\left( m+n \right)$ being odd. The matching conditions only allow a horizontal boundary
between them with the equation $\eta=\eta_{p}$, where
\begin{equation}
e\left( m+1, n \right) - e\left( m, n \right)=\eta_{p}/2.
\end{equation}
Similarly there is a vertical boundary between the patches $\left( m, n \right)$ and $\left( m-1, n \right)$,
with the equation $\xi=\xi_{p}$, where
\begin{equation}
d\left( m-1, n \right) - d\left( m, n \right)=\xi_{p}/2.
\end{equation}
The other boundaries can similarly be determined using the solution for $D\left( m, n \right)$.
The characterization of the asymptotic patterns for $\omega=\pi/2$, $3\pi/2$ and $2\pi$ is qualitatively similar and
will not be discussed here.

\section{Patterns with two sources}\label{ch3.6}
In this section we discuss patterns produced by adding $N$ grains each at two
sites placed at a distance $2\Lambda\mathbf{r_o}$ from each other along the $x$-axis,
at $\Lambda\mathbf{r}_{o}$ and $-\Lambda\mathbf{r}_{o}$ with $\mathbf{r}_{o}\equiv \left( \xi_{o}, 0 \right)$. 
Again, the diameter $2\Lambda$ is defined as the height of the smallest
rectangle enclosing all sites that have toppled at least once.
The two limits, $r_0$ close to zero and $r_0$ large are trivial:
For $r_o\rightarrow 0$, the asymptotic pattern is the same as that produced by 
adding grains at a single site. On the other hand if $r_{o} > 1$,  each source produces its own pattern, which do not overlap, 
and the final pattern is a simple superposition of the two patterns. 
\begin{figure}
\begin{center}
\includegraphics[scale=0.23]{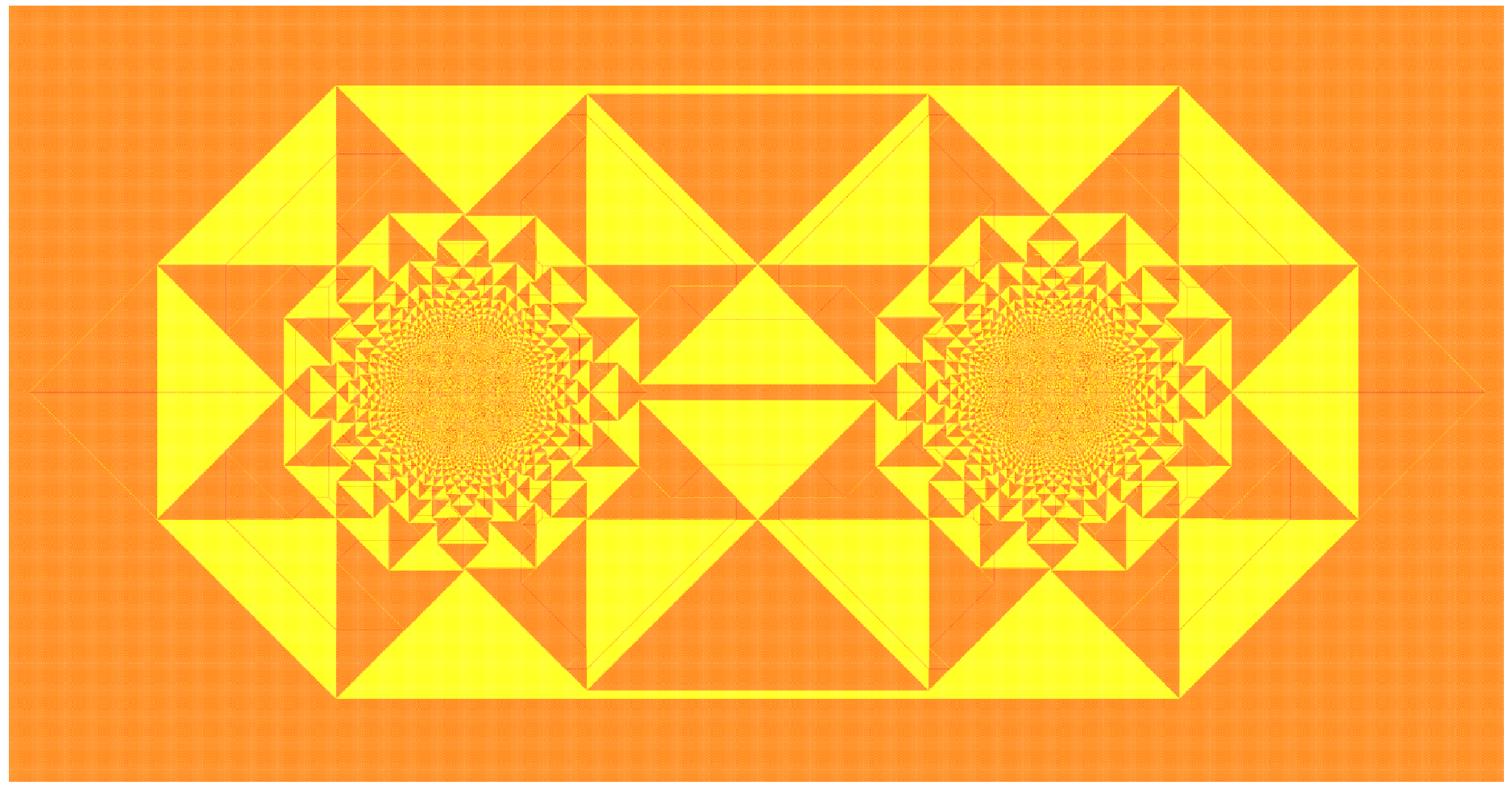}
\caption{ The pattern produced by adding $N=640000$ grains each at ($-760,0$) and 
($760,0$) on the F-lattice with the initial checkerboard distribution of grains and 
relaxing. This corresponds to $r_{o}=0.95$. Color code red=0 and yellow=1.
(Details can be seen in the electronic version using zoom in )}
\label{fig:twosource}
\end{center}
\end{figure}
\begin{figure}[t]
\begin{center}
\includegraphics[scale=0.14]{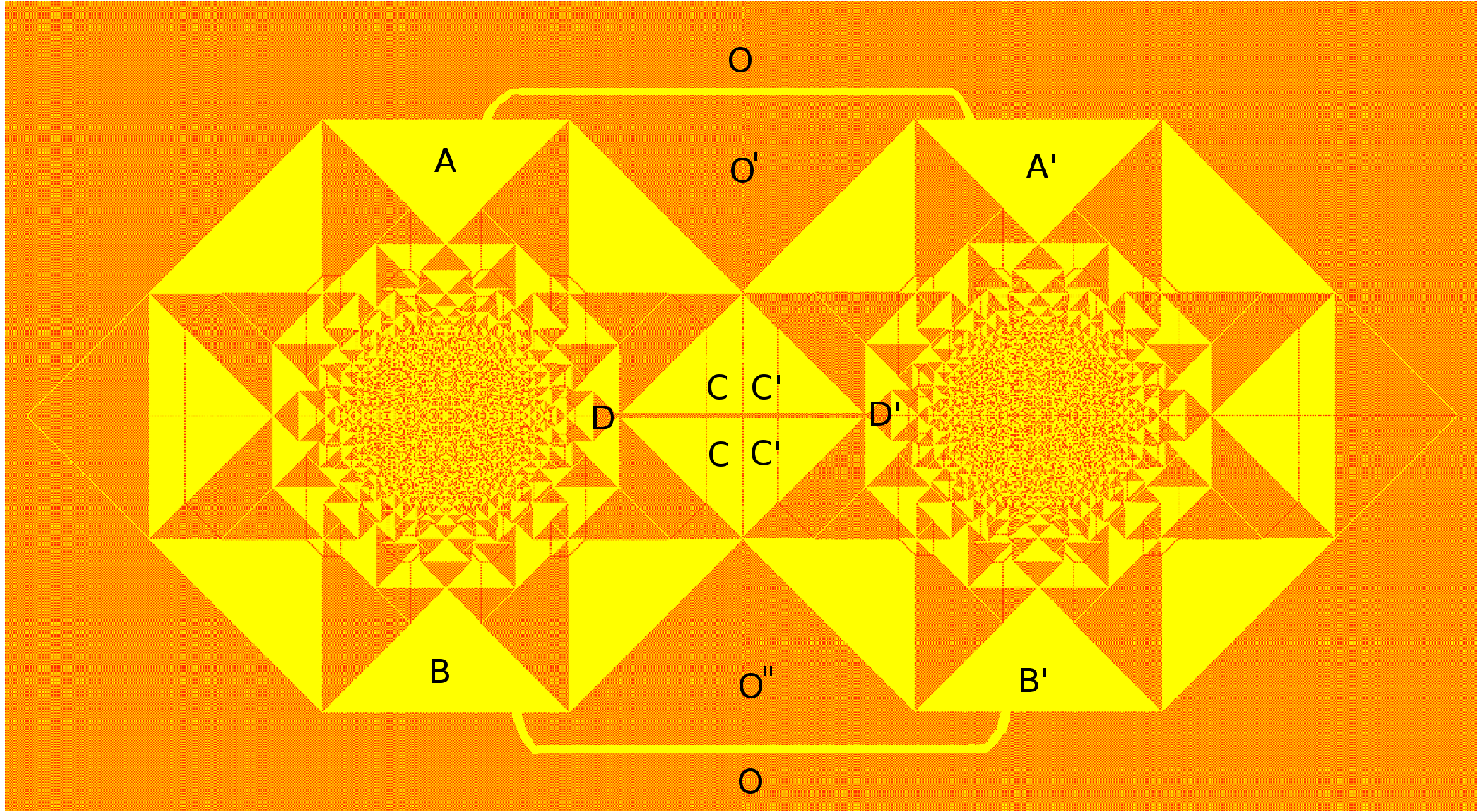}
\caption{ The pattern constructed by combining two single source patterns and drawing
connecting lines between few patches following the connectivity in the pattern in 
Fig.\ref{fig:twosource}.}
\label{fig:edited}
\end{center}
\end{figure}

As noted before, the adjacency graph for the single source
pattern has a square lattice structure on a Riemann surface of two-sheets \cite{myepl}. Then
the graph for two non-intersecting single source patterns is a square 
lattice on two disjoint Riemann surfaces, each having two-sheets 
(Fig.\ref{fig:adjnocol}). Only the vertex at the origin
represents the exterior of the pattern, which is the same for both of the
single source patterns. It has sixteen neighbors and is placed
midway between the two Riemann surfaces.
For later convenience let us associate the lower Riemann surface to the pattern
around the left source at $-\mathbf{r}_o$ and denote it by $\Gamma_L$. Similarly the
upper Riemann surface as $\Gamma_R$ corresponding to the pattern around
the right source $\mathbf{r}_o$.

For $ 0 < r_{o} <1$, the two single source patterns overlap. Using the abelian property, we first topple 
as if the second source were absent. The resulting pattern still has some
unstable sites in the region where the patterns overlap. 
Further relaxing these sites transfers these excess grains outward,
and changes the dimensions and positions of the patches: some patches become bigger,
some may merge, and sometimes a patch may break into two disjoint patches.  

The pattern produced with two sources with $r_0 = 0.95$ is shown in Fig.\ref{fig:twosource}. We see that there are still
only two types of periodic patches, corresponding to $\Delta\rho(\mathbf{r})$ values $0$
and $1/2$,  and the slope of the boundaries between patches takes the values 
$0$, $\pm1$ or $\infty$. 

The relaxation due to overlap changes the adjacency graph from the case with 
no overlap. This modified adjacency graph, for $r_{o}$ in the range
$0.70$ to $1.00$, is shown in Fig.\ref{fig:adj}.
For $r_0$ just below $1$, these changes are few and are listed below.
\begin{figure}
\begin{center}
\includegraphics[scale=0.13]{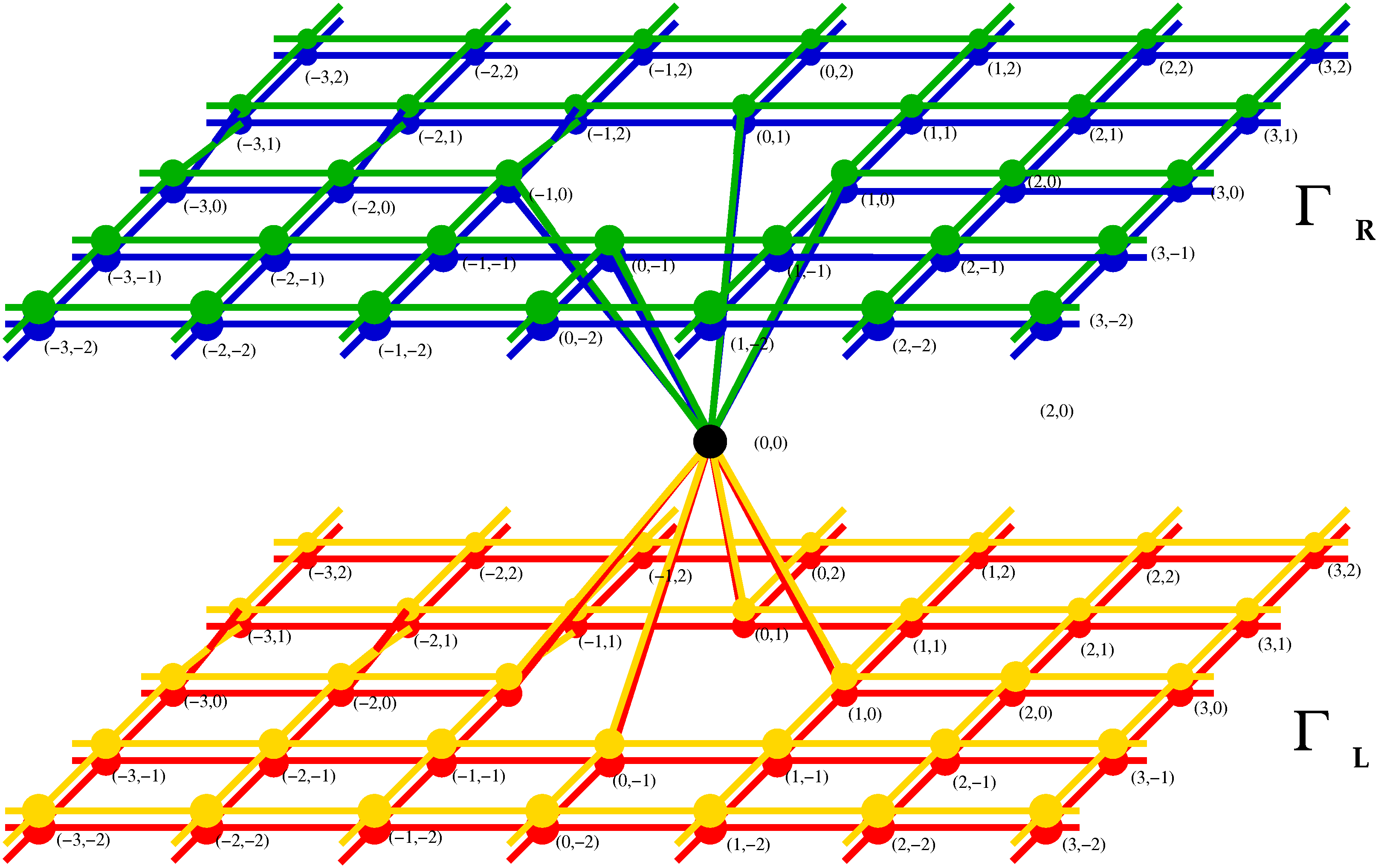}
\end{center}
\caption{ Representation of the adjacency graph of the patches for two non-overlapping single source patterns
as a square grid on two Riemann surfaces each of two-sheets. The vertices with same $(m,n)$ coordinates on different sheets are represented by different colors.}
\label{fig:adjnocol}
\end{figure}
\begin{figure}
\begin{center}
\includegraphics[scale=0.13]{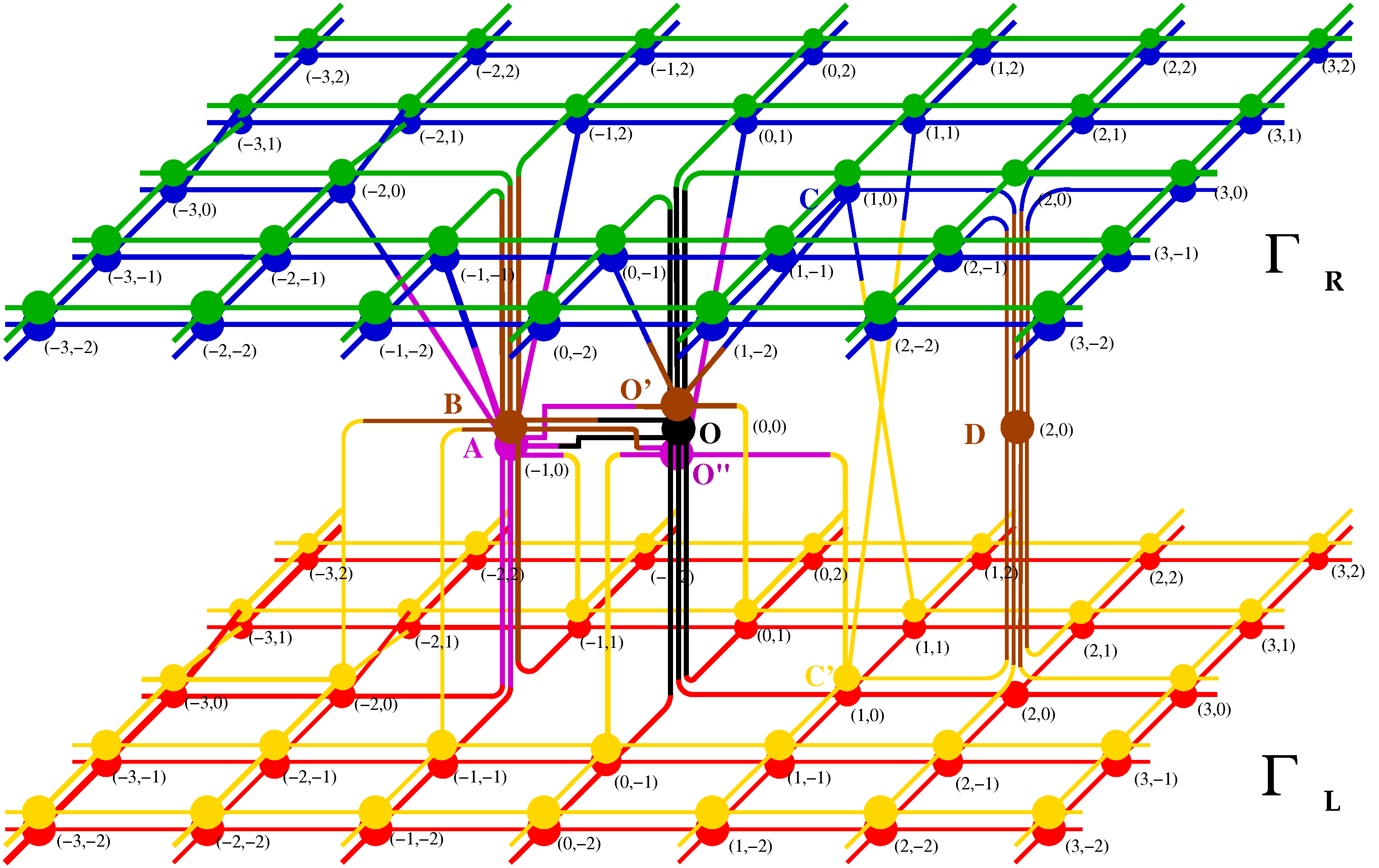}
\end{center}
\caption{ The adjacency graph for two intersecting single source patterns around two 
sites of addition placed at a distance $2 \mathbf{r}_o$ from each other. The graph
has the structure of square grids on four Riemann sheets except for a finite number
of vertices indicated by the alphabates $A$, $B$, $O$, $O'$, $O''$ and $D$ shown placed
in the middle layer. This graph remains unchanged for $r_{o}$ in the range $0.70$ to $1.00$.}
\label{fig:adj}
\end{figure}

(i) We note that the patches labelled $A$ and $A'$ in Fig.\ref{fig:edited} have the same
$\xi$ and $\eta$ dependence of the potential function $\phi$. Then, for $r_0$ just below $1$, these patterns can
join with each other by a thin  strip. This only requires a small movement in the
boundaries of nearby patches ( i.e.  only a small change  in the $d$ and $e$ values
of nearby patches). Thus, in the adjacency graph, the vertices corresponding to $A$
and $A'$ are collapsed into a single vertex $A$ in Fig.\ref{fig:adj}. 

(ii) Similarly, the vertices corresponding to the patches $B$ and $B'$ in 
Fig.\ref{fig:edited} are collapsed into a single vertex $B$ in Fig.\ref{fig:adj}.

(iii) This divides the region outside the pattern in to three parts, $O$, $O'$ and $O''$.
They are also shown in Fig.\ref{fig:adj} as separate vertices.

(iv) The patches marked $C$ and $C'$ also have the same quadratic form, and the
vertical boundary between them disappears. However, the patches $D$ and $D'$ are
also joined by a thin strip. This horizontal strip divides the joined $C$ and $C'$
into two again (Fig.\ref{fig:edited}). 

The adjacency of other patches remains unchanged. The adjacency graph of the
pattern is shown in Fig.\ref{fig:adj}. Interestingly, this new adjacency graph
remains the same for all $0.70 < r_0 < 1$, even though  for $r_{o} < 0.85$,
the sizes of different patches are substantially different. Compare the pattern
for $r_0 =0.70$ in Fig.\ref{fig:bellow}, with the pattern for $r_0=0.95$ in Fig.\ref{fig:twosource}:
The shape of the central patches in Fig.\ref{fig:bellow} is different
from that in Fig.\ref{fig:twosource}.

In Fig.\ref{fig:adj}, we have have placed  the vertices which are formed by merging
or dividing the patches, midway between the Riemann sheets corresponding
to the two sources. As $r_0$ is  decreased  below $0.70$, more collisions between the
growing patches will occur and the number of vertices in this middle region will increase.
For any nonzero $r_0$, the number of vertices in the middle layer is finite. In the
$r_{o}\rightarrow 0$ limit, vertices from both the surfaces $\Gamma_{L}$ and $\Gamma_{R}$
come together and form a single Riemann surface corresponding to a single source pattern
around $\mathbf{r}=0$. For $r_0$ small, but greater than zero, the outer patches are arranged
as in the single-source case, but closer to the sources, one has a crowded pattern near each
source. In the adjacency graph, this corresponds to the vertices near the patch $(0,0)$ roughly
arranged as on a Riemann surface of two-sheets, while the ones farther from the patch $(0,0)$
remain undisturbed on the 4-sheeted Riemann surface.
\begin{figure}
\begin{center}
\includegraphics[scale=0.20]{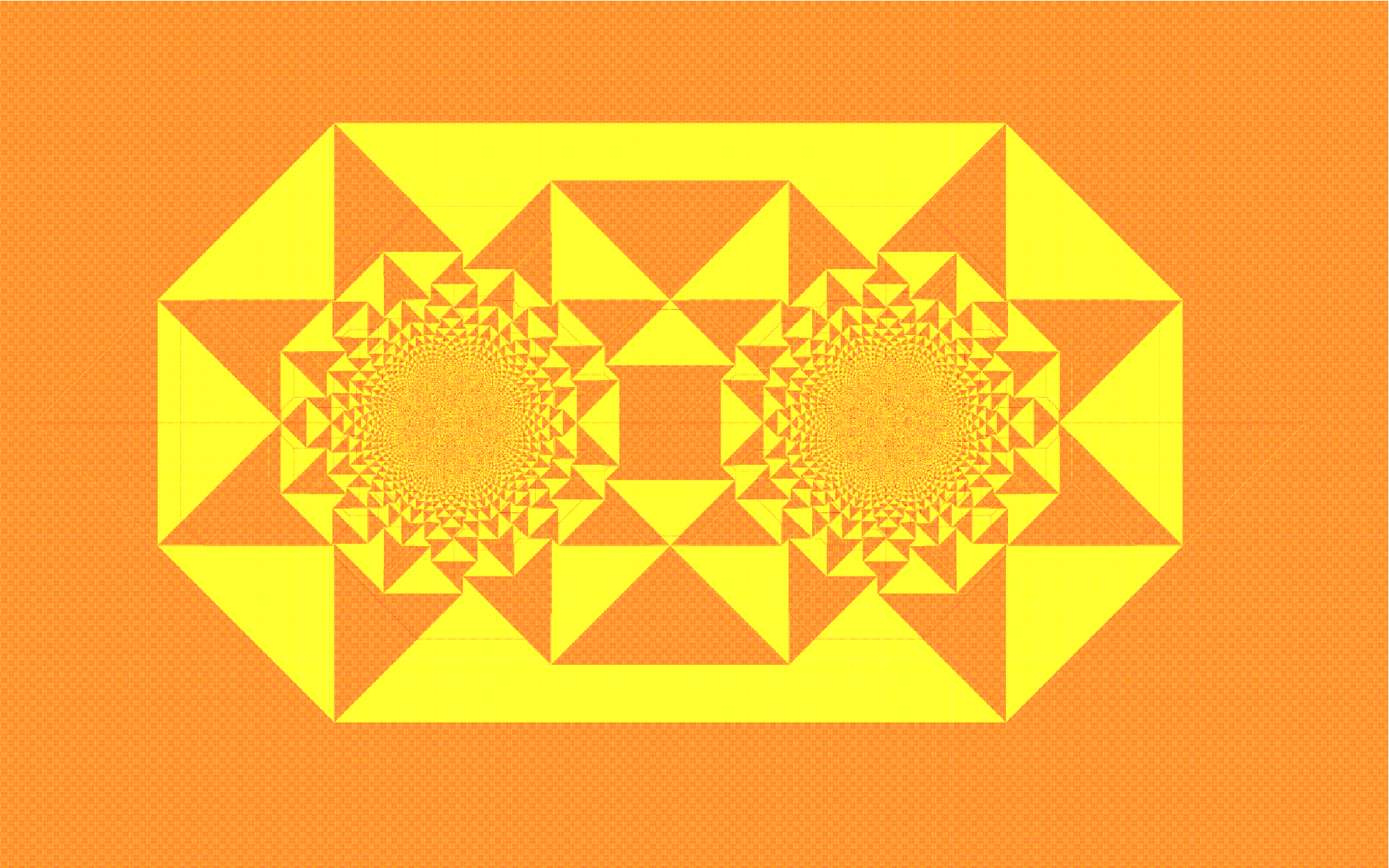}
\end{center}
\caption{ The pattern produced by adding $640000$ grains at site $\left( -600,0 \right)$
and $\left( 600,0 \right)$. Although the pattern is significantly different from
the one in Fig.\ref{fig:twosource}, their adjacency graph is same.}
\label{fig:bellow}
\end{figure}

We now characterize the pattern  with two sources and $r_0 > 0.70$ in detail by explicitly determining
the potential function on this  adjacency graph.

The Poisson equation analogous to Eq.(\ref{poisson0}) for this problem is
\begin{equation}
\nabla^2\phi(\mathbf{r})=\Delta\rho(\mathbf{r})-\frac{N}{\Lambda^2}\delta(\mathbf{r}-\mathbf{r}_o)-\frac{N}{\Lambda^2}\delta(\mathbf{r}+\mathbf{r}_o).
\label{poisson2}
\end{equation}
Let us use the same quadratic form of the potential
function given in equation (\ref{fform1}) and equation (\ref{fform2}).

Again using the same argument given in \cite{myepl}, it can be shown that
$m$ and $n$ are the coordinates of the patches in both the adjacency graphs
in Figs.\ref{fig:adjnocol} and \ref{fig:adj}. These coordinates are shown  next
to each vertex. Also, on this graph, the function $D(m,n)=d(m,n)+ie(m,n)$  satisfies the discrete Laplace equation
\begin{equation}
\sum_{m'}\sum_{n'}D(m',n')-4D(m,n)=0,
\label{laplace2}
\end{equation}
where $\left( m', n' \right)$ denote the neighbors of $\left( m, n \right)$
in the odd or even sublattice \cite{duffin}.
Let us define $z_o=\xi_o+i\eta_o$ where $(\xi_o,\eta_o)$ and
$(-\xi_{o},-\eta_{o})$ are the coordinates corresponding to $\mathbf{r}_{o}$ and $-\mathbf{r}_{o}$.
Considering that close to $\mathbf{r}_{o}$ and $-\mathbf{r}_{o}$ the potential
$\phi\left( \mathbf{r} \right)$ diverges logarithmically it can be shown (as done for single source pattern
in \cite{myepl}) that for large $|M|$,
\begin{eqnarray}
D(m,n)&=&\bar{z}_o\frac{M}{4}\pm\frac{A}{\sqrt{2\pi}}\sqrt{M}+\rm{lower ~ order ~ in ~}M \rm{, ~on~ }\Gamma_L\nonumber \\
&=&-\bar{z}_o\frac{M}{4}\pm\frac{A}{\sqrt{2\pi}}\sqrt{M}+\rm{lower ~ order ~ in ~}M \rm{, ~on~  }\Gamma_R
\label{asymptot}
\end{eqnarray}
where $A$ is a constant independent of $N$ or $\Lambda$. The solution of
the equation (\ref{laplace2}) with the boundary condition $D(0,0)=0$ and that in equation (\ref{asymptot})
for large $|M|$ determines the final pattern.

\section{Numerical analysis}\label{ch3.7}
In both the examples in section $6$ and $7$ the patterns are characterized in terms
of the solution of the standard two dimensional lattice Laplace equation on the corresponding adjacency graphs.
The solution is well-known when $\left( m, n \right)\in\mathbb{Z}^{2}$ \cite{spitzer}.
In our case where the lattice sites form surfaces of multiple sheets, we have not been able to find a 
closed-form expression for $D(m,n)$. However, the
solutions can be determined numerically to very good precision by solving it on a finite
grid $-L\le m, n\le L$ with the corresponding boundary conditions imposed exactly at
the boundary.

For the pattern with the line sink, the calculation is performed with $D=M^{2/3}$ at the
boundary and then the solution is normalized to have $d\left( 1,0 \right)=-1$.
We determined $d$ and $e$ numerically for $L=100$, $200$, $300$, $400$ and $500$ and
extrapolated our results for $L\rightarrow \infty$. 
Comparison of the results from this numerical calculation and that obtained by
measurements on the pattern is presented in Table $1$. 
We consider the four different lengths $R_{1}$, $R_{2}$, $R_{3}$ and $R_{4}$ as
defined in Fig.\ref{fig:LSlengths}. By the
definition of the diameter of the pattern $R_{1}=2\Lambda$. We present the values
of $R_{2}$, $R_{3}$ and $R_{4}$ normalized by $R_{1}$ for different $N$. The asymptotic values of these lengths
are determined from the values of $d$ and $e$. Comparison of these results
shows very good agreement between the theoretical and the measured values.
\begin{figure}
\begin{center}
\includegraphics[scale=0.4]{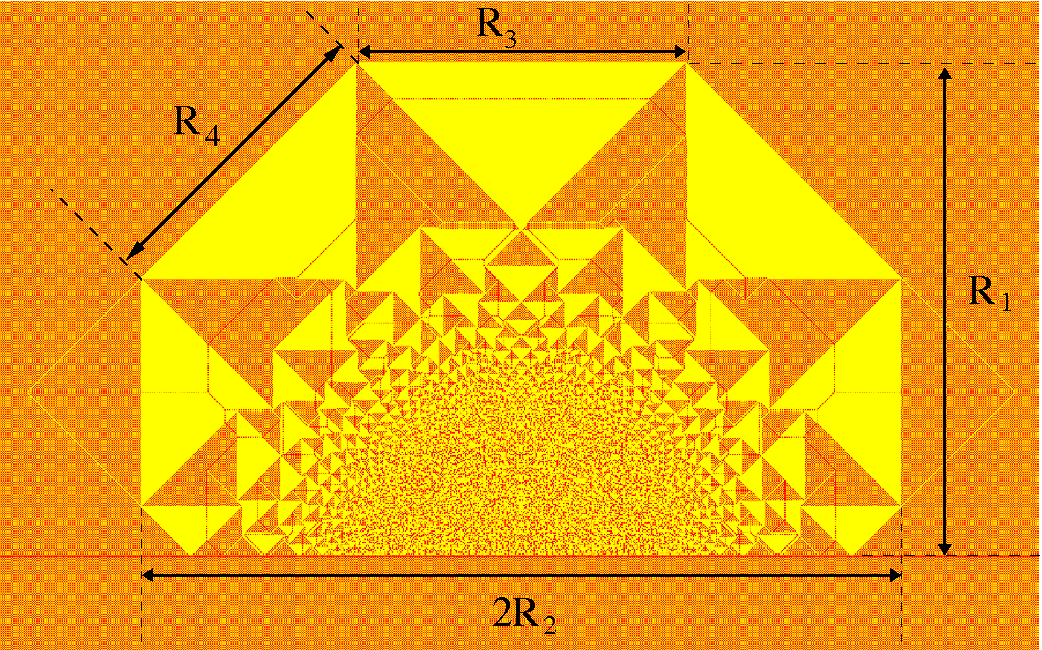}
\end{center}
\caption{ The spatial lengths $R_{1}$, $R_{2}$, $R_{3}$ and $R_{4}$ tabulated in Table $1$.}
\label{fig:LSlengths}
\end{figure}
\begin{table}
  \begin{center}
    \begin{tabular}{|c||c|c|c|c|c|}
      \hline
      ~N~ & $896k$ & $14336k$ & $57344k$ & $229376k$ & Theoretical \\
      \hline
      \hline
      $R_{2}/R_{1}$ & 0.769 & 0.768 & 0.770 & 0.770 & 0.7698  \\
      \hline
      $R_{3}/R_{1}$ & 0.675 & 0.675 & 0.667 & 0.668 & 0.6666 \\
      \hline
      $R_{4}/R_{1}$ & 0.609 & 0.609 & 0.617 & 0.616 & 0.6172 \\
      \hline
    \end{tabular}
    \caption{Comparison of different lengths measured directly from the pattern in Fig.\ref{fig:LSlengths}
for increasing values of $N$, with their theoretical values.}
  \end{center}
  \label{table:first}
\end{table}

A similar numerical calculation is done for the pattern with two sources. In this case
the boundary condition is given by equation (\ref{asymptot}). 
The value of $A$ is determined from a self 
consistency condition that the diameter of the pattern in the reduced
coordinate is $2$ which imposes $2e(-1,0)=-1$ corresponding to the vertex $A$
in Fig.\ref{fig:adj}. 
We determined $d$ and $e$ numerically
for $L=100$, $200$, $300$, $400$ and $500$ and extrapolated our results for 
$L\rightarrow \infty$. 
A comparison of the results from this numerical calculation and that obtained by
measurements on the pattern are presented in Table $2$. We
considered five different spatial lengths in the pattern, corresponding to $r_{o}=0.800$.
These different lengths are drawn in Fig.\ref{fig:lengths} and their values rescaled by $\sqrt{N}$, for the patterns
with increasing $N$, are given in Table $2$. The asymptotic values of these lengths are
obtained using the values of $d$ and $e$. The rescaled
lengths extrapolated to the infinite $N$ limit match very well with the theoretical results.
\begin{figure}
\begin{center}
\includegraphics[scale=0.3]{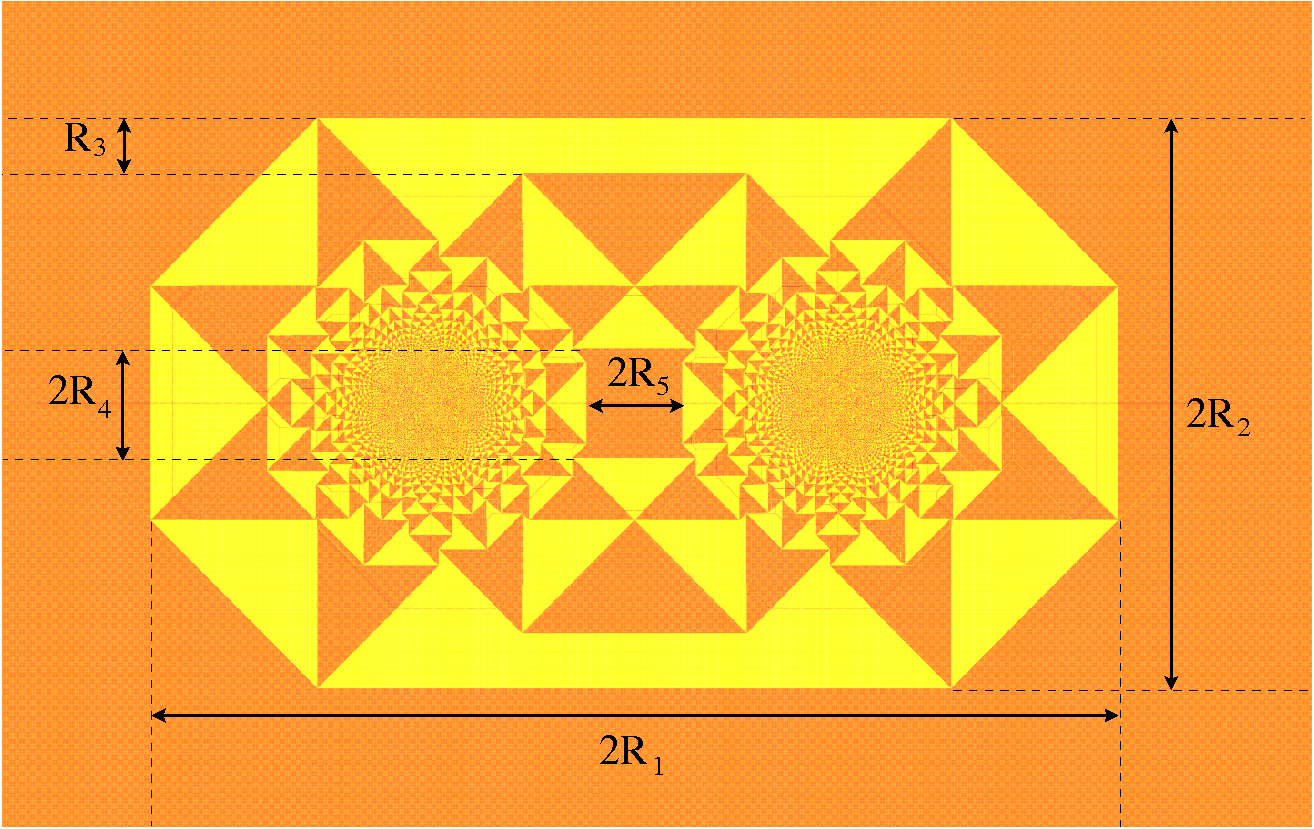}
\end{center}
\caption{ The spatial lengths $R_{1}$, $R_{2}$, $R_{3}$, $R_{4}$ and $R_{5}$ tabulated in Table.$2$}
\label{fig:lengths}
\end{figure}
\begin{table}
  \begin{center}
    \begin{tabular}{|c||c|c|c|c|c|c|}
      \hline
      ~N~ & $2.5k$ & $10k$ & $40k$ & $160k$ & $640k$ & Theoretical \\
      \hline
      \hline
      $R_{1}/\sqrt{N}$ & 1.84 & 1.84 & 1.84 & 1.83 &  1.83 & 1.82  \\
      \hline
      $R_{2}/\sqrt{N}$ & 1.06 & 1.07 & 1.07 & 1.06 & 1.05 & 1.06 \\
      \hline
      $R_{3}/\sqrt{N}$ & 0.22 & 0.21 & 0.20 & 0.19 & 0.18 & 0.18 \\
      \hline
      $R_{4}/\sqrt{N}$ & 0.18 & 0.19 & 0.19 &  0.18 & 0.18 & 0.18 \\
      \hline
      $R_{5}/\sqrt{N}$ & 0.20 & 0.22 & 0.21 & 0.21 & 0.21 & 0.21 \\
      \hline
    \end{tabular}
    \caption{Comparison of different lengths measured directly from the two source pattern
for $r_{o}=0.800$ with their theoretical values.}
  \end{center}
  \label{table:second}
\end{table}

\section{Discussion}\label{ch3.8}
While the results discussed in quantifying the patterns with growing sandpiles are presumably exact (in the
sense that $D\left( m, n \right)$ can be determined to arbitrary precision),
they have not been established rigorously. In particular, as noted before, it would be desirable to
have a direct proof of the proportional growth property from the definition of the problem. Also, we
use the observation that the asymptotic pattern consists of only two types of patches, and the
adjacency graph of the pattern is also taken as observed. It would be nice to see it following from
the definition of the problem. The unexpected accuracy of the scaling arguments giving Eqs. ($34$, $36$, $37$, $38$)
also deserves to be understood better.

We have shown that the exact characterization of the patterns
in the F-lattice on a checkerboard background reduces to solving  a discrete
Laplace equation on the adjacency  graph of the pattern. For the single source pattern
this graph is a square grid on a two-sheeted Riemann surface and in
the presence of a line sink it is on a three-sheeted Riemann surface.
This Riemann surface structure occurs for other sink geometries as well
and the number of sheets can be determined from the way
$\phi$ diverges near the origin. 

If the potential $\phi(r)$ diverges as $ r^{-a}$ near the origin, then the corresponding
complex function $\Phi(z) \sim z^{-a}$.  Then $\frac {d^2}{dz^2} \Phi \sim z^{-2-a}$. In
all the cases studied above, the patch to which point $z$ belongs is characterized by
integers $(m,n)$, where $\frac {d^2}{dz^2} \Phi  \sim m+in$. Also $\frac{d}{dz} \Phi \sim
d+ie $. Writing $D=d+i e$, and $M = m+in$, we see that $D \sim M ^{\frac{1+a}{2+a} }$.
This then gives the number of Riemann sheets. For example, for the wedge angle
$\omega=2\pi$, we have $a=1/2$. Then $D \sim M^{3/5}$, and the Riemann surface would have 5 sheets. 
\begin{figure}
\begin{center}
\includegraphics[scale=0.21]{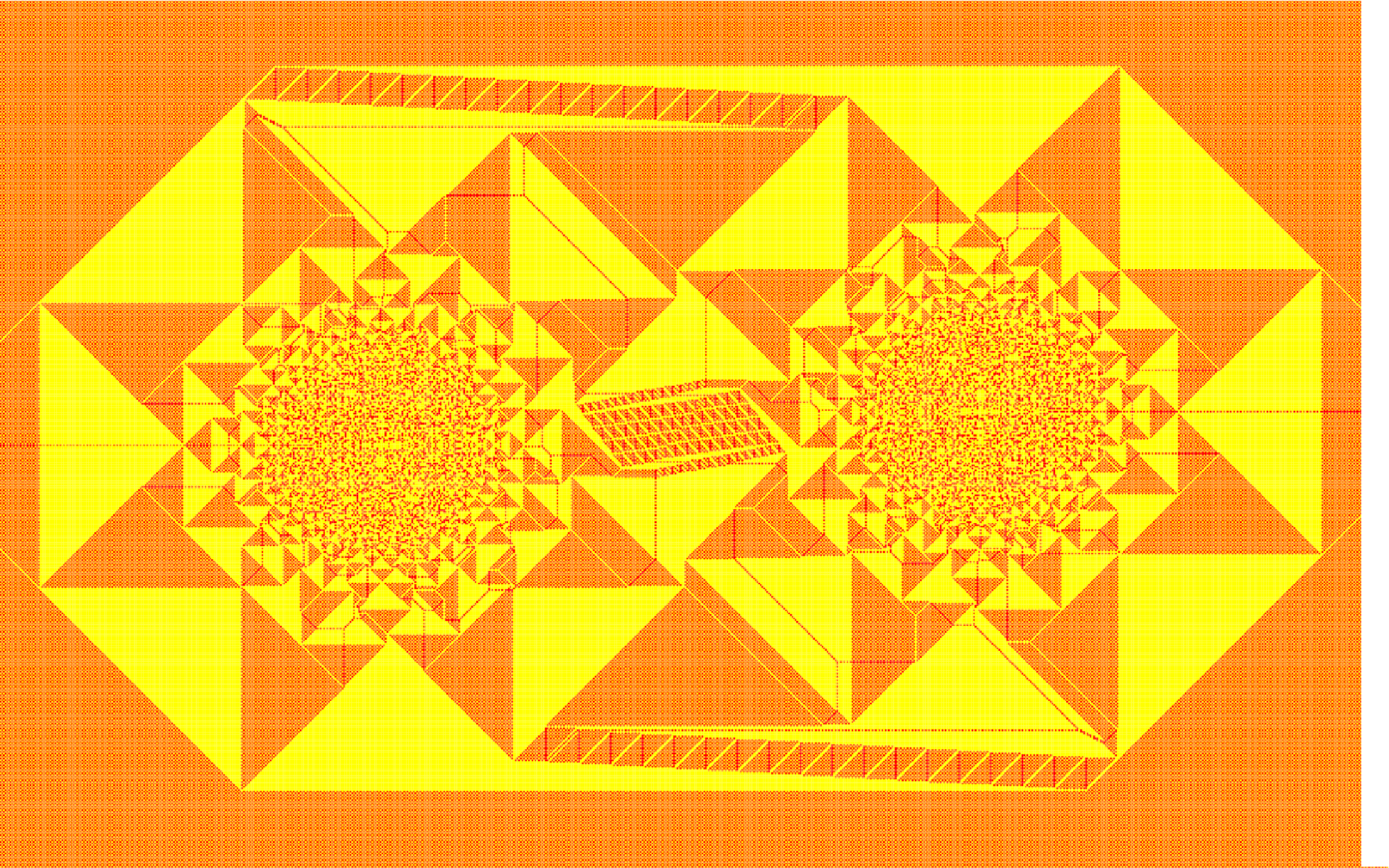}
\caption{ Pattern produced by adding $N=40000$ grains each at ($-180,0$) and 
($180,20$) on the F-lattice with initial checkerboard distribution of grains and 
relaxing. Color code red=0 and yellow=1. (Details can be seen in the online
version using zoom in )}
\label{fig:distr}
\end{center}
\end{figure}

The patterns discussed so far in this chapter have only two types of patches with densities
$1/2$ and $1$. But it is possible to have patterns with patches of other densities. For example,
for the patterns with two sources, even a slight deviation
of the position of the second source in Fig.\ref{fig:twosource} from the x-axis
introduces patches which have areal density different from $1/2$ or $1$. One
such pattern produced by adding $40000$ grains each at $\left( -180, 0 \right)$
and $\left( 180, 20 \right)$ is shown in Fig.\ref{fig:distr}. The regions with
stripes of red and yellow are patches of the new density. In addition, the boundaries of
these patches have slopes other than $0$, $\pm 1$ and $\infty$. Most of
the analysis presented here is appilcable to this pattern, except that the matching
conditions along the common boundary between two patches and the adjacency graph are different.

The cases in which the full pattern can be explicitly determined are clearly special. 
For example, one of the conditions used for the exact characterization of the patterns in this chapter is that
inside each patch the height variables are periodic and hence $\Delta\rho\left( \mathbf{r} \right)$
is constant. It is easy to check that this condition is not met for most sink geometries.
For example, patterns  of the type discussed in Section $4$
with any $\omega$ other than integer multiples of $\pi/4$ have aperiodic
patches. In such cases, the present treatment for characterization of patterns is
clearly not applicable. However, the scaling analysis for the growth of
the spatial lengths in the pattern with $N$ is still valid.

The function $D=d+ie$ satisfies the discrete Cauchy-Riemann condition (equation (\ref{cr})). These functions are known as discrete  holomorphic
functions in the mathematics literature. Usually they have been studied  for a square grid of 
points on the plane \cite{duffin,spitzer}. While more general discretizations of the plane have
been discussed  \cite{mercat,laszlo}, not much is known about the behavior of such functions 
for multi-sheeted Riemann surfaces. A perturbative approach of
determining these functions for square discretization of multi-sheeted
Riemann surface is presented in the Appendix \ref{apndx:dhf}.

In our analysis we have also used the fact that the patterns have nonzero average overall 
excess density ( i.e. $C_2$ in Eq. (\ref{nr3}) is nonzero). 
The case $C_2=0$ is quite different, and requires a substantially different treatment.
We discuss such patterns in the next chapter.

\chapter{Pattern Formation in Fast-Growing Sandpiles}
\textit{Based on the paper \cite{myjsm}} by Tridib Sadhu and Deepak
Dhar.
\begin{itemize}
\item[\textbf{Abstract}]
We study the patterns formed by adding  $N$ sand-grains at a single
site on an initial periodic background in the
Abelian sandpile models, and relaxing the configuration. When the heights at
all sites in the initial background are low enough, one gets patterns
showing proportionate growth, with the diameter of the pattern formed
growing  as $N^{1/d}$ for large $N$, in $d$-dimensions. On the other
hand, if sites with
maximum stable height in the starting configuration form an infinite
cluster, we get  avalanches that do not stop. In this chapter, we
describe our unexpected finding of an interesting class of
backgrounds in two dimensions,  that  show  an  intermediate behavior:
For any $N$, the avalanches are finite, but  the diameter of the
pattern increases   as  $N^{\alpha}$, for large $N$, with $1/2 <
\alpha \leq 1$.  Different values of $\alpha$ can be realized  on
different backgrounds, and the patterns still show proportionate
growth. The non-compact nature of growth simplifies their
analysis significantly. We characterize the asymptotic pattern exactly for one
illustrative example with $\alpha=1$.
\end{itemize}
\section{Introduction}\label{sec:ch4.1}
In the last two chapters, we studied the patterns produced by adding
grains at a single site in a Deterministic Abelian Sandpile Model (DASM), and relaxing. A complete relaxation
process, starting from the addition of sand to reaching the final
stable configuration is called an avalanche. The length of an
avalanche depends on the initial height configuration,
\textit{i e.}, the background. For some
backgrounds on an infinite lattice, topplings may continue for ever,
and the avalanches reach to infinity. For other backgrounds, where the
avalanches are finite, the toppled
sites form patterns in the spatial configuration of sites with
different values of the height variables.

The backgrounds leading to infinite
avalanches have been termed as \textit{explosive}.
\begin{figure}
\begin{center}
\includegraphics[width=10cm,angle=0]{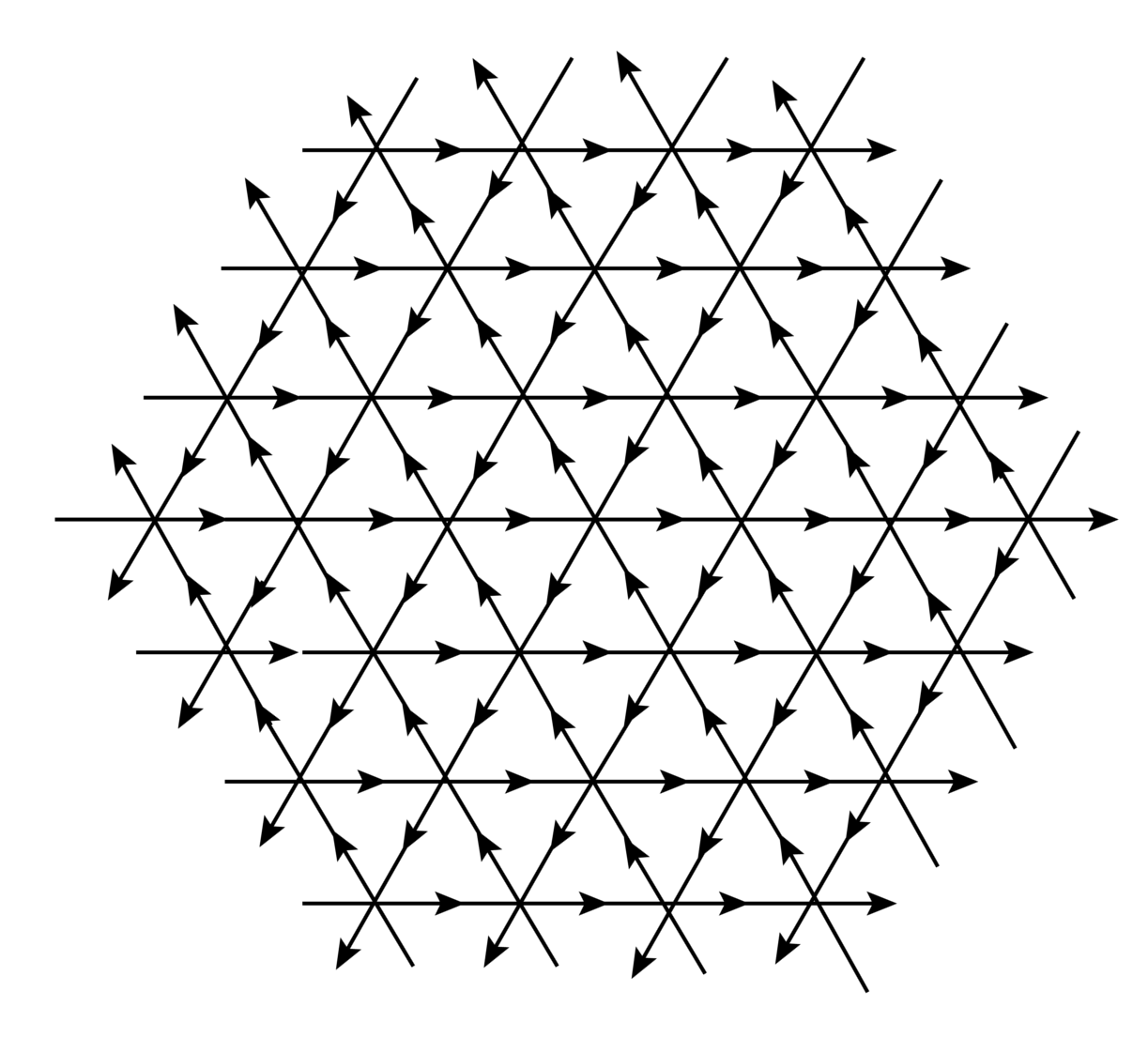}
\caption{A directed triangular lattice.}
\label{fig:trilattice}
\end{center}
\end{figure}
In the other backgrounds with finite avalanches the diameter of the pattern
usually grows as $N^{1/d}$ with increasing number of added grains in
$d-$dimensions. We call this type of growth as \textit{compact} growth. All the
patterns studied so far showed compact growth. In this chapter, we
describe a remarkable class of patterns where the diameter remains
finite for any finite $N$, but grows as $N^{\alpha}$, with
$1/d<\alpha\le 1$. We call this type of growth as \textit{non-compact} growth.
Characterization of these patterns, as will be shown, is simpler than
the ones with compact growth.

For a sandpile model with stochastic toppling rule, the size of the pattern is
determined only by the density $\rho_{o}$ of heights in the background,
and the specific arrangement of heights does not matter. There exist a
critical density $\rho_{c}$, depending on the toppling rules,
such that, on a background with sub-critical density $\rho_{o} < \rho_{c}$,
finite avalanches occur with probability $1$. The corresponding
asymptotic pattern is a simple circle,
with density $\rho=\rho_{c}$, inside, and
the diameter $2\Lambda$ of the circle growing as $\sqrt{N}$. For densities
$\rho_{o}>\rho_{c}$, probability of finite avalanches vanishes in the
large $N$ limit.

For a deterministic sandpile a similar critical density can not be defined.
One can construct backgrounds with densities very close to zero, and
still there are infinite avalanches. A simple example of such
backgrounds, on any lattice, is the one where the sites with height
$z_{c}-1$ form an infinite connected cluster, with $z_{c}$ being the threshold height.
Height at other sites could be zero, and thus the density
$\rho_{o}$ of the background could be made very small.
On the other hand, it is
possible to construct backgrounds with mean density arbitrarily close to
$z_{c}-1$, and yet the avalanches are always finite \ref{explosion}.
Absence of an infinite avalanche depends on the detailed arrangement
of heights in the background, and not on the density alone.
\begin{figure}
\begin{center}
\includegraphics[scale=0.155,angle=0]{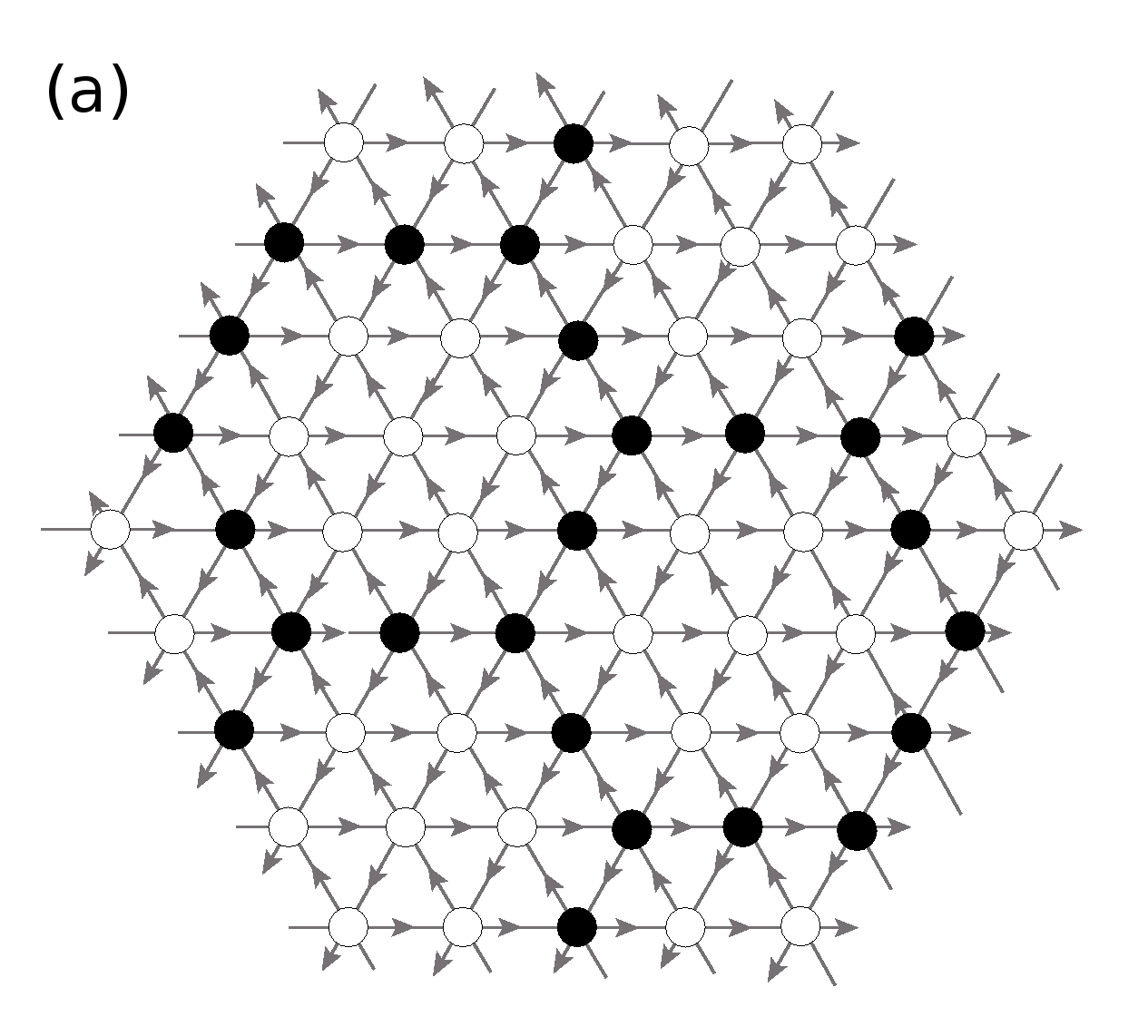}
\includegraphics[scale=0.155,angle=0]{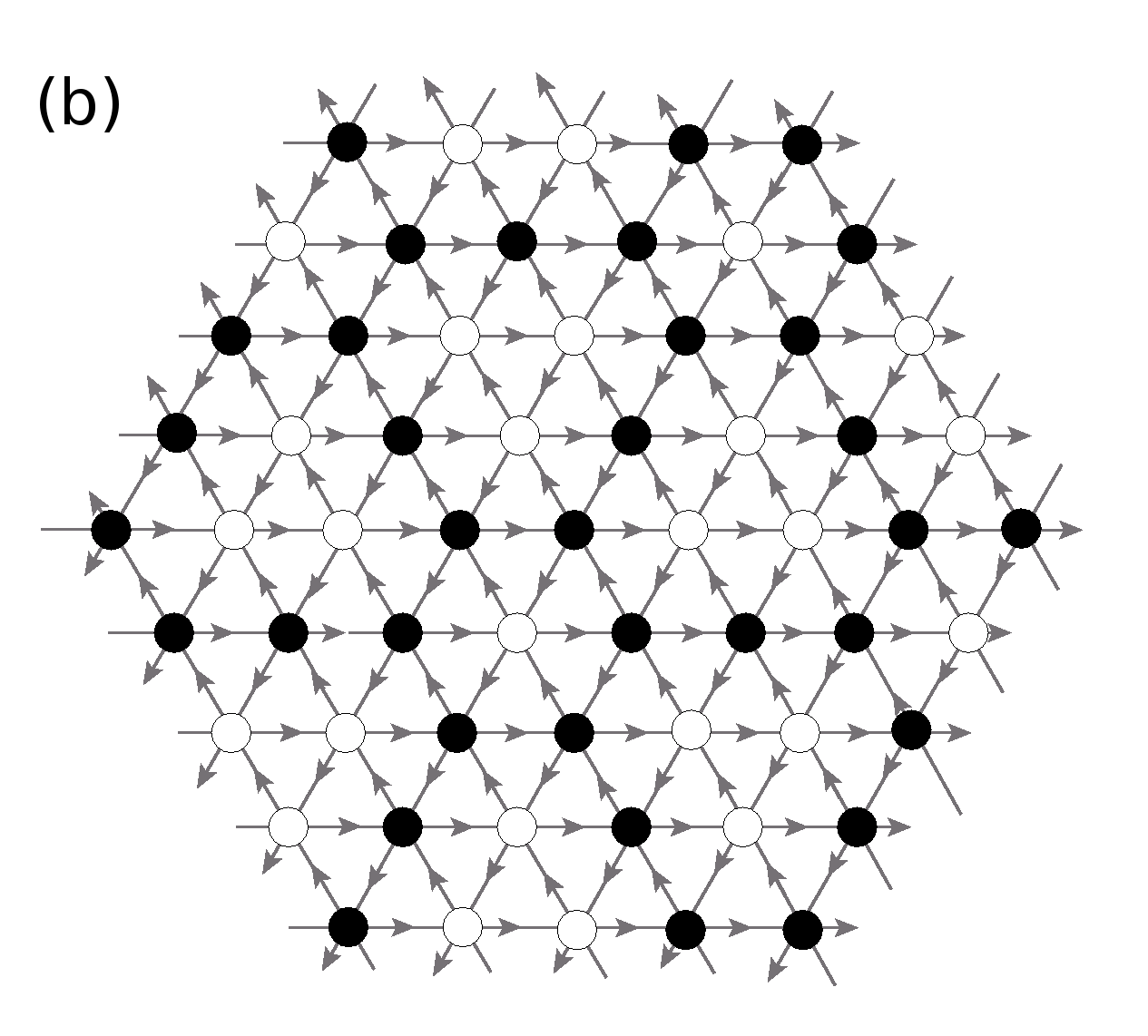}
\caption{Examples of the background of class I and II, respectively. The filled
circles represent height $1$ and unfilled ones $2$.}
\label{fig:tribg}
\end{center}
\end{figure}

There are some earlier work on the growth rate of the sandpile
patterns. Some backgrounds of both types, explosive and non-explosive, for a
deterministic ASM were studied in
\cite{explosion,fey}. In all the examples, studied so far,
the background is either explosive or the growth of the patterns is compact.
For a deterministic ASM on a square lattice, it was shown \cite{fey}, that the pattern produced on a background of constant height
$z\le z_{c}-2$, is always enclosed inside a square whose width grows as $\sqrt{N}$.
Given the absence of any
critical density, it is non-trivial to find a background on which the
patterns grow with a rate faster than $\sqrt{N}$, but finite. In fact,
for the sandpile models on a standard square
lattice there are no known examples of patterns with non-compact
growth.
\begin{figure}
\begin{center}
\includegraphics[scale=0.7,angle=0]{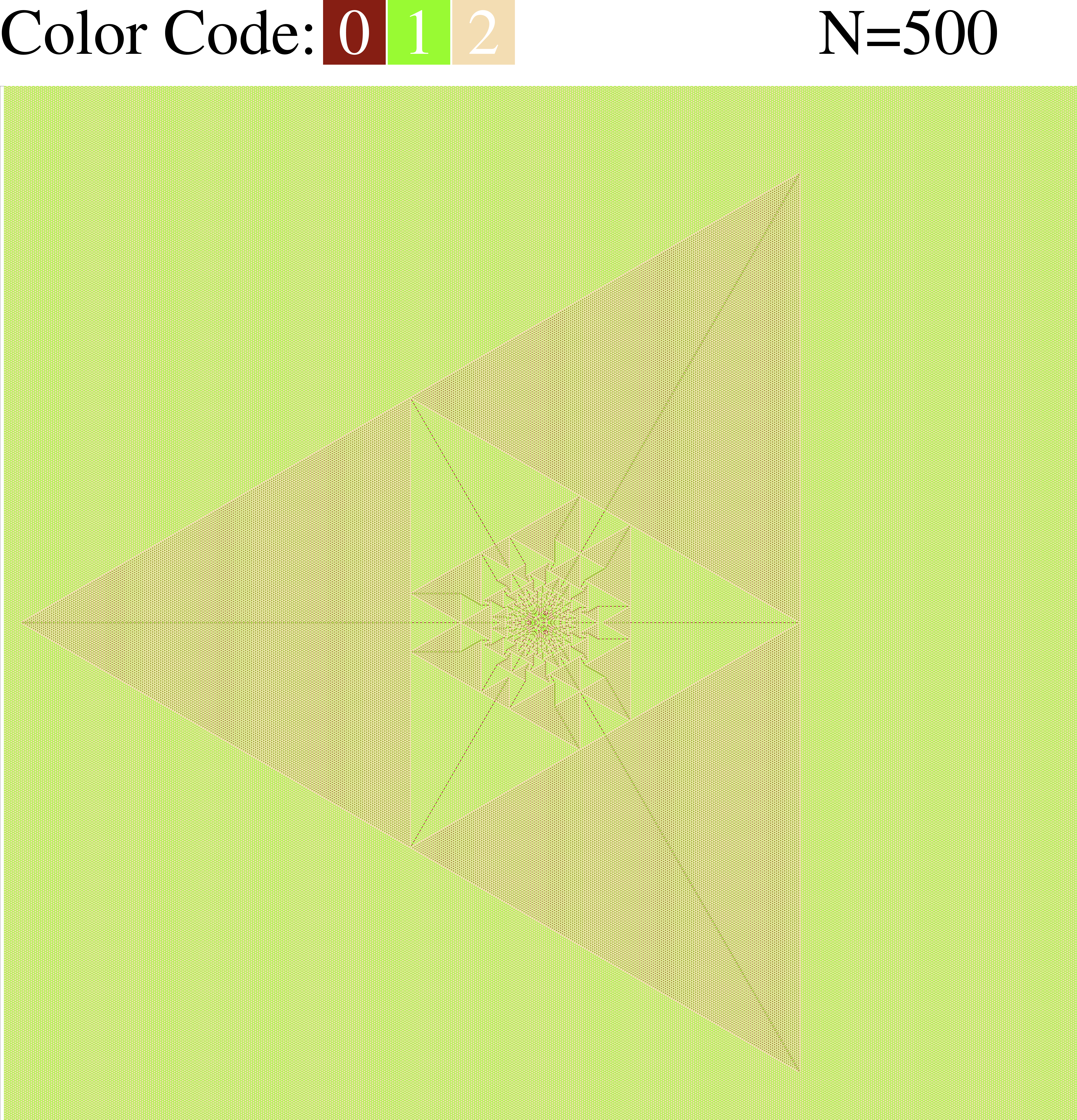}
\caption{(color in the electronic copy) The pattern formed on the background in figure
\ref{fig:tribg}$(a)$by adding $N$ particles
respectively at the origin. Details can be viewed in the electronic
version using zoom in.}
\label{fig:hex}
\end{center}
\end{figure}
\begin{figure}
\begin{center}
\includegraphics[scale=0.7,angle=0]{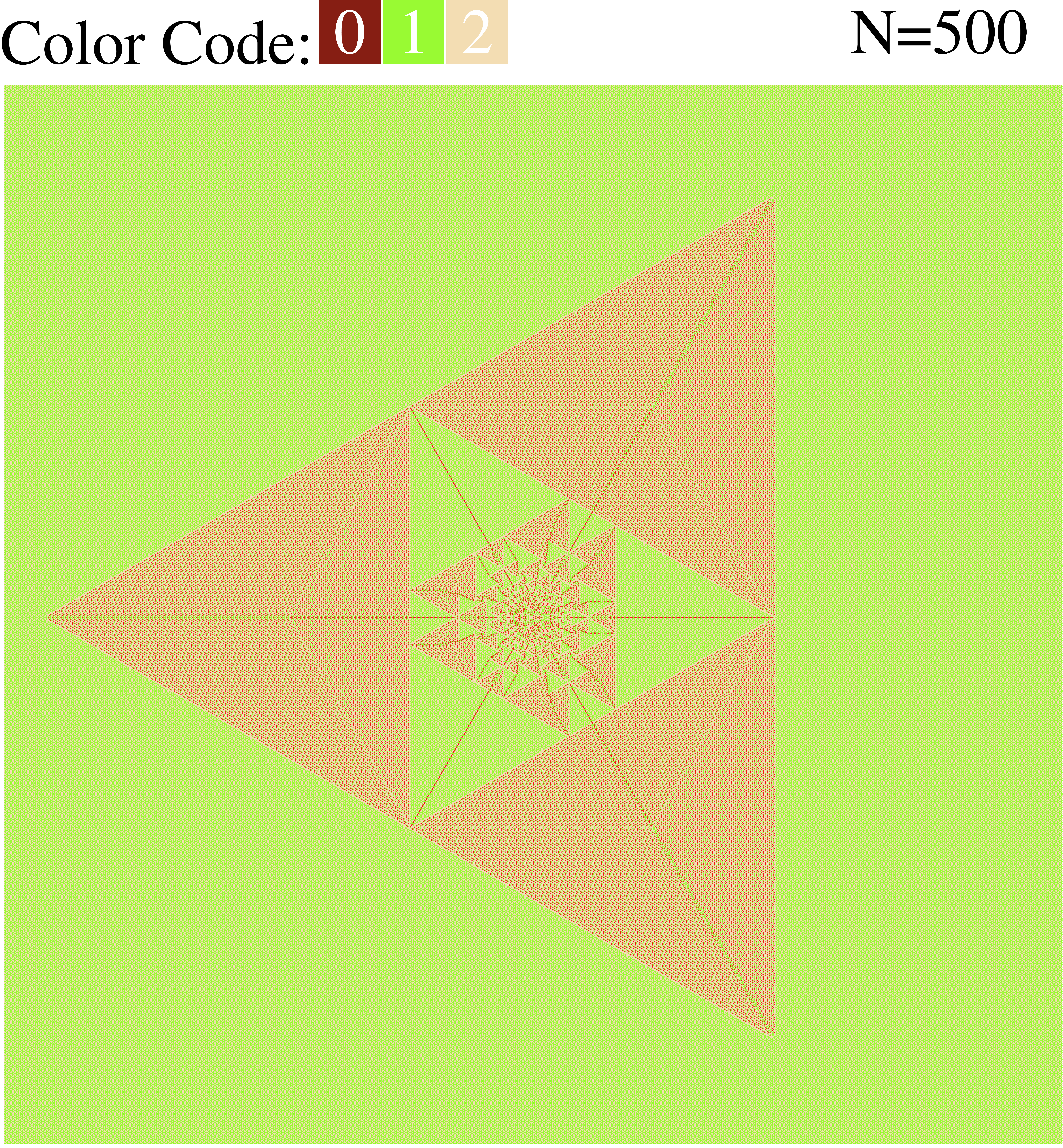}
\caption{ (color in the electronic copy) The pattern formed on the background in figure
\ref{fig:tribg}$(b)$ by adding $N$ particles
respectively at the origin. Details can be viewed in the electronic
version using zoom in.}
\label{fig:tri}
\end{center}
\end{figure}

We found two classes of backgrounds, both infinite, on a directed
triangular lattice (see Fig. \ref{fig:trilattice}), for which the
growth is proportionate, with the growth exponent $\alpha=1$. Examples
of these classes of backgrounds and patterns formed on them are shown
in Fig. \ref{fig:tribg}, Fig. \ref{fig:hex} and Fig. \ref{fig:tri}. Our
numerical study shows, but we have no formal proof, that different
backgrounds belonging to the same class produce the same asymptotic
pattern. In addition, we found infinitely many backgrounds on the
F-lattice which produce patterns with proportionate non-compact
growth. However, in these cases the growth exponent $\alpha$ takes a
different value, with $1/2 < \alpha < 1$ for each member.

We also discuss the exact characterization of the pattern shown in
Fig. \ref{fig:hex}, one of the two asymptotic patterns we have found
with $\alpha=1$. This is described, as in the earlier studied case of
compact growth (see previous chapters), in terms of the scaled
toppling function. However, the analysis of non-compact patterns is
actually simpler. Clearly, for $\alpha>1/d$, the mean excess density
of particles in the toppled region is zero, in the asymptotic
patterns. Infact, the patterns are made of large patches
where heights are periodic, and inside each patch, the mean
density is exactly the same as in the background, and the excess
grains are concentrated along the patch boundaries. There are also
some boundaries where excess grains density is negative. We show that
this leads to the scaled toppling function being a piece-wise linear
function of the rescaled coordinates. Thus, in each patch, the
potential function is specified by only three coefficients. In
contrast, for the compact patterns discussed in chapters \ref{ch2} and
\ref{ch3}, the scaled toppling function is a quadratic function of the
coordinates in each patch, and one has to determine six coefficients
for each patch, to determine the function fully.

We are able to reduce the problem of determining the asymptotic
pattern in Fig. \ref{fig:hex} to that of finding the lattice Green's
function on a hexagonal lattice. This is known to be expressible as
integrals that can be evaluated in closed form (see Appendix \ref{ap:laplace}), and
this leads to a full solution of the problem. This is in contrast to
the characterization of the compact growth patterns in previous
chapters, where one requires solution of Discrete Laplace's equation
on a square grid on Riemann surfaces of multiple sheets, and there are
no known closed form expression for the solution. 

This chapter is organized as follows: In section \ref{sec:cnc}, we discuss, in
details, how different periodic background configurations give rise to
different rates of growth.
In section \ref{sec:enc} we define a deterministic ASM on the directed triangular lattice, and describe the two
classes of periodic backgrounds that produces patterns with non-compact growth.
In section \ref{sec:pwl} we argue that, for any pattern with
non-compact proportionate growth the rescaled toppling function is
piece-wise linear. In section \ref{sec:anc}, we discuss exact
characterization of the simplest of the non-compact growth patterns
with $\alpha=1$. Patterns on the F-lattice, with $1/2<\alpha < 1$ are
discussed in section \ref{sec:ncfl}. The section \ref{sec:tropical}
contains some discussion about connection to tropical polynomials.

\section{Compact and non-compact growth}\label{sec:cnc}
The simplest growing patterns are found in the  Manna-type sandpile
models with stochastic toppling rules \cite{sasm}. In these models, when the
density of particles $\rho_{o}$ in the background is small,
the avalanches are always finite. In the relaxed configuration, the
toppled sites form a nearly circular region (see Fig.
\ref{fig:mannapattern}). The asymptotic pattern seems to be perfectly
circular disc of uniform density, with an average density $\rho^\star$
inside the circle and $\rho_{o}$ outside. The value of $\rho^\star$ is
independent of the background density $\rho_{o}$, and is equal to the
unique steady state density of the corresponding self organized
critical model with random sites of addition, and dissipation at the
boundary \cite{sasm}. The region inside the circle forgets about the
initial height configuration, and is in the self-organized critical
state. The boundary of the affected region is thin with a sharp
transition of density from $\rho^\star$ to $\rho_{o}$ (see Fig. \ref{fig:transition}). Then considering
that, for large $N$, all the added grains are confined inside the
circular region of diameter $2\Lambda$, we get
\begin{equation}
N=\left( \rho^\star-\rho_{o} \right)\pi \Lambda^{2}+\rm{Lower~ order~
in~}\Lambda.
\label{eq1}
\end{equation}
Thus the pattern has a compact growth.

For densities $\rho_{o}$ close to, but below  $\rho^\star$, sometimes a
single particle addition can lead to very large increase in the size of the toppled region. However,
probability of such large jumps decreases exponentially with size, and
for any finite $N$, with probability $1$, avalanches remain finite. As
long as $\rho_{o}$ is less than $\rho^\star$, the system relaxes, forming a
pattern whose diameter grows as $\sqrt{N}$. Adding a single grain on a
background of super-critical density ($\rho > \rho^\star$) gives rise to infinite
avalanches, with non-zero probability. Then, with probability
$1$, such backgrounds will lead to an infinite avalanche for some
finite value of $N$.
In higher dimensions also, a similar  behavior is expected.
\begin{figure}
\begin{center}
\includegraphics[width=12.0cm]{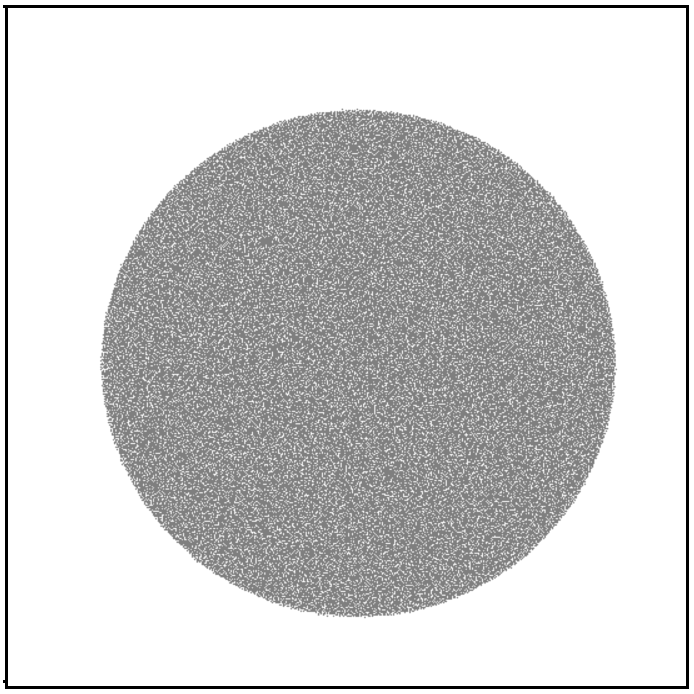}
\caption{The pattern produced by adding $N=10^5$ grains at a single site on a
stochastic ASM defined on an infinite square lattice and relaxing; Initial
configuration with all sites empty. The
threshold height $z_{c}=2$, and on toppling two grains are transfered
either to the vertical or horizontal neighbors, with equal
probability. Color code: White=0, and Black=1.}
\label{fig:mannapattern}
\end{center}
\end{figure}

\begin{SCfigure}
\includegraphics[width=10cm]{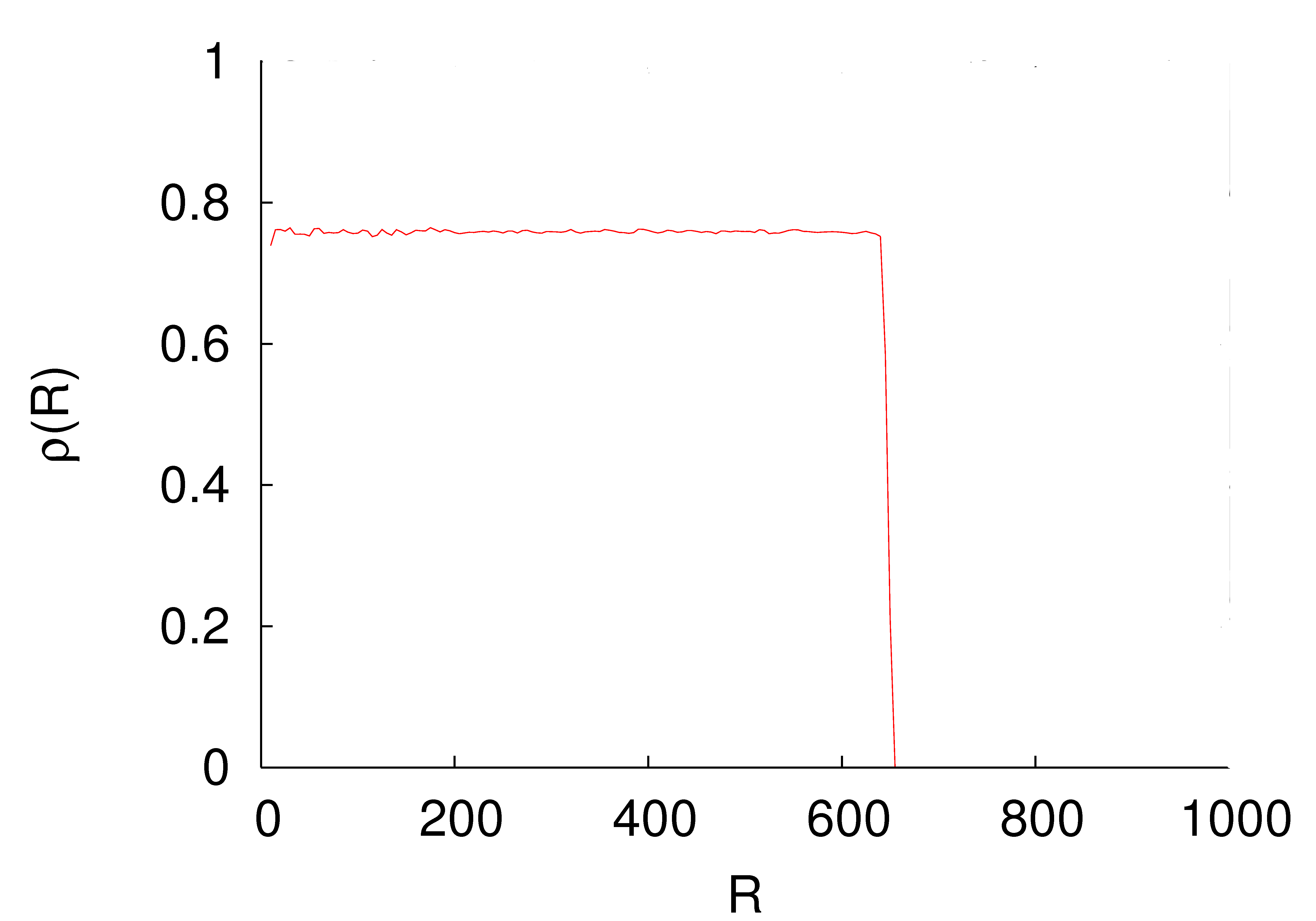}
\caption{The variation of density $\rho$ along the radius of the
circular pattern in Fig. \ref{fig:mannapattern}.}
\label{fig:transition}
\end{SCfigure}

At $\rho_{o}=\rho_{c}$ the probability of large avalanches has a power
law tail, and the above argument does not apply. It may be possible to
construct a robust background with critical density of heights. However, no such example has been
found, so far.

In the models with deterministic relaxation rules there is no well defined
critical density $\rho_{c}$, separating the explosive and
non-explosive backgrounds. The geometry of height distribution plays
the determining role in the robustness of a background.

In the deterministic models, similar circular pattern is produced for a background with random assignment of
heights $z\le z_{c}$ per site, of small average density $\rho_{o}$.
Inside the pattern, density of heights is $\rho_{s}$ which is the steady
state density of the corresponding SOC model.
However, on a background of higher densities, this picture is changed
considerably.
For example consider the BTW model on a square lattice, where the steady state density $\rho_{s}=17/8=2.125$.
It has been shown that a background with a random assignment of height $3$ with probability $\epsilon$, on
a sea of constant height $2$ is explosive, even for arbitrary
small value of $\epsilon$ \cite{explosion}, although
the average density $\rho_{o}=2+\epsilon$ is much less than $\rho_{s}$. 

For a background with periodic heights, it is possible to construct
explosive backgrounds of any density, even with arbitrary small
values. As an example, we consider the BTW model on an infinite square lattice. Define a background made of square unit cells of width $m$ with
empty sites inside the cell, and $3$ grains at each site in the
boundary. For $m>1$, the average density
$\rho_{o}=6/m-9/m^{2}$. If any of the occupied site receives a grain, it starts a chain of toppling events where
all the occupied sites topple, and the avalanche reaches infinity.

On the other hand a constant background of height $3$ at all sites is
a minimally stable configuration \textit{i.e}, addition of a single
grain will produce an infinite avalanche.
Still, it is possible to construct a robust background with density arbitrarily
close to $3$. For example, consider the background in the previous
example, and exchange the height variables: make the occupied
sites empty, and fill the empty sites with $3$ grains. It has been
shown \cite{explosion}, that this background is stable and the
pattern produced has a compact growth.

We will show in the next section that, there is a large class of
backgrounds, with a range of densities, for which the growth is less
than explosive, but more than compact.

\section{Examples of non-compact growth}\label{sec:enc}
We first discuss the patterns with $\alpha =1$. We start with an ASM
on a directed graph corresponding to an infinite two dimensional
triangular lattice, with each site having three incoming and three
outgoing arrows (see Fig. \ref{fig:trilattice}). The threshold
height $z_{c}=3$, for each site. If the height at any site is above or
equal to $z_{c}$, it is unstable, and relaxes by toppling: in each
toppling, three sand grains leave the unstable site, and are
transferred one each along the directed bonds going out of the site.

We consider two classes of backgrounds on this lattice:
\begin{itemize}
\item[]\underline{\textbf{Class $I$:}} We consider the lattice as made of triangular plaquettes, which are joined together to make tiles in the shape of regular hexagons  with edges of length $l$. We cover the two-dimensional plane with these tiles.  Sites that lie on the boundaries of these hexagons are assigned height $1$, and the rest of the sites have height $2$.  Figure 
\ref{fig:tribg}$(a)$ shows the background configuration for the case $l=2$. 

\item[]\underline{\textbf{Class $II$:}}
For these backgrounds, we cover the two-dimensional plane with tiles
in the shape of equilateral triangles of edge-length $l$. The sites
that lie on the boundaries of the triangles, and are shared by two
triangles, are assigned  height  $1$, and remaining sites are assigned
height $2$. Sites that are at the corners of triangular tiles, and
shared by six of them, are also assigned heights $2$. The background
configuration corresponding to $l=4$ is shown in figure \ref{fig:tribg}$(b)$. The pattern
made of triangular tiles with $l=3$ is same as the class I background
with hexagon of edge-length $1$. Hence, only patterns formed  with
triangles of edge-length $ l \ge 4$ will be said to be in this class.
\end{itemize}

The patterns produced by adding $N$ grains, where $N$ is large, at a single
site on the two backgrounds in Fig. \ref{fig:tribg} are shown in
Fig. \ref{fig:hex} and \ref{fig:tri}. While the patterns look quite
similar, a closer examination shows that they are not identical. In
Fig. \ref{fig:tri}, there are extra lines of particles within the
brownish patches which break each patch into smaller parts. In fact,
with the identification of some patches having only a point in common,
as discussed later, we can show that each patch breaks into exactly
three patches. These three parts have similar periodic pattern, but
with different orientations.
\begin{figure}
\begin{center}
\includegraphics[width=8.8cm]{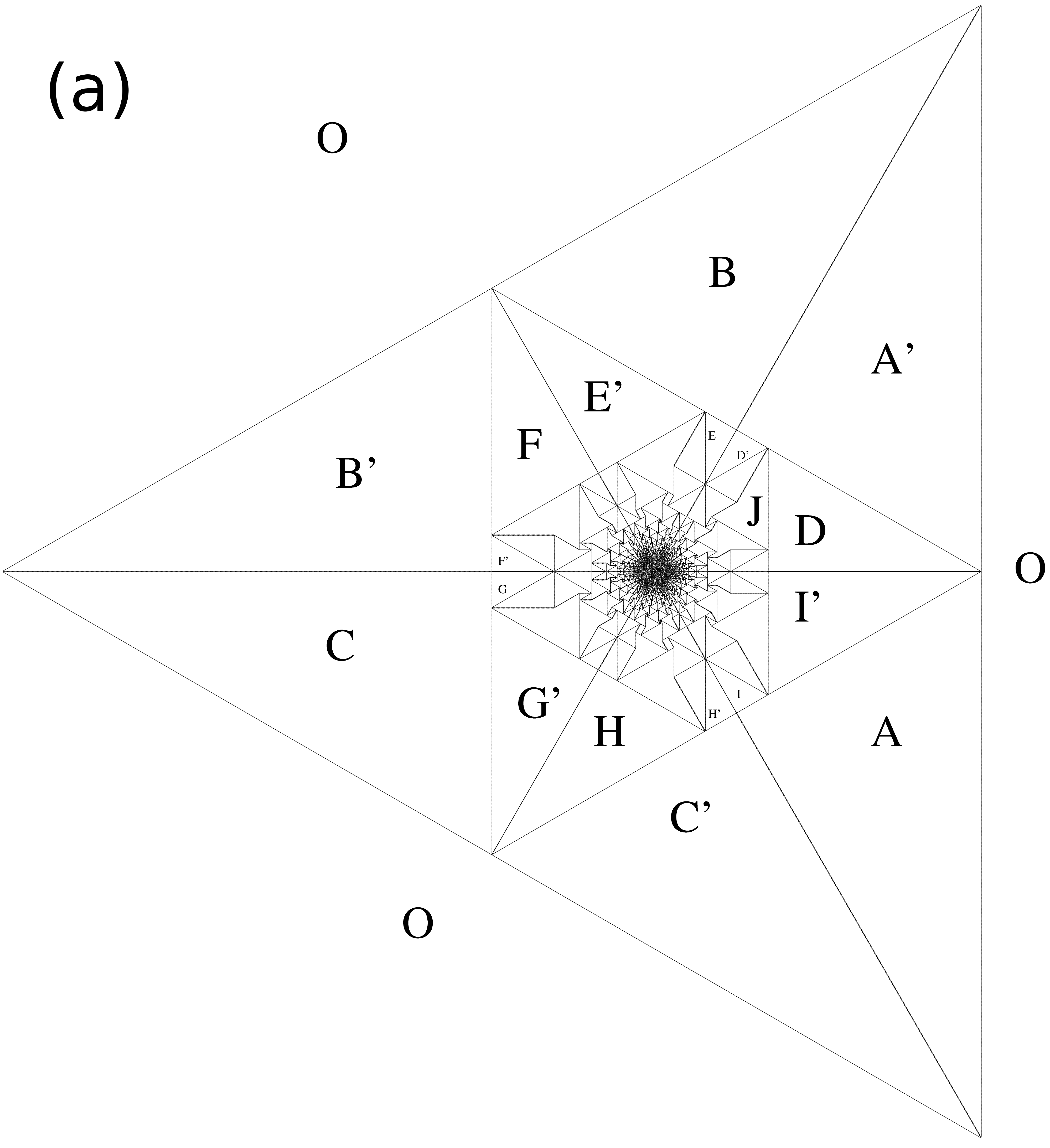}\\
\includegraphics[width=8cm]{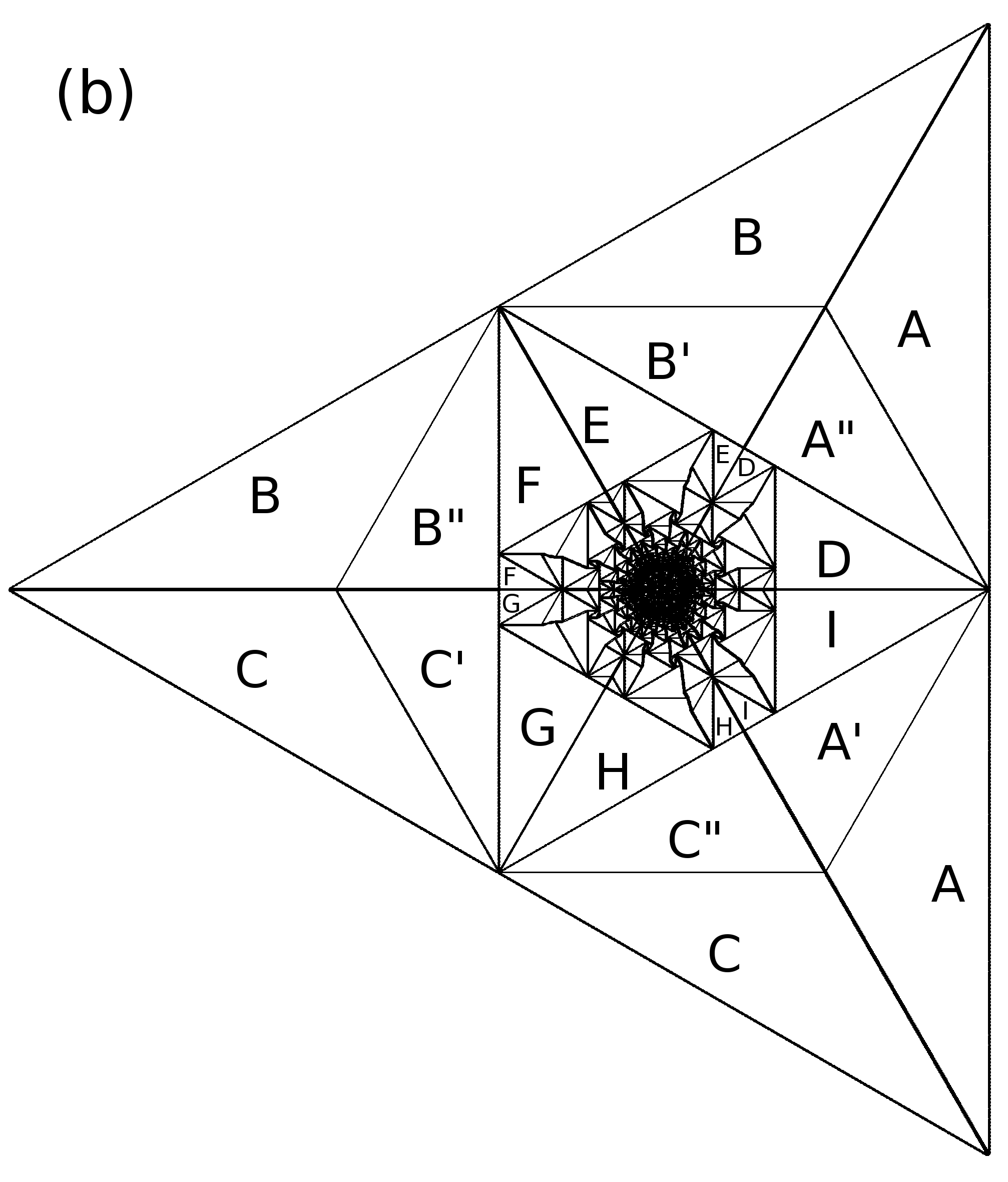}
\caption{The patterns in terms of $Q\left( r \right)$, corresponding
to those in Fig. \ref{fig:hex} and \ref{fig:tri}. Sites with zero
$Q\left( r \right)$ are colored white, and non-zero are colored black.
The larger patches are given identifying labels.}
\label{fig:line}
\end{center}
\end{figure}

This differences can be seen more clearly in terms of the net excess
change in height ${Q\left( R \right)}$ in a unit cell centered at
$R$, where the unit cell is that of the background pattern.
\begin{equation}
Q\left( R \right)=\sum_{R'\in \textrm{unit cell}}\Delta z(R+R'),
\end{equation}
where $\Delta z\left( R \right)$ is the change in height at site $R$.
For example in the first background in Fig. \ref{fig:tribg} a unit
cell is a hexagon of edge length $l=2$, and for the second background
it is a parallelogram of each side length $l=4$. A site that is on the edge of
the unit cell is counted with weight $1/2$, and a site on the corner
of the hexagon with weight $1/3$, and on the corner of the
parallelogram with weight $1/4$. By construction, the function
$Q\left(r\right)$ is zero inside each patch, and non-zero along the
boundaries between patches. The patterns in terms of these variables,
corresponding to those in Fig. \ref{fig:hex} and \ref{fig:tri} are
shown in Fig. \ref{fig:line}.
\begin{SCfigure}
\includegraphics[scale=0.73,angle=0]{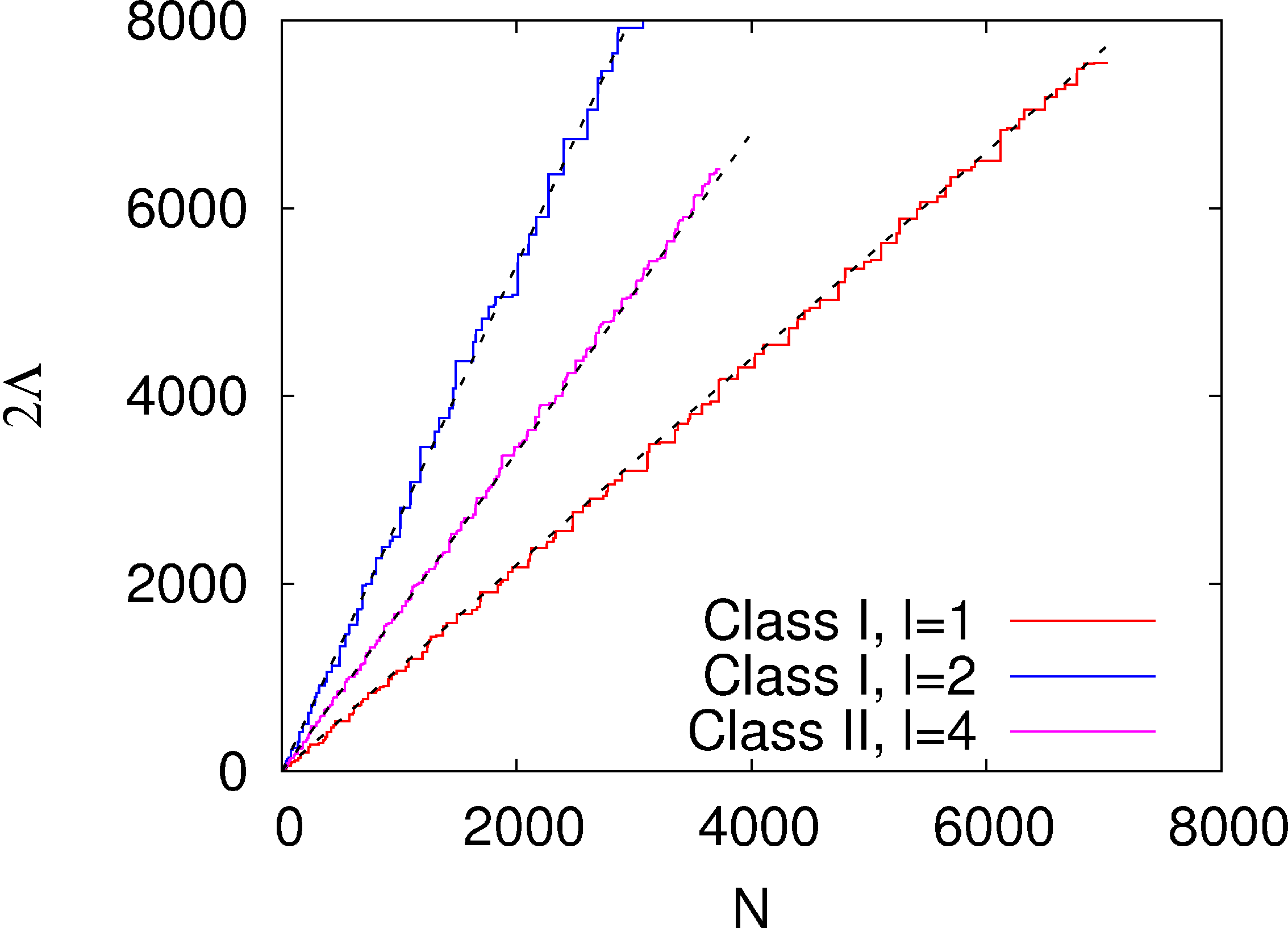}
\caption{(color in the electronic copy) The diameter $2\Lambda$ of the patterns as a function of the
number $N$ of  added  grains. The cases shown are  (i) Class I, $l=1$,
(ii) Class $I$, $l= 2$, and (iii) Class $II$, $l=4$. The corresponding
straight line fits have slopes given by $1.1$, $2.7$ and $1.7$ respectively.}
\label{fig:triln}
\end{SCfigure}

We have seen that the patterns on these two classes of backgrounds
exhibit proportionate growth, \textit{i.e.}, all the spatial features
inside the patterns for large $N$, grow at the same rate with the
diameter. We define the diameter $2\Lambda$, in general, for any pattern in
this paper, as the height of the smallest rectangle containing it. For
the patterns in Fig. \ref{fig:hex} and Fig. \ref{fig:tri}, it is then
the length of a side of the bounding
equilateral triangle.
This particular choice makes $2\Lambda$ as
an integer multiple of $\sqrt{3}$, on the triangular lattice. We find that for both types of backgrounds
in Fig. \ref{fig:tribg}, the diameter of the pattern grows linearly
with $N$ (Fig. \ref{fig:triln}).
\section{Piece-wise linearity of the toppling function\label{sec:pwl}}
Considering the proportionate growth, let us define a rescaled
coordinate $\vec r=\vec R /N^{\alpha}$,
where $\vec R\equiv\left( x,y \right)$ is the position vector of a
site on the lattice.
The number of topplings at any site inside the pattern, scales linearly with
$N$. Let us define
\begin{equation}
\phi\left( \vec r \right)=\lim_{N\rightarrow
\infty}\frac{T_{N}( \vec R )}{N}.
\end{equation}
We now show, using an extension of the argument given in chapter
\ref{ch2},
that the function $\phi$ is linear inside periodic patches
in all the patterns with non-compact growth, \textit{i.e.}, with
$\alpha>1/2$. Within a patch, the function
$\phi(\vec r)$ is expandable in Taylor series around any point
$\vec r_{o}$, not on the boundary of the patch. Defining $\vec r_{o}\equiv\left( \xi_{o},
\eta_{o} \right)$, and
$\Delta\vec r_{o}\equiv\left( \Delta\xi, \Delta\eta \right)$ we have
\begin{equation}
\phi\left( \xi_{o}+ \Delta \xi, \eta_{o}+\Delta\eta \right)-\phi\left(
\xi_{o}, \eta_{o} \right)=d\Delta\xi+e\Delta\eta+\mathcal{O}\left(
\Delta\xi^{2},\Delta\eta^{2},\Delta\xi\Delta\eta \right).
\end{equation}
Consider any term of order $\ge 2$ in the expansion, for example, the term
$\sim(\Delta\xi)^2$. This can only arise due to a term $\sim (\Delta
x)^2N^{1-2\alpha}$ in the toppling function $T_{N}(
\vec R)$.
Then, considering the fact that $T_{N}(
\vec R )$ is an integer function of $x$ and $y$, it is easy
to see that this term would lead to discontinuous changes in
$T_{N}(\vec R)$ at intervals of $\Delta x \sim
\mathcal{O}(N^{\alpha-1/2})$. As $\alpha>1/2$ for non-compact growth
patterns, this leads to a change in the periodicity of heights at such intervals
inside each patch which themselves are of size $\sim N^{\alpha}$. This
would then result in many defect lines within a patch, in the pattern at large $N$. However
there are no such features in Fig. \ref{fig:hexpicl1}.
Therefore inside each periodic patch,
$\phi(\vec r)$ must be exactly linear in $\vec r$. In fact, it turns out
that the integer toppling function $T_{N}(\vec R)$ is exactly linear inside
a patch even for any finite $N$, except for an additional periodic term of
periodicity equal to that of the heights inside the patch.
\begin{SCfigure}
\includegraphics[width=9cm,angle=0]{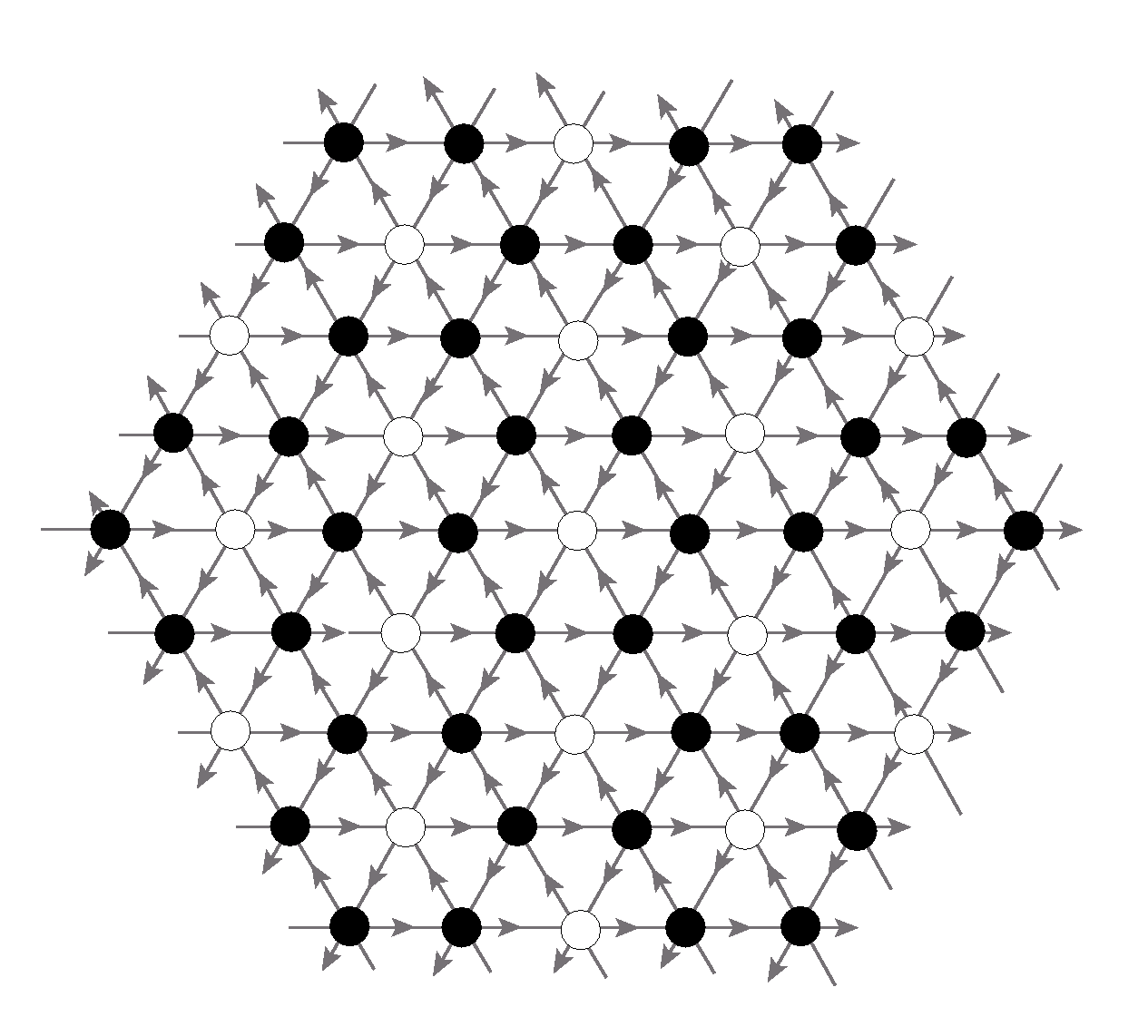}
\caption{The background of class I corresponding to $l=1$. The filled
circles represent height $z=1$ and unfilled ones $z=2$.}
\label{fig:bkg0}
\end{SCfigure}

Another consequence of the  exact linearity of the potential function in each
patch is that all patch boundaries in the asymptotic pattern  are straight lines.

The argument finally relies on the two observed (not rigorously
established) features of the
patterns, \textit{i.e.}, there is proportionate growth, and that the
patterns can be decomposed in terms of periodic patches which are
themselves of size $\mathcal{O}\left( N^{\alpha} \right)$.
\begin{figure}
\begin{center}
\includegraphics[width=12cm]{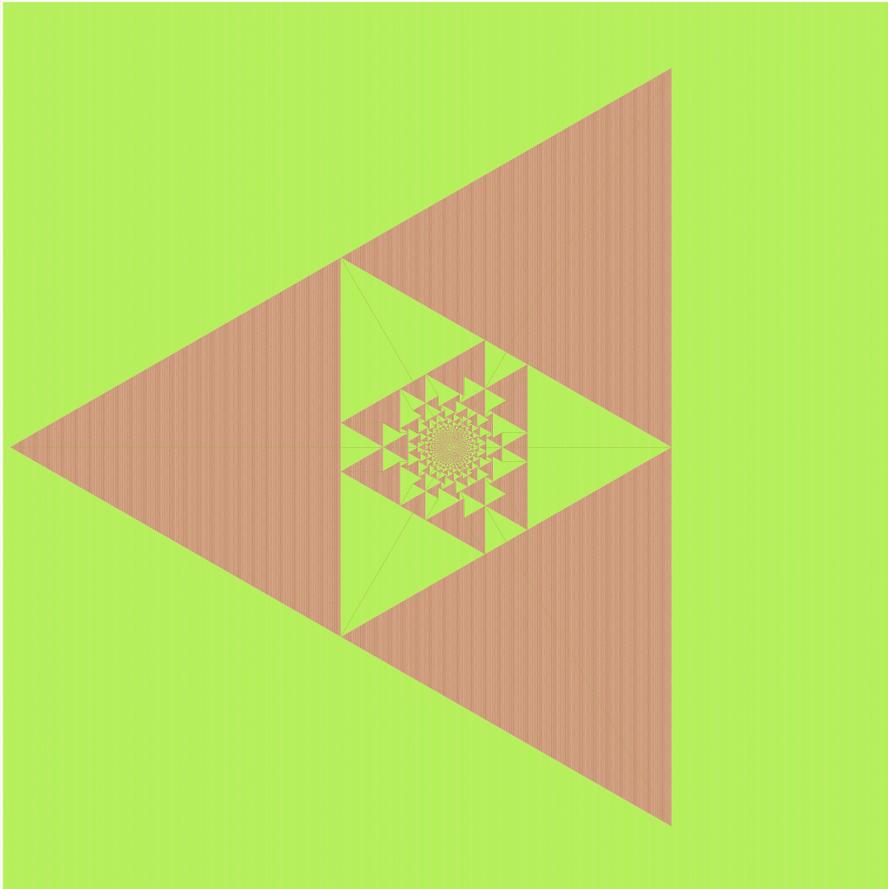}
\caption{(color in the electronic copy) The pattern produced by adding $N=3760$ grains at a single site
on the background in Fig. \ref{fig:bkg0}, and relaxing. Details can be
seen in the electronic copy using zoom-in}
\label{fig:hexpicl1}
\end{center}
\end{figure}

Let us write the toppling function $T_{N}(\vec R)$ within a single
patch $P$ as
\begin{equation}
T_{N}( \vec R )=A_P+\vec K_{P}\cdot
\vec R+T_{periodic}( \vec R ),
\label{eq:intT}
\end{equation}
where $T_{periodic}( \vec R )$ is a periodic function of
its argument with zero mean value. If $\hat{e}_{1}$ and
$\hat{e}_{2}$ are the basis vectors at the unit cell of the periodic
pattern then we have
\begin{eqnarray}
T_{N}( \vec R+\hat{e}_{1} )-T_{N}( \vec R)&=&\vec K_{P}\cdot
\hat{e}_{1},\nonumber \\
T_{N}( \vec R+\hat{e}_{2} )-T_{N}( \vec R )&=&\vec K_{P}\cdot
\hat{e}_{2}.
\end{eqnarray}

As $T_{N}( \vec R )$ are integer valued functions,
$\vec K_{P}\cdot \hat{e}_{1}$ and
$\vec K_{P}\cdot\hat{e}_{2}$ can only take integer values. If
$\hat{g}_{1}$ and $\hat{g}_{2}$ are the unit vectors in the reciprocal
space of the super lattice of the periodic pattern,
\begin{equation}
\hat{g}_{i}\cdot\hat{e}_{j}=\delta_{ij}\textrm{~ ~ ~ ~ }i,j=1,2,
\end{equation}
then $\vec K_{P}$ must be an integer linear combination of
$\hat{g}_{1}$ and $\hat{g}_{2}$, and can be written as
\begin{equation}
\vec K_{P}=n_1 \hat{g}_{1}+ n_2 \hat{g}_{2},
\label{eq:recp1}
\end{equation}
where $n_1$ and $n_2$ are some integers. For example, in the background
pattern in Fig. \ref{fig:bkg0}, a choice
of the basis vectors and its reciprocal vectors is
\begin{eqnarray}
\hat{e}_{1}\equiv\left( \frac{3}{2}, \frac{\sqrt{3}}{2}
\right);&&\mathrm{ ~ ~ ~ ~ ~   }
\hat{e}_{2}\equiv\left( \frac{3}{2}, -\frac{\sqrt{3}}{2}  \right)
\nonumber\\
\hat{g}_{1}\equiv\frac{2}{3}\left( -\frac{1}{2},
-\frac{\sqrt{3}}{2} \right);&&\mathrm{ ~ ~ ~ ~ ~   }
\hat{g}_{2}\equiv\frac{2}{3}\left( -\frac{1}{2},\frac{\sqrt{3}}{2} \right).
\label{eq:recp}
\end{eqnarray}
The fact that $K_{_P}$ is constant
inside a patch, implies that the patches can be labeled by the pair
of integers $\left( n_1, n_2 \right)$.

An interesting consequence of this
linear dependence of $T_{N}( \vec R )$ is that there are no transient structures within
the patches. On increasing $N$, if the $A_{_P}$ function
increases, all sites in the patch $P$, except possibly those at the patch
boundaries, undergo same number of additional
topplings.

\section{Characterizing the class $I$ asymptotic patterns\label{sec:anc}}
We now discuss characterization of the asymptotic pattern of class I,
showing $\alpha =1$. In this section we quantitatively characterize
the asymptotic pattern for the case $l=1$. The background
configuration is shown in Fig. \ref{fig:bkg0}.  A site on the
triangular lattice can be labeled uniquely by a pair of integers
$\left(p, q\right)$, such that its position on a complex plane can be
written as $\mathbf{R}=p+q \omega$, where $\omega=\exp\left( i2\pi/3 \right)$ is
a complex cube root of unity. Then, the height variables in the
background pattern in Fig. \ref{fig:bkg0}, can be written as
\begin{eqnarray}
h(p+q\omega) &=&2 \textrm{ if } p + q =0 \textrm{ (mod 3),}\nonumber\\
  &=&1 \rm{~ ~ ~otherwise.}
\end{eqnarray}
The average height in the background, $\langle z \rangle=4/3$. The
configuration of the pile  produced on this background, by adding
$N=3760$ grains at the origin is shown in Fig. \ref{fig:hexpicl1}.
\begin{SCfigure}
\includegraphics[width=8cm,angle=0]{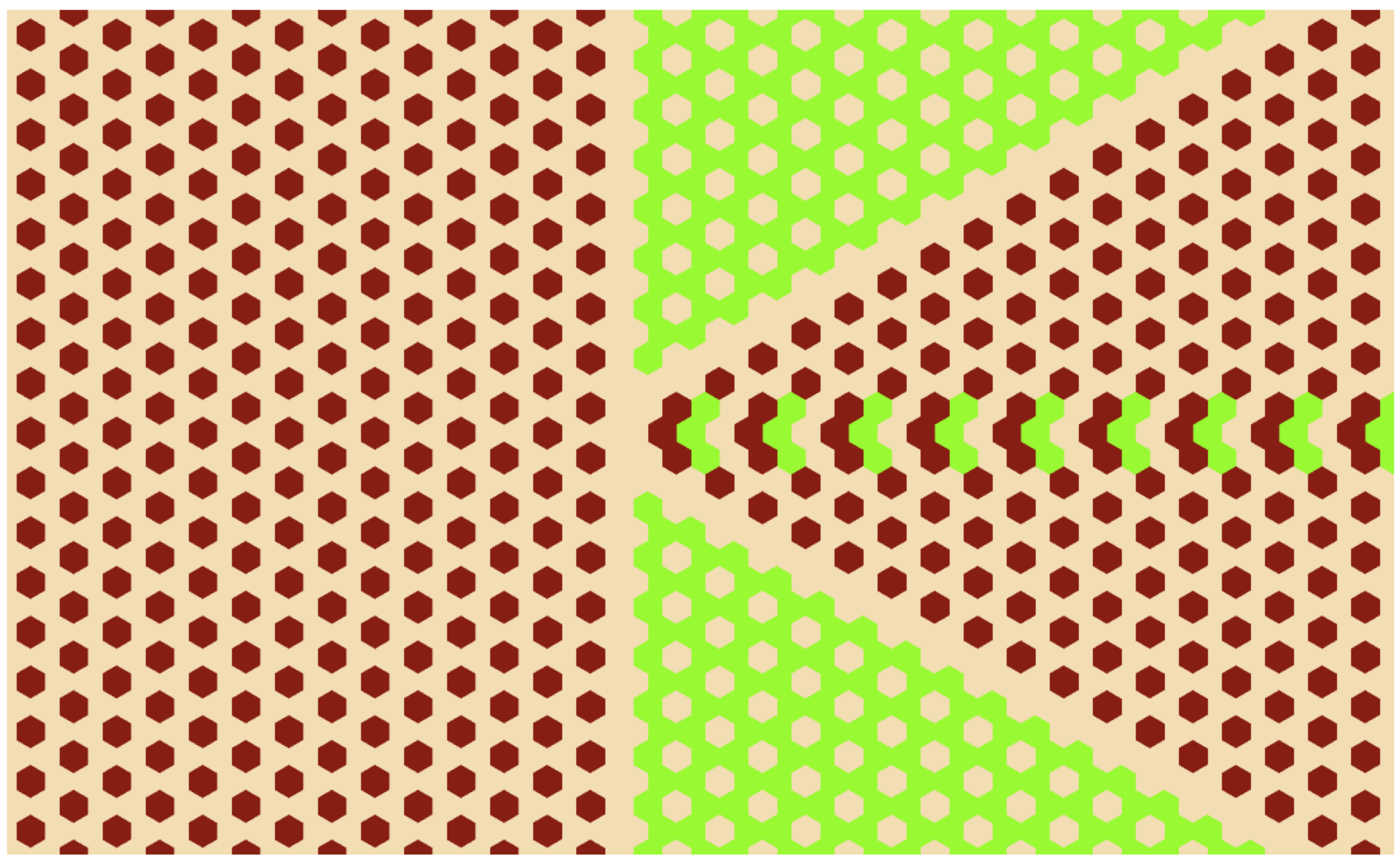}
\caption{(color in the electronic copy) An example of patch boundaries in Fig. \ref{fig:hexpicl1}
meeting each other. Each filled hexagon represents Wigner cell around
a site, and the color in them denotes height of that site. The color code is same as in
Fig. \ref{fig:hexpicl1}.}
\label{boundary}
\end{SCfigure}

We see that the sites toppled due to addition of the grains are
confined within an equilateral triangle. The pattern can be thought of
as a union of patches, inside which the heights are periodic. A
zoom-in showing the height configuration with five patches meeting at
a point is shown in Fig. \ref{boundary}. There are only two types of
periodic patches seen: one is like the background, where the sites of
height $2$ are surrounded by sites of height $1$, and the other with
heights $0$ surrounded by heights $2$. Then, the average height
$\langle z \rangle$ inside both types of patches are same. In fact,
it is equal to that of the background, $\langle z \rangle=4/3$.

The patches in the outer region of the pattern are big, and they
become smaller, and more numerous as we go inwards. Along the common
boundary of adjacent patches, we see line-like defect structures, and
only along these lines the density is different from the background. In Fig. \ref{boundary}, one
can also see the periodicity of the structures along the patch
boundaries.  Some patch boundaries, like the horizontal boundary in
Fig. \ref{boundary}, have a deficit of particles compared to the background.

The boundaries of the patches are seen more clearly in terms of
$Q(\mathbf{R})$ variables, as shown in figure
\ref{fig:line}(a), where we have labelled different patches as
$\mathbf{A, A', B, B'}...$ \textit{etc}..
\begin{SCfigure}
\includegraphics[width=8cm,angle=0]{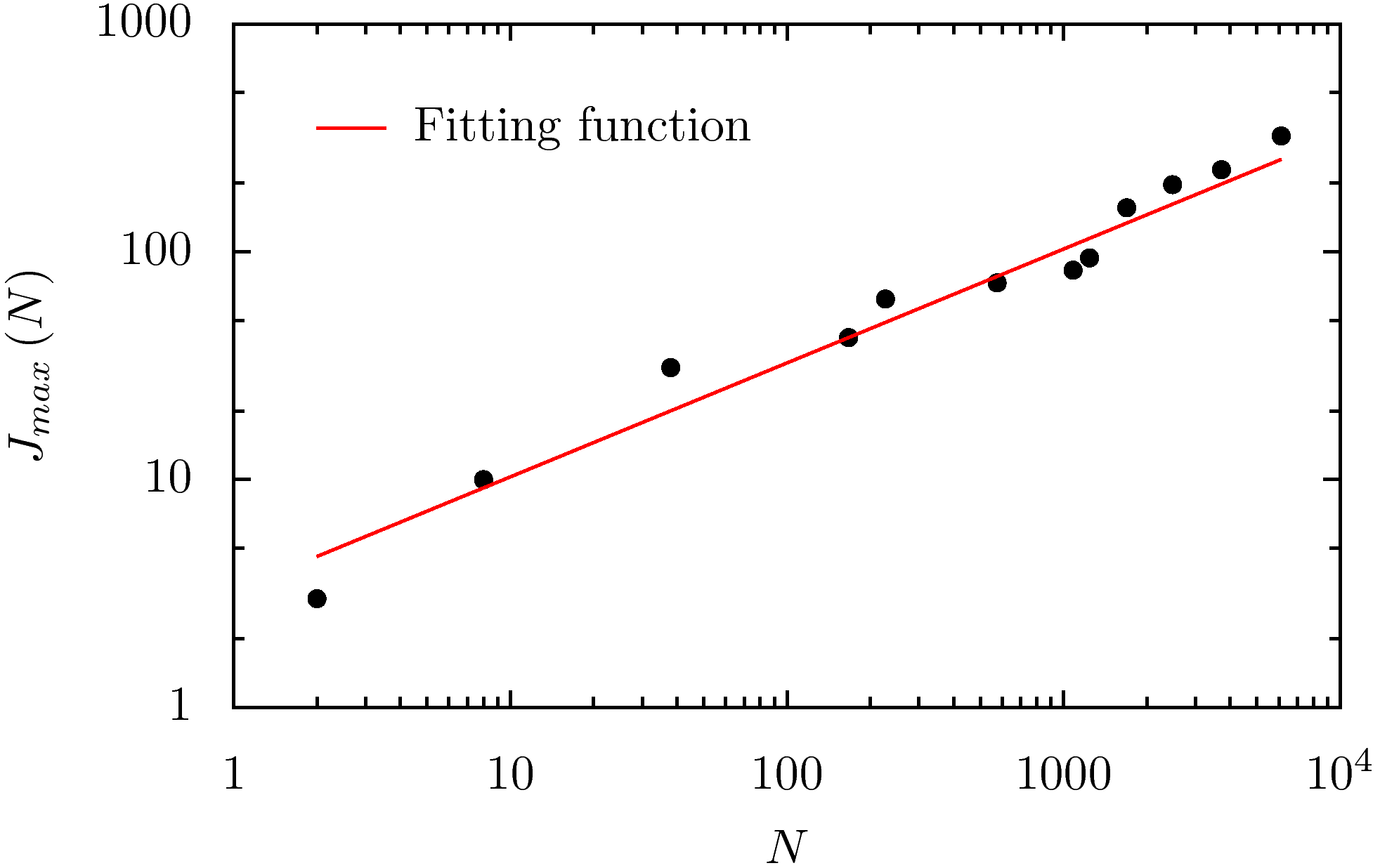}
\caption{$J_{max}(N)$ for the background in Fig.
\ref{fig:bkg0} has a square root dependence on $N$ with a fitting function $3.25\sqrt{x}$.}.
\label{jump}
\end{SCfigure}

The dependence of $2\Lambda$ on $N$ for this background is shown in
Fig. \ref{fig:triln}. We see that the diameter for the pattern  grows
asymptotically linearly with $N$, but it grows in bursts: it remains constant for a
long interval as more and more grains are added, and suddenly
increases by a large amount at certain values of $N$. For example, at
$N=3721$,
the $2\Lambda$ is $2276 \sqrt{3}$, and it jumps to a value
$2408 \sqrt{3}$ when one more
grain is added. Let $J_{max}(N_m)$ denote the size of the
maximum jump in $2\Lambda$ encountered, as $N$ is varied from $1$ to
$N_m$. In Fig. \ref{jump}, we have plotted the variation of $J_{max}(N_m)$ with $N_m$. The graph is consistent with a
power-law  growth, with a power around $1/2$. Thus the fractional size
of the bursts decreases for large $N$.
\begin{figure}
\begin{center}
\includegraphics[width=12cm,angle=0]{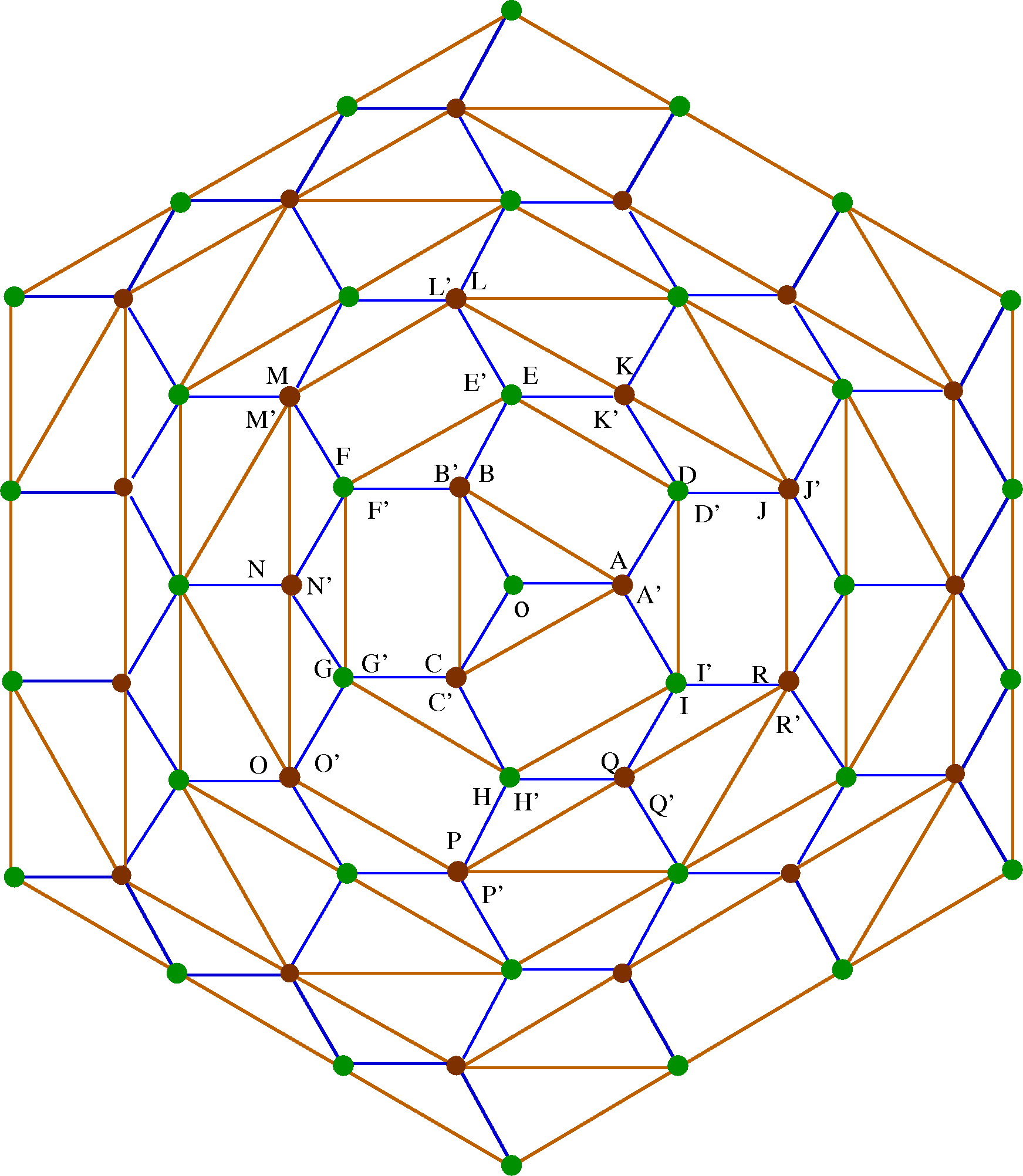}
\caption{The adjacency graph of patches in the pattern in
Fig. \ref{fig:hexpicl1}. The vertices corresponding
to the brownish and greenish patches in the pattern are denoted by
different colors. The pair of patches labeled
by the alphabets and its corresponding primed alphabets in Fig.
\ref{fig:line}(a) are
represented by same vertex on the graph.}
\label{adjhex}
\end{center}
\end{figure}

We define scaled complex coordinates $\mathbf{r}=\mathbf{R}/N$, where
$\mathbf{R}=p+q\omega$ is the complex coordinate of the site
$\left( p,q \right)$. We define the rescaled toppling function for
this pattern as 
\begin{equation}
\phi\left( \mathbf{r} \right)=\lim_{N\rightarrow
\infty}\frac{\sqrt{3}T_{N}\left( \mathbf{r}N \right)}{2 N}.
\label{eq:rscT}
\end{equation}
Then it is easy to see that $\nabla \phi=\left( \partial_{\xi}\phi,
\partial_{\eta}\phi \right)$ is equal to the mean flux of particles at
$\mathbf{r}$. If we consider a small line element
$d\mathbf{l}\equiv\left(d\xi, d\eta\right)$, then the net flux of particles across the line
$d\mathbf{l}$ equals $N \nabla\phi\cdot d\mathbf{l}$. Then, the
conservation of sand grains implies that the toppling function
$T_{N}\left( \mathbf{R} \right)$ satisfies the equation
\begin{equation}
\nabla_{\circ}^{2}T_{N}\left( \mathbf{R} \right)=\delta z\left(
\mathbf{R}\right)-N\delta\left( \mathbf{R} \right),
\end{equation}
where $\nabla_{\circ}^{2}$ is the finite-difference operator on the
lattice, corresponding to the Laplacian $\nabla^{2}$. It is easy to
see that this implies that the scaled potential function $\phi$
satisfies the Poisson equation
\begin{equation}
\nabla^{2}\phi \left( \mathbf{r} \right)=\Delta \rho\left(
\mathbf{r}\right)-\delta\left( \mathbf{r} \right),
\end{equation}
where $\Delta\rho\left( \mathbf{r} \right)$ is the areal density of excess grains at
$\mathbf{r}$. It is related to $\langle\Delta z\left( \mathbf{r}
\right)\rangle$, the mean excess grain density {\it per site} by
\begin{equation}
\Delta \rho\left( \mathbf{r} \right)= \frac{2}{\sqrt{3}} \langle \Delta
z\left( \mathbf{r} \right)\rangle.
\end{equation}

The piece-wise linearity of $\phi$ simplifies the analysis of the pattern,
significantly. The potential function can be characterized by only
three parameters.
Using Eq. (\ref{eq:recp1}), (\ref{eq:recp}) and (\ref{eq:rscT}), for each patch
$P$, we can find a pair of integers $\left( m,n \right)$ such that the
potential in patch $P$ is characterized by
\begin{equation}
\phi\left( \mathbf{r} \right)=- \frac{1}{2\sqrt{3}}\left(
\mathbf{D}_{m,n}\overline{\mathbf{r}}+ \overline{\mathbf{D}}_{m,n}r
\right)+f_{m,n},
\end{equation}
where
\begin{equation}
\mathbf{D}_{m,n}=m+n\omega,
\label{eq:Dunn}
\end{equation}
and $f_{m,n}$ is a real number, constant everywhere inside the patch.
Here $\overline{z}$ denotes the complex conjugate of $z$.

Each patch is characterized by a complex number $\mathbf{D}_{m,n}$ which is the
coefficient  in the potential function $\phi\left( \mathbf{r} \right)$ of the patch. 
In the complex $\mathbf{D}$-plane, each patch with labels as in Fig. \ref{fig:line}(a) can then  be represented by a point. We
connect two patches by a line if they share a common boundary.
Then the resulting figure, shown in Fig. \ref{adjhex}, is the
adjacency graph of the patches. 
\begin{figure}[t]
\begin{center}
\includegraphics[width=10cm,angle=0]{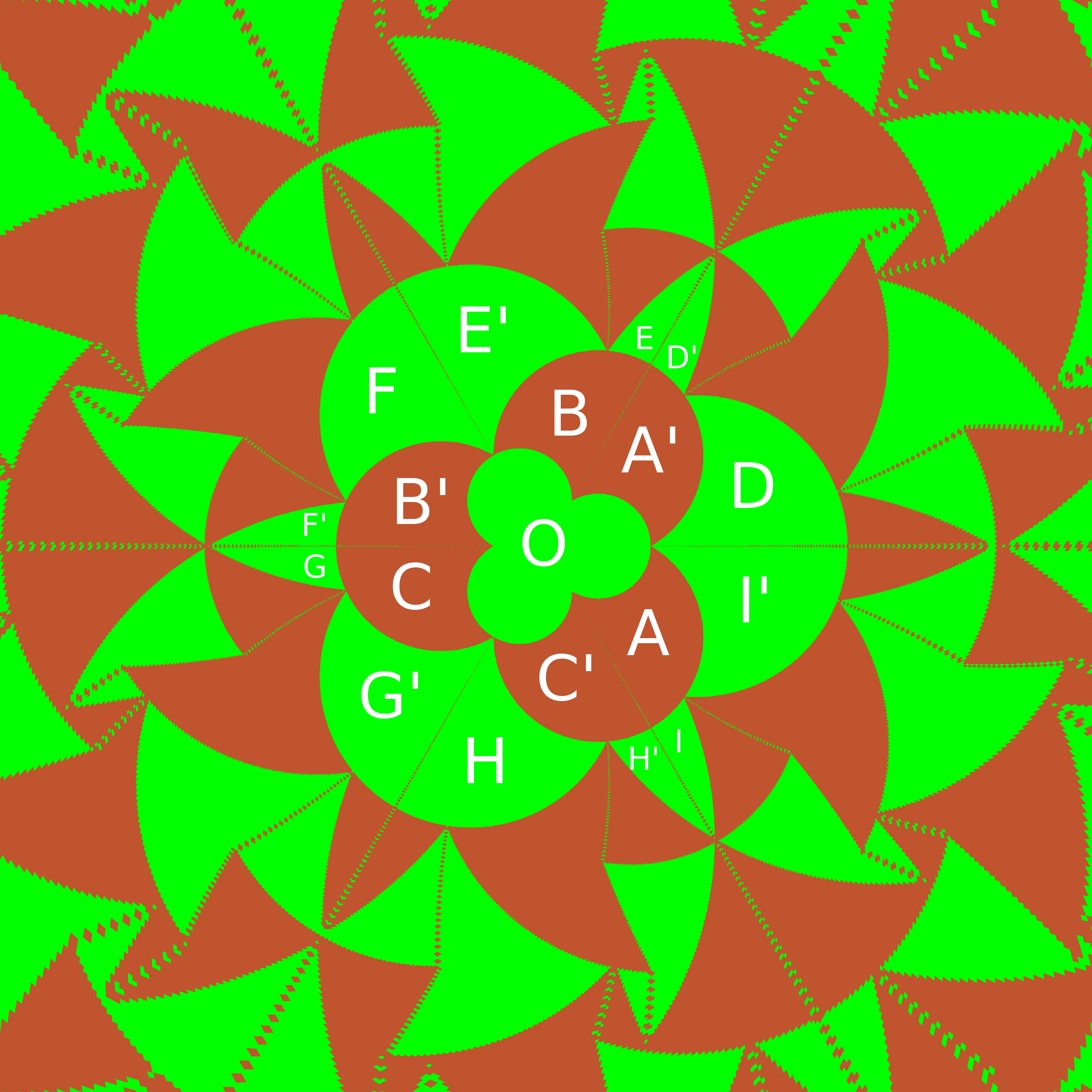}
\caption{The $1/\mathbf{\bar{R}}$ transformation of the pattern in
Fig. \ref{fig:hexpicl1}, where $\mathbf{\overline{R}}$ is the complex
conjugate of $\mathbf{R}$. Labels are the same as used in
Fig. \ref{fig:line}(a).}
\label{fig:1byz}
\end{center}
\end{figure}

We can determine the connectivity structure of this graph, without
knowing the full potential function in each patch. We first take
$1/\mathbf{\bar{R}}$
transformation of the pattern. This is shown in Fig. \ref{fig:1byz}. 
Some of the bigger patches are denoted by capital alphabets in
Fig. \ref{fig:line}(a) and their corresponding patches on the
transformed pattern in Fig. \ref{fig:1byz}.
The patches $\textbf{A}$ and $\textbf{A}'$ in Fig. \ref{fig:line}(a) are adjacent
to the outer region $\textbf{O}$ through the same vertical boundary.
Matching the values of the function $\phi( \mathbf{r})$ and fixing the discontinuity in its normal derivatives at the boundary,
it is easy to see that $\phi(\mathbf{r})$ has the same functional form in the
patches $\textbf{A}$ and $\textbf{A}'$. In fact, it is convenient to
imagine that the boundary between $\textbf{O}$ and $\textbf{A}$ moved
to the right by an infinitesimal amount, so that it
does not touch the patches $\textbf{D}$ and $\textbf{I}'$, and then
$\textbf{A}$ and $\textbf{A}'$ would actually
join to form a single connected patch $\textbf{A}$. We thus consider
$\textbf{A}$ and $\textbf{A}'$ as one patch, and both can be
represented as one point
on the $\mathbf{D}$-plane. Similarly, we identify $\textbf{B}$ and $\textbf{B}'$, $\textbf{C}$
and $\textbf{C}'$, \textit{etc}. 
Then the adjacency graph can be
constructed by joining the sites on the $\mathbf{D}$-plane, according
to the adjacency of patches in Fig. \ref{fig:1byz}.

It turns out that the patches corresponding to $m+n=2\left(
\right.$mod$\left. 3 \right)$ do not appear in the pattern, and
the adjacency graph, as shown in Fig. \ref{adjhex}, is a hexagonal lattice with some extra edges shown in
brown color. These extra edges connect all the vertices at same distance from the origin $\left( 0,0
\right)$ (in the $L^{1}$ metric), and also connect some of the
diagonally opposite sites on the rectangular faces of the graph as shown in figure.

The charge density $\Delta\rho\left( \mathbf{r} \right)$ is zero
inside the patches, and the excess grains due to addition are
distributed along the patch
boundaries, leading to nonzero line charge densities separating
neighboring patches. Then the density function $\Delta\rho\left(
\mathbf{r}
\right)$ is
a superposition of the line charge densities along the patch
boundaries. There are three kinds of line charges of charge density
$\lambda=-1/\sqrt{3}$, $1$, and $2/\sqrt{3}$.

From the electrostatic analogy, it is seen that $\phi\left(
\mathbf{r} \right)$ is
continuous across the common boundary between neighboring patches, and
its normal derivative is discontinuous by an amount equal to the line
charge density $\lambda$ along the boundary. Let $P$ and $P'$ be the two
neighboring patches with the equation of the boundary between them
\begin{equation}
\mathbf{r}=|\mathbf{r}| \exp\left( i \theta \right)+\mathbf{A},
\end{equation}
such that the patch $P'$ is on the left of the boundary.
Then using the continuity condition, it is easy to show that
\begin{eqnarray}
\mathbf{D}_{p'}-\mathbf{D}_{p}&=& i\lambda\sqrt{3} \exp\left( i\theta \right) \rm{~and~} \nonumber \\
f_{p'}-f_{p}&=&Re[\overline{\mathbf{A}}\left( \mathbf{D}_{p'}-\mathbf{D}_{p} \right)]/\sqrt{3},
\label{eq:bc}
\end{eqnarray}
where $\overline{\mathbf{A}}$ is the complex conjugate of $\mathbf{A}$.
We note that, there are only six different types of patch boundaries in the pattern, with angle $\theta$
an integer multiple of $\pi/6$.

It is easy to check that the matching conditions along the edges of hexagonal lattice
(denoted by blue solid line in Fig. \ref{adjhex}) are sufficient to
determine $D_{m,n}$ for
all the vertices. The line charge density $\lambda=-1/\sqrt{3}$ for
the patch boundaries corresponding
to these edges. 
Also, the potential function $\phi=0$, for the vertex at the origin, and
hence, $D$ and $f$ both vanishes.
Then using the matching condition, it is easy to check that, the values
of $D_{m,n}$ are consistent with the form in Eq. (\ref{eq:Dunn}).

The function $f_{m,n}$ satisfies the discrete Laplace's equation on the underlying hexagonal lattice
of the adjacency graph \textit{i.e.}
\begin{equation}
\sum_{m',n'}f_{m',n'}-3f_{m,n}=0 \textrm{ ~ ~ ~  for }\left( m,n \right)\ne 0,
\label{laplace}
\end{equation}
where $\left( m', n' \right)$ denotes the three neighbors of the vertex $\left(m,n\right)$ on the hexagonal lattice.
This can be checked from the concurrency condition of patch boundaries. For
example consider the edges $\mathbf{OA}$, $\mathbf{DA'}$ and
$\mathbf{I'A}$ on the adjacency graph. The corresponding patch boundaries in the
pattern intersect at the same point (Fig. \ref{fig:line}(a)).
Then it is easy to check using the matching condition in
Eq.(\ref{eq:bc}) that,
\begin{equation}
f_{O}+f_{D}+f_{I}=3f_{A}.
\end{equation}
Similar equations hold for the other vertices.

In the region outside the pattern, where none of the sites toppled, the potential function $\phi\left( z \right)=0$. This corresponds to $m =n =0$, and $f_{0,0}=0$. The solution of the Laplace's equation with the above
boundary condition can be written in the following integral form \cite{atkinson}
\begin{equation}
f_{m,n}=\frac{I}{4\pi^{2}}\int_{-\pi}^{\pi}\int_{-\pi}^{\pi}\frac{1-\cos\left(
k_{1}(2m-n)/3+k_{2}n \right)}{1-\left(
\cos{2k_{2}}+2\cos{k_{1}}\cos{k_{2}} \right)/3}dk_{1}dk_{2},
\label{solution}
\end{equation}
for $m+n=0$ (mod $3$),
where $I$ is a normalizing constant, which determines the pattern up to
a scale factor. For the
sites with $m+n=1$ (mod $3$), $f_{m,n}$
are the average of those corresponding to the neighboring sites. As an
example the potential function in region $\textbf{A}$, and $\textbf{C}'$ is
\begin{eqnarray}
\phi_{_{\textbf{A}}}(\mathbf{r})&=&\frac{I}{3}-\frac{\xi}{\sqrt{3}}, \\
\phi_{_{\textbf{C}'}}(\mathbf{r})&=&\frac{I}{3}+\frac{1}{\sqrt{3}}\left( \frac{1}{2}\xi +
\frac{\sqrt{3}}{2}\eta \right),
\end{eqnarray}
where $\mathbf{r}=\xi+i \eta$, and $i=\sqrt{-1}$.
Then the equation of the patch boundary between patches $\textbf{A}$ and $\textbf{O}$ is
\begin{equation}
\xi=I/\sqrt{3},
\label{eq:aob}
\end{equation}
and that of the boundary between patches $\textbf{C}'$ and $\textbf{O}$ is
\begin{equation}
\sqrt{3}\xi+3\eta+2I=0.
\end{equation}
Equivalently, the length of an edge of the bounding
equilateral triangle of the pattern is equal to $2I N$, for large
$N$.
\begin{SCfigure}
\includegraphics[width=8.0cm]{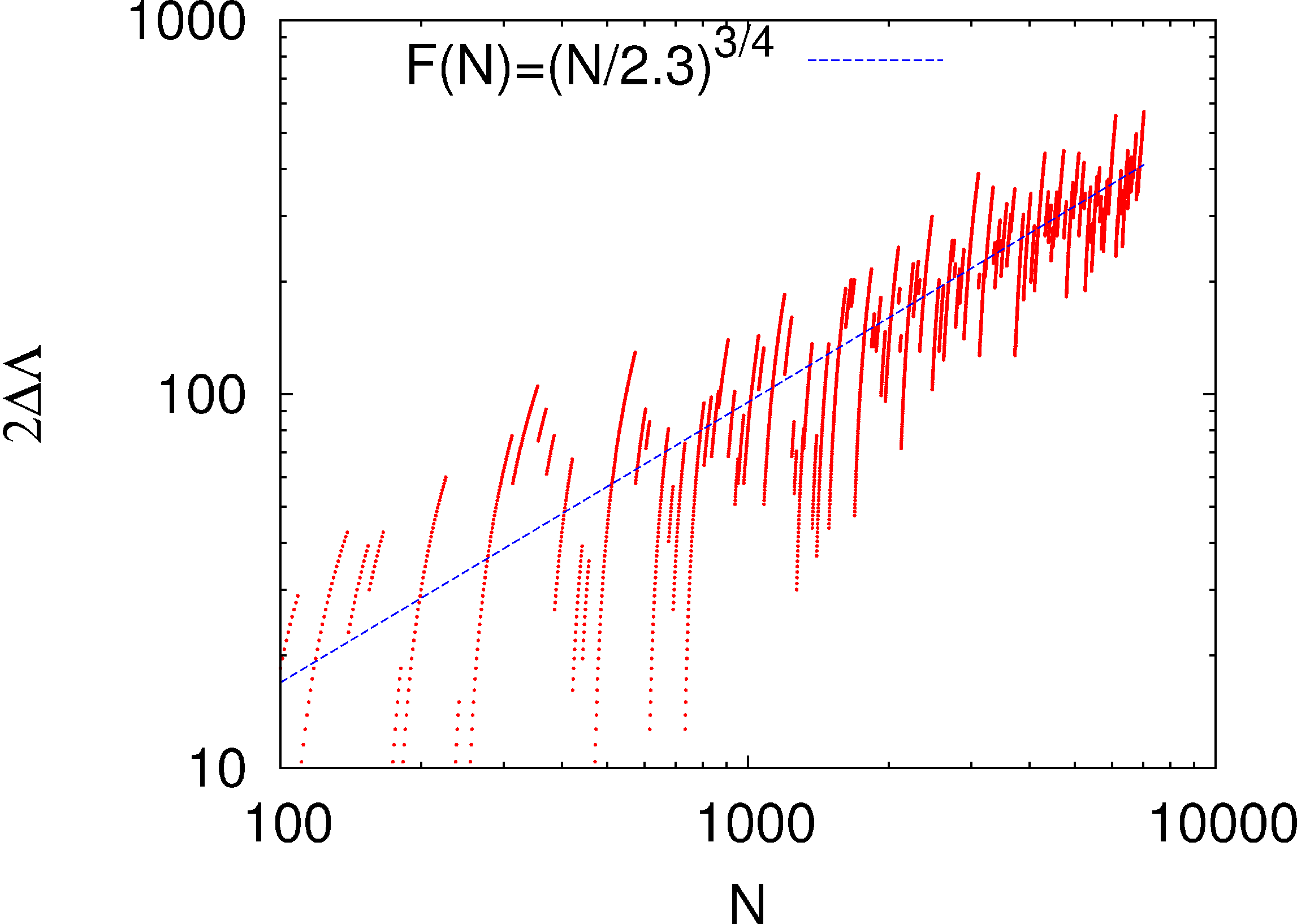}
\caption{The discrepancy $2 \Delta \Lambda$ between the actual height
of bounding triangle, and the asymptotic value $2 N/\sqrt{3}$ plotted
as a function of $N$. The straight line shows a simple power-law fit
with power $3/4$.}
\label{fig:diffdia}
\end{SCfigure}
\begin{SCfigure}
\includegraphics[width=8.0cm]{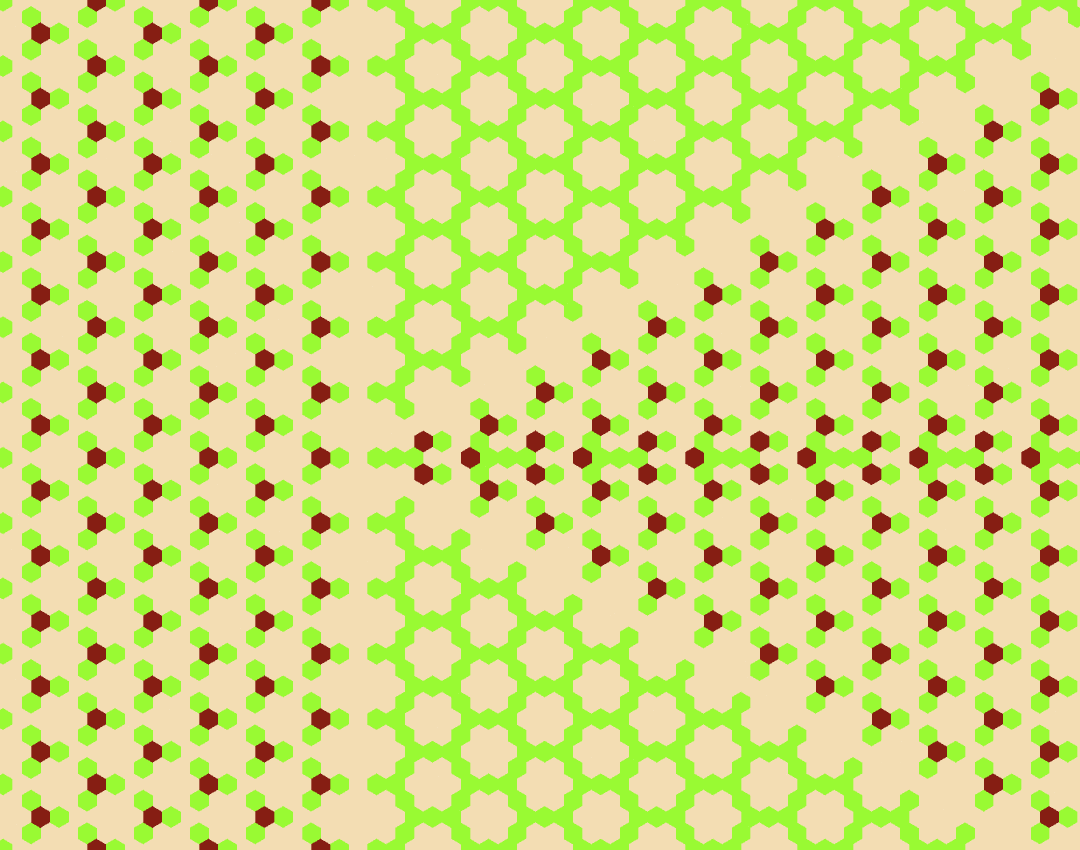}
\caption{An example of five patches meeting at a point, for a
pattern on the background of Class $I$, $\ell=2$. It is easy to check
that the line charge density for the vertical boundary
$\lambda=-1/\sqrt{3}$, same as in Fig. \ref{boundary}. The color code
is same as in Fig. \ref{fig:hexpicl1}.}
\label{fig:pb2}
\end{SCfigure}

The constant $I$ in Eq. (\ref{solution}) can be calculated using the form of the potential
function near the site of addition.
As noted, the function $\phi$ can be considered as the potential due to
line charges along the patch boundaries and a point
charge of unit amount at the origin. Then, close to the origin the
solution diverges logarithmically as $\widetilde{\phi}\left(
\mathbf{r}
\right)=-\left(2\pi\right)^{-1}\log\left( |\mathbf{r}| \right)$,
and the potential function is an approximation to this solution by a piece-wise linear function.
Then, there are coordinates $\mathbf{r_{o}}$ inside each patch $\left( m,n
\right)$ with $|m|+|n|$ large, where the $\phi$ and its first
derivatives are equal to $\tilde\phi$ and its first derivatives,
respectively. Then,
\begin{eqnarray}
2\sqrt{3}\frac{\partial}{\partial
\overline{\mathbf{r}}}\widetilde{\phi}\left( \mathbf{r}
\right)\vert_{\mathbf{r_{o}}}&\simeq&-\mathbf{D}_{m,n} \rm{~ ~ and ~ ~} \nonumber \\
-\frac{1}{2\sqrt{3}}\left\{
\mathbf{D}_{m,n}\overline{\mathbf{r_{o}}}+\overline{\mathbf{D}}_{m,n}\mathbf{r_{o}}
\right\}+f_{m,n}&\simeq&-\frac{1}{2\pi}\log\left( |\mathbf{r_{o}}|
\right).\nonumber\\
\end{eqnarray}
The above two equations imply
\begin{equation}
f_{m,n}\simeq\frac{1}{2\pi}\log\left( |m+n\omega| \right),
\label{logasymp}
\end{equation}
for $|m|+|n|$ large.
Comparing it with the Eq. (\ref{solution}) for large $|m|+|n|$ we find
that the numerical constant $I=1/\sqrt{3}$. This determines the potential function completely, and thus
characterizes the pattern. For example, as in figure
\ref{fig:line}(a), the equation of the rightmost boundary of the
pattern, using Eq. (\ref{eq:aob}) is $x=N/3$. Equations of other boundaries
of patches can be calculated similarly. For example, the reduced
coordinates of the point where the patches $D$ and $D'$ meet in Fig.
\ref{fig:line}(a), is determined by the condition that it is a
common point of patches $D$, $J$ and $A'$, and that the function
$\phi$ is continuous.
\begin{eqnarray}
f_{1,0}-\frac{1}{\sqrt{3}}\xi &=& f_{2,1}-\frac{1}{\sqrt{3}}\left( \frac{3}{2}\xi +
\frac{\sqrt{3}}{2}\eta \right) \nonumber \\
&=& f_{3,1}-\frac{1}{\sqrt{3}}\left( \frac{5}{2}\xi +
\frac{\sqrt{3}}{2}\eta \right).
\end{eqnarray}
Then using the values $f_{1,0}=1/3\sqrt{3}$, $f_{2,1}=
1/2\sqrt{3}$, and
$f_{3,1}=7/6\sqrt{3}-1/\pi$ \cite{atkinson} we get the reduced coordinates of this point as 
\begin{equation}
\left(\xi,\eta\right) =
\left(\frac{2}{3}-\frac{\sqrt{3}}{\pi},-\frac{1}{3\sqrt{3}}+\frac{1}{\pi}\right).
\end{equation}

Equivalently, the height of the bounding equilateral triangle
increases as $ 2 N/\sqrt{3} \simeq 1.154 N$. The estimated slope of
the fitting line in Fig. \ref{fig:triln} is 1.1, in reasonable agreement with the
theory. However,  even though the exact function $\Lambda(N)$  has
large  fluctuations of number theoretic origin, the estimated slope is
noticeably lower than the calculated asymptotic value.  To examine
this discrepancy closer, we have plotted in Fig. \ref{fig:diffdia} the
discrepancy $2 \Delta \Lambda  = 2 N/\sqrt{3} - 2 \Lambda(N)$ as a
function of $N$. We find that this appears to increase with $N$ as
$N^{3/4}$, for large $N$. The reason for this behavior is not
understood yet.
\begin{SCfigure}
\fbox{\includegraphics[width=7cm,angle=0]{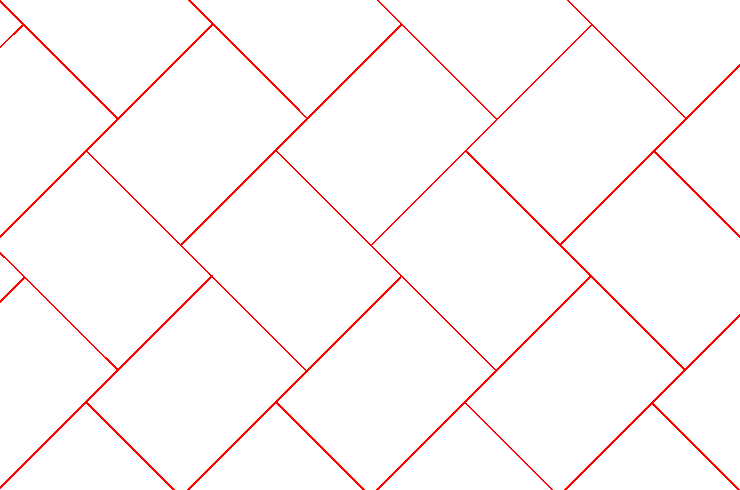}}
\caption{A schematic representation of the periodic tiling of the  plane using tilted rectangles. Background 
height patterns with such periodicities on the F-lattice give rise to non-compact growth with the growth-exponent between  $1/2$  and  $ 1$  }
\label{fig:tile_f}
\end{SCfigure}

For backgrounds, with $l > 1$, our numerical results suggest that there is a
crossover length $R^\star(l)$, and initially, for $R < R^\star(l)$, the
avalanches grow ``explosively" in size. As a result, the number of
particles  inside a disc of radius $R^\star$ in the final  pattern  is less
than that in the initial background. The net flux of particles going
out of the disc increases with $R$ until the radius becomes of order
$R^\star$. After this, the large-scale properties of the pattern
are the same as that of $l=1$ pattern, with the number of particles added
$A_{\ell} N$, where $A_{\ell}$ is an $\ell$-dependent constant. In
particular, the size of the pattern is $A_{\ell}$ -times the size
of the pattern for $l=1$ with same $N$.  The crossover length
$R^\star$ is expected to grows as $\sqrt{N}$.

For a background with   $\ell > 1$,  the
basis vectors at the unit cell are $\ell\hat{e}_1$ and
$\ell\hat{e}_2$, where $\hat{e}_1$, $\hat{e}_2$ are the basis vectors
for $\ell=1$ background (see Eq. (\ref{eq:recp})). Then the reciprocal basis vectors are
$\hat{g}_1/\ell$ and
$\hat{g}_2/\ell$.  From the observed patterns, we find  that the line charge
densities remain same for any $\ell$ (see Fig. \ref{fig:pb2} for an
example of the patch boundaries).  This implies that $n_1$ and $n_2$ in eq. (8) are constrained to be
multiples of $l$. Writing $n_1 = l m$, $n_2 = l n$, we see that 
the patches can be labeled by the same pair of integers $\left( m,n
\right)$ as in the $\ell=1$ case, and the potential function
$\phi_{(l)}\left( \mathbf{r} \right)$ for general $l$ is related to the $l=1$ case by simple scaling:
\begin{equation}
\phi_{(\ell)}\left( \mathbf{r} \right) =  A_{\ell} \phi_{(1)}\left(
\frac{\mathbf{r} }{A_{\ell} }\right),
\end{equation}
where $A_{\ell}$ is a scale factor. For $\ell=2$, $3$, $4$ and
$5$ the values of $A_\ell$ are approximately $2.34$, $3.38$, $4.41$ and $5.37$,
respectively. We note that $A_\ell$ increases approximately linearly
with $\ell$.

\section{Non-compact patterns with exponent $\alpha<1$.\label{sec:ncfl}}
On the F-lattice, after some experimentation, we found that  the background pattern 
having the periodicity of the tiling of plane with tilted rectangles,
shown in Fig. \ref{fig:tile_f}, produces patterns with interesting non-compact
growth. We studied rectangles
with aspect ratio $l:(l+1)$, and the rectangles are tilted by
$45^{\circ}$ to the x-axis. Two such periodic backgrounds are shown in
Fig. \ref{fig:fbg}. In
these background patterns, the sites with height zero, are arranged
along the boundaries of tilted rectangles with two
possible orientations, and rest of the sites have
height one. The stable height-patterns generated by adding $N$
particles and relaxing the configuration on these two backgrounds are shown in Fig.
\ref{fig:fpic1} and Fig. \ref{fig:fpic2}, respectively. The growing boundaries of the
patches in the patterns are shown, in
terms of the $Q$ variables, in Fig. \ref{fig:flinepic1} and
\ref{fig:flinepic2}, respectively. Again, we see that the patch boundaries
are straight lines, with rational slopes.
The plot of diameter $2\Lambda$ vs N, for these two patterns are shown in Fig.
\ref{fig:flinegrowth}. We see
that the growth exponent $\alpha $ is approximately $0.6$ for figure
\ref{fig:fpic1} and $0.725$ for figure
\ref{fig:fpic2}. In general, value of the exponent $\alpha$ is in range
$1/2<\alpha <1$, and approaches value $1$ as density
$\rho_{o}$ of the background becomes close to $1$.
\begin{figure}
\begin{center}
\includegraphics[width=6cm,angle=0]{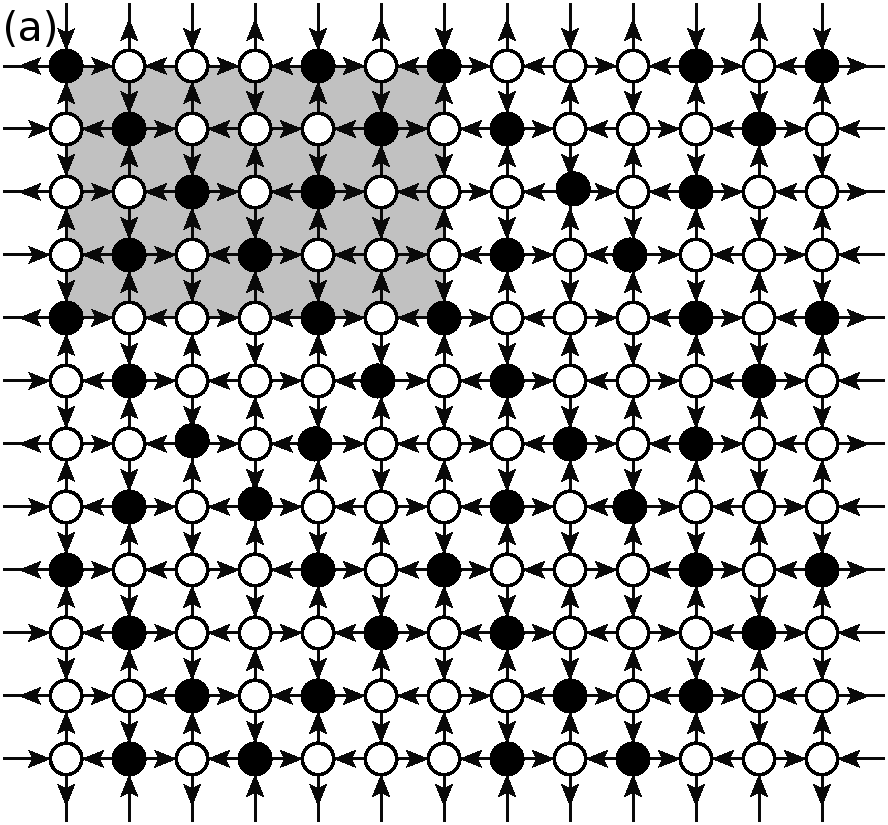}
\includegraphics[width=6cm,angle=0]{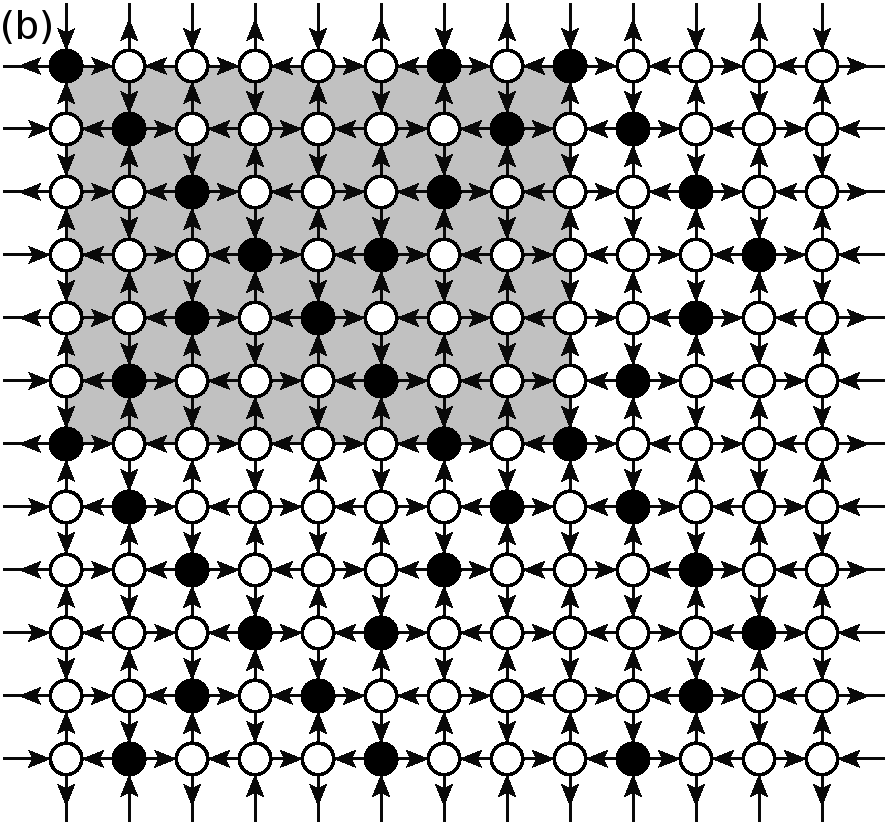}
\caption{The two backgrounds studied on the F-lattice. Unit cells of the periodic
distribution of particles are shown by gray rectangular shades. The filled
circles represent height $0$ and unfilled ones $1$.}
\label{fig:fbg}
\end{center}
\end{figure}
\begin{figure}
\begin{center}
\includegraphics[width=18cm,angle=90]{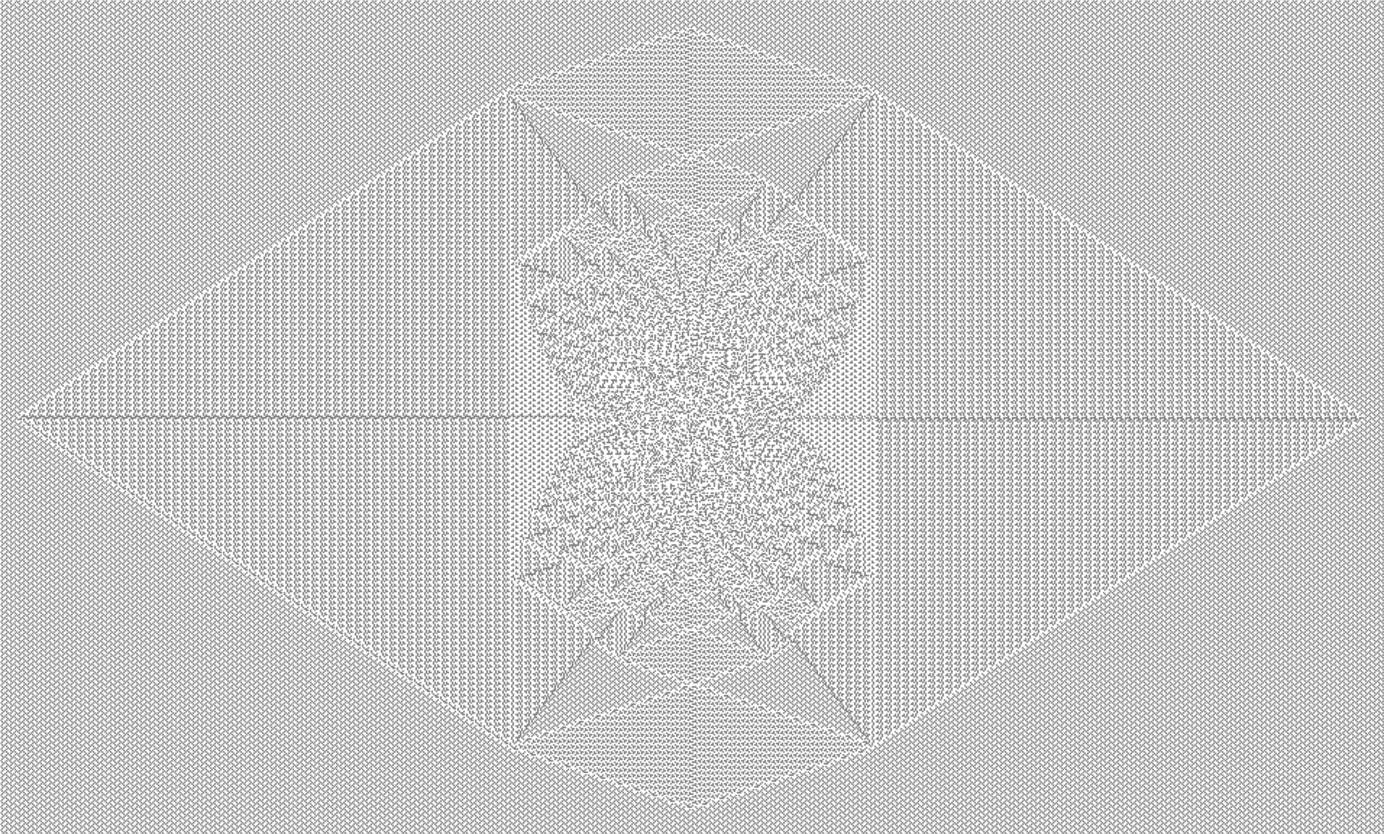}
\caption{The pattern produced on the first background in
Fig. \ref{fig:fbg}, by adding $N=2200$ grains at a single site, and
relaxing the configuration. Color code: White$=1$ and Black$=0$.
Details can be viewed in the electronic
version using zoom in.}
\label{fig:fpic1}
\end{center}
\end{figure}
\begin{figure}
\begin{center}
\includegraphics[width=18cm,angle=90]{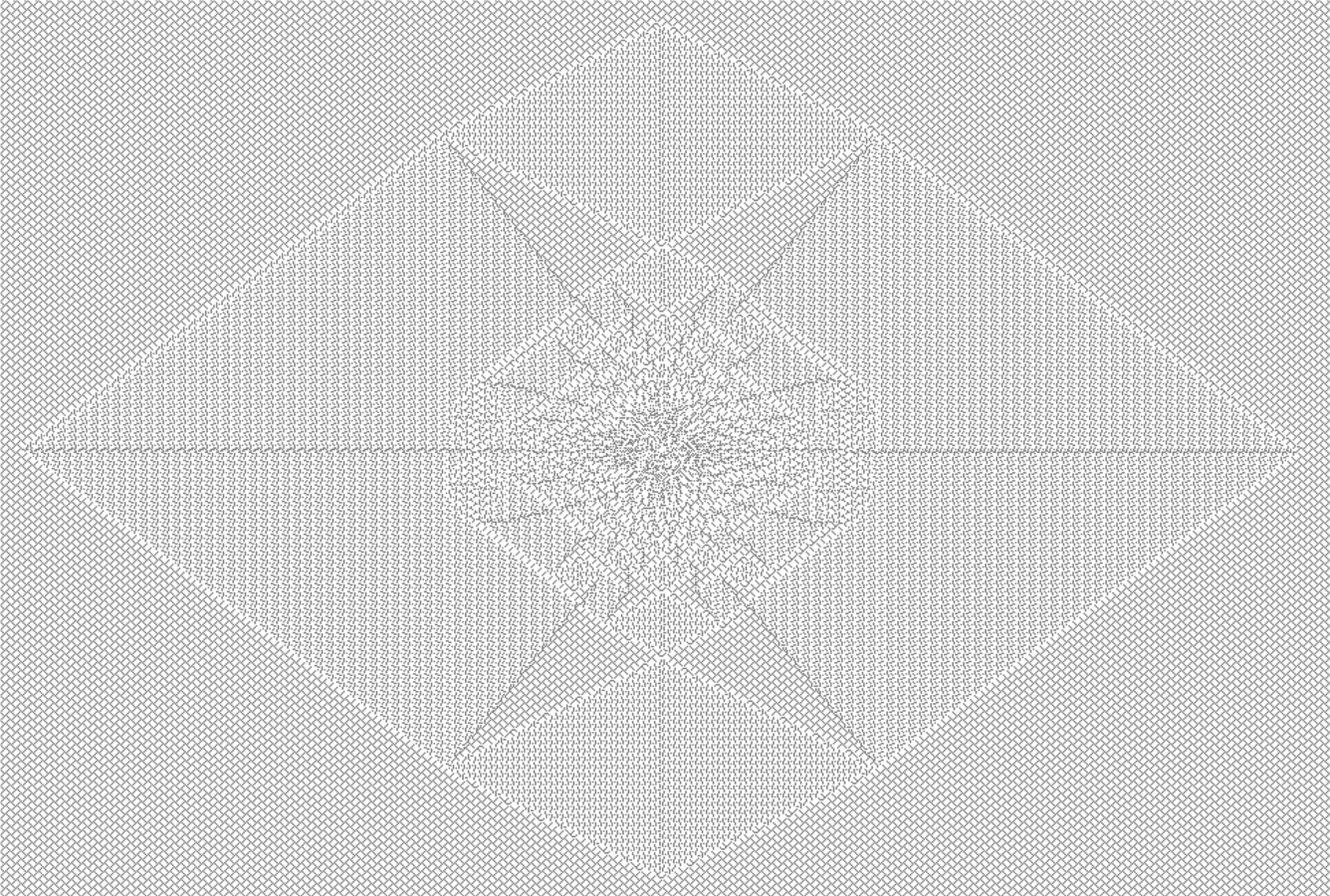}
\caption{The pattern produced on the second background in
Fig. \ref{fig:fbg} by adding $N=600$ grains at a single site, and
relaxing the configuration. Color code: White$=1$ and Black$=0$.
Details can be viewed in the electronic
version using zoom in.}
\label{fig:fpic2}
\end{center}
\end{figure}
\begin{figure}
\includegraphics[width=18cm,angle=90]{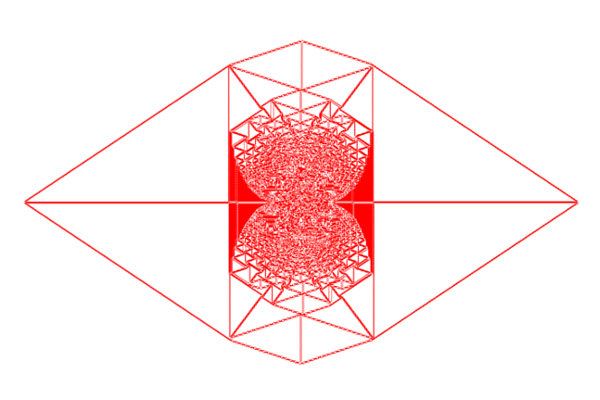}
\caption{The pattern in terms of $Q(r)$, showing the boundaries of 
patches
corresponding to Fig. \ref{fig:fpic1}. Color code: White$=0$ and Red$=$Non-zero.
Details can be viewed in the electronic
version using zoom in.}
\label{fig:flinepic1}
\end{figure}
\begin{figure}
\begin{center}
\includegraphics[width=16cm,angle=90]{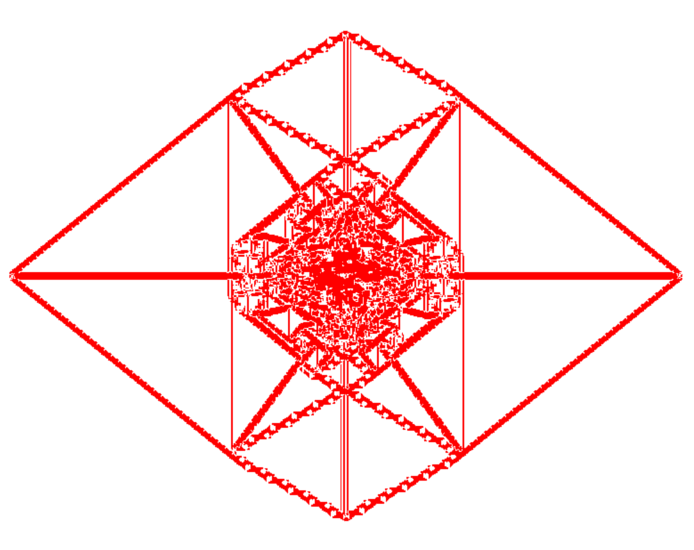}
\caption{The pattern in terms of $Q(r)$, showing the boundaries of patches
corresponding to Fig. \ref{fig:fpic2}. Color code: White$=0$ and Red$=$Non-zero.
Details can be viewed in the electronic
version using zoom in.}
\label{fig:flinepic2}
\end{center}
\end{figure}

There are unresolved areas of apparent solid color in the patterns, taking up a
sizable fraction of the total area, \textit{e.g.},
two large regions of red color on both sides of Fig.
\ref{fig:flinepic1}. In these regions, the pattern appears to be complex,
suggesting either a large number of patch boundaries, or patches of non-zero areal excess charge
density. However, the fractional area of these
regions decreases with larger $N$. Also on comparing patterns
with different $l$, we have seen that the fractional area of such
regions decreases as $l$ increases. A more detailed study of these
patterns seems like an interesting problem for future investigations.

\section{Summary and concluding remarks\label{sec:tropical}}
In this chapter, we have studied two dimensional patterns formed in
Abelian sandpile models by adding particles at one site on an initial
periodic background, where the diameter of the pattern grows as
$N^{\alpha}$, with $\alpha>1/2$. Using some features observed in the
pattern of adjacency of patches as an input, we are able to determine
the exact asymptotic pattern in the  specific case with $\alpha=1$, on
a class $I$ background.
\begin{figure}
\begin{center}
\includegraphics[width=12cm,angle=0]{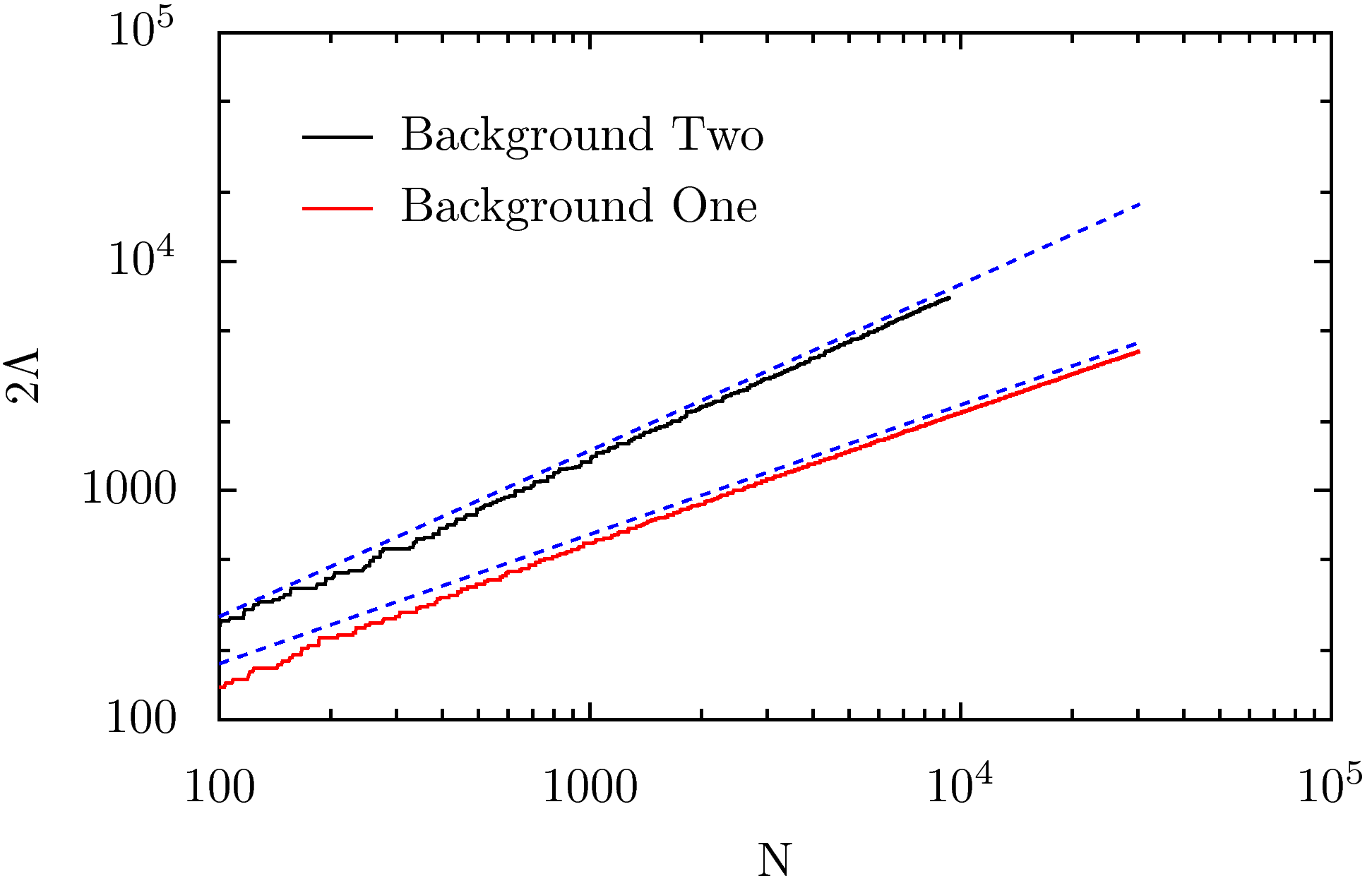}
\caption{The change in diameter as a function of $N$, for the patterns
on the two backgrounds in Fig. \ref{fig:fbg}. The numerical results fit well
with straight lines of slope $0.6$ and $0.725$, for the backgrounds one
and two, respectively.}
\label{fig:flinegrowth}
\end{center}
\end{figure}

The patterns on class $II$ backgrounds can also be characterized
similarly. As noted earlier, some of the patches split into smaller parts.
By using the $1/\overline{\mathbf{R}}$ transformation, we can again determine the structure of
the adjacency graph. The graph for the pattern in Fig.
\ref{fig:line}(b) is shown in Fig. \ref{fig:adjtwo}. It is a periodic
lattice where half of the vertices of the hexagonal lattice are
replaced by $3$
vertices (colored in brown). The exact D-values for different  patches can be easily
determined. The determination of $f_{m,n}$ for this pattern then
requires the solution of the Laplace's equation on this graph.
It can be shown that a slight alteration of the graph, by
drawing the missing edges in the small triangles shown in pink colors,
does not change the pattern. Then the solution of the Laplace's
equation can be reduced to the solution of a
resistor network on this modified graph. The later can be further
reduced to the resistor network on a
hexagonal lattice, discussed by Atkinson \textit{et.al.}
\cite{atkinson}, using the well-known $Y-\Delta$ transformation.
We present the analysis in the Appendix \ref{apndx:laplace}.
\begin{figure}[t]
\begin{center}
\includegraphics[width=10.0cm]{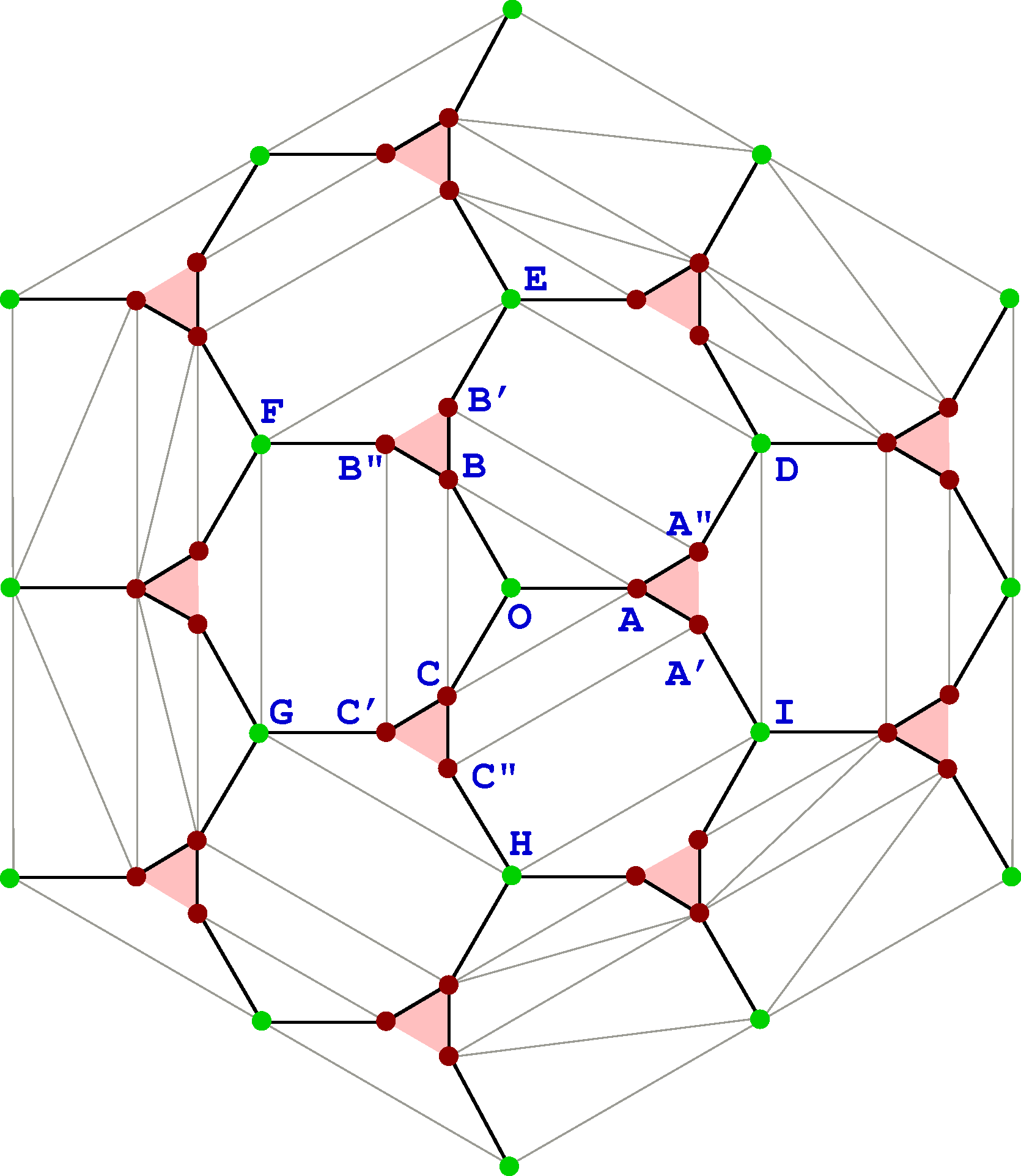}
\caption{The adjacency graph of the patches in the pattern in Fig.
\ref{fig:line}(b). The vertices corresponding to the brownish and greenish
patches in the pattern (Fig.\ref{fig:tri}) are denoted by different colors.}
\label{fig:adjtwo}
\end{center}
\end{figure}

An important feature of the non-compact patterns is that,
it can be characterized by a piece-wise
linear function. This characterization is simpler than that of the
patterns with compact growth, where one requires piece-wise quadratic
polynomials. We have shown that there are
infinitely many backgrounds, on which the patterns have non-compact growth.
It would be desirable to determine  the exact value of $\alpha$ for
different backgrounds showing non-compact growth  studied in section
\ref{sec:ncfl}.
\begin{figure}[t]
\begin{center}
\includegraphics[width=9cm]{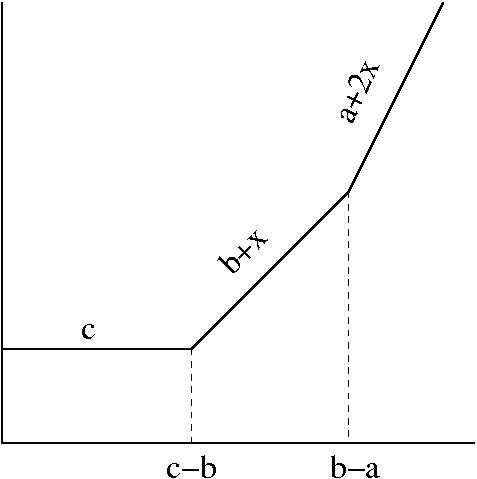}
\caption{Graph corresponding to the tropical function in equation
(\ref{eq:trop})}
\label{fig:trop}
\end{center}
\end{figure}

Another interesting question is a possible connection of this
problem to tropical algebra \cite{tropical}. In tropical mathematics, one defines
operations similar to `addition' and `multiplication' (denoted by $\oplus$ and
$\otimes$ here) by
\begin{eqnarray}
a\oplus b&=&\textrm{max}\left\{ a,b \right\},\nonumber\\
a\otimes b&=&a+b,
\end{eqnarray}
where $a$, $b$ are real parameters.
Familiar properties of addition and multiplication operators, like
commutativity, associativity, existence of identity, distributive
property continues to hold in the new definition. One can then define
polynomials in several variables. The graph of a tropical polynomial
is a piecewise linear function which is also convex. For example, consider the
tropical function
\begin{equation}
f(x)=a\otimes x^{2}\oplus b \otimes x \oplus c.
\label{eq:trop}
\end{equation}
In terms of standard algebra
\begin{equation}
f(x)=max\left\{ a+2x,bx,c\right\}.
\end{equation}
The graph corresponding to this function is shown in Fig. \ref{fig:trop}

We note that for the pattern discussed in section $4$, the potential
function is piece-wise linear. It seems plausible that  tropical polynomials may be useful
to describe this function. In fact, tropical geometry have been
discussed as possibly related to sandpile models \cite{norine,propp2}.
For small values of $N$, our numerical study showed that  $\phi$
is convex, if restricted to one sextant. However, for larger $N$, as
shown in Fig. \ref{surfaceplot}, we see that $\phi$ is not convex even
within one sextant.
\begin{figure}[t]
\begin{center}
\includegraphics[width=10cm]{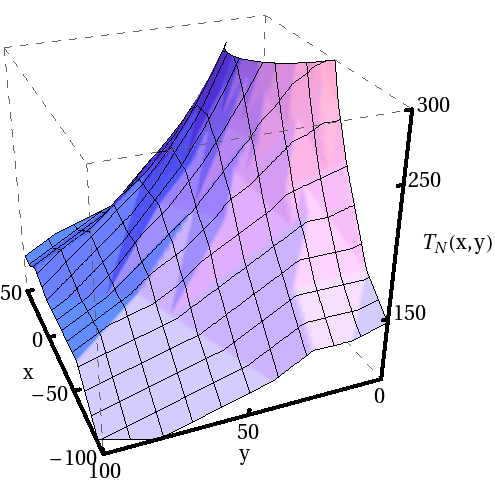}
\caption{(color in the electronic copy) Three dimension plot of the integer toppling function
$T_{N}\left( x,y \right)$ for a triangular pattern like in Fig.
\ref{fig:hexpicl1}, but with $N=800$. The plot shows a zoomed-in
section in the region $y\ge0$ and $y+\sqrt{3}x\ge0$.}
\label{surfaceplot}
\end{center}
\end{figure}
We conclude that it is not possible to
represent the potential $\phi$ as a simple tropical polynomial.
\chapter{A continuous height sandpile model}\label{ch:zhang}
\textit{Based on the paper \cite{mypre}} by Tridib Sadhu and Deepak
Dhar.

\begin{itemize}
\item[\textbf{Abstract}:]
In this chapter, we study the Zhang model of sandpile, defined in the first
chapter, on a one dimensional chain of length $L$, driven by adding a random amount of height
at a randomly chosen site at each addition step. We show that in spite
of this randomness in the input height, the probability distribution 
function of height at a site in the steady state is sharply peaked, and 
the width of the peak decreases as $ {L}^{-1/2}$ for large $L$.

We also
discuss how the height added at one time is distributed among 
different sites by topplings with time. We relate this distribution to 
the time-dependent probability distribution of the position of a 
marked grain in the 
one dimensional abelian model with discrete heights.  We argue that in 
the large $L$ limit, the variance of height at site $x$ has a scaling 
form $L^{-1}g(x/L)$, where $g(\xi)$ varies as $\ln(1/\xi)$ for 
small $\xi$, which agrees well with the results from 
numerical simulations. 
\end{itemize}

\section{Introduction}\label{ch4.1}
After Dhar first discovered the abelian property of the BTW model
\cite{dharprl}, many more models in the general class known as abelian distributed 
processors, were studied, as the abelian 
property makes their theoretical study simpler
\cite{dharphysica06}. The 
original sandpile model of Bak \textit{et al.} \cite{btw}, the Eulerian walkers 
model \cite{euler}, and the abelian variant \cite{sasm} 
of the model originally proposed by Manna \cite{manna} are all members 
of 
this class.  Models which do not have the abelian property have been 
studied mostly by numerical simulations. As discussed in chapter
\ref{ch:intro}, the Zhang model is one such model, and
this is what we study in this chapter.

In the Zhang model, the amount of height added at a randomly chosen site 
at each addition step is not fixed, but random.  In spite of 
this, the model in one dimension has the remarkable property that the 
height at a site in the steady state has a very sharply peaked 
distribution in which the width of the peak is much less than the 
spread in the input amount per time step, and the width decreases with 
increasing system size $L$. This behavior was noticed by Zhang 
using numerical simulations in one and two dimension \cite{zhang}, and he called it 
the `emergence of quasi-units' in the 
steady state of the model. He argued that for large systems, the 
behavior would be same as in the discrete model. Recently, A. Fey 
\textit{et al.} \cite{redig} have proved that for some choices of the distribution of input height, in one dimension, the 
variance of height does go to zero as the length of the chain $L$ goes 
to infinity. However, they did 
not study how fast the variance decreases with $L$.

We study this emergence of `quasi-units' in one 
dimensional Zhang sandpile by looking at how the added height is redistributed 
among different sites in the avalanche process. We show that the 
distribution function of the fraction of height, added at a site $x'$, 
reaching a site $x$
after $t$ time steps following the addition is exactly equal to the 
probability distribution that a marked grain in the 
one-dimensional height type BTW model added at site $x'$, reaches site $x$ in 
time $t$. The latter problem was
studied earlier \cite{punya}. We use this to show that the 
variance of height asymptotically vanishes as $ 1/L$.  We also 
discuss the spatial dependence of the variance along the system length.  In 
the large $L$ limit, the variance at site $x$ has a scaling form 
$L^{-1}g(x/L)$. We determine an approximate form of the scaling 
function $g(\xi)$, which agrees very well with the results of our 
numerical simulations.

There have been other studies of the Zhang model earlier. 
Blanchard \textit{et al.} \cite{blanchard} have studied the steady state of
the model where the amount of addition of height is fixed but
the site of addition is chosen randomly, and found that the 
distribution of energies even 
for the two site problem is very complicated, and has a multi-fractal 
character. In two dimensions, the distribution of height seems 
to sharpen for larger $L$, but the rate of decrease of the width is very 
slow \cite{janosi}. Most other studies have dealt with the question as to whether the critical exponents of the avalanche distribution in this 
model are the same as in the discrete abelian model \cite{lubeck,milshtein}. A. Fey \textit{et al.}'s results imply that the asymptotic 
behavior of the avalanche distribution in one dimension, for specific
cases, is identical to the discrete case, but the situation in higher dimension remains unclear 
\cite{vespignaniz, guilera}.

The plan of the chapter is as follows. In Section \ref{ch4.2}, we define the model 
precisely.  In Section \ref{ch4.3}, we show that the way the 
height added at a site is distributed among different sites by 
toppling is same as the time-dependent probability 
distribution of the position  of a marked grain in the 
discrete abelian sandpile model. This correspondence is used in Section 
\ref{ch4.4} to determine the qualitative dependence of the variance of the height 
variable at a site on its position $x$, and on the system size $L$. We 
propose a simple extrapolation form that incorporates this dependence. 
We check our theoretical arguments with numerical simulations in
Section \ref{ch4.5}.  
 Section \ref{ch4.6} contains a summary and concluding remarks. A detailed calculation
of the solution of an equation, required in Section \ref{ch4.4}, is
given in the Appendix \ref{apndx:zhang}.

\section{Definition and preliminaries}\label{ch4.2}

We consider our model on a linear chain of size $L$. The sites are 
labelled by integers $1$ to $L$ and a real continuous height variable is 
assigned to each site. Let $h(x,t)$ be the height variable at site $x$ 
at the {\it end} of the time-step $t$. We define a threshold 
height value $h_c$, same for each site, and sites with $h(x,t) \ge h_c $ are 
called unstable, while those with $h(x,t) < h_c$ are called stable.  
Starting from a configuration where all sites are stable, the 
dynamics is defined as follows.

(i) The system is driven by adding a random amount of height at the 
\textit{beginning} of every   time-step at a randomly chosen site. Let the amount 
of height added 
at time $t$ be $\Delta_t$.  We will assume that all $\Delta$'s are 
independent, identically distributed random variables, each picked 
randomly from an uniform interval $1-\epsilon \le \Delta_t \le 
1+\epsilon$. Let the site of addition chosen at time $t$ be denoted by 
$a_t$.

(ii) We make a list of all sites whose height exceeds or becomes equal to 
the critical value $h_c$.  All these sites are relaxed in parallel by 
topplings. In a toppling, the height of the site is equally distributed to its 
two neighbors and the height at that site is reset to zero. If there is 
toppling at a boundary site, half of the height at that site before 
toppling is lost. 

(iii) We iterate Step (ii) until all topplings stop. This completes one 
time step. 

This is the slow driving limit, and we have assumed that all avalanche 
activity stops before the next addition event. In this limit, the model 
is characterized by two parameters $\epsilon$ and $h_c$. In the limit 
$\epsilon =0$, and $1 < h_c \le 2$, the model reduces to the discrete 
case, where the behavior is well understood \cite{ruelle}. For non-zero 
but small $\epsilon$, the behavior does not depend on the precise value 
of $h_c$. In fact, starting with a recurrent configuration of the pile, 
and adding height at some chosen site, we get exactly the same sequence 
of topplings for a range of values of $h_c$ \cite{redig}. To be precise, 
for any fixed initial configuration, and fixed driving sequence (of 
sites chosen for addition of height), whether a site $x$ topples at time 
$t$ or not is independent of $h_c$, so long as we have $1 + \epsilon < 
h_c \le 2 -2 \epsilon$. In the following, we assume for simplicity that 
$h_c = 3/2$, and $0 \leq \epsilon \le 1/4$.

 It was shown in \cite{redig} that in this case, the stationary state has at 
most one site with height $h(x,t)=0$ and all other sites have height in the 
range $ 1 -\epsilon \le h(x,t) \le 1 +\epsilon$. The position of the empty site 
is equally distributed among all the lattice points. There are also some 
recurrent configurations in which all sites have height $h(x,t) \geq 1 - \epsilon$.   
In such cases, we shall say that the site with zero height is the site 
$L+1$.  Then, in the steady state, there is exactly one site with height equal to $0$, and the 
$L+1$ different positions of the site are equally likely.

If $ h_c$ does not satisfy the inequality $ 1 + \epsilon  <  h_c \le
2 - 2 \epsilon$,  this simple characterization of the steady state is no 
longer valid. However, our treatment can be easily extended to those 
cases. Since the qualitative behavior of the model is the same in all 
cases, we restrict ourselves to the simplest case here.

It is easy to see that the toppling rules are in general not abelian. 
For example, start with a two site model in configuration $(1.6,2.0)$ 
and $h_c=1.5$. The final configuration would be $(1.4, 0)$, or $(0, 
1.3)$,  depending on whether the first or the second site is toppled initially.  In our model, 
using the parallel update rule, the final configuration would be 
$( 1.0, 0.8)$.  A. Fey \textit{et al.} \cite{redig} have shown that only 
in one dimension, for $1+\epsilon < h_c$, the Zhang model has a restricted abelian 
character, namely, that the final state does not depend on the order of topplings 
within an avalanche. However, topplings in two different avalanches do 
not commute.

\section{The propagator, and its  relation to the discrete abelian model} \label{ch4.3}
It is useful to look at the Zhang model as a perturbation about the 
$\epsilon =0$ limit.  For sufficiently small $\epsilon$, given the site of 
addition and initial configuration,  the toppling sequence is {\it 
independent} of $\epsilon$.  It is also 
independent of the amount of height of addition $\Delta_t$, and is same as 
the model with $\epsilon = 0$, which is the $1$-dimensional abelian 
sandpile model with integer heights (hereafter referred to simply as 
ASM, without  further 
qualifiers).  
We decompose the height variables as
\begin{equation}
\label{eq:1}
  h(x,t)=\mbox{Nint}[h(x,t)] + \epsilon \eta(x,t),
\end{equation}
where Nint refers to the nearest integer value. Then the integer part of 
the height evolves as in the ASM. We write
\begin{equation}
\Delta_t = 1 + \epsilon u_t , \mathrm{~for ~~all~~}t.
\end{equation}
Here $u_t$ is uniformly distributed in the interval $[-1,+1]$. The 
linearity of height transfer in toppling implies that the evolution of the 
variables $\eta(x,t)$ is independent of $\epsilon$. Thus, $\eta(x,t)$ is a 
linear function of $u_t$; the precise function depends on the sequence 
of topplings that took place. These are determined by the sequence of
addition sites $\{a_t\}$ up to the time $t$, and the initial configuration 
$C_0$. These together will be called the evolution history of the 
system up to time $t$, and denoted by  $\mathcal{H}_t$. We assume that at 
the starting time $t = 0$, the 
variables $\eta(x, t=0)$ are zero for all $x$, and the initial 
configuration is a recurrent configuration $C_0$ of the ASM. Then, from 
the linearity of the toppling rules, we can write $\eta(x,t)$ as a 
linear function of $\{u_{t'}\}$ for $1 \leq t' \le t$, and we can write for 
a given history $\mathcal{H}_t$,
\begin{equation}
  \label{eq:3}
  \eta(x,t|\{u_t\},\mathcal{H}_t)=  
  \sum_{t'=1}^{t}G(x,t|a_{t'},t',\mathcal{H}_t)u_{t'}.
\end{equation}
This defines the matrix elements $G(x,t|a_{t'},t',\mathcal{H}_t)$. These 
can be understood in terms of the probability distribution of the 
position of a marked grain in the ASM as follows. Consider the motion of 
a marked grain in the one dimensional height type BTW model. We start 
with configuration $C_0$ and add grains at sites according to the 
sequence $\{a_t\}$. All grains are identical except the one added at 
time $t'$, which is marked. In each toppling, the marked grain jumps to 
one of its two neighbors with equal probability. Consider the 
probability that the marked grain will be found at site $x$ after a 
sequence of relaxation processes at time $t$. We denote this probability as 
${\rm Prob}(x,t|a_{t'},t',\mathcal{H}_t)$. From the toppling rules in 
both the models, it is easy to see that 
\begin{equation}
 G(x,t|a_{t'},t',\mathcal{H}_t) = {\rm Prob}(x,t|a_{t'},t',\mathcal{H}_t).
\end{equation}
Averaging over different histories $\mathcal{H}_t$, we get the probability 
that a marked grain added at $x'=a_{t'}$ at time $t'$  is found at a 
position $x$ at time $t \ge t'$ in the steady state of the ASM.
Denoting the latter probability by $\mathrm{Prob_{ASM}}(x,t| x',t')$, we get
\begin{equation}
  \label{eq:5}
\overline{ G(x,t|x'=a_{t'},t',\mathcal{H}_t)} = 
\mathrm{Prob_{ASM}}(x,t|x',t'),
\end{equation}
where the over bar denotes averaging over  different 
histories $\mathcal{H}_t$, consistent with the specified constraints. 
Here, 
the constraint is that $\mathcal{H}_t$ must satisfy  $a_{t'} = x'$. At 
other places, the constraints may be different, and will be  specified  
if not clear from the context.

We shall denote the variance of a random variable $\xi$ by  $Var[\xi ]$.
For the specific case with $\epsilon=0$, using the definition
in Eq. \eqref{eq:1}, it is easy to show that
$Var[h(x,t)]=L/(L+1)^2$. For non-zero $\epsilon$, in addition to the previous
term, there will be a term proportional to $\epsilon^2$, as the term linear in $\epsilon$ vanishes. 
Hence, we can write
\begin{equation}
Var[h(x,t)] =   L / (L+1)^2 + \epsilon^2 Var[\eta(x,t)].
\end{equation}
Different $u_t$ are independent random variables, also independent of 
$\mathcal{H}_t$ and have zero mean. Let $Var[ u_t] = \sigma^2$. For the case 
when $u_t$ has a uniform distribution between $ -1$ and $ +1$, we have 
$\sigma^2  =  1/3$. Then, from Eq. \eqref{eq:3}, 
we get
\begin{equation}
  \label{eq:7}
  Var[\eta(x,t)] =  \sigma^2
\sum_{t'=1}^{t} \overline{G^2(x,t|a_{t'},t',\mathcal{H}_t)}.
\end{equation}
As $t \rightarrow \infty$, the system tends to a steady state, and the 
average in the 
right hand side of Eq. \eqref{eq:7} becomes a function of $t-t'$. Also, for a 
given $t'$, all values of $a_{t'}$ are equally likely. We define
\begin{equation}
  \label{eq:8}
F(x,\tau) \equiv \frac{1}{L} \lim_{t' \rightarrow \infty}\sum_{x'} 
\overline{G^2(x,t'+ \tau | 
x',t',\mathcal{H}_t)}.
\end{equation}
Then, for large $L$, in the steady state ($t$ large), the variance of height at site 
$x$ is $1/L+\epsilon^2\Sigma^2(x)$, where 
\begin{equation}
  \label{eq:9}
  \Sigma^2(x) = \lim_{t \rightarrow \infty} Var[\eta(x,t)] = \sigma^2 \sum_{\tau=0}^{\infty} F(x,\tau).
\end{equation}
We define $\overline{\Sigma^2}$ to be the  average of $\Sigma^2(x)$ over 
$x$ as
\begin{equation}
\overline{\Sigma^2} = \frac{1}{L} \sum_x \Sigma^2(x).
\end{equation}
Evaluation of $G(x,t|x',t',\mathcal{H}_t)$ for a given history 
$\mathcal{H}_t$ and averaging over $\mathcal{H}_t$ is quite tedious for 
$t > 1$ or $2$.  For $\overline{G}$, the problem has been studied in 
the context of residence times of grains in sand piles, and some exact results 
are known in specific cases \cite{punya}.  For $\overline{G^2}$, 
the calculations are much more difficult.  However, some simplifications 
occur in large $L$ limit. We discuss these in the next section.
\section{Calculation of $\Sigma^2(x)$ in large-$L$ limit}\label{ch4.4}
In order to find the quantity 
$F(x,\tau)$ in Eq. (8), we have to 
average $ G^2(x,t | x',t',\mathcal{H}_t)$ over all possible histories 
$\mathcal{H}_t$, which is quite difficult to evaluate exactly. 
However, we can determine the leading behavior of $F(x,\tau)$ in this limit. 

We use the fact that the path of a marked grain in the ASM is a random 
walk \cite{punya}.  Consider a particle that starts away from the 
boundaries at $x'= \xi L$, with $L$ large, and $0 < \xi < 1$. If it 
undergoes $r(\mathcal{H}_t)$ topplings between the time $t'$ and $t = t' + 
\tau$ under 
some particular history $\mathcal{H}_t$, then its probability distribution 
is approximately a Gaussian, centered at $x'$ with width $\sqrt{r}$.  
Then, we have 
\begin{equation} 
G(x,t|x',t', \mathcal{H}_t) \simeq  \frac{1}{\sqrt{2 \pi r(\mathcal{H}_t) }} \exp\left( - 
\frac{( x - x')^2} {2 r(\mathcal{H}_t)}\right). 
\end{equation}
Using this approximation for $G$, summing over $x'$, we get
\begin{equation}
  \label{eq:12}
\sum_{x'} G^2(x,t|x',t',\mathcal{H}_t) \simeq \frac{1}{2\sqrt{\pi 
r(\mathcal{H}_t)}}.
\end{equation}
Thus, we have to calculate the average of $1/\sqrt{r(\mathcal{H}_t)}$ over 
different histories. Here $r(\mathcal{H}_t)$ was defined as the number of 
topplings undergone by the marked grain. Different possible 
trajectories of a marked grain, for a given history, do not have the 
same number of topplings. However, if the typical displacement of the 
grain is much smaller than its distance from the end, differences 
between these are small, and can be neglected. There are typically 
${\cal O}(L)$ topplings per grain per avalanche in the model, and a grain moves a 
typical distance of ${\cal O}(\sqrt{L})$ in one avalanche. Then, we can 
approximate $r(\mathcal{H}_t)$ by $N(x')$, the number of topplings at $x'$.

 Let the number of topplings at $x'$ at time steps $\tau = 0, 1 , 2, \ldots$ be 
denoted by $N_0, N_1, N_2, \ldots$. Then, $N(x') = N_0 + N_1 + N_2 + \cdots$.  
It can be shown 
that the number of topplings in 
different avalanches in the one dimensional ASM are nearly uncorrelated (In 
fact the correlation function between $N_i$ and $N_j$ varies as 
$(1/L)^{|i -j|}$.).  By 
the central limit theorem for sum of weakly correlated random variables, 
the  mean value of $N$ grows linearly with $\tau$, but the standard 
deviation increases only as $\sqrt{\tau}$. Then, for $\tau \gg 0$, the 
distribution is sharply peaked about the mean, and $\langle1/\sqrt{N} \rangle \simeq 
1/\sqrt{\langle N \rangle}$. 

Clearly, for $\tau \gg 0$,  $\langle N \rangle = \tau 
\bar{n}(x')$, where $\bar{n}(x')$ is the mean number of topplings per 
avalanche at $x'$ in the ASM,  given by
\begin{equation}
\bar{n}(x = \xi L)= L \xi  (1 - \xi)/2.
\end{equation}
The upper limit on $\tau$  for the validity of the above argument 
comes from the 
requirement that the width of the Gaussian be much less than the 
distance from the boundary, (without any loss of generality, we can assume 
that $\xi < 1/2$, so that it is the left boundary ), else we cannot neglect 
events where the marked grain 
leaves the pile. This gives  
$\sqrt{\tau \bar{n}(x)} \ll \xi L$, or equivalently, $\tau \ll \xi L. $
Thus we get,
\begin{equation}
  F(x,\tau) \simeq \frac{C_1}{L} [ \tau L \xi (1 -\xi)]^{-1/2}, \mathrm{~ for~~} 0 \ll \tau 
\ll \xi L,
\end{equation}
where $C_1$ is some constant.

Also, we know that for $\tau \gg L$, the probability that the grain 
stays in the pile decays exponentially as $\exp( -\tau/L)$ \cite{punya}.
Thus, $\overline{G}$, and also $\overline{G^2}$ will decay exponentially 
with $\tau$, for $\tau \gg L$. Thus, we have, for some constants $C_2$ and 
$a$,
\begin{equation}
  F(x,\tau) \simeq \frac{C_2}{L^2} \exp( - a \tau/L), \mathrm{~~for~} \tau \gg L.
\end{equation}
It only remains to determine the behavior of $ F(x,\tau)$, for $\xi L 
\ll \tau \ll L$.   In this 
case, in the ASM, there is a significant probability that the marked 
grain leaves 
the pile from the end. This results in a faster decay of $G$, and hence 
of $F$ with time.   We argue below that the behavior of the function 
$F(x,\tau)$ is 
given by
\begin{equation}
  F(x,\tau) \sim \frac{C_3}{L\tau}, \mathrm{~~for~~} \xi L  \ll \tau \ll L,
\end{equation}
where $C_3$ is some constant.
 This can be seen as follows:  Let us consider 
the special case when the particle starts at a site close to the boundary. 
Then $\bar{n}(x)$ is approximately a linear function of $x$ 
for small $x$. Its spatial variation cannot be neglected, and Eq. \eqref{eq:12} is 
no longer valid.  We 
will now argue that  in this case
\begin{equation}
  \overline{G(x,t' + \tau |x', t',\mathcal{H}_t)} \simeq x' \tau^{-2} \exp( -x/ \tau),
\end{equation}
for $0\ll \tau \ll L$. The time evolution of $Prob_{ASM}(x,t|x',t')$
in Eq. \eqref{eq:5} is well described as a diffusion
with diffusion coefficient proportional to $\bar{n}(x)$ which is the mean
number of topplings per avalanche at $x$ in the ASM \cite{punya}.
For understanding the long-time survival probability in this problem, we 
can equivalently  consider  the problem in a continuous-time version: 
consider a  random walk on 
a half line where sites are labelled by positive integers, and the jump 
rate out of a site $x$ is proportional to $x$. A particle starts at site 
$x =x_0$ at time $t=0$. If $P_j(t)$ is the probability that the particle is 
at $j$ at time $t$, then the equations for the time-evolution of $P_j(t)$ 
are, for all $ j>0$,
\begin{equation}
  \label{eq:18}
\frac{d}{dt} P_j(t) = (j+1) P_{j+1}(t) +(j-1) P_{j-1}(t) - 2j 
P_j(t).
\end{equation}
The long time solution starting with $P_j(0) = \delta_{j,x_0}$ is 
\begin{equation}
  P_j(t) \simeq x_0 t^{-2} \exp(- j/t)
\end{equation}
for $t \gg x_0$ and large $j$. The 
probability that the particle survives till time $t$ decreases as $1/t$ 
for large $t$. We have discussed the calculation in the Appendix
\ref{apndx:zhang}. 

Using Eq. \eqref{eq:5}, we see that $ \overline{G(j,t'+\tau|x_0,t')}$ scales as $ 
x_0/\tau^2 $.  It seems reasonable to assume that $\overline{G^2}$ will 
scale as $\overline{G}^2$. Then, each term in the 
summation for $F(x,\tau)$ in Eq. \eqref{eq:8} scales as 
$x_0^2/\tau^4$, and there are $\tau$ such terms, as the  sum over $x_0$ has an 
upper cutoff proportional to $\tau$, and so $F(x,\tau)$ varies as 
$1/\tau$ for $ L \gg \tau \gg x_0$. This concludes the argument.

We can put these three limiting behaviors into a single functional form 
that interpolates between these, as
\begin{equation}
  F(x,\tau) \simeq \frac{1}{L} \frac{K \exp( - a \tau/L)}{\tau + B \sqrt{\tau L \xi (1 
-\xi)}},
\end{equation}
where $K$, $a$ and $B$ are some constants.
In Section V, we will see that results from numerical simulation 
are consistent with this phenomenological 
expression.  
 
Using this interpolation form in Eq. \eqref{eq:9}, and converting the sum 
over $\tau$ to an integration over a variable $u = \tau/L$,  we can write
\begin{equation}
\Sigma^2(x = \xi L) \simeq \frac{\sigma^2}{L} \int_0^{\infty}du \frac{ K 
\exp( - a u)}{u + B 
\sqrt{ u \xi ( 1 - \xi)}}.
\end{equation}
This integral can be simplified by a change of variable $ au = z^2$, 
giving
\begin{equation}
\Sigma^2(x = \xi L) \simeq \frac{ K \sigma^2}{L} I\left(B'\sqrt{\xi(1 - \xi)}\right),
\end{equation}
where $K, B'$ are  constants, and $I(y)$ is a function defined by
\begin{equation}
I(y) = 2\int_0^{\infty} dz \frac{ \exp( - z^2)}{ z + y}.
\end{equation}
It is easy to verify that $I(y)$ diverges as $\ln ( 1/y)$ for small 
$y$.  In 
particular, we note that the exponential term in the integral 
expression for $I(y)$ has a significant contribution only for $z$ near 
$1$.  We may approximate this by dropping the exponential factor, and 
changing the upper limit of the integral to $1$.  The resulting integral 
is easily done, giving
\begin{equation}
  \label{eq:24}
\Sigma^2(x=\xi L) \simeq \frac{ K' \sigma^2}{ L} \ln \left( 1 + 
\frac{1}{B'\sqrt{\xi(1 -\xi)}}\right),
\end{equation} 
where $K'$ is some constant.
Averaging $\Sigma^2(x)$ over $x$, we get a behavior $ 
\overline{\Sigma^2} \simeq 1/L$. Of course, the answer is not exact, and one could have constructed 
other interpolation forms that have the same asymptotic behavior.
We will see in the next Section that results from numerical simulations
for $\Sigma^2(x)$ can be fitted very well to the phenomenological
expression in Eq. \eqref{eq:24}.
\section{Numerical results}\label{ch4.5}
We have tested  our non-rigorous  theoretical arguments against results 
obtained from numerical simulations. In Fig. \ref{patch1}, we have plotted the 
probability distribution $\mathcal{P}_L(h)$ of height at a site, averaged over all sites. We used $L= 
200$, $500$ and $1000$, and averaged over $10^8$ different configurations 
in the steady state. We plot the scaled distribution function 
$ \mathcal{P}_L(h)/\sqrt{L}$ versus the scaled height $( h - \bar{h}) 
\sqrt{L}$, where $\bar{h}$ is the average height per site. 
Using law of mass balance it is easy to
show that the average height per site is exactly equal to the average value of the 
addition of height, hence $\bar{h}= 1.0$ in our case. A good collapse is seen, which verifies the fact that the width 
of the peak varies as ${L^{-1/2}}$. 
\begin{SCfigure}
\includegraphics[scale=0.65,angle=270]{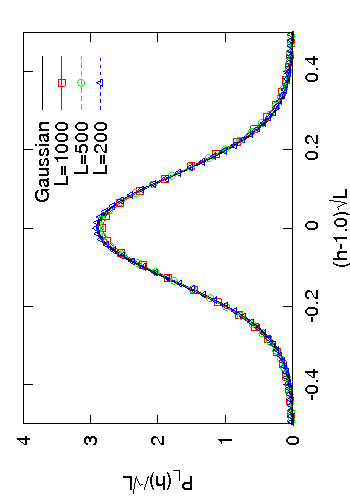}
\caption{ Scaling collapse of the probability distribution
$\mathcal{P}_L(h)$ of height per site in the steady state for different systems 
of size $200$, $500$ and $1000$. The distribution is well described by a Gaussian of width $0.136$.}
\label{patch1}
\end{SCfigure}  

The  dependence of the variance of $h(x,t)$ on
$x$ is plotted in Fig. \ref{patch2} for systems of length $200$, $300$ and 
$400$. The data was  obtained by averaging over $10^8$ 
avalanches. We plot $( L + \lambda) \Sigma^2(x)/\sigma^2$ versus 
$x_{eff}/L_{eff}$, where $x_{eff}$ differs from $x$ by an amount 
$\delta$ to take into account the corrections due to end effects. Then, for 
consistency, $L$ is replaced by $L_{eff} = L + 2 \delta$.   For 
the specific choice of $\lambda=5 \pm 1$ and $\delta=1.0 \pm 0.2$, we get 
a good collapse of the curves for different $L$.  We also show a fit to 
the proposed interpolation form in Eq. \eqref{eq:24}, with $K'=1.00 \pm
0.01$ and $B'= 1.5 \pm 0.2$. We see that the fit is very good.
\begin{SCfigure}
\includegraphics[scale=0.65,angle=270]{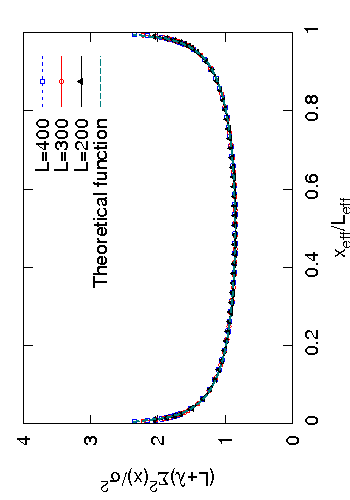}
\caption{ Scaling  collapse of $\Sigma^2(x)/\sigma^2$ at site $x$ 
for systems of different length $L$. }
\label{patch2}
\end{SCfigure}

In order to check the logarithmic dependence of $\Sigma^2(x)$
on $x$ for small $x$, we re-plot the data in Fig. \ref{patch3} using 
logarithmic scale for $x$.  We get a good collapse of the data for 
different $L$, supporting our proposed dependence in Eq. \eqref{eq:24}.
\begin{SCfigure}
\includegraphics[scale=0.65,angle=270]{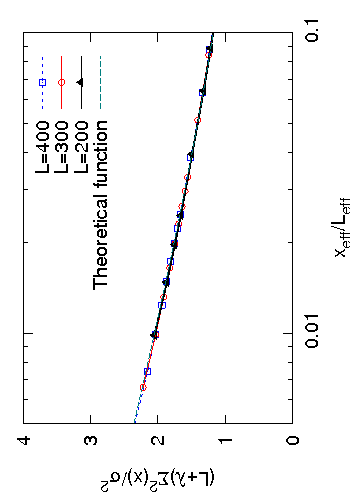}
\caption{ The same plot in Fig. \ref{patch2} resolved more at the  left boundary of 
the model and taking $x$ axis  in log scale.}
\label{patch3}
\end{SCfigure} 
\section{Concluding remarks}\label{ch4.6}
To summarize, we have studied the emergence of quasi-units in the
one-dimensional Zhang sandpile model. The variance of height variables 
in the steady state is governed by the balance between two competing 
processes. The randomness in the drive i.e., the height of addition, 
tends to increase the variance in time. On the other hand, the topplings 
of height variables tend to equalize the excess height by distributing 
it to the nearby sites. There are on an average $O(L^2)$ topplings per 
avalanche. Hence, in one dimension there are, on an average, $O(L)$ 
topplings per site per avalanche. For large system size, the second 
process dominates over the first and the variance becomes low. We have 
shown that the variance vanishes as $1/L$ with increasing system size 
and the probability distribution of height concentrates around a 
non-random value which depends on the height of addition. We have also 
proposed a functional form for the spatial dependence of variance of 
height which incorporates the correct limiting behaviors, and matches 
very well with the numerical data.

An interesting question is whether one can extend these arguments to the 
two-dimensional Zhang model. There is some numerical evidence for the sharpening of height
peaks as the system size is increased \cite{lubeck}. In the simplest scenario
\cite{lubeck}, there are $z-1$ peaks located at multiples 
of the quasi-unit $h_0=h_c(z+1)/z^2$, where 
$z$ is the lattice coordination number. This would imply
that the asymptotic behavior of the two dimensional Zhang model is same as the height type BTW model
in two dimension.
However, this simple scenario can not be fully correct. During an avalanche, there is a
finite probability that a site  receives height from two neighbors in the previous time-step. If the height at the site before
was $(z-1)h_0$, it will tranfer an height approximately equal to $h_0' = (z+1) h_0/z$ to its neighbors. As such events occur with nonzero probability,
in the one site height distribution function, there would have to be  peaks around $h_0'$, $h_0 + h_0',\cdots$ also.
These peaks then give rise to other peaks. With many peaks,
the definition of the width of a peak becomes somewhat ambiguous.
As the number of topplings per site varies only as $\ln L$,
the width is expected to decrease much more slowly with $L$ than in one dimension.

\chapter{Stochastic sandpile models}\label{ch:manna}
\textit{Based on the paper \cite{myjsp1}} by Tridib Sadhu and Deepak Dhar.

\begin{itemize}
\item[\textbf{Abstract}:]
In this chapter, we study the steady state of the abelian sandpile models with stochastic 
toppling rules. The particle addition operators commute with each other, 
but in general these operators need not be diagonalizable. We 
use their abelian algebra to determine their eigenvalues, and the 
Jordan block structure. These are then used to determine the
probability of different configurations in the steady state. We
illustrate this procedure by explicitly determining the numerically
exact steady state for a one dimensional example, for systems of
size $\le12$, and also study the change in density of particles along
the lattice, in the steady state.
\end{itemize}

\section{Introduction}
Sandpile models with stochastic toppling rules are important subclass of 
sandpile models \cite{dharphysica06}.  The first such model was studied by 
Manna \cite{manna}, and  these are usually known as Manna models in
the literature. They are able to describe the avalanche 
behavior seen experimentally in the piles of granular media much better than 
the deterministic models \cite{osloexpt}. Also, in numerical studies, one 
gets better scaling collapse, and consequently, more reliable estimates 
for the values of the critical exponents, than for models with deterministic 
toppling rules \cite{stochnum}.

Unfortunately, the theoretical understanding of models with 
stochastic toppling rules is much less than that of their deterministic 
counterparts, e.g. the Bak-Tang -Wiesenfeld (BTW) model \cite{btw}.  For 
example, there is no analogue of the burning test  
to distinguish the transient and  the recurrent states of a 
general Manna model. For the deterministic case, it is known that 
all the recurrent configurations occur with equal probability in the 
steady state. A similar characterization of the steady state is not 
known in the Manna case.  The steady state has been 
explicitly determined only for the fully directed stochastic models 
\cite{asap,rittenberg,kloster,paczuski}. In some cases, one can formally characterize the 
recurrent states of the model, e.g. the $1$-dimensional Oslo rice pile 
model, but a straightforward direct depth-first
calculation of the exact
probabilities of different configurations in the steady state takes 
$\mathcal{O}(exp{(L^3)})$ steps where $L$ is the system length
\cite{dharoslo}. While 
the exact values of the critical exponents have been conjectured for $(1+1)$ 
dimensional directed Manna model \cite{kloster,paczuski}, the 
prototypical undirected Manna model in one dimension has resisted an 
exact solution so far \cite{dickman1,dickman2,dickman3}. 
In higher 
dimensions, most of the studies are only numerical.

The conditions under which different scaling behaviors are seen in stochastic models is a long debated issue.
Initial studies suggested that the stochastic sandpile model and the BTW model exhibit
similar scaling behavior
\cite{guilera,pietronero,vespignani}. However
later large
scale simulations showed that the stochastic
sandpile models constitute a universality class different from
their deterministic counterparts: the critical exponents, scaling
functions and geometrical features are different for the two classes
of models
\cite{biham,lubeck1,milshtein,lubeck3,tebaldi,campelo}. Further evidence came from the
qualitative differences in their avalanche distribution which has a
multi-fractal nature for BTW model, whereas it follows simple finite
size scaling \cite{dickman1,tebaldi,campelo} for the Manna model. Also the directed version of the above models exhibit different scaling
behavior than their undirected counterparts \cite{satorras}.
Numerical results suggest that both the abelian and non-abelian
Manna model constitute a universality class \cite{biham2} different from
Directed-Percolation (DP) \cite{bonachela1}. 
However, the corresponding fixed points are unstable with respect to
introduction of perturbation (``stickiness'') and with stickiness the
critical behavior flows to the DP universality class
\cite{mohanty1,bonachela2,mohanty2}. While there is
a controversy about the generic DP behavior of undirected stochastic
sandpile models; for the directed case, the numerical evidence for
asymptotic DP behavior is quite convincing.

While the original Manna model did not have the abelian property of
the BTW model, one can construct stochastic toppling rules with
abelian property \cite{sasm}. In this paper, we discuss this abelian
version of the stochastic Manna model. We 
shall use the terms Deterministic abelian Sandpile Models (DASM)
(see Ch. \ref{ch:intro})  and 
Stochastic abelian Sandpile Models (SASM), if we need to distinguish 
between these two classes of models. 
In DASM the relaxation rules satisfy pair wise balance
\cite{barma}, which makes the model analytically tractable and the
recurrent configurations become equally probable in the steady state. However
the stochastic models do not have this property, and the steady
state can not be determined easily.

We use the algebra of the addition operators to determine the steady 
state of the model. This algebraic approach provides a computationally 
efficient method to determine the Markov evolution matrix of the model. 
The addition operators of SASM are not necessarily diagonalizable even 
if we restrict ourselves to the space of recurrent configurations. Using 
the abelian algebra we determine a generalized eigenvector basis in 
which the operators reduce to Jordan block form. We also define a 
transformation matrix between this basis and the configuration basis, 
and express the steady state in the latter. This procedure is 
illustrated by explicitly working out the case of a one dimensional 
Manna model. In this special case, we can show that each Jordan block is 
at most of dimension $2$. We determine the numerically exact steady 
state of the model for systems of size up to $12$ and determine the 
asymptotic density profile by extrapolating the results.

This chapter is organized as follows: In section \ref{sec:ch5.2}, we define the model 
precisely. In section \ref{sec:ch5.2.5} we recapitulate the algebra of
addition operators for DASM, and use it to determine the steady state.
For the stochastic models same definition for the addition operators
does not work and need to be redefined. We do this in section
\ref{sec:ch5.3}, and discuss their algebra. Calculation of the eigenvalues and the 
Jordan block structure of the addition operators are given in section
\ref{sec:ch5.4}. 
The transformation matrix between the generalized eigenvector basis 
and the configuration basis is determined in section \ref{sec:ch5.5}, and is used to 
determine the steady state vector in the configuration basis in section \ref{sec:ch5.6}. The exact numerical determination of the steady state is discussed 
in section \ref{sec:ch5.7} with some concluding remarks in section \ref{sec:ch5.8}.
 
\section{The Model}\label{sec:ch5.2}
 
 We define a generalized Manna model on a graph of $N$ sites with a 
non-negative integer height variable $z_i$ defined at each site $i$. Let 
the threshold height at $i$ be $z_i^c$, and the site is unstable if $z_i 
\geq z_i^c$.  If the system is stable, a sand grain is added at a 
randomly chosen site which increases the height by $1$.  For each site 
$i$, there is a set of $\alpha_i^{max}$ lists $E_{\alpha,i}$ with 
$\alpha=1,2,\cdots,\alpha_i^{max}$.  If a site is unstable, it relaxes 
by the following toppling rule: we decrease its height by $z_{i}^c$. 
Then, with probability $p_{\alpha,i}$, we select the list 
$E_{\alpha,i}$, independent of any previous selections, and then add one 
grain to each site in that list. If a site occurs more than once in the 
list, we add that many grains to that site.

 Toppling at a site can make other sites unstable and they topple in 
their turn, until all the lattice sites are stable. Using an argument
given in \cite{sasm}, it can be easily shown that the above toppling
rule is abelian. Then, it follows from the 
abelian property that the probabilities of different final 
stable configurations are independent of the order in which different 
unstable sites are toppled.

We illustrate these rules with some examples below.
\begin{itemize}
\item
\textbf{Model A} (The one dimensional Manna model): The graph is $L$
sites on a line and $z_i^c= 2$, for all sites. On toppling each
grain is transfered to its neighbors with equal probability. Hence
we have $\alpha_i^{max}=3$, for all $i$, with $E_{1,i}=\{i-1, i-1\}$, 
$E_{2, i}=\{i-1,
i+1\}$, and $E_{3, i}=\{i+1, i+1\}$ and
$p_{1,i}=p_{3,i}=1/4$ and $p_{2,i}=1/2$. Also grains can move out of the
system  if  toppling occurs  at a boundary site.
\item
\textbf{Model B} (The one dimensional dissipative Manna model): Same
as model A except that on toppling a grain
can move out of the system with probability $\epsilon$. Then
$\alpha_i^{max}=6$ and the lists of neighbors $E_1=\{i-1,i-1\}$,
$E_2=\{i-1,i+1\}$, $E_3=\{i+1,i+1\}$, $E_4=\{i-1\}$, $E_5=\{i+1\}$ and
$E_6=\Phi$, where $\Phi$ is an empty set. The corresponding probabilities are
$p_{1,i}=p_{3,i}=(1-\epsilon)^2/4$, $p_{2,i}=(1-\epsilon)^2/2$,
$p_{4,i}=p_{5,i}=\epsilon(1-\epsilon)$ and $p_{6,i}=\epsilon^2$.

In this case, one can use periodic boundary conditions, as there is 
dissipation at all sites. the steady state is critical only in the limit 
$\epsilon \rightarrow 0$.
For the models A and B, it is easy to see that all stable configurations 
occur in the steady state with non-zero probability. We can also define 
stochastic models where the recurrent configurations form only an 
exponentially small fraction of all stable configurations. An example of 
this type is 

\item
\textbf{Model C}:
The graph is a square lattice with $N$ sites and
$z_{\mathbf{i}}^c=2$. Under
toppling, with equal probability two particles are transfered to
either horizontal or vertical neighbors.  Hence $\alpha_i^{max}=2$ with
$E_{1,\mathbf{i}}=\{\mathbf{i}+\mathbf{e}_x,\mathbf{i}-\mathbf{e}_x\}$
and $E_{2,\mathbf{i}}=\{\mathbf{i}+\mathbf{e}_y,\mathbf{i}-\mathbf{e}_y\}$ with
$p_{1,\mathbf{i}}=p_{2,\mathbf{i}}=1/2$.
\end{itemize}

In the following we will mostly confine ourselves to Model A.
Extensions to other cases present no special difficulties.

\section{Determination of the steady state for a
DASM}\label{sec:ch5.2.5}
Before we carry out the analysis of the SASM, we recapitulate how the
steady state for the DASM, defined in section \ref{sec:dasm}, can be determined using the operator algebra
of addition operators \cite{dharphysica06}.
Let us denote the space of stable states (see section
\ref{sec:dasm}) in the DASM as $\Gamma$ spanned by
$\Omega=\prod_{i=1}^{N}z_{i}^{c}$ basis vectors labeled by $C$. We define $P(C,t)$ as the probability of
finding the system in the basis $C$ at time $t$. The time is in
driving steps, \textit{i.e.}, it increases by one when a grain is
added and the system is fully relaxed. To each set
$\{P(C,t)\}$, we associate a vector $|P(t)\rangle$ belonging to the
vector space $\Gamma$, and write 
\begin{equation}
  |P(t)\rangle = \sum_C P(C,t)|C\rangle.
  \label{P(t)}
\end{equation}
This defines a state of the system at time $t$.
Recall that, in section \ref{sec:dasm}, we have defined the particle addition
operators $\textbf{a}_i$, for all $i$, which correspond to adding a sand grain at site $i$
when the system is in state configuration $C$, and relaxing it until a stable
configuration is reached. Thus, these are linear operators which act
on the vector space $\Gamma$ and maps a configuration $C$, in it, to another
configuration $a_{i}C$, which is reached by the avalanche.

The time-evolution of the system is Markovian \cite{vkampen} and the evolution
operator $\mathbf{W}$ is defined by the master equation
\begin{equation}
|P\left( t+1 \right)\rangle =\mathbf{W}|P\left(t\right)\rangle,
\end{equation}
We can write the time-evolution operator in terms of the addition
operators as
\begin{equation}
\mathbf{W}=\frac{1}{L}\sum_{i}\mathbf{a}_{i},
\end{equation}
where $L$ is the number of sites on the lattice.
To solve the time evolution, in general, we have to diagonalize the
evolution operator $\mathbf{W}$. Now, we have shown in section \ref{sec:dasm},
that the addition operators commute with each other. Then all
the addition operators $\left\{ \mathbf{a}_{i} \right\}$ and hence, also $\mathbf{W}$ have a common set
of eigenvectors. Let $|\psi\rangle$ be one such simultaneous eigenvector of
the operators $\left\{ \mathbf{a}_{i} \right\}$, with eigenvalues
$\left\{e^{i\psi_{i}}\right\}$, respectively. Then
\begin{equation}
\mathbf{a}_{i}|\psi\rangle=e^{i\psi_{i}}|\psi\rangle.
\end{equation}

Recall the definition of the toppling matrix $\Delta_{i,j}$,
introduced in section \ref{sec:dasm}, for the DASM. Then, from the toppling
rule, one can easily show, that
\begin{equation}
\mathbf{a}_{i}^{\Delta_{i,i}}=\prod_{j\ne
i}\mathbf{a}_{j}^{-\Delta_{i,j}}.
\label{eq:dasmalgbr}
\end{equation}
Also, recall, we have shown in section \ref{sec:dasm}, that within the
recurrent state space, each addition operator has an inverse. Then,
the above relation can be written as
\begin{equation}
\prod_{j}\mathbf{a}_{j}^{\Delta_{i,j}}=\mathbf{1}, \textrm{  for all
}i.
\end{equation}
Applying the LHS to the eigenvector $|\psi\rangle$ gives $\exp\left(
i\sum_{j}\Delta_{i,j}\psi_{j} \right)=1$, for every $i$, so that
$\sum_{j}\Delta_{i,j}\psi_{j}=2\pi m_{i}$, where $m_{i}$'s are
arbitrary integers. Then inverting,
\begin{equation}
\psi_{j}=2\pi\sum_{i}\left[\Delta^{-1}\right]_{j,i}m_{i},
\end{equation}
where $\Delta^{-1}$ is the inverse of $\Delta$.

The particular eigenstate $|0\rangle$, corresponding to
$\psi_{j}=0$ for all $j$, is invariant under the action of all the
$\mathbf{a}$'s, \textit{i.e.}, $\mathbf{a}_{i}|0\rangle=|0\rangle$.
Thus this must be the stationary state of the system.

\section{Algebra of the addition operators for SASM}\label{sec:ch5.3}
We use the same notations, as in the last section, but this time for
a SASM. So, $\Gamma$ is the space of stable states and
$C\equiv\left\{ z_{i} \right\}$ is a stable height configuration
constituting a complete set of basis vectors.

For stochastic toppling rules, the 
resulting state from the action of $\mathbf{a}_{i}$ on a basis vector
$C$, is not necessarily another basis vector, but a linear combination of
them. So, the addition operators have to be redefined. If the resulting configuration is $C'$
with probability $P_i(C'|C)$, we define
\begin{equation}
  \mathbf{a}_i|C\rangle = \sum_{C'} P_i(C'|C) |C'\rangle,
  \label{a_i}
\end{equation}
for all $C$. Note that the action of any of these operators on a given
configuration gives a unique probability state vector. 

Eq. (\ref{a_i}) is a formal definition of the operators 
$\{\mathbf{a}_i\}$. One can think of these as $\Omega \times \Omega$ 
matrices, but, unlike the DASM, it is quite non-trivial to actually determine
the matrix elements $P_i(C'|C)$ explicitly from the toppling rules. This 
is because of the non-zero probability of an arbitrary large 
number of toppling before a steady state is reached.

For an example, consider the avalanches in model A for system of size $L=3$.
Consider the $2^3$ stable configurations as the basis vectors 
and denote them by their height values $|z_1, z_2, z_3 \rangle$.
The action of $\textbf{a}_2$ on $|0,1,0\rangle$ will
generate a unstable state $|0,2,0\rangle$. Using the toppling rules
we can write the following set of equations for three unstable states
\begin{eqnarray}
  |0, 2, 0\rangle &=& \frac{1}{4}|2, 0, 0\rangle + \frac{1}{2}|1, 0,
1\rangle + \frac{1}{4}|0, 0, 2\rangle, \nonumber \\
  |2, 0, 0\rangle &=& \frac{1}{4}|0, 2, 0\rangle + \frac{1}{2}|0, 1,
0\rangle + \frac{1}{4}|0, 0, 0\rangle, \nonumber \\
  |0, 0, 2\rangle &=& \frac{1}{4}|0, 2, 0\rangle + \frac{1}{2}|0, 1,
0\rangle + \frac{1}{4}|0, 0, 0\rangle.
  \label{|0,2,0}
\end{eqnarray}
We see that there is a nonzero probability that the avalanche
can continue for more than $s$ toppling, for any finite $s$. e.g.  in 
the sequence 
$|0,2,0 \rangle \rightarrow |2,0,0\rangle \rightarrow |0,2,0\rangle \cdots$. 
Thus straight forward application of the relaxation rules does not result 
in a finite procedure to determine the unstable vector $|0,2,0\rangle$
in terms of the stable configurations. Instead, we have to  write Eq. 
(\ref{|0,2,0})
as a matrix equation
\begin{align}
 \textbf{M} 
  \begin{bmatrix}
    |0,2,0\rangle \\
    |2,0,0\rangle \\
    |0,0,2\rangle
 \end{bmatrix}
  =
  \begin{bmatrix}
    |1,0,1\rangle \\
    |0,1,0\rangle \\
    |0,0,0\rangle
 \end{bmatrix},
\end{align}
and then invert it.
More generally, the determination of $P(C'|C)$ involves working in
a larger space of unstable configurations. 

For example in model A, there are $2^L$  stable configurations, 
where each site has $0$ or $1$ particle. Total number of particles is at 
most $L$. On adding one particle, the  number of 
particles can become $L+1$, where initially, only one site will have 
height $2$. However, it is easy to verify that by toppling one can 
generate configurations  where the number of particles at a site is much 
greater than $2$. In fact, all the $L+1$ particles could be at the same 
site. Then the total number of stable and unstable configurations  
$\Omega'$ is the number of 
ways one can distribute $L+1$ or less particles on $L$ sites. It is easily 
seen that $\Omega'$ varies as $4^L$, and one
needs to invert a matrix of size $\Omega' \times \Omega'$.

There are models, generally known as the restricted sandpile models
\cite{restr1,restr2,restr3}, where the toppling rules ensure that the heights do
not exceed a fixed value. For these, the space of allowed configurations
is much smaller. However, the height restriction makes the model
non-abelian.

In this chapter we will use the operator algebra
to obtain an efficient method to determine the probabilities $P(C'|C)$
explicitly which requires inverting a matrix only of size $2^L \times 2^L$.
It has been shown \cite{sasm} that the addition operators for different sites
commute i.e.
\begin{equation}
  [\mathbf{a}_i, \mathbf{a}_j] =0 \rm{,~ ~ ~ for ~ all ~} i,j.
\end{equation}
The proof uses the fact that any stochastic toppling event can be
simulated by a pseudo-random generator. This essentially makes the
toppling deterministic, which has the abelian property. 

However, unlike the DASM, the inverse operators $\{\mathbf{a}^{-1}_i\}$
for SASM need not exist, even if we restrict ourselves to the set of 
recurrent configurations. This is because among the recurrent
states, one can have two different
initial probability vectors that yield the same resultant vector.
This makes the determination of the matrix form of the operators
difficult for this model.

Apart from the abelian property, the operators also satisfy a set 
of algebraic equations, like the Eq. \eqref{eq:dasmalgbr} in DASM. For simplicity of presentation, now on we
consider $z_i^c=z_c$ and $p_{\alpha,i}=p_{\alpha}$ for all sites. Then
consecutive addition of $z_c$ grains at a site 
ensures that the site will topple once and transfers $z_c$ grains to
its neighbors, irrespective of the initial height. Then
the operators obey the following equation
\begin{equation}
  \mathbf{a}_i^{z_c} =
\sum_{\alpha}p_{\alpha}\mathbf{a}^{E_{\alpha, i}} \rm {~for~
}1\le i\le N, 
\end{equation}
where we have used the notation $\mathbf{a}^{E}=\displaystyle
\prod_{x\epsilon E}\mathbf{a}_{x}$ for any list $E$, and
\begin{equation}
  \mathbf{a}_i = \mathbf{1},
  \label{a_0}
\end{equation}
for sites $i$ outside the lattice. In particular for the examples
in section \ref{sec:ch5.3}, these equations are as follows
\begin{eqnarray}
\mathbf{a}_i^2&=&\frac{1}{4}(\mathbf{a}_{i-1}+\mathbf{a}_{i+1})^2
\rm{~~~~~~~~~~~~~~~~~~~~~~~for~Model~A,} \label{modA} \\
\mathbf{a}_i^2&=&[\frac{1-\epsilon}{2}\mathbf{a}_{i-1}+\frac{1-\epsilon}{2}\mathbf{a}_{i+1}+\epsilon
\mathbf{1}]^2 \rm{~~~~~for~Model~B,~and~} \\
\mathbf{a}_\mathbf{i}^2&=&\frac{1}{2}(\mathbf{a}_{\mathbf{i}-\mathbf{e}_x}\mathbf{a}_{\mathbf{i}+\mathbf{e}_x}+\mathbf{a}_{\mathbf{i}-\mathbf{e}_y}\mathbf{a}_{\mathbf{i}+\mathbf{e}_y})\rm{~~~~~~~~for~Model~C.}
\end{eqnarray}

\section{Jordan Block structure of the addition operators}\label{sec:ch5.4}
In general the matrices $\{\mathbf{a}_i\}$ need not be diagonalizable. 
However, using the abelian property, we can construct a common set of
generalized eigenvectors for all the operators $\{a_i\}$ such that in this basis
the matrices simultaneously reduce to Jordan block form. These 
generalized eigenvectors split the vector space $\Gamma$ into
disjoint subspaces, each corresponding to distinct set of eigenvalues.
\begin{lemma}
There will be at least one common eigenvector in
each subspace, for all the addition operators. 
\end{lemma}
\begin{proof}
Consider one of the operators, say $\mathbf{a}_1$.
Let $\Gamma_1$
be the subspace of $\Gamma$ spanned by the (right) generalized eigenvectors of 
$\mathbf{a}_1$ corresponding to the eigenvalue $a_1$. There is at
least one such generalized eigenvector, so $\Gamma_1$ is non-null. We pick one
of the other addition operators, say $\mathbf{a}_2$. From the fact that
$\mathbf{a}_2$ commutes with $\mathbf{a}_1$, it immediately follows
that $\mathbf{a}_2$ acting on any vector in the subspace $\Gamma_1$
leaves it within the same subspace. Diagonalizing $\mathbf{a}_2$
within this subspace, we construct a possibly smaller but still
non-null subspace $\Gamma_2$ which is spanned by simultaneous 
generalized eigenvectors of $\mathbf{a}_1$ and $\mathbf{a}_2$ with eigenvalues
$a_1$ and $a_2$. Repeating this argument with the other operators,
one can construct vectors which are simultaneous eigenvectors of 
all the $\{\mathbf{a}_i\}$. \qed
\end{proof}

Let $|\psi \rangle$ be such an eigenvector, with
\begin{equation}
  \mathbf{a}_i|\psi\rangle = a_i |\psi \rangle \rm{,~ ~ ~ for ~} \rm{~
}\rm{~}1 \le i \le N. 
\end{equation}
Then from Eq.($6$)the eigenvalues satisfy the following set of equations
\begin{equation}
  a_i^{z_c} = \sum_{\alpha}p_{\alpha}a^{E_{\alpha,i}}\rm{~ ~ ~ for
~ } \rm{~ }\rm{~}1 \le i \le N,
\end{equation}
where we have used the notation
$a^{E}=\displaystyle\prod_{x\epsilon E}a_{x}$, for any list $E$.

Rather than work with this general case, we will consider the special
case in model A for simplicity. No extra complications occur in
the more general case. Then, from Eq.(\ref{modA}), the corresponding
eigenvalue equation is
\begin{equation}
  a_i^2 = \frac{1}{4}(a_{i-1}+a_{i+1})^2 \rm{,~ ~ ~ for ~ } \rm{~ }\rm{~}1 \le i \le L 
  \label{eveqn}
\end{equation}
These are L coupled quadratic equations in $L$ complex variables $\{a_i\}$.
We can reduce them to $L$ linear equations by taking square root
\begin{equation}
  \eta_ia_i=\frac{1}{2}(a_{i-1}+a_{i+1}),
  \label{eta}
\end{equation}
where $\eta_i=\pm1$. The Eq. (\ref{a_0}) sets the values for the 
eigenvalues of $\mathbf{a}_0$ and $\mathbf{a}_{L+1}$ which are 
\begin{equation}
  a_0=a_{L+1}=1.
  \label{a_1}
\end{equation}
There are $2^L$ different choices for the set of $L$
different $\eta$'s and for each such choice, we get a set of
eigenvalues $\{a_i\}$. In general, there will be
degenerate sets of eigenvalues and the degeneracy arises
if one of the $a_i$ is zero. Using the triangular inequality in \eqref{eta}
we get
\begin{equation}
  2|a_i|\le|a_{i-1}|+|a_{i+1}|,
\end{equation}
i.e. $|a_i|$ are convex functions of discrete variables $i$. Then,
given the boundary condition in Eq. (\ref{a_1}), there could
at most be one $a_i=0$ in the solution for a given $\{\eta_i\}$,
which means that each eigenvalue set can be at most doubly degenerate.

%
Finding the number of  degeneracies of solutions is interesting but
difficult in general. We show that 
\begin{lemma}
For $L=3$ (mod $4$) the number of degenerate sets of eigenvalues $\ge2^{(L+1)/2}$.
\end{lemma}
\begin{proof}
Consider the system of length $L=4m+3$, with $m$ being a non negative
integer. For any given set $\{\eta_i\}$, $i=1$ to $2m+2$, it is
possible to construct a solution $\{b_i\}$ of Eq. (\ref{eta}) with
$i\le2m+2$ which satisfies $b_0=1$ and $b_{2m+2}=0$.
Clearly, from  Eq.(\ref{eta}), if we have the solution $\{a_i\}$ 
corresponding to a particular set $\{\eta_j\}$, one can construct the 
solution $\{a_i'\}$ corresponding to $\{ \eta_j'= 
-\eta_j\}$  using 
$a_j' = (-1)^j a_j$. Using this symmetry we
extend $\{b_i\}$ ($i = 1 $ to $(L+1)/2$) to form a set $\{a_i\}$ for 
$i = 1$ to $L$  as follows:
\begin{eqnarray}
a_{i}&=&b_{i} \rm{~~~~~~~~~~~~~~~~~for~}i\le 2m+2, \\
     &=&(-1)^{i}b_{L+1-i} \rm{~~~for~}i>2m+2. 
\end{eqnarray}
This is a solution of Eq.(\ref{eta}) for the set $\{\eta'_i\}$ with 
\begin{eqnarray}
\eta_{i}'&=&\eta_{i} \rm{~~~~~~~~~~~~~~~~~for~} i\le 2m+2, \\
         &=&-\eta_{L+1-i} \rm{~~~~~~~~for~}i>2m+2.
\end{eqnarray}
and this solution $\{a_i\}$ satisfies the boundary conditions
$a_0=1$, $a_{L+1}=1$, and $a_{2m+2}=0$ (Fig.$1$). There are
$2^{2m+2}$ such solutions possible corresponding to all possible sets
of $\{\eta_i'\}$, and this gives the lower bound for the number of
degenerate solutions.\qed
\end{proof}

\begin{figure}
   \begin{center}
    \includegraphics[scale=0.25,angle=0]{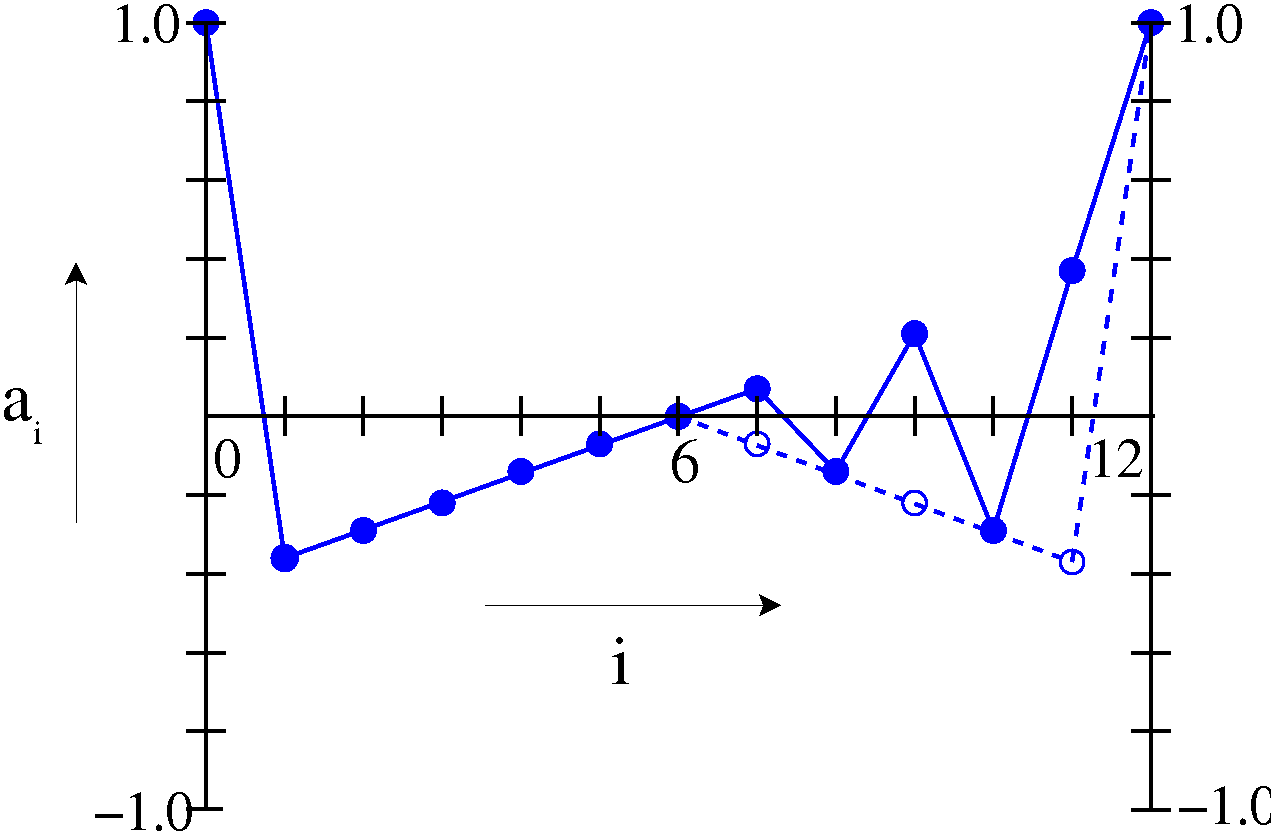}
    \caption{The filled circles denote a dependence of the eigenvalues
$a_i$ on $i$ for $L=11$ with $a_6=0$. The unfilled circles show the
results obtained by an inversion of $a_i$, with $i\le6$, around
the central site. The eigenvalues on the right-half of the lattice are obtained by
multiplying $(-1)^i$ to the values of the unfilled circles.}
    \end{center}
    \label{figure1}
\end{figure}

A direct numerical calculation for $L\le20$ shows
that if $L\ne3$ (mod $4$), all $2^L$ sets of eigenvalues are
distinct. We present the degeneracies of the solutions in
Table.$1$. Calculation for simple choices of $\eta$ shows that
the degeneracies are possible only if $L=3$ (mod $4$).

For example,
consider $\eta_i=-1$ for $i=L$ and for the rest of the sites it is
$1$.  Then $a_i$ is of the form $a_i=1- \alpha i$, for all $i$. If we 
want this to be zero for $i=k$, we must have $\alpha = 1/k$.  Then, 
requiring that the equation (\ref{eta}) be satisfied at $i=L$, gives   
$3L=4k+1$, i.e. $L=3$ (mod $4$). Similarly for
the set with $\eta_{L-1}=-1$ and $1$ for rest of the sites imposes a
condition on length $3L=8k+1$, which is also a subset of $L=3$ (mod
$4$). Finding a general proof that degeneracies occur only if $L 
=3$ (mod $4$) remain an  open problem.

For each degenerate subspace there is a generalized eigenvector
linearly independent of the eigenvector corresponding to the
eigenvalue of the subspace. In general, let us denote them by
$|\{a_i\};n \rangle$, where $n=1$ for the eigenvector and $n=2$ for
the generalized eigenvector. For non-degenerate subspace $n$ can only
be $1$. The vectors satisfy the following equations
\begin{eqnarray}
  \mathbf{a}_i|\{a_j\};1\rangle &=& a_i|\{a_j\};1\rangle, \nonumber \\
  \mathbf{a}_i|\{a_j\};2\rangle &=& a_i|\{a_j\};2\rangle + \alpha_i
|\{a_j\};1\rangle,
  \label{gen}
\end{eqnarray}
where $\alpha$'s are complex numbers. Then using the Eq. (\ref{eta}) it can be shown
easily that $\alpha$'s satisfy the following equation
\begin{equation}
  \eta_i\alpha_i=\frac{1}{2}(\alpha_{i-1}+\alpha_{i+1}).
  \label{etaalpha}
\end{equation}
This is similar to the Eq. (\ref{eta}), except the boundary conditions
which are
\begin{equation}
  \alpha_0=\alpha_{L+1}=0.
\end{equation}
For a given set of $\{\eta_i\}$, these are $L$ simultaneous set of homogeneous 
linear equations. If
$\left\{ \alpha_{i} \right\}$ is a solution, then $\left\{ \lambda
\alpha_{i}\right\}$ is also a solution. Thus there are infinitely many possible
solutions, each corresponding to different choices of $\lambda$.
In order to get a single solution we choose $\alpha_i=1$ if $a_i=0$, without
loss of generality. This corresponds to choosing a particular 
normalization of the rank $2$ eigenvectors.
The solution of both the equations (\ref{eta}) and
(\ref{etaalpha}) can be easily obtained numerically. 
The generalized eigenvectors and the Jordan block
form of the addition operators for the
system of size $L=3$ are given in the appendix.
\begin{table}
  \begin{center}
    \begin{tabular}{|c||c||c|c|c|c|c|c|c|c|c|c|}
      \hline
      ~L~ & $g$ & $N_1$ & $N_2$ & $N_3$ & $N_4$ & $N_5$ & $N_6$ & $N_7$ &
$N_8$ & $N_9$ & $N_{10}$ \\
      \hline
      3 & 4 & 0 & 4 & & & & & & & &  \\
      \hline
      7 & 40 & 0 & 0 & 8 & 24 & & & & & & \\
      \hline
      11 & 136 & 0 & 0 & 0 & 8 & 0 & 120 & & & & \\
      \hline
      15 & 1304 & 0 & 0 &  0 & 4 & 32 & 48 & 288 & 560 & & \\
      \hline
      19 & 3024 & 0 & 0 & 0 & 0 & 0 & 8 & 0 & 288 & 0 & 2432 \\
      \hline
    \end{tabular}
    \caption{Degeneracies arise if one of the $a_i$ is
zero in a solution of
Eq.(\ref{eta}). In the table, $g$ denotes the total number of solutions with one of the $a_i=0$  i.e. the
total number of degenerate sets of solution. $N_i$ is the
number of solutions with the eigenvalue $a_i=0$.
Values for the other half of the system can be obtained using symmetry.}
  \end{center}
  \label{first}
\end{table}

\section{Matrix representation in the configuration basis}\label{sec:ch5.5}
Given the well-defined action of the addition operators on the
generalized eigenvectors it is
possible to define a transformation
matrix $\mathbf{M}$ between the  configuration basis and the generalized 
eigenvector basis. 
\begin{equation}
|\{z_i\}\rangle = \sum_{j}\mathbf{M}_{\{z_i\},j}|\psi_j\rangle,
\end{equation}
where $|\{z_i\}\rangle$ is the  basis vector of $\Gamma$  corresponding to 
the height configuration
$\{z_i\}$ and $|\psi_j\rangle$ is the $j$th generalized eigenvector. Let us
express the configuration $|\{0\}\rangle$, with all sites empty, as a 
linear combination of all the generalized eigenvectors.
\begin{equation}
  |\{0\}\rangle = \sum_{j}c_j|\psi_j\rangle,
\end{equation}
where $c_j$s are constants. Then all the stable configurations can be
obtained by adding grains at properly chosen sites in $|\{0\}\rangle$.
\begin{equation}
  |\{z_i\}\rangle = \prod_i\mathbf{a}_i^{z_i}|\{0\}\rangle =
\sum_jc_j\prod_i\mathbf{a}_i^{z_i}|\psi_j\rangle,
\end{equation}
and hence
\begin{equation}
\mathbf{M}_{\{z_i\},j}=\langle \{z_i\}|\prod_i\mathbf{a}_i^{z_i}|\psi_j\rangle.
\end{equation}

The action of the addition operators on the generalized eigenvectors, 
for example Eq.(\ref{gen}) for model A, would generate the elements of 
the matrix $\mathbf{M}$. 
Given $\mathbf{M}$, we can get the eigenvectors of $\mathbf{a}_i$, in
the configuration basis, in 
particular, the steady state vector, by the inverse transformation 
\begin{equation}
|\psi_j\rangle = \mathbf{M}^{-1} |\{z_i\}\rangle.
\end{equation}
The  addition operators in the  
configuration basis  are obtained using the similarity transformation 
$\mathbf{M}\mathbf{a}^J_i\mathbf{M}^{-1}$. An explicit form of 
$\mathbf{M}$ for model A of length $L=3$ is given in the appendix.

\section{Determination of the steady state vector}\label{sec:ch5.6}
Just like the DASM in section \ref{sec:ch5.2.5},
the time-evolution of the system is Markovian and the
evolution operator can be written in terms of the addition operators as
\begin{equation}
  \mathbf{W}=\frac{1}{L}\sum_i\mathbf{a}_i.
\end{equation}
Then the common eigenvector of all the addition operators
corresponding to eigenvalue $1$ is the steady
state vector of the system. The steady state vector can be determined
in the stable configuration basis using the matrix $\mathbf{M}^{-1}$.
For model A of length $L=3$ the steady state vector is 
\begin{eqnarray}
  |S\rangle = &&\frac{13}{392}|0, 0, 0\rangle + \frac{1}{16} |1, 0, 0\rangle 
+ \frac{47}{392} |0, 1, 0\rangle + \frac{3}{16} |1, 1, 0\rangle \nonumber \\
   &&+  \frac{1}{16} |0, 0, 1\rangle + \frac{13}{98} |1, 0, 1\rangle 
+ \frac{3}{16} |0, 1, 1\rangle + \frac{3}{14} |1, 1, 1\rangle,
\end{eqnarray}
where the stable configurations are denoted by $|z_1,z_2,z_3\rangle$ with
$z_i$ as the height of the $i$th site.
The amplitude of each term in the expansion is the probability
of finding the corresponding height configuration in the steady state.

\begin{SCfigure}
    \includegraphics[scale=0.70,angle=0]{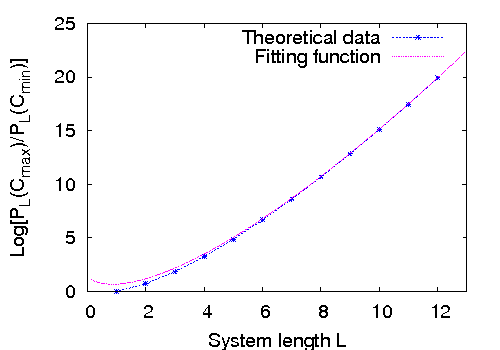}
    \caption{The ratio of the probability of the most probable
configuration $C_{max}$ (all occupied) and the least probable
configuration $C_{min}$ (all sites empty) plotted as a function of 
the system
length $L$. The fitting function $f(x)=a-bL+cL\log L$, with $a=1.50,
b=0.80$ and $c=0.94$.}
    \label{fig2}
\end{SCfigure}

\section{Numerical Results}\label{sec:ch5.7}
Here we describe the results of numerical calculations for the exact steady state of model A for
different system lengths and discuss its properties.
As shown in Eq. (\ref{eta}) and Eq. (\ref{etaalpha}) the
eigenvalues $\{a_i\}$ and the off-diagonal matrix elements
$\{\alpha_i\}$ form sets of linear equations
for a given set of $\{\eta_i\}$. We solve them by LU
decomposition method. Because of the tridiagonal structure
of the equations, only $\mathcal{O}(2^L)$ number of steps are required
to get the solution. The maximum number of steps ($\mathcal{O}(2^{3L})$)
are required for the inversion of the transformation matrix
$\mathbf{M}$. We have used the Gauss-Jordan elimination method
for the inversion. It is important to note that, the
maximum system length $L$, possible to treat by this method,
is determined by the limited memory of the computers,
and not by the computation time. Using desktop computers we were able
to determine $\mathbf{M}$ exactly for systems of size $L\le12$. 
\begin{figure}
 \begin{center}
    \includegraphics[scale=0.75,angle=0]{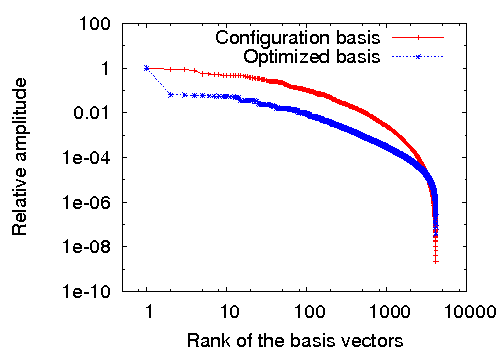}
    \caption{The amplitudes, normalized with its largest value, corresponding to the basis vectors
 in the steady state plotted
as a function of the rank of the basis vectors.
The vectors are arranged in decreasing orders of their
amplitudes. The plot is given for the configuration basis and the
optimized basis for model A of size $L=12$.}
    \label{fig3}
  \end{center}
\end{figure}

We note that as $L$ is increased, the second largest eigenvalue of
$\mathbf{W}$
tends to $1/2$. Thus, the gap between the largest and the next largest
eigenvalue of the relaxation matrix does not tend to zero. 
This gap measures the relaxation time of the system in terms of the
macro-time unit of interval between addition of grains. However, the
average duration of an avalanche measured in terms of micro-time unit
of duration of a single toppling event does diverge, as system size
increases.

An interesting question is the extent of  variation between 
probabilities of
different configurations in the steady state. In the one-dimensional
Oslo model, for a system of $L$ sites, the ratio of probabilities of
the most probable to the least probable configuration varies as
$\exp(L^3)$\cite{dharoslo}.
However in model A, we find that the ratio is not quit as large, and it
only varies approximately as $\exp(0.94L\log L)$ (Fig.\ref{fig2}) for 
large $L$. 

This suggests that the
exact steady state is not a product measure state. To check this we define a 
product basis $|\psi'\rangle =
\displaystyle\prod_i|\psi_i'\rangle$,
where $|\psi_i'\rangle$ could be any one of the two orthogonal vectors
\begin{eqnarray}
  |1'\rangle &=& \cos\phi_i|1\rangle + \sin\phi_i|0\rangle, \nonumber \\
  |0'\rangle &=& \sin\phi_i|1\rangle - \cos\phi_i|0\rangle,
\end{eqnarray}
with $\phi_i$ a real number. Then in this basis the steady state can
be written as
\begin{equation}
  |S\rangle = \sum_{\psi'} P(\psi')|\psi'\rangle.
\end{equation}
We choose $\{\phi_i\}$ so that the ratio between the amplitudes
of basis vectors with next-largest and largest 
amplitudes becomes as small as possible (this would become zero, 
if the state was a product measure state). 
In Fig.\ref{fig3}, we have plotted for system of 
size $L=12$, the relative amplitudes in both configuration
basis and the optimized product basis
as a function of the rank of the basis vectors with the vectors
arranged in decreasing orders of their amplitudes. In the
optimized basis the second highest probability is only $10$ times
smaller than the highest probability. This shows that the steady
state measure is not a product measure.

\begin{figure}
  \begin{center}
    \includegraphics[scale=0.75,angle=0]{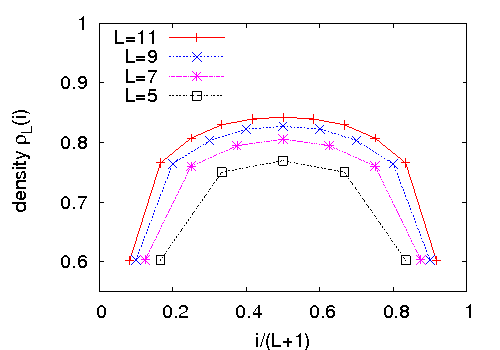}
    \caption{Average steady state density $\rho_{L}(i)$ at site $i$ for
model A of different length $L$.}
    \label{fig4}
  \end{center}
\end{figure}
\begin{table}
  \begin{center}
    \begin{tabular}{|c|c|c|}
      \hline
      $ ~ ~ ~ ~ ~ L ~ ~ ~ ~ ~ $ & $ ~ ~ ~ ~ ~ ~ \bar{\rho}_L ~ ~ ~ ~ ~ ~ $ & $ ~ ~ ~ ~ ~ ~ \rho_L(x_m) ~ ~ ~ ~ ~ ~ $ \\
      \hline
      2 & 0.583333 & 0.583333 \\
      \hline
      3 & 0.634354 & 0.709184 \\
      \hline
      4 & 0.669262 & 0.737000 \\
      \hline
      5 & 0.695210 & 0.769704 \\
      \hline
      6 & 0.715472 & 0.786491 \\
      \hline
      7 & 0.731879 & 0.805897 \\
      \hline
      8 & 0.745514 & 0.816009 \\
      \hline
      9 & 0.757080 & 0.827217 \\
      \hline
      10 & 0.767051 & 0.834600 \\
      \hline
      11 & 0.775760 & 0.842665 \\
      \hline
      12 & 0.783451 & 0.848054 \\
      \hline
    \end{tabular}
    \caption{The values of particle density in the steady state 
    for the model A of different length $L$. Here $\bar{\rho}_L$ denotes the
    steady state particle density averaged over all sites and $\rho_L(x_m)$ denotes the
    steady state particle density at the central site.}
  \end{center}
  \label{second}
\end{table}
\begin{table}
  \begin{center}
    \begin{tabular}{|c|c|c|c|}
      \hline
      ~L~ & $1/\bar{\rho}_\infty$ & $B$ & $\nu_\perp$ \\
      \hline
      \hline
      3 & ~ ~ 1.061 ~ ~ & ~ ~ 1.128 ~ ~ & ~ ~ 0.656 ~ ~ \\
      \hline
      4 & 1.049 & 1.132 & 0.641 \\
      \hline
      5 & 1.053 & 1.132 & 0.646 \\
      \hline
      6 & 1.049 & 1.131 & 0.640 \\
      \hline
      7 & 1.050 & 1.131 & 0.641 \\
      \hline
      8 & 1.049 & 1.130 & 0.639 \\
      \hline
      9 & 1.049 & 1.130 & 0.639 \\
      \hline
      10 & 1.049 & 1.130 & 0.639 \\
      \hline
      11 & 1.049 & 1.130 & 0.639 \\
      \hline
    \end{tabular}
    \caption{The sequential fit of the functional form in Eq. (\ref{rhobar}) to the
    data for average particle density for model A of different length $L$ given in Table $2$.}
  \end{center}
  \label{third}
\end{table}

The steady state density for different sites
are plotted in Fig. \ref{fig4} for different system sizes.
Amongst the different fitting forms that we tried,  the following 
functional form gives the best fit
\begin{equation}
  \frac{1}{\rho_L(x)} = \frac{1}{\bar{\rho}_\infty} + b[\frac{1}{(x+d)^{\nu_\perp}} + \frac{1}{(L+1-x+d)^{\nu_\perp}}],
\end{equation}
where $\bar{\rho}_\infty, b, \nu_\perp$ and $d$ are real numbers. 
Using this functional form the steady state particle density
averaged over all sites for system of size $L$ can be written as
\begin{equation}
  \frac{1}{\bar{\rho}_L} = \frac{1}{\bar{\rho}_\infty} + \frac{B}{(L+\delta)^{\nu_\perp}},
  \label{rhobar}
\end{equation}
where $B$ is a real number and $\bar{\rho}_\infty$
is the asymptotic value of the average particle density.
The exact value of $\bar{\rho}_L$ and the particle density
at the central site $\rho_L(x_m)$ are listed in the Table $2$
for different system sizes. 
The sequential fitting method is used to find the values of 
$\bar{\rho}_\infty$, $B$, $\nu_\perp$ and $\delta$ from these data. For a given choice
of $\delta$, these values are obtained numerically by solving the Eq. (\ref{rhobar})
for three consecutive lengths $L-1$, $L$ and $L+1$. Best
convergence of the values of $\bar{\rho}_\infty$, $B$ and $\nu_\perp$ are obtained for
$\delta=1.1$, which are tabulated in Table $3$.
The asymptotic value of the average particle density
converges to $\bar{\rho}_\infty =0.953$ which is close
to the more precise estimate $0.94885(7)$, from Monte 
Carlo simulations  \cite{dickman1}.

\section{Concluding remarks.}\label{sec:ch5.8}
For a general SASM with $N$ sites, the calculation of eigenvalues 
involves solving $N$ coupled polynomial equations in $N$ variables. This 
can be done in time polynomial in $\Omega$, the number of stable 
configurations of the model. These are then used to construct the 
transformation matrix ${\bf M}$ of size $\Omega \times \Omega$. Finally 
inverting the matrix ${\bf M}$ gives us the eigenvectors of the 
evolution operator, in particular the steady state.

Of course, to determine the steady state of any Markov chain on $\Omega$ 
states, we need to determine the eigenvectors of the evolution matrix of 
size $\Omega \times \Omega$. The point here is that the specification of 
the toppling rules does not directly specify the evolution matrix, and 
determining the matrix elements of the latter from the toppling rules is 
computationally very nontrivial. Using  the 
abelian property, we are able to tackle this problem.

For a generic model with some parameters, e.g. the model B, except for 
special symmetries, one does not expect degeneracies in eigenvalues to
occur 
for a generic value of the parameters.  For special values of the 
parameters, if there is a  non-trivial Jordan 
block structure of the evolution operator, it would  show up in the 
time-dependent correlation
functions of the model by the presence of terms of the type $ t \exp( 
-\lambda_j 
t)$, in addition to the usual sum of terms of the type $\exp( 
-\lambda_j 
t)$.

In particular we have
explicitly calculated the steady state for a specific model (model A
in section 2) of length L$\le12$. Extrapolating the results
we determined the asymptotic density profile in the steady state.
The power-law profile of deviations from the mean value near the ends
would be important for determining the avalanche exponents of the model 
\cite{lubeck4}. This remains an interesting open problem.

\appendix
\chapter{Solution of Laplace's equation on hexagonal
lattice\label{ap:laplace}}
We follow the calculation given in \cite{atkinson}.
Let $F_{r}$ be the solution of the Poisson equation
\begin{equation}
\sum_{r'}F_{r'}-mF_{r}=I_{r},
\end{equation}
defined on a lattice, where the sum is carried over the nearest
neighbors of the vertex $r$
and $m$ is the total number of them. For a hexagonal lattice $m=3$ and for a triangular lattice $m=6$.
Clearly a hexagonal lattice can be constructed from a triangular
lattice by a $\Delta-Y$
transformation as shown in the Fig.\ref{fig:dy}. Let $F_{r}^{hex}$ and
$F_{r}^{tri}$ are the solutions of the
Poisson equation with same $I_{r}$ defined at the common vertices and
$I_{r}=0$ for the additional vertices in the
hexagonal lattice. It is easy to check that $F_{r}^{hex}=3F_{r}^{tri}$ for the common vertices and for
the additional ones the solution is the average of its value at the neighboring sites (see Fig.\ref{fig:tran}).
We first detrmine $F^{tri}$ and then use the relation to determine $F^{hex}$.
\begin{SCfigure}
\includegraphics[width=8cm,angle=0]{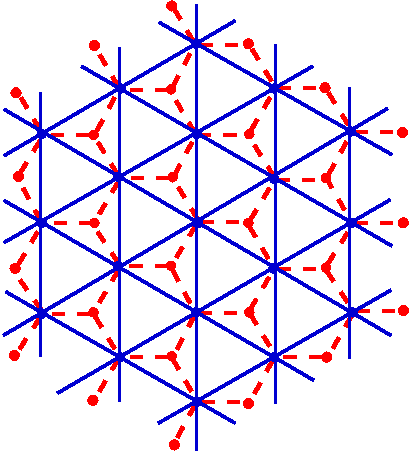}
\label{fig:dy}
\caption{A hexagonal lattice can be constructed from a triangular lattice by a $\Delta-Y$ transformation. Blue solid lines form a triangular lattice and the red dotted lines form a hexagonal lattice.}
\end{SCfigure}
\begin{figure}
\begin{center}
\includegraphics[width=8cm,angle=0]{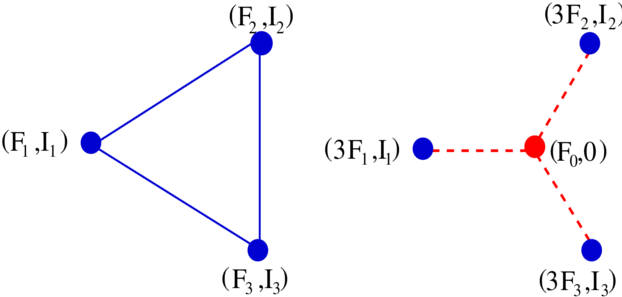}
\label{fig:tran}
\caption{If $F_{r}$ is a solution of the Poisson equation on the
triangle with charge $-I_{r}$ defined at the vertices, then $3F_{r}$ is a solution on the $Y$ shaped graph with the same charges.}
\end{center}
\end{figure}

Consider a triangular lattice with lattice edges of length $1$. Each
vertex on the lattice is denoted by a pair of integers $\left( p, q \right)$ where the complex coordinate
of a site is
\begin{equation}
z=p\frac{\sqrt{3}}{2}+iq\frac{1}{2},
\end{equation}
with $p+q=$ even.
The Poisson equation on this lattice is
\begin{equation}
F^{tri}_{p+1,q+1}+F^{tri}_{p+1,q-1}+F^{tri}_{p-1,q+1}+F^{tri}_{p-1,q-1}+F^{tri}_{p,q+2}+F^{tri}_{p,q-2}-6F^{tri}_{p,q}=I_{p,q}
\label{poisson2}
\end{equation}
Define the discrete Fourier transformation
\begin{equation}
V\left( k_{1}, k_{2} \right)=\sum_{p}\sum_{q}\exp\left[-i\left( \frac{\sqrt{3}}{2}k_{1}p+\frac{1}{2}k_{2}q \right)\right]F_{p,q}
\end{equation}
\begin{equation}
F_{p,q}=\sqrt{3}\int_{-\frac{2\pi}{\sqrt{3}}}^{\frac{2\pi}{\sqrt{3}}}\frac{dk_{1}}{4\pi} \int_{-2\pi}^{2\pi}\frac{dk_{2}}{4\pi}\exp\left[i\left( \frac{\sqrt{3}}{2}k_{1}p+\frac{1}{2}k_{2}q \right)\right]V\left( k_{1},k_{2} \right)
\label{ftrans}
\end{equation}
Using the expression in equation \ref{ftrans} it can be shown that
\begin{equation}
V\left( k_{1},k_{2} \right)=\frac{\widetilde{I}\left( k_{1}, k_{2} \right)}{4\cos\left( \sqrt{3}k_{1}/2 \right)\cos\left( k_{2}/2 \right)+2\cos\left( k_{2} \right)-6},
\end{equation}
where $\widetilde{I}\left( k_{1},k_{2} \right)$ is the discrete Fourier transformation of $I_{p, q}$. We are interested in a solution with
$I_{p,q}=I\delta_{p,0}\delta_{q,0}$, where $I$ is normalization constant. Also
considering the condition that the function $F_{p,q}^{tri}=0$, at the origin, it can be shown that
\begin{equation}
F^{tri}_{p,q}=\frac{I\sqrt{3}}{96\pi^{2}}\int_{-\frac{2\pi}{\sqrt{3}}}^{\frac{2\pi}{\sqrt{3}}}\int_{-2\pi}^{2\pi}\frac{1-\cos\left( \sqrt{3}k_{1}p/2+k_{2}q/2 \right)}{1-\left( 2\cos\left( \sqrt{3}k_{1}/2 \right)\cos\left( k_{2}/2 \right)+\cos\left( k_{2} \right) \right)}.
\end{equation}
Using the relation between $F_{r}^{hex}$ and $F_{r}^{tri}$ and with a simple change of variables we get the expression in equation \ref{solution}.

\chapter{Relation to the theory of discrete analytic functions
\label{apndx:dhf}}
The sandpile patterns we studied are characterized in terms of
discrete analytic functions (DAF) on different discretizations of the
complex plane. For the pattern in Fig. \ref{fig:hexpicl1}, it is the DAF
on a hexagonal lattice, which increases logarithmically at large
distances as in Eq. (\ref{logasymp}).

Studies of DAF started with the work of Kirchhoff on resistor
networks \cite{hu,cserti,Doyle}, and has been studied subsequently by many others
\cite{duffin,mercat}. However, we have not encountered any work on DAF
on many sheeted Riemann surfaces. In the following we present
a way to determine DAF on a square discretization of
Riemann surfaces.

Consider a square grid of points $z=m \epsilon+in\epsilon$, where $m,n$ are
integers and $\epsilon$ is the lattice spacing. Let $f\left(
m\epsilon, n\epsilon \right)$ be a complex function
defined at every site on the grid.
The function $f$ is defined to be discrete analytic
\cite{laszlo} if it satisfies
the discrete Cauchy Riemann condition
\begin{equation}
\frac{f\left(z_{3}\right)-f\left(z_{1}\right)}{z_{3}-z_{1}}=\frac{f\left(z_{4}\right)-f\left(z_{2}\right)}{z_{4}-z_{2}},
\label{eq:dcr}
\end{equation}
at all elementary squares on the grid as shown in Fig. \ref{fig:square}.

In complex analysis, simple examples of analytic functions are
$z^{n}$, and any polynomial of $z^{n}$ is also analytic. For
DAF, it is clear, using the linearity of equation (\ref{eq:dcr}), that
sum of DAF is also discrete analytic. However, not all positive integer powers of
$z$ are discrete analytic. It is easy to check that the functions
$1$, $z$, $z^{2}$, $z^{3}$ are discrete analytic, but $z^{4}$ is not.
We can however construct polynomial functions of $Re(z)$ and
$Im(z)$, that are discrete analytic. Two such examples are
$z^{4}-z\overline{z}\epsilon^{2}$ and $z^{5}-5 z^{2}\overline{z}\epsilon^{2}/2$.
\begin{figure}
\begin{center}
\includegraphics[width=6cm,angle=0]{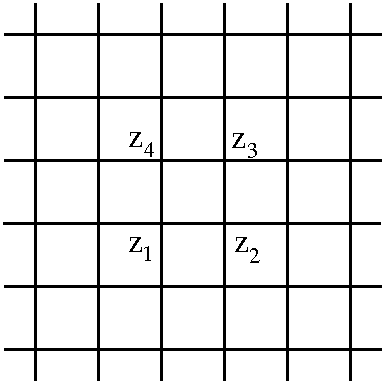}
\caption{\label{fig:square} A square grid on the complex plane.}
\end{center}
\end{figure}

We define a function $F_{n}(z,\epsilon)$ as a homogeneous polynomial in
$z$, $\overline{z}$ and $\epsilon$, of degree $n$, which is a DAF. Using homogeneity, we have
\begin{equation}
F_n(z,\epsilon)=a^n F_n(\frac{z}{a},\frac{\epsilon}{a}),
\end{equation}
and then using $a=\epsilon$, we can express $F_n(z,\epsilon)$ in terms
of $F_n(z,1)$.
This fixes $F_{n}(z,\epsilon)$ up to a multiplicative constant. The
normalization is fixed by requiring that as $\epsilon$ tends to zero,
$F_{n}\left( z, \epsilon \right)\rightarrow z^{n}$.
Then using Cauchy Riemann conditions it is easily seen that
$F_{n}(z, \epsilon)$, for all integers $n\ge0$, has a series expansion in
$\epsilon^{2}$ of the form
\begin{equation}
F_{n}(z,\epsilon)=z^{n}\left[ 1+
\frac{\epsilon^{2}}{z^{2}}g_{1}^{(n)}(\frac{\overline{z}}{z})+\frac{1}{2!}
\frac{\epsilon^{4}}{z^{4}}g_{2}^{(n)}(\frac{\overline{z}}{z})+\cdots
\right],
\end{equation}
where
\begin{eqnarray}
g_{1}^{(n)}(x)&=&-\frac{1}{n-3}\binom{n}{4} x,\\
g_{2}^{(n)}(x)&=&\frac{7!}{(4!)^{2}}\frac{1}{n-6}\binom{n}{7}x^{2}, \\
g_{3}^{(n)}(x)&=&-\frac{10!}{(4!)^{3}}\frac{1}{n-9}\binom{n}{10}x^{3}-\frac{27}{n-7}\binom{n}{8}x,\nonumber\\
\end{eqnarray}
and so on. For an integer $n$, this series will terminate after a
finite number of terms, and all of them can be determined
iteratively.

It is possible to analytically continue the functions for
rational values of $n$. For example,
\begin{equation}
g_{1}^{(n)}(x)=-\frac{\Gamma(n+1)}{4!\Gamma(n-2)}x.
\end{equation}
Then, this analytic continuation of $F_{n}(z,\epsilon)$ provides us the discrete
analytic functions which in the limit $|z|\rightarrow \infty$ grows as
$z^{n}$, for any real positive values of $n$. It is interesting to
note that  the function $D_{m,n}$, used in \cite{myepl} to characterize
the pattern in Fig. \ref{fig:flatone}(b), is equal to
$F_{1/2}(z=m+in,\epsilon=1)$, up to a multiplicative constant. The patterns in the presence of a line
of sinks, or near wedges studied in \cite{myjsp} involve other rational
values. For example, for the pattern near a line sink, one requires
the function $F_{1/3}(z,1)$.

\chapter{Solution of the Eq. (5.18)}\label{apndx:zhang}

Here we discuss the solution of the Eq. \eqref{eq:18} for the starting values 
given by  
\begin{equation}
  \label{eq:A1}
  P_j(t =0)=\alpha^{j-1}.
\end{equation}
We start with an ansatz $P_j(t)=b_t\exp(-a_tj)$, where both $a_t$ and $b_t$ are functions only of 
$t$. This form satisfies the Eq. \eqref{eq:18} for all $j$,
$t>0$, if $a_t$ and $b_t$ satisfy
\begin{eqnarray}
  \label{eq:A2}
  \frac{da_t}{dt}&=& 2-e^{a_t}-e^{-a_t}, \\
  \label{eq:A3}
  \frac{db_t}{dt}&=& b_t (e^{-a_t}-e^{a_t}).
\end{eqnarray}
To solve the Eq. \eqref{eq:A2}, we first make a change of variable $z=e^{-a_t}$. 
In terms of $z$, the equation becomes $dz/dt=(1-z)^2$, which can be easily integrated
to give
\begin{equation}
  \label{eq:A4}
  e^{-a_t}=\frac{t+A-1}{t+A},
\end{equation}
where $A$ is an integration constant. To satisfy the
initial condition in Eq. \eqref{eq:A1}, we choose 
\begin{equation}
  A=(1-\alpha)^{-1}.
\end{equation}
Similarly, to solve the equation for $b_t$, we use the form of $e^{-a_t}$ given in 
Eq. \eqref{eq:A4} and get
\begin{equation}
  \frac{db_t}{dt}= b_t\frac{1-2(t+A)}{(t+A)(t+A-1)}.
\end{equation}
This can be integrated to give
\begin{equation}
  b_t=\frac{B}{(t+A)(t+A-1)},
\end{equation}
\\
where $B$ is an integration constant. 
Then the probability can be written as 
\begin{equation}
  P_j(t)=B\frac{(t+A-1)^{j-1}}{(t+A)^{j+1}},
\end{equation}
where to satisfy the initial condition at $t=0$,
we have used the integration constant $B=(1-\alpha)^{-2}$.
Then, with these values of $A$ and $B$, we have the solution for all $j$, $t>0$, given by
\begin{eqnarray}
  \label{eq:A9}
  P_j(t)&=&\frac{\left[(1 -\alpha) t+\alpha \right]^{j-1}}{\left[ ( 1-\alpha) t+ 1 \right]^{j+1}}\nonumber\\
  &=&\phi_j(\alpha,t),\mathrm{~ ~ ~ ~ ~ say}.
\end{eqnarray}
Now, as $\phi_j(\alpha,t)$ satisfies the Eq. \eqref{eq:18},
\begin{equation}
  \psi_{j,n}(\alpha,t)=\frac{1}{(n-1)!}\frac{\partial^{n-1}\phi_j(t)}{\partial\alpha^{n-1}}
\end{equation}
will also satisfy the equation for any natural number $n$. In addition,
\begin{equation}
  \psi_{j,n}(\alpha=0,t=0)=\delta_{j,n}.
\end{equation}
Hence, we see that the solution of the Eq. \eqref{eq:18}, starting
with $P_j(t)=\delta_{j,n}$ at $t=0$ is
\begin{equation}
P_j(t)=\psi_{j,{n}}(\alpha=0,t)=\frac{1}{(n-1)!}\frac{\partial^{n-1}\phi_j(\alpha,t)}{\partial\alpha^{n-1}}\mid_{\alpha=0},
\label{eq:B12}
\end{equation}
for all $j$, $t>0$, where $\phi_j(\alpha,t)$ is given in Eq. \eqref{eq:A9} and 
$n$ is any natural number.

From \eqref{eq:A9} and \eqref{eq:B12} it easily follows that for large $t$ and $j$, the solution
asymptotically is $P_j(t)\simeq n t^{-2}\exp(-j/t)$.

\chapter{Jordon Block for L=3 SASM}
Here we give some details of the explicit calculation of the steady 
state, and the matrix representation of addition operators for model A 
of length $L=3$. 

The eight sets of eigenvalues  obtained by solving Eq.(\ref{eta})  are 
$(1, 1, 1)$, $(-1, 1, -1)$, $(\frac{1}{3}, -\frac{1}{3},\frac{1}{3})$, 
$(-\frac{1}{3}, -\frac{1}{3}, -\frac{1}{3})$, $ 
(\frac{1}{2}, 0, -\frac{1}{2})$, and $(-\frac{1}{2}, 0, \frac{1}{2})$ 
with the last two sets repeated twice.

For writing the matrix structure of the addition operators, we choose 
the  order of the eigenvectors same  
as the order of the eigenvalues mentioned above. For the degenerate subspace we 
order the eigenvector $|\{a_i\};1\rangle$, defined in 
Eq.(\ref{gen}), before the generalized eigenvector $|\{a_i\};2\rangle$.  Then in 
this basis the matrices corresponding to the addition operators have the 
following Jordan block form. \begin{align}
 \textbf{a}^J_1 = 
  \begin{pmatrix}
    1 ~&~ 0 ~&~ 0 ~&~ 0 ~&~ 0 ~&~ 0 ~&~ 0 ~&~ 0 \\
    0 ~& -1 &~ 0 ~&~ 0 ~&~ 0 ~&~ 0 ~&~ 0 ~&~ 0 \\
    0 ~&~ 0 ~&~ \frac{1}{3} ~&~ 0 ~&~ 0 ~&~ 0 ~&~ 0 ~&~ 0 \\
    0 ~&~ 0 ~&~ 0 ~& -\frac{1}{3} &~ 0 ~&~ 0 ~&~ 0 ~&~ 0 \\
    0 ~&~ 0 ~&~ 0 ~&~ 0 ~&~ \frac{1}{2} ~&~ \frac{1}{2} ~&~ 0 ~&~ 0 \\
    0 ~&~ 0 ~&~ 0 ~&~ 0 ~&~ 0 ~&~ \frac{1}{2} ~&~ 0 ~&~ 0 \\
    0 ~&~ 0 ~&~ 0 ~&~ 0 ~&~ 0 ~&~ 0 ~& -\frac{1}{2} & -\frac{1}{2} \\
    0 ~&~ 0 ~&~ 0 ~&~ 0 ~&~ 0 ~&~ 0 ~&~ 0 ~& -\frac{1}{2} 
 \end{pmatrix},
\end{align}
\begin{align}
 \textbf{a}^J_2 = 
  \begin{pmatrix}
    1 ~&~ 0 ~&~ 0 ~&~ 0 ~&~ 0 ~&~ 0 ~&~ 0 ~&~ 0 \\
    0 ~&~ 1 ~&~ 0 ~&~ 0 ~&~ 0 ~&~ 0 ~&~ 0 ~&~ 0 \\
    0 ~&~ 0 ~& -\frac{1}{3} &~ 0 ~&~ 0 ~&~ 0 ~&~ 0 ~&~ 0 \\
    0 ~&~ 0 ~&~ 0 ~& -\frac{1}{3} &~ 0 ~&~ 0 ~&~ 0 ~&~ 0 \\
    0 ~&~ 0 ~&~ 0 ~&~ 0 ~&~ 0 ~&~ 1 ~&~ 0 ~&~ 0 \\
    0 ~&~ 0 ~&~ 0 ~&~ 0 ~&~ 0 ~&~ 0 ~&~ 0 ~&~ 0 \\
    0 ~&~ 0 ~&~ 0 ~&~ 0 ~&~ 0 ~&~ 0 ~&~ 0 ~&~ 1 \\
    0 ~&~ 0 ~&~ 0 ~&~ 0 ~&~ 0 ~&~ 0 ~&~ 0 ~&~ 0 
 \end{pmatrix},
\end{align}
\begin{align}
 \textbf{a}^J_3 = 
  \begin{pmatrix}
    1 ~&~ 0 ~&~ 0 ~&~ 0 ~&~ 0 ~&~ 0 ~&~ 0 ~&~ 0 \\
    0 ~& -1 &~ 0 ~&~ 0 ~&~ 0 ~&~ 0 ~&~ 0 ~&~ 0 \\
    0 ~&~ 0 ~&~ \frac{1}{3} ~&~ 0 ~&~ 0 ~&~ 0 ~&~ 0 ~&~ 0 \\
    0 ~&~ 0 ~&~ 0 ~& -\frac{1}{3} &~ 0 ~&~ 0 ~&~ 0 ~&~ 0 \\
    0 ~&~ 0 ~&~ 0 ~&~ 0 ~& -\frac{1}{2} & -\frac{1}{2} &~ 0 ~&~ 0 \\
    0 ~&~ 0 ~&~ 0 ~&~ 0 ~&~ 0 ~& -\frac{1}{2} &~ 0 ~&~ 0 \\
    0 ~&~ 0 ~&~ 0 ~&~ 0 ~&~ 0 ~&~ 0 ~&~ \frac{1}{2} ~&~ \frac{1}{2} \\
    0 ~&~ 0 ~&~ 0 ~&~ 0 ~&~ 0 ~&~ 0 ~&~ 0 ~&~ \frac{1}{2} 
  \end{pmatrix},
\end{align}
The transformation matrix $\mathbf{M}$, discussed in section $5$, between
the generalized eigenvector basis and the
configuration basis has the following form
\begin{align}
  \textbf{M} = 
  \begin{pmatrix}
    1 ~&~ 1 ~&~ 1 ~&~ 1 ~&~ 1 ~&~ 1 ~&~ 1 ~&~ 1 \\
    1 ~& -1 ~&~ 1/3 ~& -1/3 ~&~ 1 ~&~ 1/2 ~& -1 ~& -1/2 \\
    1 ~&~ 1 ~& -1/3 ~& -1/3 ~&~ 1 ~&~ 0 ~&~ 1 ~&~ 0 \\
    1 ~& -1 ~& -1/9 ~&~ 1/9 ~&~ 1/2 ~&~ 0 ~& -1/2 ~&~ 0 \\
    1 ~& -1 ~&~ 1/3 ~& -1/3 ~& -1 ~& -1/2 ~&~ 1 ~&~ 1/2 \\
    1 ~&~ 1 ~&~ 1/9 ~&~ 1/9 ~& -3/4 ~& -1/4 ~& -3/4 ~& -1/4 \\
    1 ~& -1 ~& -1/9 ~&~ 1/9 ~& -1/2 ~&~ 0 ~&~ 1/2 ~&~ 0 \\
    1 ~&~ 1 ~& -1/27 ~& -1/27 ~& -1/4 ~&~ 0 ~& -1/4 ~&~ 0 
  \end{pmatrix},
\end{align}
where the configuration basis vectors are chosen in the following order
$(0, 0, 0)$, $(1, 0, 0)$, $(0, 1, 0)$, $(1, 1, 0)$, $(0, 0, 1)$,
$(1, 0, 1)$, $(0, 1, 1)$, and $(1, 1, 1)$. The matrix is non-singular,
and the inverse can be calculated numerically.
Using the similarity transformation $\mathbf{M}\mathbf{a}_1^J\mathbf{M}^{-1}$ we find matrix
representation of the addition operator $\mathbf{a}_1$ in the configuration
basis.
\begin{align}
  \textbf{a}_1 =
  \begin{pmatrix}
    0 ~&~ \frac{2}{7} ~&~ 0 ~&~ \frac{4}{49} ~&~ 0 ~&~ 0 ~&~ 0 ~&~ 0 \\
    1 ~&~ 0 ~&~ 0 ~&~ 0 ~&~ 0 ~&~ \frac{1}{24} ~&~ 0 ~&~ \frac{1}{9} \\
    0 ~&~ \frac{4}{7} ~&~ 0 ~&~ \frac{22}{49} ~&~ 0 ~&~ 0 ~&~ 0 ~&~ 0 \\
    0 ~&~ 0 ~&~ 1 ~&~ 0 ~&~ 0 ~&~ \frac{1}{12} ~&~ 0 ~&~ \frac{19}{72} \\
    0 ~&~ 0 ~&~ 0 ~&~ 0 ~&~ 0 ~&~ \frac{7}{24} ~&~ 0 ~&~ \frac{1}{9} \\
    0 ~&~ \frac{1}{7} ~&~ 0 ~&~ \frac{16}{49} ~&~ 1 ~&~ 0 ~&~ 0 ~&~ 0 \\
    0 ~&~ 0 ~&~ 0 ~&~ 0 ~&~ 0 ~&~ \frac{7}{12} ~&~ 0 ~&~ \frac{37}{72} \\
    0 ~&~ 0 ~&~ 0 ~&~ \frac{1}{7} ~&~ 0 ~&~ 0 ~&~ 1 ~&~ 0 
  \end{pmatrix},
\end{align}
The other operators can also be determined similarly.

\bibliographystyle{amsalpha}	
\bibliography{mybib}{}
\end{document}